\documentclass[10pt,a5paper]{ETHthesis}

\usepackage[ngerman,english,british]{babel}
\usepackage{E}
\usepackage{ETHthesis}
\usepackage{makeidx}
\usepackage{enumerate}
\usepackage{graphicx}
\usepackage{subfig}
\usepackage{amssymb,amsmath,latexsym}
\usepackage{amsthm}
\usepackage[plainpages=false]{hyperref}
\usepackage{palatino}
\usepackage{pstricks-add}
\usepackage{pst-plot}
\usepackage{pst-xkey}
\usepackage{multirow}
\usepackage{slashbox}
\usepackage{dsfont}
\usepackage{pifont}
\usepackage{bm}
\usepackage{textcomp}
\usepackage{microtype}
\usepackage{url}
\usepackage{units}

\allowdisplaybreaks[2]

\theoremstyle{plain}
\newtheorem{theorem}{Theorem}
\newtheorem{corollary}{Corollary}
\newtheorem{lemma}{Lemma}

\theoremstyle{definition}
\newtheorem{definition}{Definition}
\newtheorem{example}{Example}
\newtheorem{protocol}{Protocol}
\newtheorem{condition}{Condition}

\numberwithin{equation}{chapter}
\numberwithin{theorem}{chapter}
\numberwithin{lemma}{chapter}
\numberwithin{definition}{chapter}
\numberwithin{corollary}{chapter}

\newcommand{\ep}{\varepsilon}
\newcommand{\bof}[1]{\bm{#1}}
\newcommand{\st}{\operatorname{s.t.}}

\newcommand{\tr}{\operatorname{tr}}

\DeclareMathOperator*{\expect}{\mathbb{E}}

\newcommand \brakketsmall[3]{\mathinner{\langle{#1}|{#2}|{#3}\rangle}}

\newcommand{\ket}[1]{\ensuremath{\left| #1 \right>}} % for Dirac bras
\newcommand{\bra}[1]{\ensuremath{\left< #1 \right|}} % for Dirac kets
\newcommand{\braket}[2]{\mathinner{\left< #1 \vphantom{#2}  \middle|  #2 \vphantom{#1} \right>}} % for Dirac brackets
\newcommand{\brakket}[3]{\mathinner{\left< #1 \vphantom{#2}\vphantom{#3} \right|
 \vphantom{#1} #2\vphantom{#3} \left|
 \vphantom{#1}\vphantom{#2}#3 \right>}} % for Dirac brackets
 \newcommand{\smallbrakket}[3]{\mathinner{\bigl< #1 \vphantom{#2}\vphantom{#3} \bigr|
 \vphantom{#1} #2\vphantom{#3} \bigl|
 \vphantom{#1}\vphantom{#2}#3 \bigr>}} % for Dirac brackets

 \newcommand{\bigbrakket}[3]{\mathinner{\Bigl< #1 \vphantom{#2}\vphantom{#3} \Bigr|
 \vphantom{#1} #2\vphantom{#3} \Bigl|
 \vphantom{#1}\vphantom{#2}#3 \Bigr>}} % for Dirac brackets

 \newcommand{\bigerbrakket}[3]{\mathinner{\biggl< #1 \vphantom{#2}\vphantom{#3} \biggr|
 \vphantom{#1} #2\vphantom{#3} \biggl|
 \vphantom{#1}\vphantom{#2}#3 \biggr>}} % for Dirac brackets

\initETHthesis

\makeindex

\begin{document}

\dissnum{19226}
\title{Device-Independent\\ Quantum Key Distribution}
\degree{Doctor of Sciences}
\author{Esther H\"anggi}
\acatitle{MSc in Physics, EPFL}
\dateofbirth{February 11, 1982}
\citizen{Nunningen, SO, Switzerland} 
\examiner{Prof. Dr. Stefan Wolf} 
\coexaminera{Prof. Dr. Artur Ekert}
\coexaminerb{Prof. Dr. Renato Renner}

\maketitle

\thispagestyle{empty}
\cleardoublepage

\pagestyle{plain} 

\chapter*{Acknowledgments}

This thesis would never have been possible without the support of many people. 

First of all I would like to thank Stefan Wolf who has been a great advisor, giving us a lot of freedom in our research. Still, his door was always open for discussions.  
I would like to thank Renato Renner for an intense and fruitful collaboration which was the basis of a large part of this thesis, and 
also for
being my co-examiner.  
I am grateful to Artur Ekert for being co-examiner and investing his time in reviewing this 
thesis. 

My studies at ETH would not have been the same without the present and past members of the 
Quantum Information Group, Daniel Burgarth, Roger Colbeck, Dejan Dukaric, Matthias Fitzi, 
Manuel Forster, Viktor {Galliard}, Melanie Raemy, Severin Winkler, Stefan Wolf, and J\"urg Wullschleger. I thank them all for countless research discussions, as well as enjoyable lunch and coffee breaks, and beers at bqm. Special 
thanks go to Viktor Galliard, who was my office mate during most of my time at ETH. 

During the last years, we had a close collaboration with the Quantum Information Theory Groups from 
the physics department. Thank you for 
the joint research days and seminars, as well as group hikes, to 
Normand Beaudry, Mario Berta, Matthias Christandl, Oscar Dahlsten, Fr\'ed\'eric \linebreak[4] Dupuis, David Gross, Stefan Hengl, Renato Renner, L\'idia del Rio, {Christian} Schilling, Cyril Stark, Marco Tomamichel, and Johan {\AA}berg. 

I also thank the members of the Information Security and Cryptography Research Group of Ueli Maurer, whose offices were close to ours during 
 most of the last few years. These include Divesh Aggarwal, Dominik Raub, and 
 Vassilis Zikas. Special thanks go to Stefano Tessaro for answering many questions related to (classical) cryptography.

I am grateful to the members of the Complexity and Algorithms Group, that is, Chandan Dubey, Thomas Holenstein, and Robin 
K\"unzler who often joined us for lunch and coffee breaks during which we had many interesting discussions about 
research.

I have learnt a lot from my co-authors and collaborators. 
It was (and still is) a great pleasure to work with
Gilles Brassard, 
Anne Broadbent,
Roger Colbeck, 
Sandro Coretti, 
Matthias Fitzi, 
Andr\'e Allan M\'ethot, 
Renato Renner, 
Valerio Scarani, 
Alain Tapp, 
Stefano Tessaro, 
Stefan Wolf, and
J\"urg Wullschleger.

Many thanks to Matthias Fitzi, Cyril Stark, and Severin Winkler for reading preliminary versions of this thesis and 
for their valuable comments. Special thanks go 
to Matthias Fitzi for Alice, Bob and Eve. I am grateful to Beate Bernhard for taking such good care of all 
administrative things. 

Finally, I would like to thank my family, J\"urg, Margrit and Silja for their continuous  support. Last but not least, I want to thank Andr\'e for everything. 

This research was partially supported by the Swiss National Science \linebreak[4] Foun\-da\-tion (SNF) and the ETH research commission.

\chapter*{Abstract}

\emph{Quantum key distribution} allows two parties connected by a quantum \linebreak[4] channel to establish a secret key that is unknown to 
any unauthorized third party. The secrecy of this key is based on the laws of quantum physics. For security, however, it is 
crucial that the honest parties are able to control their physical devices accurately and completely. The goal of 
\emph{device-independent} quantum key distribution is to remove this requirement and base security only on the (observable) 
behaviour of the devices, i.e., the probabilities of the measurement results given the choice of measurement. 

In this thesis, we study two approaches to achieve device-independent quantum key distribution: in the first approach, the adversary can distribute any system to the honest parties that cannot be used to communicate between the three of them, i.e., it must be \emph{non-signalling}. This constraint is strictly weaker than the ones imposed by quantum physics, i.e., the adversary is strictly stronger. Security can then be concluded \emph{only} based on the observed correlations. In the second approach, we limit the adversary to strategies which can be implemented using quantum physics. More precisely, we demand that the behaviour of the system shared between the honest parties and the adversary can be obtained by measuring \emph{some} kind of entangled quantum state. Security is then based on the laws of quantum physics, but it does not rely on the exact details of the physical systems and devices used to create the observed correlations. In particular, it is independent of the dimension of the Hilbert space describing them. 

For both approaches, we show how device-independent quantum key distribution can be achieved when imposing an additional 
condition. In the non-signalling case this additional requirement is that communication by means of the quantum system  is 
impossible between \emph{all} subsystems, while, in the quantum case, we demand that measurements on different subsystems must commute. 
We give a generic security proof for device-independent quantum key distribution in these cases and apply it to an 
 explicit quantum key distribution protocol, thus proving its security. 
We also show that, \emph{without} any additional such requirement there exist means of non-signalling adversaries to attack several systems \emph{jointly}. Some extra constraints are, hence, necessary for efficient device-independent secrecy. 

\selectlanguage{ngerman}

\chapter*{Zusammenfassung}

\emph{Quan\-ten-Schl\"us\-sel\-ver\-tei\-lung} erlaubt zwei durch einen Quantenkanal ver\-bun\-de\-nen Parteien einen 
Schl\"us\-sel zu erzeugen, der vor jeder unberechtigten Drittpartei geheim ist. Die Sicherheit dieses 
Schl\"us\-sels basiert auf den Gesetzen der Quantenphysik. Sie kann aber nur garantiert wer\-den, wenn die 
ehrlichen Parteien die physikalischen Apparate genau und voll\-st{\"a}n\-dig kon\-trol\-lie\-ren k\"on\-nen. Das Ziel 
\emph{ge\-r{\"a}\-te\-un\-ab\-h{\"a}n\-gi\-ger} Quan\-ten-\linebreak[4]Schl\"us\-sel\-ver\-tei\-lung ist, diese Be\-din\-gung zu 
lockern, und die Sicherheit nur auf das (testbare) Verhalten der Apparate zu basieren, genauer gesagt, auf die 
Wahrscheinlichkeiten von Messresultaten, gegeben die Wahl einer bestimmten Messung. 

In dieser Arbeit betrachten wir zwei m\"og\-liche Vorgehensweisen um 
ge\-r\"a\-te\-un\-ab\-h\"an\-gige Quan\-ten-Schl\"us\-sel\-ver\-tei\-lung zu erreichen: in der ersten kann der Gegner den ehrlichen Parteien jede beliebige Art von Systemen zukommen lassen, die nicht zur Kommunikation verwendet werden kann. Diese Bedingung ist strikte schw\"acher als diejenigen, die durch die Quantenphysik vorgegeben sind, der tolerierte Gegner ist also 
 st\"ar\-ke\-rer. Sicher\-heit wird in dieses Fall \emph{nur} von den beobachteten Korrelationen hergeleitet. In der zweiten Vorgehensweise be\-schr\"an\-ken wir die m\"og\-lichen Strategien des Gegners auf solche, die durch Quantensysteme implementiert werden k\"on\-nen. Genauer gesagt verlangen wir, dass das System der ehrlichen Parteien und des Gegners durch das Messen eines ver\-schr\"ank\-ten Quantenzustandes erzeugt werden kann. Sicherheit beruht in diesem Fall auf den Gesetzen der Quantenphysik, ist aber un\-ab\-h\"an\-gig von den Details der physikalischen Systemen und der Apparate, mit Hilfe derer die Korrelationen zustande kamen. Insbesondere ist die Dimension des Hilbertraumes, der die Systeme beschreibt, beliebig. 

F\"ur beide Vorgehensweisen zeigen wir, wie ge\-r\"a\-te\-un\-ab\-h\"an\-gi\-ge Quan\-ten-Schl\"us\-sel\-ver\-tei\-lung erreicht werden 
kann, falls noch eine weitere Bedingung eingehalten wird: f\"ur den Fall, wo die Systeme 
nicht zur Kommunikation gebraucht werden k\"on\-nen, entspricht diese der Vorgabe, 
dass Kommunikation auch zwischen Teilsystemen un\-m\"og\-lich ist; w\"ah\-rend im 
quantenmechanischen Fall Messungen auf verschiedenen Teilsystemen kommutieren m\"us\-sen. 
Wir geben in beiden F\"al\-len einen allgemeinen Si\-cher\-heits\-be\-weis f{\"u}r ge\-r{\"a}\-te\-un\-ab\-h{\"a}n\-gi\-ge Quan\-ten-Schl{\"u}s\-sel\-ver\-tei\-lung \linebreak[4] und wenden diesen auf ein konkretes Protokoll an, 
von dem wir zeigen, dass es auch unter diesen schwachen Annahmen sicher ist. Wir 
zeigen weiter, dass \emph{ohne} eine solche zu\-s\"atz\-liche Bedingung gute 
Strategien existieren, mit denen ein Gegner, der nur durch die Un\-m\"og\-lich\-keit 
von Kommunikation bes\-chr\"ankt ist, mehrere Systeme \emph{gemeinsam} attackieren 
kann. Weitere Ein\-schr\"an\-kungen sind deshalb im Allgemeinen notwendig f\"ur 
effiziente ge\-r\"a\-te\-un\-ab\-h\"an\-gi\-ge Sicherheit.

\selectlanguage{british}

\tableofcontents
\cleardoublepage

\pagenumbering{arabic}
\pagestyle{headings}

\chapter{Introduction}

\section{Quantum Key Distribution}

Key agreement is a protocol among two parties, Alice and Bob, to
produce local strings such that, ideally, both strings are equal and
no adversary can get any information about this string by
eavesdropping the protocol. This task can only be realized based on
certain assumptions (such as assuming that computing power~\cite{dh,rsa}  or mem\-o\-ry \linebreak[4] \cite{maurer}
of the adversary are bounded) or the availability of resources (such
as noisy channels~\cite{csikor}). Wiesner~\cite{wiesner} 
observed that a quantum channel can serve as such a resource in context of various cryptographic tasks. 
The reason is that a quantum channel obeys the uncertainty principle of quantum mechanics, which states that there exist certain properties of quantum mechanical systems that cannot be known (exactly) simultaneously and that measuring one of them necessarily disturbs the other. Wiesner~\cite{wiesner} proposes a scheme for sending two messages `either but not both of which may be received'\footnote{It later turned out that, unfortunately, the laws of quantum mechanics alone are not enough to achieve this functionality, called \emph{oblivious transfer}~\cite{Lo98insecurityof,mayers,lochau}.} and a way of making `money that is physically impossible to counterfeit'. The idea of basing security on the laws of quantum physics was further developed and combined with ideas from public-key cryptography by Bennett, \linebreak[4] Brassard, Breidbart, and Wiesner~\cite{bbbw} and finally made into a key-dis\-tri\-bu\-tion scheme by Bennett and Brassard~\cite{bb84}. 

Roughly, the BB84 key-distribution scheme~\cite{bb84} works as follows (see Figure~\ref{fig:bb84}): Alice and Bob are connected by an (insecure) quantum channel and a public but authenticated classical channel.\footnote{An authenticated channel can be built from an insecure classical channel using a short key~\cite{stinson,GN}. To account for the need of this initial key, quantum key distribution is sometimes called \emph{key expansion}.} Alice encodes a bit by sending a photon that is polarized in the direction of either of two basis states. However, she chooses not only the value of the encoded bit at random, but also the encoding is done in either
 the \emph{horizontal} or \emph{diagonal} basis.\footnote{Instead of photons, Alice could also use another two-level quantum system and for the encoding another set of two \emph{mutually unbiased bases}, i.e., two bases where measuring a basis-state of one basis in the other basis gives a random outcome.}
Bob receives the photon, chooses one of the two bases at random and measures the polarization of the photon in this basis. They repeat this process several times, each time taking note of the basis and the encoded bit, or measurement result, respectively. 
Later, Alice uses the classical authenticated channel to tell Bob which basis she used to encode the bit.  If Bob measured in the `wrong' (i.e., other) basis, he obtained a random bit uncorrelated with what Alice sent. They discard exactly these bits. Wherever Bob measured in the same basis Alice used for the encoding, he should have received exactly the bit Alice had sent. Alice and Bob randomly select some of the bits and check this. If Bob received the correct bits, they use the remaining bits as a key. 

Why is this secure? Assume that Eve intercepts the quantum channel between Alice and Bob and measures the photon. Since she does not know the basis in which the bit was encoded, with probability $1/2$, she measures in the wrong basis, in which case Bob's bit will be random even when he measures in the same basis Alice used for the encoding. These `errors' introduced by an eavesdropper will (with high probability) be noticed by Alice and Bob when they check their results and they will abort the protocol. 

\begin{figure}[h]
\centering
\pspicture*[](-5.4,-4.25)(5.4,2.75)
\rput[c]{0}(-4.75,0.75){\includegraphics[width=1.25cm]{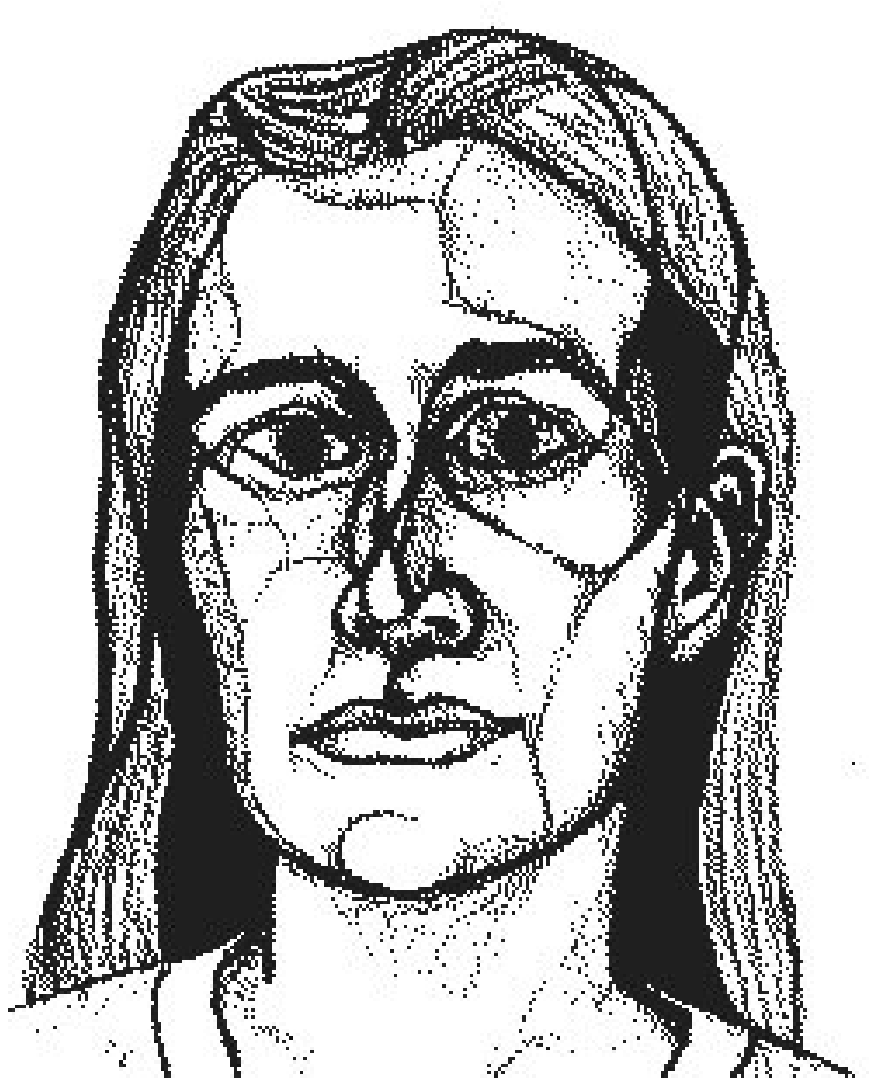}}
\rput[c]{0}(4.75,0.75){\includegraphics[width=1.25cm]{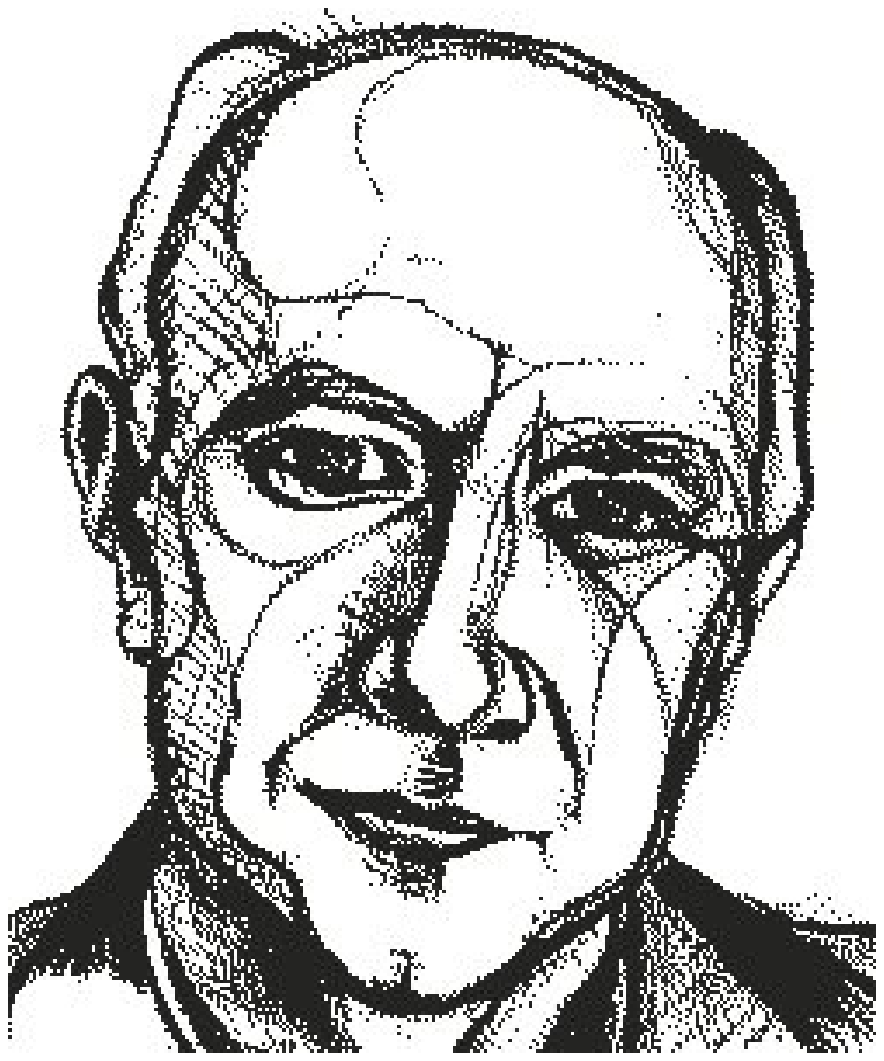}}
\rput[c]{0}(0,2){\includegraphics[width=1.4cm]{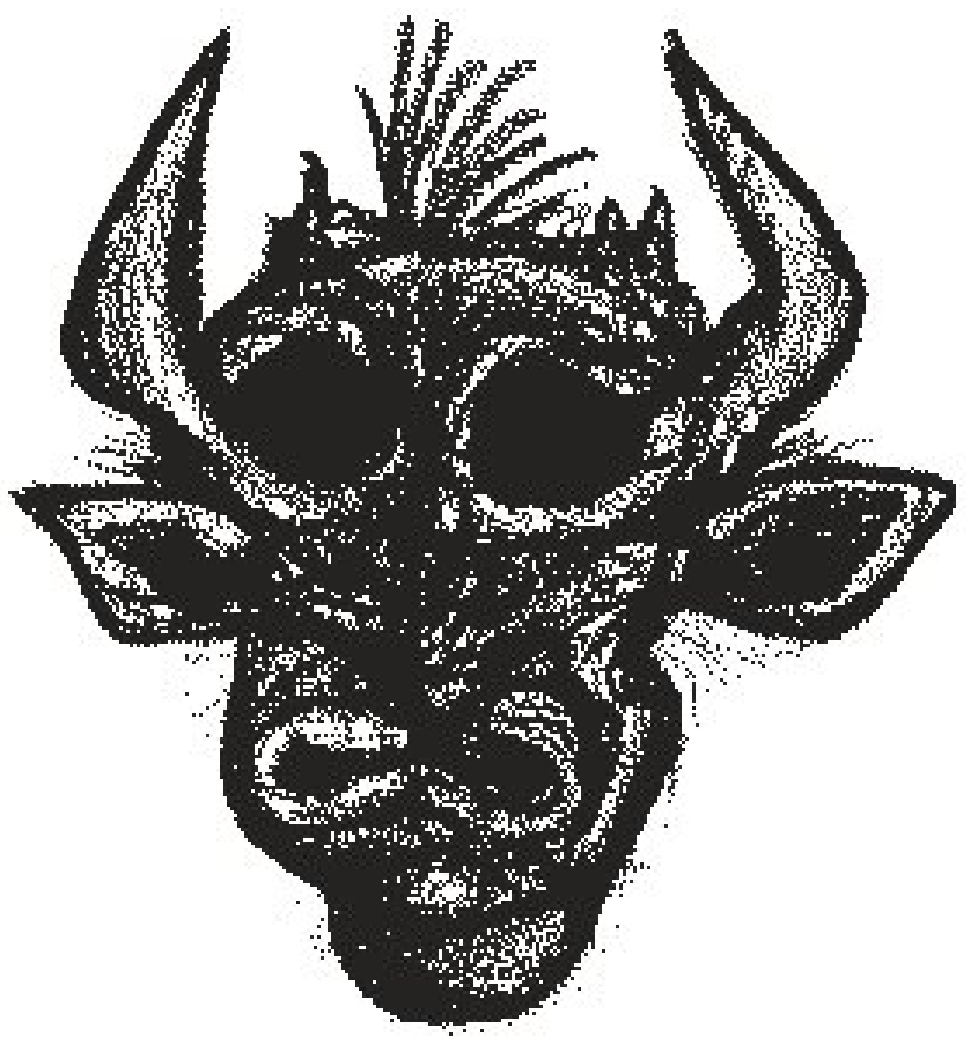}}
\pspolygon[linewidth=2pt](-4,0)(-2.5,0)(-2.5,1.5)(-4,1.5)
\pspolygon[linewidth=2pt](4,0)(2.5,0)(2.5,1.5)(4,1.5)
\psline[linewidth=2pt,linecolor=gray]{->}(-2.5,0.75)(-0.5,0.75)
\psline[linewidth=2pt,linecolor=gray]{->}(0.5,0.75)(2.5,0.75)
\rput[c]{0}(0,0.75){\color{gray}{\Huge{$\gamma$}}}
\rput[c]{0}(-0.25,0){
\rput[c]{0}(-4.4,-0.75){\Huge{$\{$}}
\rput[c]{0}(-4,-1){
\psline[linewidth=0.5pt]{->}(0,0)(0,0.5)
\psline[linewidth=0.5pt]{->}(0,0)(0.5,0)
\rput[c]{0}(0.7,0){{$0$}}
\rput[c]{0}(0,0.7){{$1$}}
}
\rput[c]{0}(-3,-1){\Huge{$,$}}
\rput[c]{0}(-1.6,-0.75){\Huge{$\}$}}
\rput[c]{45}(-2.353553391,-1){
\psline[linewidth=0.5pt]{->}(0,0)(0,0.5)
\psline[linewidth=0.5pt]{->}(0,0)(0.5,0)
\rput[c]{-45}(0.7,0){{$0$}}
\rput[c]{-45}(0,0.7){{$1$}}
}
}
\rput[c]{0}(0.25,0){
\rput[c]{0}(1.6,-0.75){\Huge{$\{$}}
\rput[c]{0}(2,-1){
\psline[linewidth=0.5pt]{->}(0,0)(0,0.5)
\psline[linewidth=0.5pt]{->}(0,0)(0.5,0)
}
\rput[c]{0}(3,-1){\Huge{$,$}}
\rput[c]{0}(4.4,-0.75){\Huge{$\}$}}
\rput[c]{45}(3.646446609,-1){
\psline[linewidth=0.5pt]{->}(0,0)(0,0.5)
\psline[linewidth=0.5pt]{->}(0,0)(0.5,0)
}
}
\rput[c]{0}(-4,-1.75){
\rput[c]{0}(-0.6,-0.25){basis}
\rput[c]{0}(-0.6,-0.75){bit}
\rput[c]{0}(-0.6,-1.25){useful}
\psline[linewidth=0.5pt]{-}(0,0)(3,0)
\psline[linewidth=0.5pt]{-}(0,-0.5)(3,-0.5)
\psline[linewidth=0.5pt]{-}(0,-1)(3,-1)
\psline[linewidth=0.5pt]{-}(0,-1.5)(3,-1.5)
\psline[linewidth=0.5pt]{-}(0,0)(0,-1.5)
\psline[linewidth=0.5pt]{-}(0.5,0)(0.5,-1.5)
\psline[linewidth=0.5pt]{-}(1,0)(1,-1.5)
\psline[linewidth=0.5pt]{-}(1.5,0)(1.5,-1.5)
\psline[linewidth=0.5pt]{-}(2,0)(2,-1.5)
\psline[linewidth=0.5pt]{-}(2.5,0)(2.5,-1.5)
\psline[linewidth=0.5pt]{-}(3,0)(3,-1.5)
\rput[c]{0}(0.25,-0.25){$+$}
\rput[c]{0}(0.75,-0.25){$+$}
\rput[c]{0}(1.25,-0.25){$\times$}
\rput[c]{0}(1.75,-0.25){$+$}
\rput[c]{0}(2.25,-0.25){$\times$}
\rput[c]{0}(2.75,-0.25){$\times$}
\rput[c]{0}(0.25,-0.75){$1$}
\rput[c]{0}(0.75,-0.75){$0$}
\rput[c]{0}(1.25,-0.75){$0$}
\rput[c]{0}(1.75,-0.75){$0$}
\rput[c]{0}(2.25,-0.75){$1$}
\rput[c]{0}(2.75,-0.75){$1$}
\rput[c]{0}(0.25,-1.25){\ding{55}}
\rput[c]{0}(0.75,-1.25){\ding{51}}
\rput[c]{0}(1.25,-1.25){\ding{51}}
\rput[c]{0}(1.75,-1.25){\ding{51}}
\rput[c]{0}(2.25,-1.25){\ding{55}}
\rput[c]{0}(2.75,-1.25){\ding{51}}
}
\rput[c]{0}(2,-1.75){
\rput[c]{0}(-0.6,-0.25){basis}
\rput[c]{0}(-0.6,-0.75){result}
\rput[c]{0}(-0.6,-1.25){useful}
\psline[linewidth=0.5pt]{-}(0,0)(3,0)
\psline[linewidth=0.5pt]{-}(0,-0.5)(3,-0.5)
\psline[linewidth=0.5pt]{-}(0,-1)(3,-1)
\psline[linewidth=0.5pt]{-}(0,-1.5)(3,-1.5)
\psline[linewidth=0.5pt]{-}(0,0)(0,-1.5)
\psline[linewidth=0.5pt]{-}(0.5,0)(0.5,-1.5)
\psline[linewidth=0.5pt]{-}(1,0)(1,-1.5)
\psline[linewidth=0.5pt]{-}(1.5,0)(1.5,-1.5)
\psline[linewidth=0.5pt]{-}(2,0)(2,-1.5)
\psline[linewidth=0.5pt]{-}(2.5,0)(2.5,-1.5)
\psline[linewidth=0.5pt]{-}(3,0)(3,-1.5)
\rput[c]{0}(0.25,-0.25){$\times$}
\rput[c]{0}(0.75,-0.25){$+$}
\rput[c]{0}(1.25,-0.25){$\times$}
\rput[c]{0}(1.75,-0.25){$+$}
\rput[c]{0}(2.25,-0.25){$+$}
\rput[c]{0}(2.75,-0.25){$\times$}
\rput[c]{0}(0.25,-0.75){$0$}
\rput[c]{0}(0.75,-0.75){$0$}
\rput[c]{0}(1.25,-0.75){$1$}
\rput[c]{0}(1.75,-0.75){$0$}
\rput[c]{0}(2.25,-0.75){$1$}
\rput[c]{0}(2.75,-0.75){$1$}
\rput[c]{0}(0.25,-1.25){\ding{55}}
\rput[c]{0}(0.75,-1.25){\ding{51}}
\rput[c]{0}(1.25,-1.25){\ding{51}}
\rput[c]{0}(1.75,-1.25){\ding{51}}
\rput[c]{0}(2.25,-1.25){\ding{55}}
\rput[c]{0}(2.75,-1.25){\ding{51}}
}
\psline[linewidth=0.5pt]{-}(-2.75,-3.75)(3.25,-3.75)
\psline[linewidth=0.5pt]{->}(-2.75,-3.75)(-2.75,-3.25)
\psline[linewidth=0.5pt]{->}(3.25,-3.75)(3.25,-3.25)
\rput[c]{0}(0.25,-4){same basis, different bit $\rightarrow$ eavesdropper}
\endpspicture
\caption{The~\cite{bb84} quantum key-distribution protocol.}
\label{fig:bb84}
\end{figure}

Of course, however, Eve does not need to measure the photon going through the quantum channel, but she can do a more sophisticated attack. For example, she can \emph{entangle} a system with the photon, store it, and delay her measurement until after Alice and Bob have revealed the basis used for the encoding. Indeed, it took several years until it was shown that the scheme remains secure in this case and a full security proof against the most general attacks was made~\cite{secmayers,seclochau,bbbmr,shorpreskill,ilm}. A further difficulty is that a physical implementation of the protocol will never be perfect and always contains noise. It is, therefore, necessary to allow for noisy channels and unreliable detectors in order to establish a key. It was also realized that the possibility of Eve delaying her measurement until the key is actually used in an application could pose a serious problem. The definition of secrecy of a key needs to be made carefully to hold in this situation~\cite{AKMR07}. Meanwhile, these issues have been considered in the security 
proofs and 
it can be shown that quantum key distribution remains secure despite of them~\cite{rennerphd}.

Some years after Bennett and Brassard, Ekert~\cite{ekert} proposed a quantum key distribution protocol the security of which is based on a different property of quantum physics: the monogamy of entanglement. In fact, two quantum systems which are strongly entangled (correlated) can at most be weakly entangled with a third system~\cite{terhal}. 
The idea of Ekert's protocol is the following (see Figure~\ref{fig:ekert}): Alice prepares two photons\footnote{The original protocol~\cite{ekert} uses spin-$(1/2)$ particles. For simplicity, we stick with the formulation in terms of photons.} in an entangled quantum state, more precisely the singlet state, i.e.,  $\ket{\Psi^-} =(\ket{01}-\ket{10})/\sqrt{2}$. She sends one of the two particles to Bob. They both measure their particle in a basis chosen at random. Alice chooses from a basis rotated by an angle of $0$, $\pi/8$, or $\pi/4$, Bob chooses from 
$\pi/8$, $\pi/4$, and $3\pi/8$. She takes her outcome to be the measurement result, while he outputs the opposite of the measurement result. 
Note that the probability that Alice and Bob obtain the same outcome is 
$\cos^2\alpha$ 
when they measure in bases mutually enclosing an angle of $\alpha$, i.e., for 
$0$ 
they obtain perfectly correlated outcomes and for 
$\pi/2$ they obtain
perfectly anti-correlated outcomes. 
After all measurements are completed, they compare their results over the public authenticated channel and estimate the correlation of their outcomes given they measured in each possible combination of bases. They add the correlations for the bases pairs 
$(0,\pi/8)$, $(\pi/4, \pi/8)$ and $(\pi/4, 3\pi/8)$ (where the first angle is associated with Alice and the second with Bob) and subtract the correlation of the bases $(0, 3\pi/8)$.   
When this value is $2\sqrt{2}$, they continue, otherwise they abort.\footnote{Note that this test corresponds to testing the value of the Bell inequality given in Section~\ref{subsec:localsystem}.} If they do not abort, they take exactly those outcomes as key where they measured in the same direction. 
The security of the protocol is based on the fact that a value of $2\sqrt{2}$ can only be reached by the singlet state and because this state is \emph{pure}, the eavesdropper's system cannot be correlated with it. 

\begin{figure}[h]
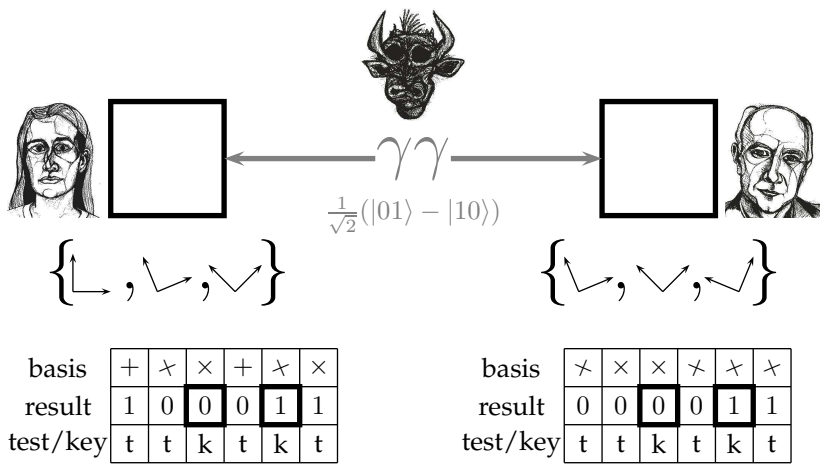

\centering
\pspicture*[](-5.4,-4)(5.4,2.75)
\rput[c]{0}(-4.75,0.75){\includegraphics[width=1.25cm]{alice_small_dither.eps}}
\rput[c]{0}(4.75,0.75){\includegraphics[width=1.25cm]{bob_small_dither.eps}}
\rput[c]{0}(0,2){\includegraphics[width=1.4cm]{eve_small_dither.eps}}
\pspolygon[linewidth=2pt](-4,0)(-2.5,0)(-2.5,1.5)(-4,1.5)
\pspolygon[linewidth=2pt](4,0)(2.5,0)(2.5,1.5)(4,1.5)
\psline[linewidth=2pt,linecolor=gray]{<-}(-2.5,0.75)(-0.5,0.75)
\psline[linewidth=2pt,linecolor=gray]{->}(0.5,0.75)(2.5,0.75)
\rput[c]{0}(0,0){\color{gray}{{$\frac{1}{\sqrt{2}}(\ket{01}-\ket{10})$}}}
\rput[c]{0}(0,0.75){\color{gray}{\Huge{$\gamma \gamma$}}}
\rput[c]{0}(-0.25,0){
\rput[c]{0}(-4.4,-0.75){\Huge{$\{$}}
\rput[c]{0}(-4.25,-1){
\psline[linewidth=0.5pt]{->}(0,0)(0,0.5)
\psline[linewidth=0.5pt]{->}(0,0)(0.5,0)
}
\rput[c]{22.5}(-3.132683432,-1){
\psline[linewidth=0.5pt]{->}(0,0)(0,0.5)
\psline[linewidth=0.5pt]{->}(0,0)(0.5,0)
}
\rput[c]{0}(-3.5,-1){\Huge{$,$}}
\rput[c]{0}(-2.5,-1){\Huge{$,$}}
\rput[c]{0}(-1.6,-0.75){\Huge{$\}$}}
\rput[c]{45}(-2.103553391,-1){
\psline[linewidth=0.5pt]{->}(0,0)(0,0.5)
\psline[linewidth=0.5pt]{->}(0,0)(0.5,0)
}
}
\rput[c]{0}(0.25,0){
\rput[c]{0}(1.6,-0.75){\Huge{$\{$}}
\rput[c]{67.5}(4.061939766,-1){
\psline[linewidth=0.5pt]{->}(0,0)(0,0.5)
\psline[linewidth=0.5pt]{->}(0,0)(0.5,0)
}
\rput[c]{22.5}(1.941341716,-1){
\psline[linewidth=0.5pt]{->}(0,0)(0,0.5)
\psline[linewidth=0.5pt]{->}(0,0)(0.5,0)
}
\rput[c]{0}(3.5,-1){\Huge{$,$}}
\rput[c]{0}(2.5,-1){\Huge{$,$}}
\rput[c]{0}(4.4,-0.75){\Huge{$\}$}}
\rput[c]{45}(3.053553391,-1){
\psline[linewidth=0.5pt]{->}(0,0)(0,0.5)
\psline[linewidth=0.5pt]{->}(0,0)(0.5,0)
}
}
\rput[c]{0}(-4,-1.75){
\rput[c]{0}(-0.7,-0.25){basis}
\rput[c]{0}(-0.7,-0.75){result}
\rput[c]{0}(-0.7,-1.25){test/key}
\psline[linewidth=0.5pt]{-}(0,0)(3,0)
\psline[linewidth=0.5pt]{-}(0,-0.5)(3,-0.5)
\psline[linewidth=0.5pt]{-}(0,-1)(3,-1)
\psline[linewidth=0.5pt]{-}(0,-1.5)(3,-1.5)
\psline[linewidth=0.5pt]{-}(0,0)(0,-1.5)
\psline[linewidth=0.5pt]{-}(0.5,0)(0.5,-1.5)
\psline[linewidth=0.5pt]{-}(1,0)(1,-1.5)
\psline[linewidth=0.5pt]{-}(1.5,0)(1.5,-1.5)
\psline[linewidth=0.5pt]{-}(2,0)(2,-1.5)
\psline[linewidth=0.5pt]{-}(2.5,0)(2.5,-1.5)
\psline[linewidth=0.5pt]{-}(3,0)(3,-1.5)
\rput[c]{0}(0.25,-0.25){$+$}
\rput[c]{22.5}(0.75,-0.25){$+$}
\rput[c]{0}(1.25,-0.25){$\times$}
\rput[c]{0}(1.75,-0.25){$+$}
\rput[c]{22.5}(2.25,-0.25){$+$}
\rput[c]{0}(2.75,-0.25){$\times$}
\rput[c]{0}(0.25,-0.75){$1$}
\rput[c]{0}(0.75,-0.75){$0$}
\rput[c]{0}(1.25,-0.75){$0$}
\rput[c]{0}(1.75,-0.75){$0$}
\rput[c]{0}(2.25,-0.75){$1$}
\rput[c]{0}(2.75,-0.75){$1$}
\rput[c]{0}(0.25,-1.25){t}
\rput[c]{0}(0.75,-1.25){t}
\rput[c]{0}(1.25,-1.25){k}
\rput[c]{0}(1.75,-1.25){t}
\rput[c]{0}(2.25,-1.25){k}
\rput[c]{0}(2.75,-1.25){t}
\pspolygon[linewidth=2pt](1,-1)(1.5,-1)(1.5,-0.5)(1,-0.5)
\pspolygon[linewidth=2pt](2,-1)(2.5,-1)(2.5,-0.5)(2,-0.5)
}
\rput[c]{0}(2,-1.75){
\rput[c]{0}(-0.7,-0.25){basis}
\rput[c]{0}(-0.7,-0.75){result}
\rput[c]{0}(-0.7,-1.25){test/key}
\psline[linewidth=0.5pt]{-}(0,0)(3,0)
\psline[linewidth=0.5pt]{-}(0,-0.5)(3,-0.5)
\psline[linewidth=0.5pt]{-}(0,-1)(3,-1)
\psline[linewidth=0.5pt]{-}(0,-1.5)(3,-1.5)
\psline[linewidth=0.5pt]{-}(0,0)(0,-1.5)
\psline[linewidth=0.5pt]{-}(0.5,0)(0.5,-1.5)
\psline[linewidth=0.5pt]{-}(1,0)(1,-1.5)
\psline[linewidth=0.5pt]{-}(1.5,0)(1.5,-1.5)
\psline[linewidth=0.5pt]{-}(2,0)(2,-1.5)
\psline[linewidth=0.5pt]{-}(2.5,0)(2.5,-1.5)
\psline[linewidth=0.5pt]{-}(3,0)(3,-1.5)
\rput[c]{67.5}(0.25,-0.25){$+$}
\rput[c]{0}(0.75,-0.25){$\times$}
\rput[c]{0}(1.25,-0.25){$\times$}
\rput[c]{67.5}(1.75,-0.25){$+$}
\rput[c]{22.5}(2.25,-0.25){$+$}
\rput[c]{22.5}(2.75,-0.25){$+$}
\rput[c]{0}(0.25,-0.75){$0$}
\rput[c]{0}(0.75,-0.75){$0$}
\rput[c]{0}(1.25,-0.75){$0$}
\rput[c]{0}(1.75,-0.75){$0$}
\rput[c]{0}(2.25,-0.75){$1$}
\rput[c]{0}(2.75,-0.75){$1$}
\rput[c]{0}(0.25,-1.25){t}
\rput[c]{0}(0.75,-1.25){t}
\rput[c]{0}(1.25,-1.25){k}
\rput[c]{0}(1.75,-1.25){t}
\rput[c]{0}(2.25,-1.25){k}
\rput[c]{0}(2.75,-1.25){t}
\pspolygon[linewidth=2pt](1,-1)(1.5,-1)(1.5,-0.5)(1,-0.5)
\pspolygon[linewidth=2pt](2,-1)(2.5,-1)(2.5,-0.5)(2,-0.5)
}
\endpspicture
\caption{Ekert's quantum key distribution protocol~\cite{ekert}. The marked bits form the key.  
}
\label{fig:ekert}
\end{figure}

At first it seemed that Ekert's protocol relying on the monogamy of entanglement and Bennett and Brassard's protocol based on the uncertainty principle could be brought in a similar form~\cite{bbm}. In both cases, the only property necessary for security seemed to be the fact that when Alice and Bob use bases pointing in the same direction, they obtain perfectly correlated outcomes. Indeed, a key distribution scheme of the type \emph{prepare and measure} as the one of Bennett and Brassard, can usually be formulated in terms of an \emph{entanglement-based} protocol, as the one of Ekert, by considering as state the superposition of the random choice of basis and encoded bit on Alice's side with the state corresponding to this random choice on Bob's side, i.e., the state $\ket{\Psi}=\sum_r \sqrt{P(r)}\ket{r}\otimes \ket{\phi_r}$, where $r$ is the random value and $\ket{\phi_r}$ the state that is sent conditioned on $r$ (see~\cite{rennerphd} for a more detailed explanation). Alice then measures (in the computational basis) to obtain the random value $r$, while Bob does the same as in the original protocol.

However, after a more detailed investigation, it turned out that Ekert's protocol had an advantage over the BB84 protocol. Namely, in Ekert's protocol the key bits do not have any associated `element of  reality' \linebreak[4]\cite{ekert}. This implies that the eavesdropper `is in the hopeless position of trying to intercept non-existent information'~\cite{bbm}. This property can be very useful to overcome attacks taking advantage of flaws in the physical implementation and to 
create a key distribution protocol with untrusted devices, as we will explain below.

\section{The Need for Device-Independence}\label{sec:whydi}

It had been discovered that quantum key-distribution protocols are vulnerable to imperfections in the physical 
implementation in a way that an adversary could easily manipulate the apparatus such that the key-distribution scheme becomes completely insecure. 

Imagine, for example, that in the BB84 protocol, several photons are sent from Alice to Bob~\cite{blms,luetkenhaus}. Eve could easily attack this system by storing some of the photons in a memory. Later, she can measure it in the basis announced by Alice and know the encoded bit with certainty. The scheme, therefore,  crucially relies on the source to emit single photons to be 
secure. 
In practice, on the other hand, the photons are usually emitted by a laser with a Poissonian photon-number distribution and these are neither theoretically nor practically a single-photon source. 

As a second possible way to attack the system, imagine that the devices encoding the bit and measuring the photon are faulty: Instead of encoding and measuring in two different bases chosen at random, they always use the same basis. The eavesdropper can measure the photon in this basis without disturbing it. She can learn the bit perfectly, but will remain completely unnoticed by Alice and Bob.  

The BB84 scheme is particularly vulnerable to the problem that the bit or basis might not only be encoded in the photon, but additionally in other carriers. This problem was already noticed when the BB84 protocol was implemented for the very first time~\cite{bbbss}: The devices responsible for the choice of the polarization angle made a loud noise and this noise 
was different depending on the angle, 
such that the scheme could only reach security against a completely deaf eavesdropper~\cite{gilles}.

In fact, in the security analysis of quantum key distribution, the \emph{dimension} of the systems, i.e., their Hilbert spaces, always enters into the calculations, both in the estimation of the entropy the adversary has about the raw key, as well as in the reduction of coherent to collective attacks (de Finetti theorem)~\cite{rennerphd}.  These security proofs, therefore, only hold when the dimension of the system is known, which cannot be assumed if the adversary can tamper with the devices. For the security proof of quantum key distribution, it is, therefore, assumed that the devices are trustworthy and work exactly as specified.

This shows that even though quantum key distribution is often claimed to be \emph{unconditionally secure} (meaning that it 
does not rely on computational hardness assumptions) it actually \emph{does} make certain assumptions. The first of these assumptions --- always present in key agreement --- is, that Alice and Bob have secure laboratories. If the eavesdropper can look over Alice's shoulder when she is typing the key into her computer to use it for encryption, or if the physical device contains a transmitter sending all raw data to Eve, it is clear that no security is possible.\footnote{In classical cryptography it has recently been investigated how to construct encryption schemes which are robust against (partial) leakage of the key~\cite{DHLW10,BKKV10}.} This assumption is crucial and cannot be removed. Even though it might seem clear that such attacks need to be prevented somehow, this might not always be trivial in practice. There are examples of successful attacks where critical information about the key has been read from the screen via reflections~\cite{teapot}, from acoustic disk noise~\cite{acoustic}, protocol response time~\cite{timing1,timing2} or from the electromagnetic waves emitted by the screen~\cite{screen}. In quantum key distribution, information about the raw key could be inferred from timing information exchanged over the public authenticated channel~\cite{quantumtiming}. Alice and Bob, therefore,  need to shield their laboratories securely. 

A further assumption usually present in quantum key distribution 
is that Alice and Bob have complete control over their physical devices (i.e., only the quantum channel is corrupted) and know their \emph{exact} and \emph{complete} specification. For example, if the device is supposed to emit a single photon with an encoded bit, it cannot emit another particle where this bit is \emph{also} encoded. We have argued above that a failure of this assumption can directly lead to possible attacks on the quantum key-distribution scheme. These attacks are not only theoretical constructions, but can be implemented in practice and used to break even commercially available quantum key-distribution schemes~\cite{makarov}. 

Additionally, Alice and Bob need to be able to toss coins, i.e., have local trusted sources of randomness. In particular, it is important that they can choose their measurement bases at random and independent from the eavesdropper, and that they can choose random samples to test their systems. It is clear that if the eavesdropper could know beforehand, or even choose, the randomness used for either of these two processes, it would be easy to attack successfully.\footnote{If Alice and Bob can build quantum devices, they can, of course use quantum physics to build a random number generator.}

Finally, it is normally assumed that Alice and Bob are able to do classical computation (perfectly). For example, they need to be able to calculate the statistics of their measurement outcomes. In the case of the BB84 protocol this corresponds to counting correctly the number of bits which were incorrectly received by Bob when measuring in the same basis. This is crucial to estimate the error rate and to abort in case the eavesdropper intercepted too many messages. The classical post-processing of Alice's and Bob's data is usually also assumed to be error-free.

The goal of \emph{device-independent quantum key distribution} is to reduce the above assumptions to a minimum, in particular, to remove all assumptions about the exact working of the physical devices.\footnote{We will not consider the case where Alice and Bob do not trust their random number generator, but assume that they can toss random coins. For a proposal how to build device-independent sources of randomness starting from a small random seed, see~\cite{colbeckphd}.} The devices could then even be manufactured by the adversary. Ideally, the security should only rely on \emph{testable} features of the devices, for example, the statistics of their behaviour. The honest parties would then only need to trust their ability to do classical calculations (to compute the statistics) and the shielding of their laboratories.

\section{Possible Approaches}\label{sec:approaches}

Mayers and Yao~\cite{my} noted that in the context of device-in\-de\-pen\-dent key distribution, \emph{entanglement-based protocols} have a major advantage  \linebreak[4]  compared to \emph{prepare-and-measure protocols}. They propose a source with an additional testing device --- taking purely classical inputs and outputs --- such that these classical inputs and outputs can be used to test whether the source is suitable for quantum key distribution. They call this a \emph{self-checking} source. They noted that there exist certain correlations of the measurement results of quantum states which can \emph{only} be achieved by a state equivalent to the singlet state. In particular, the correlations used in the entanglement-based protocol (Figure~\ref{fig:ekert}) to test for entanglement are of this type. Security follows because the singlet state necessarily needs to be independent of any state the eavesdropper might have.
The argument of Mayers and Yao was made robust against noise~\cite{mmmo}
and extended to self-checking of circuits and other devices. In~\cite{abgs}, a device-independent quantum key-distribution protocol secure against collective\footnote{In a collective attack each of the systems is attacked independently and individually, but a joint measurement can be performed on Eve's system in the end. } attacks was given. 
Under a plausible, but unproven conjecture, this protocol can even be made secure against the most general  attacks if the measurement devices are memoryless~\cite{mckaguephd}. 

The idea used in the security proof of these device-independent schemes is that, for binary outcomes, the Hilbert space is in some sense equivalent to the Hilbert space of qubits. It is then sufficient to restrict to the case of qubits in the security analysis, which means that eavesdropping can be detected using a Bell test. If this test gives a value close to $2\sqrt{2}$ (for the case of the 
\emph{CHSH inequality}, see Section~\ref{subsec:localsystem}), the state must also be close to the singlet state (potentially embedded into higher dimensions). The realization of these key-agreement protocols are, therefore, very similar to an entanglement-based protocol.

Barrett, Hardy, and Kent~\cite{bhk} observed that the correlations obtained from measuring an entangled quantum system can be used to prove the security of key distribution based on the non-existence of hidden variables describing this physical system. 
In fact, Bell~\cite{bellinequality} had shown that it is not possible to describe the correlations obtained from measurements on certain entangled quantum states in a way that each of the  measurements has a well-defined pre-determined outcome. 
Barrett, Hardy, and Kent show that there exist certain quantum correlations such that the measurement outcomes even need to be \emph{completely random} before the measurement is actually carried out. This property can be used to show that the measurement outcomes need to be completely independent of any information the eavesdropper can possibly hold. 

Note that the scheme Barrett, Hardy, and Kent propose uses quantum physics to \emph{create} these (observable) correlations. However, the security is based only on the requirement that no information can be exchanged between the three parties via the system and it is, therefore, independent of quantum physics. The scheme they propose works as follows (see Figure~\ref{fig:bhk}). Alice and Bob measure $n$ singlet states using one out of $N$ possible bases on a circle (where the $N$\textsuperscript{th} basis corresponds to a $\pi/2$ rotation compared to the $0$\textsuperscript{th} basis). Bob inverts his outcome bit. They announce the measurement bases over the public authenticated channel and keep only the results for which they have measured in the same or in neighbouring bases modulo $N$ (i.e., where they either had a very small angle between the measurement directions or an angle of almost $\pi/2$). From the remaining measurements, they uncover all but one result. 
They check whether all the results where they measured almost in the same direction were equal and all the results where they measured in almost orthogonal direction were different. If this is not the case, they abort. If they did not abort, they take the remaining measurement outcome as key bit (with Bob inverting the value in case they measured at almost $\pi/2$). 
The scheme works because measuring a quantum system gives a higher probability of passing the test than what could be achieved by classical shared randomness.

\begin{figure}[h]
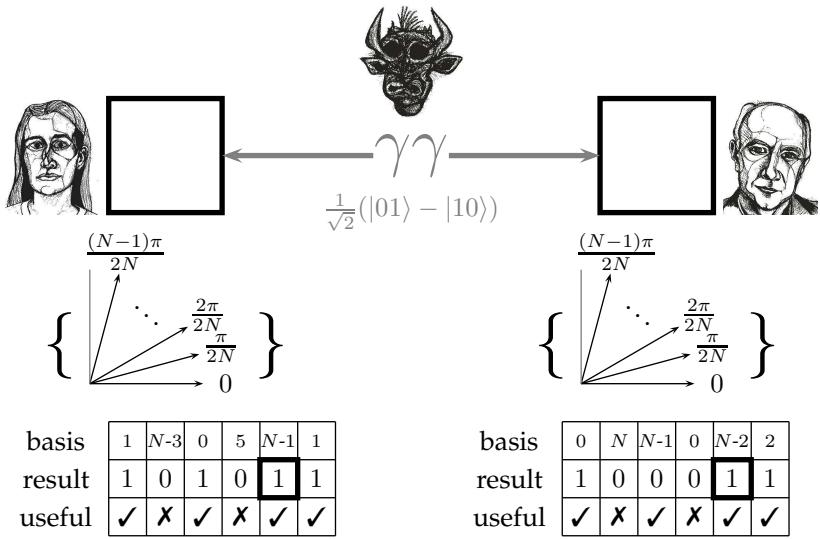

\centering
\pspicture*[](-5.4,-4.5)(5.4,2.75)
\rput[c]{0}(-4.75,0.75){\includegraphics[width=1.25cm]{alice_small_dither.eps}}
\rput[c]{0}(4.75,0.75){\includegraphics[width=1.25cm]{bob_small_dither.eps}}
\rput[c]{0}(0,2){\includegraphics[width=1.4cm]{eve_small_dither.eps}}
\pspolygon[linewidth=2pt](-4,0)(-2.5,0)(-2.5,1.5)(-4,1.5)
\pspolygon[linewidth=2pt](4,0)(2.5,0)(2.5,1.5)(4,1.5)
\psline[linewidth=2pt,linecolor=gray]{<-}(-2.5,0.75)(-0.5,0.75)
\psline[linewidth=2pt,linecolor=gray]{->}(0.5,0.75)(2.5,0.75)
\rput[c]{0}(0,0){\color{gray}{{$\frac{1}{\sqrt{2}}(\ket{01}-\ket{10})$}}}
\rput[c]{0}(0,0.75){\color{gray}{\Huge{$\gamma \gamma$}}}
\rput[c]{0}(-6.25,0){
\rput[c]{0}(1.6,-1.75){\Huge{$\{$}}
\rput[c]{0}(2.75,-1.25){{$\ddots$}}
\rput[c]{0}(2.0,-2.25){
\psline[linewidth=0.5pt]{->}(0,0)(1.5,0)
\rput[c]{0}(1.8,0){$0$}
}
\rput[c]{15}(2.0,-2.25){
\psline[linewidth=0.5pt]{->}(0,0)(1.5,0)
\rput[c]{-15}(1.8,0){$\frac{\pi}{2N}$}
}
\rput[c]{30}(2.0,-2.25){
\psline[linewidth=0.5pt]{->}(0,0)(1.5,0)
\rput[c]{-30}(1.8,0){$\frac{2\pi}{2N}$}
}
\rput[c]{75}(2.0,-2.25){
\psline[linewidth=0.5pt]{->}(0,0)(1.5,0)
\rput[c]{-75}(1.8,0){$\frac{(N-1)\pi}{2N}$}
}
\rput[c]{90}(2.0,-2.25){
\psline[linecolor=gray,linewidth=0.5pt]{-}(0,0)(1.5,0)
}
\rput[c]{0}(4.4,-1.75){\Huge{$\}$}}
}
\rput[c]{0}(0.25,0){
\rput[c]{0}(1.6,-1.75){\Huge{$\{$}}
\rput[c]{0}(2.75,-1.25){{$\ddots$}}
\rput[c]{0}(2.0,-2.25){
\psline[linewidth=0.5pt]{->}(0,0)(1.5,0)
\rput[c]{0}(1.8,0){$0$}
}
\rput[c]{15}(2.0,-2.25){
\psline[linewidth=0.5pt]{->}(0,0)(1.5,0)
\rput[c]{-15}(1.8,0){$\frac{\pi}{2N}$}
}
\rput[c]{30}(2.0,-2.25){
\psline[linewidth=0.5pt]{->}(0,0)(1.5,0)
\rput[c]{-30}(1.8,0){$\frac{2\pi}{2N}$}
}
\rput[c]{75}(2.0,-2.25){
\psline[linewidth=0.5pt]{->}(0,0)(1.5,0)
\rput[c]{-75}(1.8,0){$\frac{(N-1)\pi}{2N}$}
}
\rput[c]{90}(2.0,-2.25){
\psline[linecolor=gray,linewidth=0.5pt]{-}(0,0)(1.5,0)
}
\rput[c]{0}(4.4,-1.75){\Huge{$\}$}}
}
\rput[c]{0}(-4,-2.75){
\rput[c]{0}(-0.7,-0.25){basis}
\rput[c]{0}(-0.7,-0.75){result}
\rput[c]{0}(-0.7,-1.25){useful}
\psline[linewidth=0.5pt]{-}(0,0)(3,0)
\psline[linewidth=0.5pt]{-}(0,-0.5)(3,-0.5)
\psline[linewidth=0.5pt]{-}(0,-1)(3,-1)
\psline[linewidth=0.5pt]{-}(0,-1.5)(3,-1.5)
\psline[linewidth=0.5pt]{-}(0,0)(0,-1.5)
\psline[linewidth=0.5pt]{-}(0.5,0)(0.5,-1.5)
\psline[linewidth=0.5pt]{-}(1,0)(1,-1.5)
\psline[linewidth=0.5pt]{-}(1.5,0)(1.5,-1.5)
\psline[linewidth=0.5pt]{-}(2,0)(2,-1.5)
\psline[linewidth=0.5pt]{-}(2.5,0)(2.5,-1.5)
\psline[linewidth=0.5pt]{-}(3,0)(3,-1.5)
\rput[c]{0}(0.25,-0.25){\scriptsize{$1$}}
\rput[c]{0}(0.74,-0.25){\scriptsize{$N$-$3$}}
\rput[c]{0}(1.25,-0.25){\scriptsize{$0$}}
\rput[c]{0}(1.75,-0.25){\scriptsize{$5$}}
\rput[c]{0}(2.24,-0.25){\scriptsize{$N$-$1$}}
\rput[c]{0}(2.75,-0.25){\scriptsize{$1$}}
\rput[c]{0}(0.25,-0.75){$1$}
\rput[c]{0}(0.75,-0.75){$0$}
\rput[c]{0}(1.25,-0.75){$1$}
\rput[c]{0}(1.75,-0.75){$0$}
\rput[c]{0}(2.25,-0.75){$1$}
\rput[c]{0}(2.75,-0.75){$1$}
\rput[c]{0}(0.25,-1.25){\ding{51}}
\rput[c]{0}(0.75,-1.25){\ding{55}}
\rput[c]{0}(1.25,-1.25){\ding{51}}
\rput[c]{0}(1.75,-1.25){\ding{55}}
\rput[c]{0}(2.25,-1.25){\ding{51}}
\rput[c]{0}(2.75,-1.25){\ding{51}}
\pspolygon[linewidth=2pt](2,-1)(2.5,-1)(2.5,-0.5)(2,-0.5)
}
\rput[c]{0}(2,-2.75){
\rput[c]{0}(-0.7,-0.25){basis}
\rput[c]{0}(-0.7,-0.75){result}
\rput[c]{0}(-0.7,-1.25){useful}
\psline[linewidth=0.5pt]{-}(0,0)(3,0)
\psline[linewidth=0.5pt]{-}(0,-0.5)(3,-0.5)
\psline[linewidth=0.5pt]{-}(0,-1)(3,-1)
\psline[linewidth=0.5pt]{-}(0,-1.5)(3,-1.5)
\psline[linewidth=0.5pt]{-}(0,0)(0,-1.5)
\psline[linewidth=0.5pt]{-}(0.5,0)(0.5,-1.5)
\psline[linewidth=0.5pt]{-}(1,0)(1,-1.5)
\psline[linewidth=0.5pt]{-}(1.5,0)(1.5,-1.5)
\psline[linewidth=0.5pt]{-}(2,0)(2,-1.5)
\psline[linewidth=0.5pt]{-}(2.5,0)(2.5,-1.5)
\psline[linewidth=0.5pt]{-}(3,0)(3,-1.5)
\rput[c]{0}(0.25,-0.25){\scriptsize{$0$}}
\rput[c]{0}(0.75,-0.25){\scriptsize{$N$}}
\rput[c]{0}(1.24,-0.25){\scriptsize{$N$-$1$}}
\rput[c]{0}(1.75,-0.25){\scriptsize{$0$}}
\rput[c]{0}(2.24,-0.25){\scriptsize{$N$-$2$}}
\rput[c]{0}(2.75,-0.25){\scriptsize{$2$}}
\rput[c]{0}(0.25,-0.75){$1$}
\rput[c]{0}(0.75,-0.75){$0$}
\rput[c]{0}(1.25,-0.75){$0$}
\rput[c]{0}(1.75,-0.75){$0$}
\rput[c]{0}(2.25,-0.75){$1$}
\rput[c]{0}(2.75,-0.75){$1$}
\rput[c]{0}(0.25,-1.25){\ding{51}}
\rput[c]{0}(0.75,-1.25){\ding{55}}
\rput[c]{0}(1.25,-1.25){\ding{51}}
\rput[c]{0}(1.75,-1.25){\ding{55}}
\rput[c]{0}(2.25,-1.25){\ding{51}}
\rput[c]{0}(2.75,-1.25){\ding{51}}
\pspolygon[linewidth=2pt](2,-1)(2.5,-1)(2.5,-0.5)(2,-0.5)
}
\endpspicture
\caption{The protocol of Barrett, Hardy, and Kent. Alice and Bob choose a number $i$ at random from $\{0,\dotsc , N-1\}$ and measure the singlet in a basis turned by an angle ${i\pi}/{2 N}$. The marked bit is the key. 
}
\label{fig:bhk}
\end{figure}

The scheme proposed by Barrett, Hardy, and Kent is secure against the most general attacks. However, it only works if the quantum system and the measurement are perfectly noiseless, as otherwise the scheme will abort. Furthermore, its security is at most directly proportional to the number of systems used, which implies that it only reaches a zero key rate. The reason 
for this is that the measurement outcomes are directly used as part of the key (without doing privacy amplification). 

One proposition to overcome this problem is to use an entanglement-based scheme as given in Figure~\ref{fig:ekert}. Indeed, it can be shown that the outputs of such a system are also \emph{partially} secret against non-signalling eavesdroppers.  This system corresponds, in fact, to the case $N=2$ in the scheme of Barrett, Hardy, and Kent.  The idea is to use several of these partially secure bits to create a highly secure bit using  \emph{privacy amplification}, i.e., by applying a function to them. 
Of course, when Alice's and Bob's measurements enclose a certain angle, they will, in general, not obtain highly correlated outcomes and they will also need to do \emph{information reconciliation} to correct the errors in their raw keys.  
Such classical post-processing does indeed work, if the eavesdropper's attacks are restricted to \emph{individual attacks}~\cite{AcinGisinMasanes,AcinMassarPironio,SGBMPA}, i.e., the eavesdropper is assumed to attack and measure each system independently. For general attacks, privacy amplification against non-signalling adversaries is, however, only possible if additional non-signalling conditions are imposed between the subsystems~\cite{lluis,eurocrypt,HRW08}. 
 
The implementation of these protocols then works again along similar lines as an entanglement-based protocol.

\section{Outline and Main Results}

In this thesis, we study both approaches to device-independent quantum key distribution, using the whole of quantum physics and using only the impossibility of signalling via the physical devices (non-signalling principle). Below, we give an outline of the thesis with an overview of the main results. We include an informal description of the `proof idea' and point to the locations where the formal statements and proofs can be found. 

\subsubsection{Preliminaries}

In the next chapter, we establish the notation and review  the techniques we will use. The basics of probability theory are explained in Section~\ref{sec:prob} and the notion of (computational) efficiency in Section~\ref{sec:eff}. We will show security based on \emph{random systems} and by comparing our system to an ideal system. This approach and what it means for a key to be secure is explained in Section~\ref{sec:randomsystems}. As a tool, we will use \emph{convex optimization} in the security analysis, which we review in Section~\ref{sec:optimization}. We then introduce the basic laws of quantum physics (Section~\ref{sec:quantum}). In Section~\ref{sec:systems}, we study which systems can be realized using different resources, in particular shared randomness (Section~\ref{subsec:localsystem}), quantum mechanics (Section~\ref{subsec:quantum}), and general non-signalling theories (Section~\ref{subsec:nssystem}).

\subsubsection{Key distribution secure against non-signalling adversaries}

In Chapter~\ref{ch:nsadversaries}, we study key agreement in the presence of adversaries only limited by the non-signalling condition. This means that the adversary can interact with the physical system in an arbitrary way as long as this interaction does not imply communication between the different subsystems. Even though this non-signalling condition might be inspired by quantum mechanics, this approach does not require the validity of quantum mechanics for the security proof. The systems are implemented by quantum physics (i.e., we think that such systems \emph{exist}, because quantum mechanics predicts them), but for the security analysis this is completely irrelevant. Security is based only on the observed correlations.

\subparagraph*{Main results:} We show that for \emph{any} type of partial non-sig\-nal\-ling secrecy, privacy amplification against a non-signalling adversaries is possible using a \emph{deterministic} privacy amplification function (the XOR) if the non-signalling condition holds between all subsystems. This insight leads to a device-independent key-distribution scheme which is efficient in \linebreak[4]
terms of classical and quantum communication. 

\subparagraph*{Informal proof sketch:} Assume Alice and Bob share some kind of physical system. They can choose a measurement and obtain a result. We model this abstractly as a non-signalling system $P_{XY|UV}$ (see the left-hand side of Figure~\ref{fig:boxpartition}) taking inputs and giving outputs. 
The attack a non-signalling adversary can make on such a system corresponds exactly to the choice of a convex decomposition (input) and obtaining one of the elements (output) (see Lemmas~\ref{lemma:nsbox}, p.~\pageref{lemma:nsbox}, and~\ref{lemma:boxns}, p.~\pageref{lemma:boxns}). 

\begin{figure}[b]
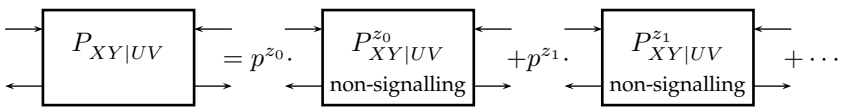

\centering
\pspicture*[](-1.1,-0.1)(10.2,1.6)
\psset{unit=1cm}
\rput[c]{0}(-0.5,0){
\pspolygon[linewidth=1pt](0,0)(0,1.25)(2,1.25)(2,0)
\rput[c]{0}(1,0.75){$P^{\phantom{z}}_{XY|UV}$}
\psline[linewidth=0.5pt]{->}(-0.5,1)(0,1)
\psline[linewidth=0.5pt]{<-}(-0.5,0.25)(0,0.25)
\psline[linewidth=0.5pt]{<-}(2,1)(2.5,1)
\psline[linewidth=0.5pt]{->}(2,0.25)(2.5,0.25)
}
\rput[c]{0}(3.2,0){
\rput[c]{0}(-0.85,0.6){$=p^{z_0}\cdot $}
\pspolygon[linewidth=1pt](0,0)(0,1.25)(2,1.25)(2,0)
\rput[c]{0}(1,0.75){$P^{z_0}_{XY|UV}$}
\rput[c]{0}(1,0.25){\footnotesize{non-signalling}}
\psline[linewidth=0.5pt]{->}(-0.5,1)(0,1)
\psline[linewidth=0.5pt]{<-}(-0.5,0.25)(0,0.25)
\psline[linewidth=0.5pt]{<-}(2,1)(2.5,1)
\psline[linewidth=0.5pt]{->}(2,0.25)(2.5,0.25)
}
\rput[c]{0}(6.9,0){
\rput[c]{0}(-0.85,0.6){$+p^{z_1}\cdot $}
\pspolygon[linewidth=1pt](0,0)(0,1.25)(2,1.25)(2,0)
\rput[c]{0}(1,0.75){$P^{z_1}_{XY|UV}$}
\rput[c]{0}(1,0.25){\footnotesize{non-signalling}}
\psline[linewidth=0.5pt]{->}(-0.5,1)(0,1)
\psline[linewidth=0.5pt]{<-}(-0.5,0.25)(0,0.25)
\psline[linewidth=0.5pt]{<-}(2,1)(2.5,1)
\psline[linewidth=0.5pt]{->}(2,0.25)(2.5,0.25)
}
\rput[c]{0}(9.7,0.6){$+\dotsb $}
\endpspicture
\caption{\label{fig:boxpartition} By Lemmas~\ref{lemma:nsbox} and~\ref{lemma:boxns}, an attack of the eavesdropper corresponds to a choice of convex decomposition. Her outcome is an element in the convex decomposition. }
\end{figure}

The question how much Eve can know about Alice's output bit $X$, therefore, corresponds to finding the best convex decomposition of Alice's and Bob's system, such that,  given $Z$, Eve can guess $X$. 

Since the conditions on a non-signalling system are linear, we can characterize this quantity by a linear program (see Lemma~\ref{lemma:distanceislp}, p.~\pageref{lemma:distanceislp}), i.e., an optimization problem of the form 
\begin{align} 
\nonumber \mathrm{PRIMAL}\\
\nonumber \max : &\quad b^T\cdot x\\
\nonumber \st &\quad
A
\cdot x \leq
c
\end{align}
where $x$ is a vector, $A$ contains, amongst others, the non-signalling conditions and $c$ contains the probabilities $P_{XY|UV}$ of the marginal system as seen by Alice and Bob. The maximal distance from uniform of $X$, from a non-signalling adversary's point of view, is $ b^T  x^*/2$, where $x^*$ is the optimal solution of this linear program. 

As an example, consider a system with binary inputs and outputs such that $\Pr[X\oplus Y=U\cdot V]=1-\ep$.\footnote{Note that this corresponds to the Bell test performed in~\cite{ekert}. A value of $B$ in the Bell test --- the maximum quantum value being $2\sqrt{2}$ --- corresponds to $1-\ep=1/2+B/8$, see Section~\ref{subsec:localsystem}.} In this case, the distance from uniform of Alice's output bit $X$ is at most $2\ep$, i.e., the more \emph{non-local} the system is, the more secret is the output bit. 

Alternatively to the primal form, we can consider the \emph{dual} form of the linear program,  given by 
\begin{align}
\nonumber \mathrm{DUAL}\\
\nonumber  \min :&\quad c^T\cdot \lambda\\
\nonumber  \st &\quad A^T\cdot \lambda = b\\
\nonumber &\quad \lambda \geq 0\ . 
\end{align}
Any dual feasible $\lambda$ gives an upper bound on the primal value ($  b^T x\leq c^T \lambda$) and, therefore, on 
Eve's knowledge about the bit.  The dual value is of the form $c^T \lambda$, where $c$ contains the marginal probabilities, and it, therefore, corresponds to an \emph{event} defined by the inputs and outputs of Alice's and Bob's system. This implies that Alice and Bob can `read' the secrecy of the bit from the behaviour of their system (Lemma~\ref{lemma:dualevent}, p.~\pageref{lemma:dualevent}). In the above example of a system with binary inputs and outputs, there exists a $\lambda$ (the optimal one) such that $c^T \lambda/2=2 \Pr[X\oplus Y\neq U\cdot V]=2\ep$ (see Example~\ref{ex:dual}, p.~\pageref{ex:dual}).

Now consider the case where Alice and Bob share $n$ copies of a bipartite non-signalling system. This can be seen as a $(2n)$-party non-signalling system,  where the non-signalling condition must hold between all subsystems. 
Our main insight, stated in Lemma~\ref{lemma:product_form}, p.~\pageref{lemma:product_form}, is that a system is $(2n)$-party non-signalling if and only if it fulfils $A^{\otimes n} x=0$, where $A$ are the conditions a bipartite non-signalling system must fulfil. 

We can then show that the security of the XOR of several (partially secure) bits $X_i$ can be calculated by the linear program (in its dual form)
\begin{align}
\nonumber \mathrm{DUAL}\\
\nonumber  \min : &\quad c_n^T\cdot \lambda_n\\
\nonumber  \st &\quad (A^{\otimes n})^T\cdot \lambda_n = b^{\otimes n}\\
\nonumber &\quad \lambda_n \geq 0\ ,
\end{align}
i.e., it is the `tensor product' of the individual linear programs. 
This implies that for any $\lambda$ which is feasible for a single system, $\lambda_n=\lambda^{\otimes n}$ is feasible for $n$ systems (Lemma~\ref{lemma:dual_product}, p.~\pageref{lemma:dual_product}) and this gives an upper bound on the distance from uniform. 
When the $n$ bipartite marginal systems behave independently, i.e., they are of the form $c^{\otimes n}$ this gives an upper bound on Eve's knowledge of $(c^T  \lambda)^n/2$, and the $(2n)$-party 
system is \emph{as secure as} if Eve had attacked each of the partial systems individually. In our example, the maximal distance from uniform of the XOR of $n$ bits is $(4\ep)^n/2$. The general statement for systems which do not necessarily have product form is given in Theorem~\ref{th:nsxor}, p.~\pageref{th:nsxor}. 

The insight that the XOR can be used to create a highly secure bit can be used to construct a key-agreement scheme where the key bits and the error-correction information are formed by the XOR of random subsets of the physical bits $X_i$. Such a scheme is analysed in Section~\ref{sec:nskeyagreement}. An explicit protocol, implementable roughly as the one in Figure~\ref{fig:ekert}, is shown to be secure against a non-signalling adversary in Section~\ref{sec:protocol}.

\subsubsection{Key distribution secure against quantum adversaries}

We then turn to the analysis of device-independent key agreement secure against quantum adversaries in Chapter~\ref{ch:quantum}. In this scenario, all systems have to be implemented by quantum physics, but we do not make any assumptions on \emph{how} they are implemented (Hilbert space dimension, etc.). The reason to consider this scenario is that a non-signalling adversary is stronger than 
(realistically) necessary, which gives lower key rates.

The difficulties arising in device-independent key agreement in the presence of quantum adversaries are different from the ones in the non-sig\-nal\-ling case. In the non-sig\-nal\-ling case, the difficulty was \emph{privacy amplification}, i.e., to show that an adversary cannot attack the key bit created from several bits significantly better than when each of these bits is attacked  individually. 
On the other hand,  it is already known that 
a highly secure string can be created from a partially secure one by privacy amplification~\cite{rennerkoenig} even when the adversary can hold quantum information.
The crucial question in this case is therefore, to determine the secrecy contained in the initial string. This secrecy is quantified by the min-entropy, which in turn directly relates to the probability with which the adversary can guess the value of the string correctly. This will, therefore, be the quantity we are interested in bounding. 

\subparagraph*{Main results:} We show how the probability that an adversary can guess the output of a quantum system can be calculated using a \emph{semi-definite program}. We then show that the guessing probability of the outputs of several quantum systems, where measurements on different subsystems commute, follows a \emph{product theorem}, in the sense that the probability to guess the whole string correctly is the product of the guessing probability of each subsystem. Using this property, we can construct a device-independent key-agreement scheme secure against the most general attacks by a quantum adversary.

\subparagraph*{Informal proof sketch:} Conceptually, our approach is similar to the one in the case of non-signalling adversaries. We will also show that the conditions several quantum systems must fulfil are the tensor product of the conditions of the individual systems. 

Alice and Bob share a quantum system characterized by a probability distribution $P_{XY|UV}.$ An adversary trying to guess Alice's string $X$ can choose a measurement on her part of the system and obtain a measurement result. What measurement she performs can, of course, depend on additional information. Any measurement induces a convex decomposition of Alice's and Bob's system, where each element needs to be a quantum system, i.e., Eve's possibilities are given by Figure~\ref{fig:quantumboxpartition}. 
\begin{figure}[h]
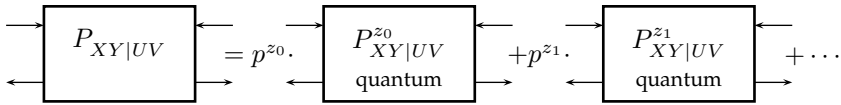

\centering
\pspicture*[](-1.1,-0.1)(10.2,1.6)
\psset{unit=1cm}
\rput[c]{0}(-0.5,0){
\pspolygon[linewidth=1pt](0,0)(0,1.25)(2,1.25)(2,0)
\rput[c]{0}(1,0.75){$P^{\phantom{z}}_{XY|UV}$}
\psline[linewidth=0.5pt]{->}(-0.5,1)(0,1)
\psline[linewidth=0.5pt]{<-}(-0.5,0.25)(0,0.25)
\psline[linewidth=0.5pt]{<-}(2,1)(2.5,1)
\psline[linewidth=0.5pt]{->}(2,0.25)(2.5,0.25)
}
\rput[c]{0}(3.2,0){
\rput[c]{0}(-0.85,0.6){$=p^{z_0}\cdot $}
\pspolygon[linewidth=1pt](0,0)(0,1.25)(2,1.25)(2,0)
\rput[c]{0}(1,0.75){$P^{z_0}_{XY|UV}$}
\rput[c]{0}(1,0.25){\footnotesize{quantum}}
\psline[linewidth=0.5pt]{->}(-0.5,1)(0,1)
\psline[linewidth=0.5pt]{<-}(-0.5,0.25)(0,0.25)
\psline[linewidth=0.5pt]{<-}(2,1)(2.5,1)
\psline[linewidth=0.5pt]{->}(2,0.25)(2.5,0.25)
}
\rput[c]{0}(6.9,0){
\rput[c]{0}(-0.85,0.6){$+p^{z_1}\cdot $}
\pspolygon[linewidth=1pt](0,0)(0,1.25)(2,1.25)(2,0)
\rput[c]{0}(1,0.75){$P^{z_1}_{XY|UV}$}
\rput[c]{0}(1,0.25){\footnotesize{quantum}}
\psline[linewidth=0.5pt]{->}(-0.5,1)(0,1)
\psline[linewidth=0.5pt]{<-}(-0.5,0.25)(0,0.25)
\psline[linewidth=0.5pt]{<-}(2,1)(2.5,1)
\psline[linewidth=0.5pt]{->}(2,0.25)(2.5,0.25)
}
\rput[c]{0}(9.7,0.6){$+\dotsb $}
\endpspicture
\caption{\label{fig:quantumboxpartition}
A quantum adversary's possibilities to attack a system (see Lemma~\ref{lemma:qpartition}, p.~\pageref{lemma:qpartition}). The choice of measurement induces a convex decomposition of Alice's and Bob's system. }
\end{figure}

Finding the maximal guessing probability, therefore, corresponds to the optimization problem of finding the sum of quantum systems with a fixed marginal system of Alice and Bob that gives the best value, as stated in Lemma~\ref{lemma:qattackopt}, p.~\pageref{lemma:qattackopt}.

We then use a semi-definite criterion that any quantum system must fulfil~\cite{npa07}, more precisely, a sequence of semi-definite criteria which can approximate the set of quantum systems arbitrarily well~\cite{dltw08,npa}. Using this sequence as condition on the elements of the convex decomposition, we can bound the guessing probability by a semi-definite program (Lemma~\ref{lemma:guessissdp}, p.~\pageref{lemma:guessissdp}), i.e., an optimization problem of the form 
\begin{align}
\nonumber\mathrm{PRIMAL}\\
\nonumber  \max : &\quad b^T\cdot x\\
\nonumber  \st &\quad A\cdot x = c\\
\nonumber &\quad x   \succeq 0\ ,
\end{align}
where `$\succeq 0$' means that the matrix corresponding to $x$ must be positive-semi-definite. The matrix $A$ contains the condition that the measurement operators on different parts of the system commute, that all measurement operators are orthogonal projectors, and that the operators associated with the same measurement sum up to the identity. The vector $c$ contains the marginal system of Alice and Bob. We can then write the probability that Eve correctly guesses Alice's value as $P_{\mathrm{guess}}\leq b^T x^*$. 

The dual of the above semi-definite program is of the form 
\begin{align}
\nonumber\mathrm{DUAL}\\
\nonumber  \min : &\quad c^T\cdot \lambda \\
\nonumber  \st &\quad A^T\cdot \lambda \succeq b
\end{align}
and any dual feasible solution gives an upper bound on the possible guessing probability of a quantum adversary in terms of the probabilities describing the system shared between Alice and Bob. 

Strictly speaking, the vector $c$ above contains certain entries which can be calculated knowing the state and measurements of Alice and Bob, but which do not correspond to an observable quantity. The above semi-definite program can, therefore, be used to calculate security in the \emph{device-dependent} scenario. To obtain the \emph{device-independent} scenario, we modify the program to optimize additionally over all the unknown entries which are compatible with the \emph{observable} behaviour of the system (i.e., the probabilities $P_{XY|UV}$). This is done in Section~\ref{subsec:observableprob}.

Our main technical insight is that a $(2n)$-party quantum system (where the quantum state is arbitrary but measurements act on a specific subsystem) must necessarily fulfil the conditions $A^{\otimes n}$ characterized by the tensor product of the conditions associated with a bipartite system  (Lem\-ma~\ref{lemma:qproductconditions}, p.~\pageref{lemma:qproductconditions}). This implies that the \emph{dual} of the semi-definite program calculating the guessing probability of the output of $n$ systems is of the form (Lemma~\ref{lemma:qproduct}, p.~\pageref{lemma:qproduct})
\begin{align}
\nonumber\mathrm{DUAL}\\
\nonumber  \min : &\quad c_n^T\cdot \lambda_n \\
\nonumber  \st &\quad (A^{\otimes n})^T\cdot \lambda_n \succeq b^{\otimes n}\ .
\end{align}
Since $b\succeq 0$, using a criterion from~\cite{mittalszegedy}, this implies that for any $\lambda$ that is feasible for a single system, $\lambda^{\otimes n}$ is feasible for $n$ systems, as stated in Lemma~\ref{lemma:qdualproduct} p.~\pageref{lemma:qdualproduct}.

If the $n$ marginal systems are independent (i.e., $c_n=c^{\otimes n}$), this implies that the probability that Eve correctly guesses the value of Alice's string is the product of the probabilities that she guesses each output correctly. More precisely, if the guessing probability of an individual system is \linebreak[4] boun\-ded by $P_{\mathrm{guess}}\leq c^T  \lambda$, then  for $n$ systems it is bounded by the product $P_{\mathrm{guess}\ n}\leq (c^T \lambda)^n$. In terms of the min-entropy this means that the min-entropy of $n$ systems is $n$ times the min-entropy of the individual systems. 
The general statement for arbitrary marginals is given in Theorem~\ref{th:guessprod}, p.~\pageref{th:guessprod}.

Using this insight, it is possible to create a secure key agreement scheme. We first consider the case where the $n$ bipartite marginal systems behave independently (Section~\ref{sec:qkd}) before considering the general case in Section~\ref{sec:notindependent}. 
Finally in Section~\ref{sec:qprotocol}, we give a protocol similar to~\cite{ekert} and analyse its security in the device-independent scenario with commuting measurements.

\subsubsection{Necessity of non-signalling condition}

In the last chapter (Chapter~\ref{ch:impossibiltiy}) we study the question whether an additional non-signalling condition between the subsystems is necessary. The setup we consider is the one where Alice and Bob share $n$ systems such that $\Pr[X_i\oplus Y_i=U_i\cdot V_i]=1-\ep$ (and $X_i$ and $Y_i$ are uniform random bits). As seen above, the output $X_i$ of each of these systems is partially secure. We ask the question whether Alice can create a bit $B=f(X)$ (where $X=X_1\dotso X_n$) from her outputs that is highly secure, even when Eve can attack all systems at once and only needs to respect a non-signalling condition between Alice, Bob and Eve.

\subparagraph*{Main results:} We first show that two partially secure systems are as local as a single one. This implies that they cannot be more secure. We then give a general attack for any number of systems such that the information a non-signalling adversary can gain about any bit $B=f(X)$ is large. More precisely, there exists a constant lower bound independent of the number of systems. This shows that privacy amplification is not possible in this setup.

\subparagraph*{Informal proof sketch:} We first consider the case of one or two systems and calculate their so-called \emph{local part}. The local part is the maximal weight a local system can have in a convex decomposition of the system. This corresponds to the \emph{fraction} of runs that need to give rise to non-local correlations when repeating an experiment and is a way of quantifying non-locality as a resource. 
Since for any local system, a non-signalling adversary can always have perfect knowledge about the outcomes (when the inputs are public), the local part gives an upper bound on the extractable secrecy of a non-signalling system (and a lower bound on the knowledge of the eavesdropper). 

We show that two systems are as local as a single one (Lemma~\ref{lemma:two_symm_boxes}, p.~\pageref{lemma:two_symm_boxes}) and that they can, therefore, not be more secure
(see Figure~\ref{fig:localparttwosyst}).
\begin{figure}[h]
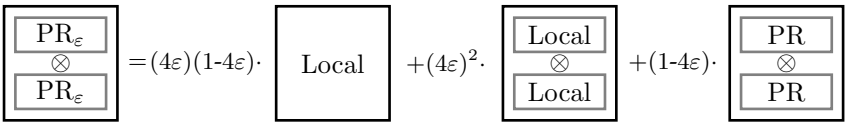

\centering
\psset{unit=0.3cm}
\pspicture*[](-1.2,-8.1)(37.5,0)
\rput[c]{0}(4.9,-4.5){\small{$=$}}
\rput[c]{0}(7.9,-4.5){\small{$(4\varepsilon)(1$-$4\varepsilon)\cdot$}}
\rput[c]{0}(18.5,-4.5){\small{$+(4\varepsilon)^2\cdot$}}
\rput[c]{0}(28.5,-4.5){\small{$+(1$-$4\varepsilon)\cdot$}}
\rput[c]{0}(-2,0){
 \pspolygon[linewidth=1pt](1,-7)(6,-7)(6,-2)(1,-2)
 \pspolygon[linewidth=1pt,linecolor=gray](1.5,-6.5)(5.5,-6.5)(5.5,-5)(1.5,-5)
 \pspolygon[linewidth=1pt,linecolor=gray](1.5,-4)(5.5,-4)(5.5,-2.5)(1.5,-2.5)
\rput[c]{0}(3.5,-4.5){$\otimes$}
\rput[c]{0}(3.5,-3.25){$\mathrm{PR}_{\varepsilon}$}
\rput[c]{0}(3.5,-5.75){$\mathrm{PR}_{\varepsilon}$}
}
\rput[c]{0}(10,0){
 \pspolygon[linewidth=1pt](1,-7)(6,-7)(6,-2)(1,-2)
\rput[c]{0}(3.5,-4.5){$\mathrm{Local}$}
}
\rput[c]{0}(20,0){
 \pspolygon[linewidth=1pt](1,-7)(6,-7)(6,-2)(1,-2)
 \pspolygon[linewidth=1pt,linecolor=gray](1.5,-6.5)(5.5,-6.5)(5.5,-5)(1.5,-5)
 \pspolygon[linewidth=1pt,linecolor=gray](1.5,-4)(5.5,-4)(5.5,-2.5)(1.5,-2.5)
\rput[c]{0}(3.5,-4.5){$\otimes$}
\rput[c]{0}(3.5,-3.25){$\mathrm{Local}$}
\rput[c]{0}(3.5,-5.75){$\mathrm{Local}$}
}
\rput[c]{0}(30,0){
 \pspolygon[linewidth=1pt](1,-7)(6,-7)(6,-2)(1,-2)
 \pspolygon[linewidth=1pt,linecolor=gray](1.5,-6.5)(5.5,-6.5)(5.5,-5)(1.5,-5)
 \pspolygon[linewidth=1pt,linecolor=gray](1.5,-4)(5.5,-4)(5.5,-2.5)(1.5,-2.5)
\rput[c]{0}(3.5,-4.5){$\otimes$}
\rput[c]{0}(3.5,-3.25){$\mathrm{PR}$}
\rput[c]{0}(3.5,-5.75){$\mathrm{PR}$}
}
\endpspicture
\caption{\label{fig:localparttwosyst}Local part of two systems.}
\end{figure}

For more than two systems, we give an attack directly, not using the local part. In Section~\ref{subsec:wbar}, we show that for any number of systems and for any function, there exists a specific good attack. Intuitively, this attack corresponds to a convex decomposition of Alice's and Bob's system, such that each element has weight $1/2$ (for an impossibility proof this is sufficient), and $P^{z_0}_{XY|UV}$ is such that the bit $B=f(X)$ is maximally biased towards $0$ (note that $P_{XY|UV}$ looks like $n$ systems). 

\begin{figure}[h]
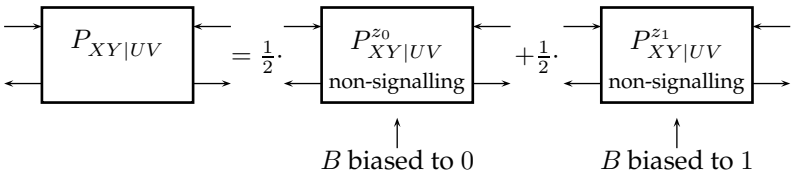

\centering
\pspicture*[](-1.1,-1)(10.2,1.6)
\psset{unit=1cm}
\rput[c]{0}(-0.5,0){
\pspolygon[linewidth=1pt](0,0)(0,1.25)(2,1.25)(2,0)
\rput[c]{0}(1,0.75){$P^{\phantom{z}}_{XY|UV}$}
\psline[linewidth=0.5pt]{->}(-0.5,1)(0,1)
\psline[linewidth=0.5pt]{<-}(-0.5,0.25)(0,0.25)
\psline[linewidth=0.5pt]{<-}(2,1)(2.5,1)
\psline[linewidth=0.5pt]{->}(2,0.25)(2.5,0.25)
}
\rput[c]{0}(3.2,0){
\psline[linewidth=0.5pt]{->}(1,-0.6)(1,-0.2)
\rput[c]{0}(1,-0.75){$B$ biased to $0$}
\rput[c]{0}(-0.85,0.6){$=\frac{1}{2}\cdot $}
\pspolygon[linewidth=1pt](0,0)(0,1.25)(2,1.25)(2,0)
\rput[c]{0}(1,0.75){$P^{z_0}_{XY|UV}$}
\rput[c]{0}(1,0.25){\footnotesize{non-signalling}}
\psline[linewidth=0.5pt]{->}(-0.5,1)(0,1)
\psline[linewidth=0.5pt]{<-}(-0.5,0.25)(0,0.25)
\psline[linewidth=0.5pt]{<-}(2,1)(2.5,1)
\psline[linewidth=0.5pt]{->}(2,0.25)(2.5,0.25)
}
\rput[c]{0}(6.9,0){
\psline[linewidth=0.5pt]{->}(1,-0.6)(1,-0.2)
\rput[c]{0}(1,-0.75){$B$ biased to $1$}
\rput[c]{0}(-0.85,0.6){$+\frac{1}{2}\cdot $}
\pspolygon[linewidth=1pt](0,0)(0,1.25)(2,1.25)(2,0)
\rput[c]{0}(1,0.75){$P^{z_1}_{XY|UV}$}
\rput[c]{0}(1,0.25){\footnotesize{non-signalling}}
\psline[linewidth=0.5pt]{->}(-0.5,1)(0,1)
\psline[linewidth=0.5pt]{<-}(-0.5,0.25)(0,0.25)
\psline[linewidth=0.5pt]{<-}(2,1)(2.5,1)
\psline[linewidth=0.5pt]{->}(2,0.25)(2.5,0.25)
}
\endpspicture
\caption{\label{fig:impboxpartition}The successful attack in the tripartite non-signalling case is such that, with probability $1/2$, Eve obtains an outcome such that the bit $B$ is biased to $0$. }
\end{figure}

In order to define an attack, it is sufficient to construct a non-signalling system $P_{XY|UV}^{z_0}$ such that 
\[P_{XY|UV}^{z_0}(x,y,u,v)\leq 2\cdot  P_{XY|UV}(x,y,u,v)\ .\]
This corresponds exactly to the condition that there exists a second non-signalling system $P_{XY|UV}^{z_1}$ summing up to the correct marginal (Lem\-ma~\ref{lemma:zweihi}, p.~\pageref{lemma:zweihi}). 

Intuitively, we construct $P_{XY|UV}^{z_0}$ starting from $P_{XY|UV}$ and by `moving around probabilities' such that the system remains non-signalling and the above condition is fulfilled. (This intuition is explained in Figure~\ref{fig:intuition}, the formal definition is given in Definition~\ref{def:pz0}, p.~\pageref{def:pz0}). We prove that this indeed defines a convex decomposition of Alice's and Bob's joint system in Lemma~\ref{lemma:pz0nonsig}, p.~\pageref{lemma:pz0nonsig}, and Lemma~\ref{lemma:z0boxpartition}, p.~\pageref{lemma:z0boxpartition}. 
\begin{figure}[h]
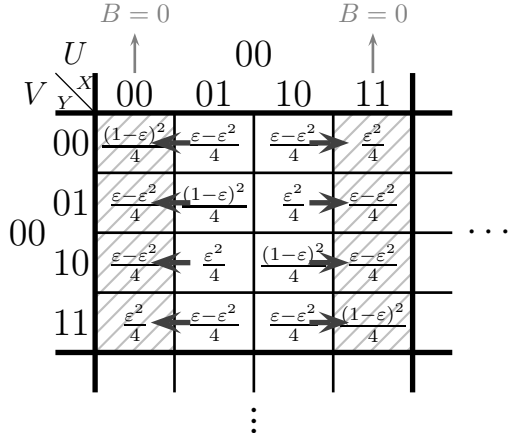

\centering
\psset{unit=0.525cm}
\pspicture*[](-4,-2)(10.5,8.75)
\rput[c]{0}(6,0){
\pspolygon[linewidth=0pt,fillstyle=hlines,hatchcolor=lightgray](0,0)(2,0)(2,6)(0,6)
\psline[linewidth=1pt,linecolor=gray]{->}(1,7)(1,8)
\rput[c]{0}(1,8.5){\color{gray}{$B=0$}}
}
\pspolygon[linewidth=0pt,fillstyle=hlines,hatchcolor=lightgray](0,0)(2,0)(2,6)(0,6)
\psline[linewidth=1pt,linecolor=gray]{->}(1,7)(1,8)
\rput[c]{0}(1,8.5){\color{gray}{$B=0$}}
\psline[linewidth=0.5pt]{-}(0,6)(-1,7)
\rput[c]{0}(10,3){\Large{$\cdots$}}
\rput[c]{0}(4,-1.5){\Large{$\vdots$}}
\rput[c]{0}(-0.25,6.75){\scriptsize{$X$}}
\rput[c]{0}(-0.75,6.25){\scriptsize{$Y$}}
\rput[c]{0}(-0.5,7.5){\large{$U$}}
\rput[c]{0}(-1.5,6.5){\large{$V$}}
\rput[c]{0}(4,7.5){\Large{$00$}}
\rput[c]{0}(1,6.5){\Large{$00$}}
\rput[c]{0}(3,6.5){\Large{$01$}}
\rput[c]{0}(5,6.5){\Large{$10$}}
\rput[c]{0}(7,6.5){\Large{$11$}}
\rput[c]{0}(-1.7,3){\Large{$00$}}
\rput[c]{0}(-0.6,5.25){\Large{$00$}}
\rput[c]{0}(-0.6,3.75){\Large{$01$}}
\rput[c]{0}(-0.6,2.25){\Large{$10$}}
\rput[c]{0}(-0.6,0.75){\Large{$11$}}
\psline[linewidth=2pt]{-}(-1,0)(9,0)
\psline[linewidth=2pt]{-}(-1,6)(9,6)
\psline[linewidth=1pt]{-}(0,3)(9,3)
\psline[linewidth=1pt]{-}(0,1.5)(9,1.5)
\psline[linewidth=1pt]{-}(0,4.5)(9,4.5)
\psline[linewidth=2pt]{-}(0,-1)(0,7)
\psline[linewidth=2pt]{-}(8,-1)(8,7)
\psline[linewidth=1pt]{-}(4,-1)(4,6)
\psline[linewidth=1pt]{-}(2,-1)(2,6)
\psline[linewidth=1pt]{-}(6,-1)(6,6)
\rput[c]{0}(1,5.25){{$\frac{(1-\ep)^2}{4}$}}
\rput[c]{0}(3,3.75){{$\frac{(1-\ep)^2}{4}$}}
\rput[c]{0}(5,5.25){{$\frac{\ep-\ep^2}{4}$}}
\rput[c]{0}(7,3.75){{$\frac{\ep-\ep^2}{4}$}}
\rput[c]{0}(1,2.25){{$\frac{\ep-\ep^2}{4}$}}
\rput[c]{0}(3,0.75){{$\frac{\ep-\ep^2}{4}$}}
\rput[c]{0}(5,0.75){{$\frac{\ep-\ep^2}{4}$}}
\rput[c]{0}(7,2.25){{$\frac{\ep-\ep^2}{4}$}}
\rput[c]{0}(3,5.25){{$\frac{\ep-\ep^2}{4}$}}
\rput[c]{0}(1,3.75){{$\frac{\ep-\ep^2}{4}$}}
\rput[c]{0}(7,5.25){{$\frac{\ep^2}{4}$}}
\rput[c]{0}(5,3.75){{$\frac{\ep^2}{4}$}}
\rput[c]{0}(3,2.25){{$\frac{\ep^2}{4}$}}
\rput[c]{0}(1,0.75){{$\frac{\ep^2}{4}$}}
\rput[c]{0}(5,2.25){{$\frac{(1-\ep)^2}{4}$}}
\rput[c]{0}(7,0.75){{$\frac{(1-\ep)^2}{4}$}}
\psline[linewidth=2pt,linecolor=darkgray,arrowscale=1.5]{<-}(1.4,0.75)(2.4,0.75)
\psline[linewidth=2pt,linecolor=darkgray,arrowscale=1.5]{<-}(1.4,2.25)(2.4,2.25)
\psline[linewidth=2pt,linecolor=darkgray,arrowscale=1.5]{<-}(1.4,3.75)(2.4,3.75)
\psline[linewidth=2pt,linecolor=darkgray,arrowscale=1.5]{<-}(1.4,5.25)(2.4,5.25)
\psline[linewidth=2pt,linecolor=darkgray,arrowscale=1.5]{->}(5.4,0.75)(6.4,0.75)
\psline[linewidth=2pt,linecolor=darkgray,arrowscale=1.5]{->}(5.4,2.25)(6.4,2.25)
\psline[linewidth=2pt,linecolor=darkgray,arrowscale=1.5]{->}(5.4,3.75)(6.4,3.75)
\psline[linewidth=2pt,linecolor=darkgray,arrowscale=1.5]{->}(5.4,5.25)(6.4,5.25)
\endpspicture
\caption{\label{fig:intuition} For the successful attack, we construct $P_{XY|UV}^{z_0}$ starting from $P_{XY|UV}$ and shifting around probabilities. }
\end{figure}

Using this attack, we show that the distance from uniform of the bit $f(X)$ as seen from Eve is at least
\begin{multline}
\nonumber  \max \left\{ \vphantom{\sum\limits_{{y}}}\right. \left. \frac{1}{2}\cdot 
 \left| P(f(X)=0)-P(f(X)=1) \right|\ , 
\right.  \\
 \left.
 \sum\limits_{{y}} \min\left\{ P(f(X|Y=y)=0,P(f(X|Y=y)=0\right\}
\right\} 
 \ .
\end{multline}
(see Lemma~\ref{lemma:distancewbar}, p.~\pageref{lemma:distancewbar}), where the first term is the \emph{bias} of the bit $f(X)$ and the second term is the sum over all possible outputs on Bob's side of the minimal probability that the bit $B$ is $0$ or $1$ given this specific value $y$. 

In case the function $f$ is linear, we can explicitly calculate this value  (Lem\-ma~\ref{lemma:imppalin}, p.~\pageref{lemma:imppalin}). It is always at least $\ep$, but when taking the XOR of many bits it becomes even larger. 

In Section~\ref{subsec:anyhash}, we show that this same attack can also be used against \emph{any} function. We do this in several steps: First, we show that, doing this attack, Eve always gains a substantial amount of information unless Alice and Bob have highly correlated bits (Lemma~\ref{lemma:knowledgereduce}, p.~\pageref{lemma:knowledgereduce}). Then we show that if Alice applies a biased function to obtain her secret bit, Eve can also attack (Lemma~\ref{lemma:knowledgebiggerthanhalfdelta}, p.~\pageref{lemma:knowledgebiggerthanhalfdelta}). We can finally use a result from~\cite{Yang07} on non-interactive correlation distillation stating that it is not possible to produce an unbiased highly correlated bit from several weakly correlated bits by applying a function. 
This leads to a constant lower bound on Eve's information about the key bit, independent of the number of systems used (Theorem~\ref{th:imppa}, p.~\pageref{th:imppa}) and implies that privacy amplification is not possible in this setting.

\chapter{Preliminaries}\label{ch:preliminaries}

\section{Probability Theory}\label{sec:prob}

\subsubsection{Probabilities}

The result of a \emph{random experiment} is called an \emph{event} and, roughly speaking, the chance that such an event is realized is its \emph{probability}. In order to be able to define the probability of an event, it is necessary to know what events can actually occur. The set of possible outcomes of a random experiment is called \emph{sample space} and denoted by $\Omega$. 
Every subset $A$ of $\Omega$, i.e., $A\in \mathcal{P}(\Omega)$, is an \emph{event}. 

We will only encounter discrete probability spaces, i.e., the case when $\Omega$ is a finite or countably infinite set and restrict to this case hereafter. For a more detailed introduction to probability theory we refer to textbooks, such as~\cite{feller,prob}.
\begin{definition}
 A \emph{discrete probability space} is a triple $(\Omega,\mathcal{A},P)$, 
where $\Omega$ is a set, $\mathcal{A}\subset \mathcal{P}(\Omega)$, and $P\colon\mathcal{A}\rightarrow [0,1]$ is a function such that 
\begin{itemize}
 \item $P(\Omega)=1$,
 \item for every sequence of events $A_i$ such that $A_i \cap A_j=\emptyset$ holds for $i\neq j$, we have $P\left(\bigcup_i A_i \right)=\sum_i P(A_i)$\ .
\end{itemize}
$P$ is called \emph{probability} on $(\Omega,\mathcal{A})$.
\end{definition}
When $\Omega$ is discrete, it is actually sufficient for the definition of a probability space to associate a positive number $p_i$ with each $\omega_i\in \Omega$, called \emph{elementary event},  such that $\sum_i p_i=1$. The probability of any event ${A}$ is then the sum of the probabilities associated with the elementary events in~${A}$.

We define the \emph{conditional probability} of an event $A$ given another event $B$. This probability is different from the probability of the event $A$, because the set of possible events is now restricted to the subsets of $B$, instead of~$\Omega$. 
\begin{definition}
Let $(\Omega,\mathcal{A},P)$ be a discrete probability space and $B\in \mathcal{A}$ an event with $P(B)>0$. The conditional probability of an event $A\in \mathcal{A}$ is
\begin{align}
\nonumber P(A|B)&=\frac{P(A\cap B)}{P(B)}\ .
\end{align}
\end{definition}

Two events $A$ and $B$ are called \emph{independent} if the probability that both events happen is the product of the two probabilities.
\begin{definition}
 Two events $A$ and $B$ are called \emph{independent} if
\begin{align}
\nonumber P(A\cap B)&= P(A)\cdot P(B)\ .
\end{align}
\end{definition}
Conditioning $A$ on an independent event $B$ leaves its probability unchanged, i.e.,
\begin{align}
\nonumber P(A|B)&=P(A)\ .
\end{align}

\subsubsection{Random variables}

A \emph{random variable} is a way of encoding the events in $\Omega$ by a number. 
\begin{definition}
A \emph{discrete random variable} $X$ 
on a 
probability space \linebreak[4] $(\Omega,\mathcal{A},P)$ 
 is a map $X\colon \Omega \rightarrow \mathbb{R}$ such that $X(\Omega)$ is countable. Furthermore, 
\begin{align}
\nonumber X^{-1}(x)&=\{\omega \in \Omega | X(\omega)=x\}\in \mathcal{A}\ .
\end{align}
The function $P_X\colon X\rightarrow [0,1]$, such that
\begin{align}
\nonumber P_X(x)&=P(A)\ ,\ \text{where}\ A=X^{-1}(x)\ ,
\end{align}
is called \emph{probability distribution} of the random variable $X$. 
\end{definition}
We will denote random variables by capital letters, such as $X$, the range of the random variable by calligraphic letters, $\mathcal{X}$, and the value a random variable has taken by lower-case letters $x$. The probability that the random variable $X$ takes value $x$ is $P_X(x)$. Sometimes we will drop the index when the random variable is clear from the context. 

We can also define the \emph{joint probability distribution} of two (or more) random variables. 
\begin{definition}
Consider two random variables $X$ and $Y$ defined on the same sample space. The function 
$P_{XY}\colon X\times Y\rightarrow [0,1]$ defined as 
\begin{align}
\nonumber P_{XY}(x,y)&=P(A\cap B)\ \text{where}\ A=X^{-1}(x)\ \text{and}\ B=Y^{-1}(y)\ ,
\end{align}
is called the \emph{joint probability distribution} of $X$ and $Y$. 
\end{definition}
When the joint probability distribution of two (or more) random variables is given, we will sometimes consider the \emph{marginal distribution} of $X$. This is the distribution of the random variable $X$ of a joint distribution when one ignores the value of the second 
random variable $Y$. 
\begin{definition}
 Given the joint probability distribution $P_{XY}$ of two random variables $X$ and $Y$, the marginal distribution of $X$ is 
\begin{align}
\nonumber  P_X(x)&= \sum_y P_{XY}(x,y)\ .
\end{align}
\end{definition}
In analogy with the case of events, we also define the conditional probability distribution as the distribution of the random variable $X$, given that another random variable $Y$ has taken the value $y$. 
\begin{definition}
The \emph{conditional probability} of $X=x$  given $Y=y$ with $P_Y(y)>0$ is 
\begin{align}
\nonumber  P_{X|Y=y}(x)&=\frac{P_{XY}(x,y)}{P_Y(y)}\ .
\end{align}
The \emph{conditional probability distribution} $P_{X|Y}$ is 
\begin{align}
\nonumber  P_{X|Y}(x,y)&= P_{X|Y=y}(x)\ .
\end{align}
\end{definition}

A conditional probability distribution can be seen as a system taking as input the random variable $Y$ and giving a (probabilistic) output $X$, depending on the input $y$.

When considering the conditional probability distribution of several random variables, $P_{XY|UV}$, the marginal conditional distribution $P_{X|UV}$ can be defined in the same way as the marginal distribution. Furthermore, if it holds that $P_{X|UV}(x,u,v)=P_{X|UV}(x,u,v^{\prime})$ for all $v,v^{\prime}\in V$,\footnote{This is in particular the case when considering \emph{non-signalling systems}, which will be defined in Section~\ref{subsec:nssystem}.} then we drop the second conditional random variable in the notation and simply write $P_{X|U}$, where
\begin{align}
\nonumber P_{X|U}(x,u)&=\sum_y P_{XY|UV}(x,y,u,v)\ .
\end{align}

We also define the \emph{expectation value} of a random variable $X$.
\begin{definition}
Let $X$ be a random variable with distribution $P_X$. The \emph{expectation value} of $X$ is 
\begin{align}
\nonumber \langle X \rangle &=\sum_{x\in \mathcal{X}}P_X(x)\cdot x\ .
\end{align}
\end{definition}

A special probability distribution is the \emph{uniform} distribution, i.e., the one where all possible outcomes are equally likely. 
\begin{definition}
Let $U$ be a random variable of range $\mathcal{U}$. The \emph{uniform distribution} over $\mathcal{U}$ is
\begin{align}
\nonumber P_U(u) &= \frac{1}{|\mathcal{U}|}\ .
\end{align}
\end{definition}
We will often use the letter $U$ (for `uniform') to denote a random variable which is uniformly distributed. 

\subsubsection{Distance between distributions}

The distance between two distributions of the same random variable can be measured by the \emph{variational distance}. The variational distance is exactly the minimal probability that the random variable drawn from one or the other distribution takes a different value. 
\begin{definition}
Let $P$ and $Q$ be distributions 
over $\mathcal{X}$. The \emph{variational distance} between $P$ and $Q$ is
\begin{align}
\nonumber d(P,Q)&=\frac{1}{2} \sum_{x\in \mathcal{X}}\left|P(x)-Q(x) \right|\ .
\end{align}
\end{definition}
Of particular importance for us is the distance of a distribution $P_X$ from the uniform one. We denote this distance by $d(X)$. 
\begin{definition}\label{def:distancefromuniform}
The \emph{distance from uniform} of a random variable $X$ over $\mathcal{X}$ with distribution $P_X$ is the variational distance between $P_X$ and the uniform distribution over  $\mathcal{X}$, i.e.,
\begin{align}
\nonumber d(X)&=\frac{1}{2} \sum_{x\in \mathcal{X}}\left|P_X(x)-\frac{1}{|\mathcal{X}|}\right|\ .
\end{align}
\end{definition}
Note that when $X$ is a bit, i.e., $\mathcal{X}=\{0,1\}$, then 
\begin{align}
\nonumber d(X)&=\frac{1}{2} \sum_{x=0}^1\left|P_X(x)-\frac{1}{2}\right|
=
\left|P_X(0)-\frac{1}{2}\right|
=
\frac{1}{2}\cdot \left|P_X(0)-P_X(1)\right|\ .
\end{align}

\begin{figure}[h]
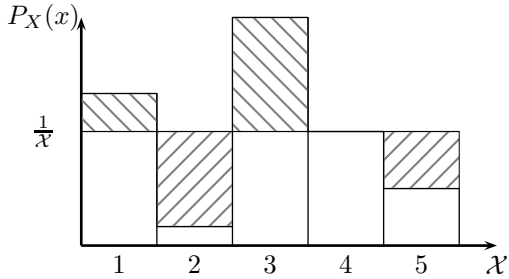

\centering
\pspicture*[](-1,-1)(6,3.5)
\psline[linewidth=1pt]{->}(0,0)(5.5,0)
\psline[linewidth=1pt]{->}(0,0)(0,3)
\rput[c]{0}(-0.5,3){{$P_X(x)$}}
\rput[c]{0}(-0.5,1.5){{$\frac{1}{\mathcal{X}}$}}
\rput[c]{0}(5.5,-0.25){{$\mathcal{X}$}}
\rput[c]{0}(0.5,-0.25){{$1$}}
\rput[c]{0}(1.5,-0.25){{$2$}}
\rput[c]{0}(2.5,-0.25){{$3$}}
\rput[c]{0}(3.5,-0.25){{$4$}}
\rput[c]{0}(4.5,-0.25){{$5$}}
\psline[linewidth=0.5pt,linecolor=gray]{-}(0,1.5)(5,1.5)
\psline[linewidth=0.5pt]{-}(0,2)(1,2)
\psline[linewidth=0.5pt]{-}(1,0)(1,2)
\psline[linewidth=0.5pt]{-}(1,0.25)(2,0.25)
\psline[linewidth=0.5pt]{-}(2,3)(3,3)
\psline[linewidth=0.5pt]{-}(2,0)(2,3)
\psline[linewidth=0.5pt]{-}(3,0)(3,3)
\psline[linewidth=0.5pt]{-}(3,1.5)(4,1.5)
\psline[linewidth=0.5pt]{-}(4,0)(4,1.5)
\psline[linewidth=0.5pt]{-}(4,0.75)(5,0.75)
\psline[linewidth=0.5pt]{-}(5,0)(5,0.75)
\pspolygon[linewidth=0pt,fillstyle=vlines,hatchcolor=gray](0,1.5)(1,1.5)(1,2)(0,2)
\pspolygon[linewidth=0pt,fillstyle=vlines,hatchcolor=gray](2,1.5)(3,1.5)(3,3)(2,3)
\pspolygon[linewidth=0pt,fillstyle=hlines,hatchcolor=gray](1,1.5)(2,1.5)(2,0.25)(1,0.25)
\pspolygon[linewidth=0pt,fillstyle=hlines,hatchcolor=gray](4,1.5)(5,1.5)(5,0.75)(4,0.75)
\endpspicture
\caption{The distance from the uniform distribution is either of the two shaded areas.}
\end{figure}

\subsubsection{Chernoff bounds and sampling}

\begin{lemma}[Chernoff~\cite{chernoff}, Hoeffding~\cite{hoeffding}]\label{lemma:chernoff}
Let $X_1,\dotsc,X_n$ \linebreak[4] $\in \{0,1\}$ be $n$ independent random variables such that for each $i$, $X_i$ is drawn according to the distribution $P_X$ with $P_X(1)=p$. Then for any $\ep>0$ it holds that
\begin{align}
\nonumber \Pr\biggl[\frac{1}{n}\sum_i x_i \geq p+\ep \biggr] &\leq e^{-2n\ep^2}\ \text{ and}\\
\nonumber \Pr\biggl[\frac{1}{n}\sum_i x_i \leq p-\ep \biggr] &\leq e^{-2n\ep^2}\ .
\end{align}
\end{lemma}

The following bound on the sum of the binomial coefficients is well-known. 
\begin{lemma}
For any $0<p<1/2$, 
\begin{align}
\nonumber \sum_{i=0}^{\lfloor p\cdot n \rfloor}\binom{n}{i} &\leq  2^{h(p)\cdot n}  \ ,
\end{align}
where $h(p)=-p\log_2 p-(1-p)\log_2(1-p)$ is the binary entropy function. 
\end{lemma}

\begin{lemma}[Sampling Lemma~\cite{KoeRen04b}]\label{lemma:sampling}
Let $Z$ be an $n$-tuple and $Z^{\prime}$ a $k$-tuple of random variables over 
 $\mathcal{Z}$, with symmetric joint probability $P_{ZZ^{\prime}}$. Let $Q_{z^{\prime}}$ be the relative frequency distribution of a fixed sequence $z^{\prime}$ and $Q_{(z,z^{\prime})}$ be the relative frequency distribution of a sequence $(z,z^{\prime})$, drawn according to $P_{ZZ^{\prime}}$. Then for every $\ep\geq 0$ it holds that
\begin{align}
\nonumber  P_{ZZ^{\prime}}\left[\Vert Q_{(z,z^{\prime})}-Q_{z^{\prime}}\Vert\geq \ep\right]&\leq |\mathcal{Z}|\cdot e^{-k\ep^2/8 | \mathcal{Z}|}\ .
\end{align}
\end{lemma}

\section{Efficiency}\label{sec:eff}

Some cryptographic tasks cannot be achieved with perfect security. For these cases, we have to accept some probability of error, or even rely on computational hardness. 
\pagebreak[4]
   
\begin{definition}
Let $g\colon \mathbb{N}\rightarrow \mathbb{R}$ be a function. 
\begin{itemize}
\item The set of functions $f\colon \mathbb{N}\rightarrow \mathbb{R}$ upper-bounded by $g$ is called \emph{$O(g)$} (\emph{$O$-notation}) 
\begin{align}
\nonumber O(g)&= \lbrace f\colon \mathbb{N}\rightarrow \mathbb{R}| \exists c>0,\ n_0:\  f(n)\leq c\cdot g(n)\ \text{ for all}\ n>n_0\rbrace\ .
\end{align}
\item The set of functions $f\colon \mathbb{N}\rightarrow \mathbb{R}$ lower-bounded by $g$ is called \emph{$\Omega(g)$} (\emph{$\Omega$-notation}) 
\begin{align}
\nonumber \Omega(g)&= \lbrace f\colon \mathbb{N}\rightarrow \mathbb{R}| \exists c>0,\ n_0:\ f(n)\geq c\cdot g(n)\ \text{ for all}\ n>n_0\rbrace \ .
\end{align}
\end{itemize}
A function $f\colon \mathbb{N}\rightarrow \mathbb{R}$ is called \emph{polynomially upper-bounded} (or \emph{polynomial}) if there exists a constant $k\geq 0$ such that $f\in O(n^k)$. 
\end{definition}
In computational complexity, algorithms that run in time at most polynomial in the input size are called \emph{efficient}, and \emph{inefficient} otherwise. 
\begin{definition}
A function $f\colon \mathbb{N}\rightarrow \mathbb{R}$ is called \emph{negligible}  if for every positive polynomial $p(\cdot)$, there exists an $n_0$ such that for all $n>n_0$
\begin{align}
\nonumber f(n)&< \frac{1}{p(n)}\ .
\end{align}
\end{definition}
For example, in key distribution, we are interested in schemes where (ideally) the probability that an adversary succeeds in breaking it is negligible in some security parameter. On the other hand, the probability that the honest parties succeed in achieving their task should be high, for example \emph{overwhelming}, as defined below.  
\begin{definition}
A probability $p\colon \mathbb{N}\rightarrow \mathbb{R}$ is called \emph{overwhelming} if $1-p(n)$ is negligible. 
\end{definition}

Note that the definition of polynomial and negligible have the nice property that they are \emph{closed under composition}.  More precisely, if $f$ and $g$ are polynomial, then so are $f\circ g$, $f+g$, and $f\cdot g$; if $f$ and $g$ are negligible, then so is $f+g$; and even if $f$ is polynomial and $g$ is negligible, then $f\circ g$ and $f\cdot g$ are negligible.

\section{Random Systems}\label{sec:randomsystems}

Most cryptographic tasks can abstractly be modelled as \emph{random systems} \cite{Maurer02}. A \emph{system} is an object taking inputs and giving outputs. The way this system is physically implemented is often irrelevant in the cryptographic context, and we can consider the system to be defined in terms of its behaviour, i.e., the probabilities that it gives a certain output given a specific input. 
\begin{definition}
 An \emph{$(\mathcal{X},\mathcal{Y})$-random system} $\mathcal{S}$ is a sequence of conditional probability distributions $P^{\mathcal{S}}_{Y_i|X_i\dotso X_1Y_{i-1}\dotso Y_1}$ for $i\geq 1$. 
\end{definition}
Even though the sequence of probability distributions defining a system could potentially be infinite, we will only consider systems defined by finite sequences and with a finite number of inputs and outputs. 
Two random systems characterized by the same probability distributions are, with the above definition, defined to be the same system. 

The different interfaces, number of interactions, and, if there is, the time-wise ordering of these inputs and outputs is described in the definition of the system. 
\begin{figure}[h]
\centering
\pspicture*[](-4,0)(4,2.5)
\psset{unit=0.75cm}
\psline[linewidth=1pt]{->}(-2.2,3)(-2.2,2)
\psline[linewidth=1pt]{<-}(-1.8,3)(-1.8,2)
\psline[linewidth=1pt]{->}(0,3)(0,2)
\psline[linewidth=1pt]{<-}(1.8,3)(1.8,2)
\psline[linewidth=1pt]{->}(2.2,3)(2.2,2)
\rput[c]{0}(0,1.25){\huge{$\mathcal{S}$}}
\pspolygon[linewidth=2pt](-2.5,0.5)(2.5,0.5)(2.5,2)(-2.5,2)
\psline[linewidth=1pt]{->}(-3,0.1)(3,0.1)
\rput[b]{0}(3.2,0){\large{$t$}}
\endpspicture
\caption{A system.}
\end{figure}

\begin{example}
The identity channel can be seen as the system taking as input a value $x\in \mathcal{X}$ and outputting the value $y=x$, i.e., $P^{\mathcal{S}}_{Y|X}(x,y)=1$ for $y=x$ and $0$ otherwise. 
\end{example}

Note that any \emph{protocol} taking as input $X$ and calculating a certain value $Y$ can also be seen as a random system with input $X$ and output $Y$.

\subsection{Indistinguishability}

The closeness of two systems $\mathcal{S}_0$ and $\mathcal{S}_1$ can be measured by introducing a so-called \emph{distinguisher}. A distinguisher $\mathcal{D}$ is itself a system and it has the same interfaces as the system $\mathcal{S}_0$, with the only difference that wherever $\mathcal{S}_0$ takes an input, $\mathcal{D}$ gives an output and vice versa. In addition, $\mathcal{D}$ has an extra output. The distinguisher $\mathcal{D}$ has access to \emph{all} interfaces of $\mathcal{S}_0$, even though these interfaces might not be in the same location when the protocol is executed (for example, one of the interfaces might be the one seen by Alice, while the other is the one seen by Eve). 

\begin{definition}
 A \emph{distinguisher} $\mathcal{D}$ for an $(\mathcal{X},\mathcal{Y})$-random system is \linebreak[4] a~$(\mathcal{Y},\mathcal{X})$-random system defined by the distributions $P^{\mathcal{S}}_{X_i|X_{i-1}\dotso X_1Y_{i-1}\dotso Y_1}$ for $i\geq 1$ (i.e., it is one query ahead). Additionally, it outputs a bit $B$ after $q$ queries based on the transcript $(X_1\dotso X_qY_1\dotso Y_q)$. 
\end{definition}

\begin{figure}[h]
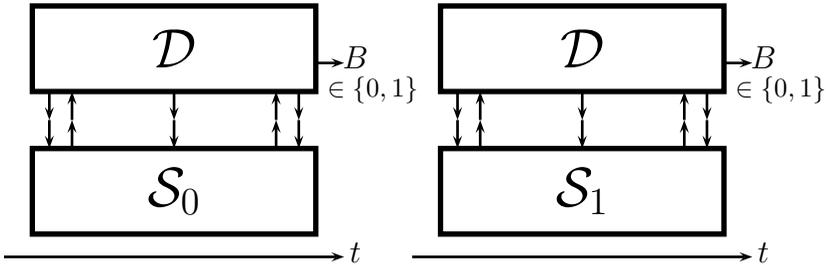

\centering
\pspicture*[](-5.1,0)(6,3.5)
\psset{unit=0.75cm}
\rput[b]{0}(-3.6,0){
\psline[linewidth=1pt]{<-}(-2.2,2.5)(-2.2,3)
\psline[linewidth=1pt]{->}(-1.8,2.5)(-1.8,3)
\psline[linewidth=1pt]{<-}(0,2.5)(0,3)
\psline[linewidth=1pt]{->}(1.8,2.5)(1.8,3)
\psline[linewidth=1pt]{<-}(2.2,2.5)(2.2,3)
\psline[linewidth=1pt]{->}(2.5,3.5)(3,3.5)
\rput[b]{0}(3.2,3.4){\large{$B$}}
\rput[b]{0}(3.5,2.8){{$\in\{0,1\}$}}
\rput[c]{0}(0,3.75){\huge{$\mathcal{D}$}}
\pspolygon[linewidth=2pt](-2.5,3)(2.5,3)(2.5,4.5)(-2.5,4.5)
\psline[linewidth=1pt]{->}(-2.2,2.5)(-2.2,2)
\psline[linewidth=1pt]{<-}(-1.8,2.5)(-1.8,2)
\psline[linewidth=1pt]{->}(0,2.5)(0,2)
\psline[linewidth=1pt]{<-}(1.8,2.5)(1.8,2)
\psline[linewidth=1pt]{->}(2.2,2.5)(2.2,2)
\rput[c]{0}(0,1.25){\huge{$\mathcal{S}_0$}}
\pspolygon[linewidth=2pt](-2.5,0.5)(2.5,0.5)(2.5,2)(-2.5,2)
\psline[linewidth=1pt]{->}(-3,0.1)(3,0.1)
\rput[b]{0}(3.2,0){\large{$t$}}
}
\rput[b]{0}(3.6,0){
\psline[linewidth=1pt]{<-}(-2.2,2.5)(-2.2,3)
\psline[linewidth=1pt]{->}(-1.8,2.5)(-1.8,3)
\psline[linewidth=1pt]{<-}(0,2.5)(0,3)
\psline[linewidth=1pt]{->}(1.8,2.5)(1.8,3)
\psline[linewidth=1pt]{<-}(2.2,2.5)(2.2,3)
\psline[linewidth=1pt]{->}(2.5,3.5)(3,3.5)
\rput[b]{0}(3.2,3.4){\large{$B$}}
\rput[b]{0}(3.5,2.8){{$\in\{0,1\}$}}
\rput[c]{0}(0,3.75){\huge{$\mathcal{D}$}}
\pspolygon[linewidth=2pt](-2.5,3)(2.5,3)(2.5,4.5)(-2.5,4.5)
\psline[linewidth=1pt]{->}(-2.2,2.5)(-2.2,2)
\psline[linewidth=1pt]{<-}(-1.8,2.5)(-1.8,2)
\psline[linewidth=1pt]{->}(0,2.5)(0,2)
\psline[linewidth=1pt]{<-}(1.8,2.5)(1.8,2)
\psline[linewidth=1pt]{->}(2.2,2.5)(2.2,2)
\rput[c]{0}(0,1.25){\huge{$\mathcal{S}_1$}}
\pspolygon[linewidth=2pt](-2.5,0.5)(2.5,0.5)(2.5,2)(-2.5,2)
\psline[linewidth=1pt]{->}(-3,0.1)(3,0.1)
\rput[b]{0}(3.2,0){\large{$t$}}
}
\endpspicture
\caption{\label{fig:distinguisher} The distinguisher}
\end{figure}
Now consider the following game:The distinguisher~$\mathcal{D}$ is given one out of two systems at random --- either 
$\mathcal{S}_0$ or $\mathcal{S}_1$ --- but the distinguisher does not know which one. It can interact with 
the system and then has to output a bit $B$, guessing which system it has interacted with. The \emph{distinguishing 
advantage between system $\mathcal{S}_0$ and $\mathcal{S}_1$} is the maximum  guessing advantage any distinguisher 
can have in this game (see Figure~\ref{fig:distinguisher}).
Equivalently, the distance between two systems can be defined as the maximum difference in probability that a distinguisher outputs the value $B=1$ given it has interacted with system $\mathcal{S}_0$ compared to when it has interacted with $\mathcal{S}_1$. 
\begin{definition}
The \emph{distinguishing advantage between two systems $\mathcal{S}_0$ and $\mathcal{S}_1$ }is 
\begin{align}
 \nonumber \delta(\mathcal{S}_0, \mathcal{S}_1)&= \max_{\mathcal{D}}[P(B=1|\mathcal{S}=\mathcal{S}_0)-P(B=1|\mathcal{S}=\mathcal{S}_1)].
\end{align}
Two systems $\mathcal{S}_0$ and $\mathcal{S}_1$ are called $\epsilon$-indistinguishable if $\delta(\mathcal{S}_0, \mathcal{S}_1)\leq \epsilon$.
\end{definition}

The probability of any event $\mathcal{E}$ when the distinguisher $\mathcal{D}$ is interacting with $\mathcal{S}_0$ or $\mathcal{S}_1$ cannot differ by more than this quantity. 
\begin{lemma}\label{lemma:event}
Let  $\mathcal{S}_0$ and $\mathcal{S}_1$ be two $\epsilon$-indistinguishable systems.
Denote by $\Pr[\mathcal{E}|\mathcal{S}_0,\mathcal{D}]$ the probability of an event $\mathcal{E}$, defined by any 
of the input and output variables, given the distinguisher~$\mathcal{D}$ is interacting with the system $\mathcal{S}_0$. Then
\begin{align}
 \nonumber \Pr[\mathcal{E}|\mathcal{S}_0,\mathcal{D}] &\leq  \Pr[\mathcal{E}|\mathcal{S}_1,\mathcal{D}]+ \epsilon
\end{align}
\end{lemma}
\begin{proof}
Assume $\Pr[\mathcal{E}|\mathcal{S}_0,\mathcal{D}]> \Pr[\mathcal{E}|\mathcal{S}_1,\mathcal{D}]+ \epsilon$ and define 
the distinguisher $\mathcal{D}$ such that it outputs $B=0$ whenever the event $\mathcal{E}$ has happened and whenever 
$\mathcal{E}$ has not happened it outputs $B=1$. Then this distinguisher reaches a distinguishing advantage of $\delta(\mathcal{S}_0, \mathcal{S}_1)>\epsilon$ contradicting the assumption that the two systems are $\epsilon$-indistinguishable.
\end{proof}

The distinguishing advantage is a \emph{pseudo-metric}, that is, it fulfils similar properties as a metric, in particular, the triangle inequality.\footnote{Since we identify the system with the probability distributions describing it, the distinguishing advantage is actually a metric, i.e., for any two systems with distance $0$, the two systems are the same. In general, it is possible to introduce the distinguishing advantage restricting the set of distinguishers to a certain class, for example, the ones which are computationally efficient. In this case, the weaker properties of a pseudo-metric remain fulfilled.} 
\begin{lemma}\label{lemma:distance} 
The distinguishing advantage fulfils 
\begin{itemize}
\item $\delta(\mathcal{S},\mathcal{S})=0$\ ,
\item $\delta(\mathcal{S}_0,\mathcal{S}_1)=\delta(\mathcal{S}_1,\mathcal{S}_0)$\ ,\ \text{and}
\item $\delta(\mathcal{S}_0,\mathcal{S}_1)+\delta(\mathcal{S}_1,\mathcal{S}_2)\geq \delta(\mathcal{S}_0,\mathcal{S}_2)$\ .
\end{itemize}
\end{lemma}
\begin{proof}
\begin{align}
 \nonumber \delta(\mathcal{S}_0, \mathcal{S}_1)&= \max_{\mathcal{D}}[P(B=1|\mathcal{S}=\mathcal{S}_0)-P(B=1|\mathcal{S}=\mathcal{S}_0)]= \max_{\mathcal{D}}[0]=0\ .
\end{align}
For the second equality, call the distinguisher that reaches the maximal value on the right-hand side (i.e., $\delta(\mathcal{S}_0,\mathcal{S}_1)$) $\mathcal{D}_0$.  
Define another distinguisher $\mathcal{D}_1$ to be the same as $\mathcal{D}_0$, but flipping the bit $B$ before outputting it. This implies
\begin{align}
 \nonumber \delta(\mathcal{S}_0, \mathcal{S}_1)
  &= [P^{\mathcal{D}_0}(B=1|\mathcal{S}=\mathcal{S}_0)-P^{\mathcal{D}_0}(B=1|\mathcal{S}=\mathcal{S}_1)]\\
\nonumber &=
[1-P^{\mathcal{D}_1}(B=1|\mathcal{S}=\mathcal{S}_0)-(1-P^{\mathcal{D}_1}(B=1|\mathcal{S}=\mathcal{S}_1))]
\\
\nonumber &= 
 [P^{\mathcal{D}_1}(B=1|\mathcal{S}=\mathcal{S}_1)-P^{\mathcal{D}_1}(B=1|\mathcal{S}=\mathcal{S}_0)]\\
 \nonumber &\leq  
  \delta(\mathcal{S}_1,\mathcal{S}_0)\ .
\end{align}
The inverse inequality follows by the same argument with the roles of $\mathcal{S}_0$ and $\mathcal{S}_1$ exchanged. Finally note that
\begin{align}
 \nonumber \delta(\mathcal{S}_0, \mathcal{S}_1)+\delta(\mathcal{S}_1, \mathcal{S}_2)
 &= \max_{\mathcal{D}}[P(B=1|\mathcal{S}=\mathcal{S}_0)-P(B=1|\mathcal{S}=\mathcal{S}_1)]\\
 \nonumber &\quad +\max_{\mathcal{D}}[P(B=1|\mathcal{S}=\mathcal{S}_1)-P(B=1|\mathcal{S}=\mathcal{S}_2)]\\
 \nonumber &\geq  
 \max_{\mathcal{D}}[P(B=1|\mathcal{S}=\mathcal{S}_0)-P(B=1|\mathcal{S}=\mathcal{S}_1)\\
 \nonumber &\quad +
P(B=1|\mathcal{S}=\mathcal{S}_1)-P(B=1|\mathcal{S}=\mathcal{S}_2)]\\
 \nonumber &=
  \delta(\mathcal{S}_0,\mathcal{S}_2)\ .\qedhere
\end{align}
\end{proof}

\subsection{Security of a key}\label{subsec:securitykey}

The security of a cryptographic primitive can be measured by the distance of this system from 
an \emph{ideal} system, which is secure by definition \linebreak[4] \cite{MaReWo07}. For example, in the case of key distribution 
the ideal system is the one which outputs a uniform and random key (bit string) to the honest parties and for which all other 
input/output interfaces are completely independent of this first interface. This key is secure by construction. 
If the \emph{real} key-distribution protocol is $\epsilon$-indistinguishable from the ideal one, then, by Lemma~\ref{lemma:event}, the 
key obtained from the real system needs to be secure except with probability $\epsilon$. 
This is true because in the ideal case the adversary knows nothing about the key.  

\begin{figure}[h]
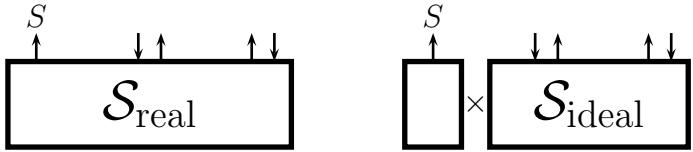

\centering
\pspicture*[](-5,0)(5,3)
\psset{unit=0.75cm}
\rput[b]{0}(3.5,0){
\psline[linewidth=1pt]{<-}(-2.0,2.5)(-2.0,2)
\pspolygon[linewidth=2pt](-2.5,0.5)(-1.5,0.5)(-1.5,2)(-2.5,2)
\rput[b]{0}(-2,2.6){\large{$S$}}
\rput[c]{0}(-1.25,1.25){\large{$\times$}}
\psline[linewidth=1pt]{->}(-0.2,2.5)(-0.2,2)
\psline[linewidth=1pt]{<-}(0.2,2.5)(0.2,2)
\psline[linewidth=1pt]{<-}(1.8,2.5)(1.8,2)
\psline[linewidth=1pt]{->}(2.2,2.5)(2.2,2)
\rput[c]{0}(0.75,1.25){\huge{$\mathcal{S}_{\mathrm{ideal}}$}}
\pspolygon[linewidth=2pt](-1,0.5)(2.5,0.5)(2.5,2)(-1,2)
}
\rput[b]{0}(-3.5,0){
\psline[linewidth=1pt]{<-}(-2.0,2.5)(-2.0,2)
\rput[b]{0}(-2,2.6){\large{$S$}}
\psline[linewidth=1pt]{->}(-0.2,2.5)(-0.2,2)
\psline[linewidth=1pt]{<-}(0.2,2.5)(0.2,2)
\psline[linewidth=1pt]{<-}(1.8,2.5)(1.8,2)
\psline[linewidth=1pt]{->}(2.2,2.5)(2.2,2)
\rput[c]{0}(0,1.25){\huge{$\mathcal{S}_{\mathrm{real}}$}}
\pspolygon[linewidth=2pt](-2.5,0.5)(2.5,0.5)(2.5,2)(-2.5,2)
}
\endpspicture
\caption{The real and ideal system for the case of key distribution.}
\end{figure}

\begin{definition}
A \emph{perfect key} of length $|\mathcal{S}|$ is a system which outputs two equal uniform random variables $S_A$ and $S_B$ (i.e., $P_{S_AS_B}(s_A,s_B)=1/|\mathcal{S}|$ for $s_A=s_B$ and $0$ otherwise) and for which all other interfaces are uncorrelated with $S_A$ and $S_B$. 
\end{definition}

\begin{definition}
A key 
 is \emph{$\epsilon$-secure} if the system outputting 
$S_A$ and $S_B$  
 is $\epsilon$-indistinguishable from 
a perfect key.  
\end{definition}

This definition implies that the resulting security is \emph{universally composable}~\cite{pw,bpw,canetti}, i.e., no matter in which context the key is used it always behaves like a perfect key, except with probability at \linebreak[4] most~$\epsilon$. In fact, 
assume by contradiction that there exists any way of using the key (or any other part of the system which generates the key) such that 
the result is insecure, i.e., distinguishable with probability larger than~$\epsilon$ from the ideal system. 
This process could be used to distinguish the key-generation scheme from an ideal one with probability larger 
than $\epsilon$, which is impossible by definition.

Often, the analysis of the security of a key is subdivided into several parts because the different properties are achieved by different sub-protocols. For example, the bound on the information an eavesdropper can have about Alice's key $S_A$ is called the \emph{secrecy} of the protocol. Secrecy is usually achieved by \emph{privacy amplification}. The probability that Alice's and Bob's key differ can then be considered separately; this is called the \emph{correctness} of the protocol. The part of the protocol responsible for correctness is \emph{information reconciliation}. By the triangle inequality, the security of the protocol is bounded by the sum of the secrecy and correctness. 

Note that the above requirements of secrecy and correctness do not exclude a trivial protocol: one that always outputs a key of zero length. Such a protocol is, of course, not useful (although it \emph{is} secure). The property that the protocol should output a key (of non-zero length) when the eavesdropper is passive is called the \emph{robustness} of the protocol. 

When key agreement is studied in an asymptotic scenario, where the number of quantum system, channel uses, random variables etc.\ used can be arbitrarily large, we are interested in the length of the key that can (asymptotically) be achieved per number of systems. 
\begin{definition}
The \emph{rate} $q$ of a key-distribution protocol is the length of the key per number of systems, i.e., $\log |S|=q\cdot n$. 
\end{definition}
Of course, we will be interested in protocols which are secure and output a certain key length when the adversary is passive. The secret key rate is then defined as the key length that can be generated when the channel is noisy according to a certain noise model.

\section{Convex Optimization}\label{sec:optimization}
\subsection{Linear programming}\label{subsec:lin_prog}

A \emph{linear program} (see, e.g.,~\cite{bv}) is an optimization problem with a linear objective function and linear inequality (and equality) constraints, i.e., it can be expressed as 
\begin{align}
\label{eq:lpprimal} \max: &\quad b^T\cdot x\\
\nonumber  \st &\quad A\cdot x\leq c
 \ ,
\end{align}
where $x$, $b$, and $c$ are real vectors, $A$ is a real matrix, and
$x$ is the variable we want to optimize. The inequality is meant to be the component-wise inequalities of the entries.  An $x$ which fulfils the constraints is called \emph{feasible}. The set of feasible $x$ is convex, more precisely, a convex polytope, i.e., a convex set with a finite number of extremal points (vertices). A feasible $x$ which maximizes the \emph{objective function} $b^T x$ is called \emph{optimal solution} and is denoted by $x^*$. The value of $b^T x^*$, i.e., the maximal value of the objective function for a feasible $x$, is called \emph{optimal value} and denoted by $q^*$. The program is called \emph{feasible}, if there exists a feasible $x$. If this is the case and the optimal value is finite, there is always a vertex of the polytope defined by the constraints at which the optimal value is attained.

Any linear program can be brought in the form given above (\ref{eq:lpprimal}), i.e., there exists a problem of the above form that is equivalent to the original optimization problem. For example, if the objective function should be minimized instead of maximized, this is equivalent to maximizing the objective function and replacing $b$ by $-b$. In the same way, constraints of the form $a x\geq c$ can be brought into the above form by multiplying them with $-1$ and equality constraints $a x= c$ can be replaced by the two constraints $a x\leq  c$ and $-a x\leq  -c$. On the other hand, an inequality constraint $a x\leq c$ can be replaced by an equality and an inequality constraint by introducing a so-called \emph{slack variable} $s$ and writing $a x+s= c$ and  $s \geq 0$.

An important feature of linear programming is duality: The linear program (\ref{eq:lpprimal}) is called the \emph{primal} problem. From this linear program, another linear program can be derived, defined by 

\begin{align}
 \label{eq:lpdual} \min: &\quad c^T\cdot \lambda\\
\nonumber  \st &\quad A^T\cdot \lambda = b\\*
\nonumber &\quad \lambda \geq 0 \ .
\end{align}
This problem is called the \emph{dual}, its optimal solution is denoted by $\lambda^*$ and its optimal value by $d^*=c^T\lambda^*$. The weak duality theorem states, that the value of the primal objective function for every feasible $x$ is smaller or equal to the value of the dual objective function for every feasible $\lambda$. The strong duality theorem states that the two optimal values are equal, i.e.,~$q^*=d^*$. 
\begin{theorem}[Strong duality for linear programming]\label{th:strongdualitylp}
Consider a linear program, defined by $A$, $b$, and $c$, and assume that either the primal or dual is feasible. Then $q^*=b^Tx^*=c^T\lambda^*=d^*$.
\end{theorem}

It is therefore possible to solve a linear program either by solving the linear program (\ref{eq:lpprimal}) itself, or by solving its dual (\ref{eq:lpdual}). 

\begin{figure}[h]
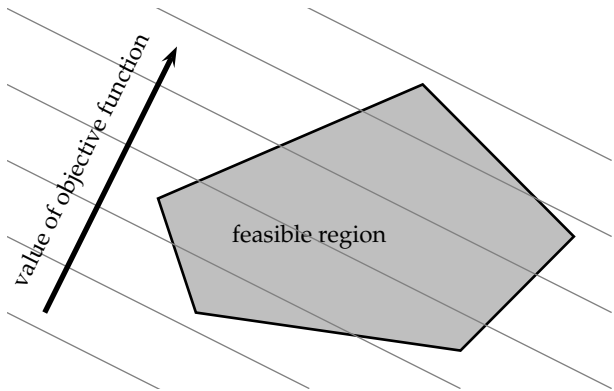

\centering
\pspicture*[](-1,-1)(7,4)
\pspolygon[linewidth=1pt,fillstyle=solid,fillcolor=lightgray](5,-0.5)(6.5,1)(4.5,3)(1,1.5)(1.5,0)
\rput[c]{0}(3,1){\small{feasible region}}
\psline[linewidth=0.5pt, linecolor=gray]{-}(-1,6)(7,2)
\psline[linewidth=0.5pt, linecolor=gray]{-}(-1,5)(7,1)
\psline[linewidth=0.5pt, linecolor=gray]{-}(-1,4)(7,0)
\psline[linewidth=0.5pt, linecolor=gray]{-}(-1,3)(7,-1)
\psline[linewidth=0.5pt, linecolor=gray]{-}(-1,2)(5,-1)
\psline[linewidth=0.5pt, linecolor=gray]{-}(-1,1)(3,-1)
\psline[linewidth=0.5pt, linecolor=gray]{-}(-1,0)(3,-2)
\psline[linewidth=2pt]{->}(-0.5,0)(1.25,3.5)
\rput[c]{63.4349488}(0,2){\small{value of objective function}}
\endpspicture
\caption{A linear programming problem.}
\end{figure}

\subsection{Conic programming}

The notion of linear programming can be generalized 
to \emph{conic programming}~\cite{btn}. 
In linear programming, the constraints are of the form $A x\leq c$, where $A x\leq c$ means that every entry of the vector $A x$ must be smaller or equal the corresponding entry of the vector $c$. The relation `$\leq$', therefore, defines a partial order on the set of vectors in $\mathbb{R}^n$. Many of the properties of linear programming follow from properties of this partial ordering  `$\leq$', namely that it is reflexive, anti-symmetric, transitive, and compatible with linear operations (homogeneous and additive).  
However, other ordering relations also have these properties. In fact, it turns out that an ordering relation with the above properties (which we denote by `$\preceq$') is completely defined by its non-negative elements. Furthermore, the non-negative elements must form a pointed convex cone. 
\begin{definition}
A set $K$ of elements of a Euclidean space $E$, i.e., a real inner product space, is called a \emph{pointed convex cone} if
\begin{itemize}
 \item $K$ is non-empty and closed under addition: $a,a^{\prime}\in K$ $\rightarrow a+a^{\prime}\in K$~.
 \item $K$ is a conic set: $a\in K,\lambda\geq 0$ $\rightarrow \lambda a\in K$~.
 \item $K$ is pointed: $a\in K$ and $-a\in K$ $\rightarrow a=0$~.
\end{itemize}
\end{definition}
A pointed convex cone in a Euclidean space $E$ induces a partial ordering `$\preceq_K$' by defining
\begin{align}
\nonumber a \preceq_K b \quad &\leftrightarrow\quad  b-a\in K\ .
\end{align}
This ordering relation has the 
properties described above. 
A \emph{conic program} is then defined as the optimization problem
\begin{align}
\label{eq:coneprimal} \max : &\quad b^T\cdot x\\
\nonumber  \st &\quad A\cdot x\preceq_K c
 \ ,
\end{align}
where $K$ is a cone in a Euclidean space $E$, and $A$ is a linear map from $\mathbb{R}^n$ to $E$. 

The fact that the constraints are defined by a cone implies, 
for example, that the feasible region is convex (unlike in the linear programming case it does, however, not need to be a polytope); the optimization problem, therefore, does not have any local optima.

\begin{figure}
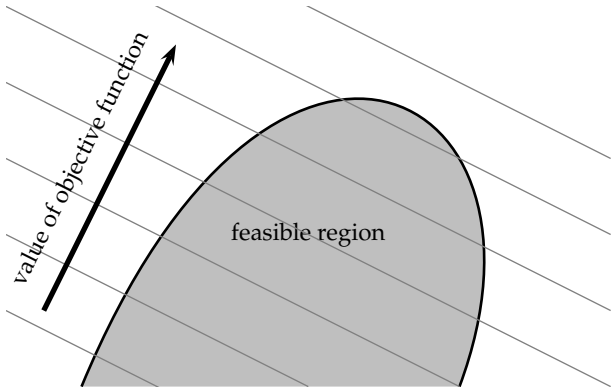

\centering
\pspicture*[](-1,-1)(7,4)
\psbezier[linewidth=1pt,fillstyle=solid,fillcolor=lightgray]{-}(5,-1)(6.5,3)(2.5,5)(0,-1)
\rput[c]{0}(3,1){\small{feasible region}}
\psline[linewidth=0.5pt, linecolor=gray]{-}(-1,6)(7,2)
\psline[linewidth=0.5pt, linecolor=gray]{-}(-1,5)(7,1)
\psline[linewidth=0.5pt, linecolor=gray]{-}(-1,4)(7,0)
\psline[linewidth=0.5pt, linecolor=gray]{-}(-1,3)(7,-1)
\psline[linewidth=0.5pt, linecolor=gray]{-}(-1,2)(5,-1)
\psline[linewidth=0.5pt, linecolor=gray]{-}(-1,1)(3,-1)
\psline[linewidth=0.5pt, linecolor=gray]{-}(-1,0)(3,-2)
\psline[linewidth=2pt]{->}(-0.5,0)(1.25,3.5)
\rput[c]{63.4349488}(0,2){\small{value of objective function}}
\endpspicture
\caption{A convex optimization problem.}
\end{figure}
 
\begin{example}
A linear program can be interpreted as the conic program where the Euclidean space is $\mathbb{R}^m$ and the cone $K$ is $\mathbb{R}^m_+$, the non-negative orthant of $\mathbb{R}^m$, i.e., 
\begin{align}
\nonumber \mathbb{R}^m_+ &=\{a=(a_1,\dotsc,a_m)^T\in \mathbb{R}^m| a_i\geq 0, i=1,\dotsc, m \}\ .
\end{align}
\end{example}
\begin{example}
A semi-definite program corresponds to the case where \linebreak[4] $E=S^m$, the space of $m\times m$ symmetric matrices with the inner product $\langle A,B\rangle=\tr (A B)=\sum_{i,j}A_{ij}B_{ij}$.\footnote{Note that this inner product transforms to the usual inner product between two vectors if the matrices $A$ and $B$ are transformed into vectors by `stacking the columns on top of each other'. In the context of semi-definite programming we will often use matrices and the vectors which can be obtained from them interchangeably.}  The cone $K$ is the set of symmetric matrices which are positive semi-definite, i.e.,
\begin{align}
\nonumber S^m_+ &=\{A\in S^m| x^T A x\geq 0\ \text{for all}\ x\in \mathbb{R}^m\}\ .
\end{align}
\end{example}

The \emph{dual} of the above conic program (\ref{eq:coneprimal}) is 
\begin{align}
 \label{eq:conedual} \min : &\quad \langle c , \lambda \rangle \\
\nonumber  \st &\quad A^T\cdot \lambda = b\\
\nonumber &\quad \lambda \succeq_{K^*} 0 \ ,
\end{align}
where $K^*$ is the \emph{dual cone} of $K$ (see Definition~\ref{def:dualcone} below) and $\langle c,\lambda \rangle$ denotes the inner product of $c$ and $\lambda$.  
\begin{definition}\label{def:dualcone}
Let $K$ be a pointed convex cone in a Euclidian space $E$. The \emph{dual cone} $K^*$ of $K$ is
\begin{align}
\nonumber K^* &=\{\lambda \in E| \langle \lambda , a \rangle \geq 0\ \text{for all}\ a\in K \}\ .
\end{align}
\end{definition}
The dual program gives an upper bound on the value of the primal program, i.e., the value of any feasible primal solution is always lower or equal the value of any feasible dual solution. 
\begin{theorem}[Weak duality for conic programming]
Consider a conic program, defined by $A$, $b$, and $c$, and a cone $K$.  Then $q^*=b^T x^* \leq \langle c,\lambda^*  \rangle=d^*$.
\end{theorem}
Unlike in the linear programming case, there exist special cases of conic programs where the optimal value of the primal and dual program are different, i.e., there exists a so-called \emph{duality gap}. Often, it can, however, be shown that the two values are indeed equal. This is, for example, the case when there exists a \emph{strictly feasible} solution of the primal or dual problem, i.e., there exists an $x$ such that $Ax \prec_K c$, where `$\prec$' denotes the fact that $Ax$ lies in the interior of the cone.

\section{Quantum Physics}\label{sec:quantum}

We first give the postulates of quantum mechanics and then the necessary definitions. 
For a more detailed introduction to quantum mechanics, we refer to~\cite{feynman3,CohenTannoudji1978a}, and for more information about quantum information to~\cite{nielsenchuang}. 

\subsubsection{Postulates of quantum physics}

\begin{enumerate}
\item The (pure) \emph{state} of a system is represented by a vector $\ket{\psi}$, element of a Hilbert space $\mathcal{H}$. For all $c\neq 0 \in \mathbb{C}$, $\ket{\psi}$ and $c\ket{\psi}$ represent the same state.\footnote{Alternatively, we could choose to normalize the vectors such that for any non-zero vector $\ket{\psi}$ it holds $\| \ket{\psi} \|=1$.} 
\item An \emph{observable} $\mathcal{A}$ is represented by a self-adjoint linear operator $A$ on $\mathcal{H}$, i.e., 
$A=A^\dagger$. 
\item The result of a measurement of the observable $\mathcal{A}$ is a real number $a$ that is an eigenvalue of $A$.
\item If a system is in state $\ket{\psi}$, then 
the probability to obtain $a$ when the observable $\mathcal{A}$ is measured, is 
\begin{align}
\nonumber \Pr_{\psi}[\text{measurement of }\mathcal{A}=a] &= \frac{\brakket{\psi}{P_a}{\psi}}{\braket{\psi}{\psi}}\ ,
\end{align}
where $P_a$ is the projector onto the subspace spanned by the eigenvectors of $A$ with associated eigenvalue $a$. 
The expectation value is 
$\langle A\rangle={\brakket{\psi}{A}{\psi}}/{\braket{\psi}{\psi}}$. 
\item If the system is in state $\ket{\psi}$, then immediately after the measurement of $\mathcal{A}$ having given result $a$, the system is in the state $\ket{\phi}$, where
\begin{align}
\nonumber \ket{\phi} &= P_a\ket{\psi}\ ,
\end{align}
and $\ket{\phi}$ is an eigenvector of $A$ associated with the eigenvalue $a$. 
\item The temporal evolution of an isolated system is\footnote{Sometimes this postulate is stated as the requirement that the evolution of the system is described by a \emph{unitary} operator, i.e, $\ket{\psi}^{\prime}=U\ket{\psi}$ in the Schr\"odinger picture.}
\begin{align}
\nonumber \langle A\rangle(t) &=\brakket{\psi_t}{A}{\psi_t}=\brakket{\psi_0}{A_t}{\psi_0}\ .
\end{align}
\begin{itemize}
\item In the Schr\"odinger picture 
\begin{align}
\nonumber \imath \hbar \frac{d}{dt}\ket{\psi_t} &=H\ket{\psi_t}\ \text{, i.e.,}\ \ket{\psi_t}=e^{-\imath Ht/\hbar}\ket{\psi_0}\ .
\end{align}
\item In the Heisenberg picture 
\begin{align}
\nonumber \frac{d}{dt}A_t=\frac{\imath}{\hbar}[H,A_t]\ \text{, i.e.,}\ A_t &=e^{\imath Ht/\hbar}A e^{-\imath Ht/\hbar}\ .
\end{align}
with $H=H^\dagger$. 
\end{itemize}
\end{enumerate}

In the above, we have used the \emph{Dirac notation}, i.e., elements of the Hilbert space are denoted by $\ket{\psi}$, called \emph{ket}, while elements of the dual space are denoted by $\bra{\psi}$, called \emph{bra}. The \emph{bracket} $\braket{\phi}{\psi}$ is the scalar product of an element $\ket{\psi}$ of $\mathcal{H}$ with the element of the dual $\bra{\phi}$. And $\brakket{\phi}{A}{\psi}$, where $A$ is a self-adjoint linear operator, can be seen equivalently as the case where the vector is $\ket{A \psi}$ or where the dual vector is $ \bra{ A\phi}$. $\langle A \rangle$ stands for the expectation value and $[A,B]$ denotes the \emph{commutator} of two operators $A$ and $B$, i.e., $[A,B]=AB-BA$.

\subsubsection{Definitions and properties}

Most of the following definitions and properties can be found in books 
on functional analysis, such as~\cite{reedsimon}.  

\begin{definition}
A \emph{Hilbert space} $\mathcal{H}$ is a complex vector space, i.e., 
\begin{align}
\nonumber \ket{\psi},\ket{\phi}\in \mathcal{H}\ \text{and}\ \lambda_1,\lambda_2\in \mathbb{C} \quad &\rightarrow \quad
\lambda_1 \ket{\psi}+\lambda_2\ket{\phi} \in \mathcal{H}
\end{align}
 with a positive Hermitian sesquilinear form, i.e., for all $\ket{\psi},\ket{\phi}\in \mathcal{H}$, there exists $\braket{\phi}{\psi}\in \mathbb{C}$  such that
\begin{enumerate}
\item it is linear in $\ket{\psi}$: $\braket{\phi}{\lambda_1\psi_1+\lambda_2 \psi_2}= \lambda_1 \braket{\phi}{\psi_1}+\lambda_2 \braket{\phi}{\psi_2}$\ ,
\item $\overline{\braket{\phi}{\psi}}=\braket{\psi}{\phi}$, where the bar  denotes the complex conjugate ,
\item for all $\ket{\psi} \in \mathcal{H}$ $\braket{\psi}{\psi}\geq 0$ and $\braket{\psi}{\psi}= 0 \leftrightarrow \ket{\psi}=0$, the \emph{norm} of a vector $\ket{\psi}$ is defined as $\| \ket{ \psi}\|=\sqrt{\braket{\psi}{\psi}}$\ .
\end{enumerate}
Furthermore, $\mathcal{H}$ is complete, i.e., for all $\ket{\psi_n}\in \mathcal{H}$ with $n=1,2,3,\dotsc$ such that $\lim_{n,m\rightarrow \infty} \|\ket{\psi_n}-\ket{\psi_m}\|=0$, there exists a $\ket{\psi} \in \mathcal{H}$ such that $\lim_{n\rightarrow \infty} \|\ket{\psi_n}-\ket{\psi}\|=0$, i.e., $\lim_{n\rightarrow \infty} \ket{\psi_n}= \ket{\psi}$. 
\end{definition}

\begin{definition}
Let $\mathcal{H}$ be a Hilbert space. The \emph{dual} of $\mathcal{H}$, denoted by $\mathcal{H}^*=\{\omega\}$,  
is the complex vector space of linear forms  
on $\mathcal{H}$, i.e., for all $\omega \in \mathcal{H}^*$
\begin{align}
\nonumber \omega\colon \mathcal{H} &\rightarrow  \mathbb{C}\\
\nonumber \ket{\psi} &\mapsto  \omega [\psi]
\end{align}
with $ \omega [\lambda_1\psi_1+\lambda_2 \psi_2]=\lambda_1\omega[\psi_1]+\lambda_2 \omega[\psi_2]$. 
\end{definition}

With every $\ket{\phi} \in \mathcal{H}$, it is possible to associate an element of the dual $\omega_{\phi}\in \mathcal{H}^*$ via the relation 
\begin{align}
\nonumber \omega_{\phi}\colon \ket{\psi} & \mapsto  \omega_{\phi}[\psi]=\braket{\phi}{\psi}
\end{align}
with $\omega_{\lambda \phi}=\bar{\lambda}\omega_{\phi}$. 
And for every element $\omega$ of $\mathcal{H}^*$, there also exists an element $\ket{\phi}$ of $\mathcal{H}$ such that 
\begin{align}
\nonumber \omega[\psi] &=\braket{\phi}{\psi}\ .
\end{align}

\begin{definition}
An \emph{orthonormal basis} of $\mathcal{H}$ is a set of vectors $\{ \ket{\phi_i}\}_{i\in I}$ such that
\begin{itemize}
\item $\braket{\phi_i}{\phi_j}= \delta_{ij}$\ {for all}\ $i,j\in I$\ { and}
\item $\braket{\psi}{\phi_i}=0$\ {for all}\ $i\in I$ $\rightarrow  \psi=0$\ .
\end{itemize}
\end{definition}
Every Hilbert space has an orthonormal basis, but the Hilbert spaces usually considered in quantum physics have an additional 
property, namely that they have a countable orthonormal basis. 
\begin{definition}
$\mathcal{H}$ is called \emph{separable} 
if it has a countable orthonormal basis. 
\end{definition}
For separable Hilbert spaces it can be checked whether a set $\{ \ket{\phi_i} \}_{i=1,2,\dotsc}$ of vectors in $\mathcal{H}$ forms an orthonormal basis, by testing if 
for all $i,j$, \linebreak[4] $\braket{\phi_i}{\phi_j}=\delta_ {ij}$ and $\sum_i \ket{\phi_i}\bra{\phi_i}= \mathds{1}_{\mathcal{H}}$.
This implies that for any $\ket{\psi},\ket{\phi}$ $\in \mathcal{H}$ and orthonormal basis $\{\ket{\phi_i}\}$, it holds that
\begin{align}
 \nonumber \ket{\psi}&=\sum_i\braket{\phi_i}{\psi}\ket{\phi_i} &&\text{(Fourier formula)}\\
 \nonumber \|{\psi}\|^2 &=\sum_i| \braket{\phi_i}{\psi}|^2  &&\text{(Plancherel formula)}\\
 \nonumber \braket{\phi}{\psi}&= \sum_i\overline{\braket{\phi_i}{\phi}} \braket{\phi_i}{\psi} &&\text{(Parceval formula).} 
\end{align}

\begin{definition}
$A$ is a \emph{bounded 
linear operator} on $\mathcal{H}$, denoted by $A\in \mathcal{B}(\mathcal{H})$, if 
\begin{align}
\nonumber A\colon \mathcal{H} &\rightarrow  \mathcal{H} \\
\nonumber \ket{\psi} &\mapsto  {A[\psi]}
\end{align}
with $A[\lambda_1\psi_1+\lambda_2 \psi_2]=\lambda_1 A[\psi_1]+\lambda_2 A[\psi_2]$ and
\begin{align}
\nonumber \sup_{\psi \neq 0}\frac{\| A [\psi] \|}{\| \ket{\psi} \|} &<\infty\ .
\end{align}
\end{definition}
Observables corresponding to physical quantities are bounded. 
Through the relation
\begin{align}
\nonumber \left(\ket{\psi}\bra{\phi^{\prime}}\right) \ket{\phi} &= \ket{\psi}\braket{\phi^{\prime}}{\phi}=\braket{\phi^{\prime}}{\phi}\ket{\psi}
\end{align}
we can express the operator mapping $\ket{\phi^{\prime}}$ to $\ket{\psi}$ multiplied by $\braket{\phi}{\phi^{\prime}}$ as $\ket{\psi}\bra{\phi^{\prime}}$. This leads to the \emph{outer product notation} of $A$. 

\begin{definition}
The \emph{adjoint} $A^\dagger$ of a bounded operator $A$ is defined such that 
\begin{align}
\nonumber \braket{\phi}{A^\dagger \psi} &=\braket{A \phi}{ \psi}\ .
\end{align}
\end{definition}
It further holds  that $(\lambda A)^\dagger=\bar{\lambda} A^\dagger$; $(AB)^\dagger=B^\dagger A^\dagger$; $\|A\| =\|A^\dagger \|$; ${A^\dagger}^\dagger=A$ and $\braket{\psi}{A^\dagger A \psi}=\braket{A\psi}{ A \psi}=\| A\psi\|^2\geq 0$. 

\begin{definition}
A bounded operator $A$ is called \emph{self-adjoint} if $A=A^\dagger$ and it is called \emph{unitary} if $A A^\dagger=A^\dagger A=\mathds{1}$.
\end{definition}

\begin{definition}
Let $A=A^\dagger \in  \mathcal{B}(\mathcal{H})$. If $A\ket{\psi}=a\ket{\psi}$ then $\ket{\psi}$ is an \emph{eigenvector} of $A$ with \emph{eigenvalue} $a$ and $a\in \mathbb{R}$. 
\end{definition}

\begin{definition}
Let $\mathcal{H}$ be a Hilbert space and $\mathcal{H}^{\prime}$ a subspace of $\mathcal{H}$ with $ \{\ket{\phi_i}\}_{i\in I}$ an orthonormal basis of $\mathcal{H}^{\prime}$. The \emph{projector} of $\mathcal{H}$ 
onto $\mathcal{H}^{\prime}$ is the operator 
\begin{align}
\nonumber P_{\mathcal{H}^{\prime}} &= \sum_{i\in I}\ket{\phi_i}\bra{\phi_i}\ .
\end{align}
\end{definition}
The projector onto a subspace of $\mathcal{H}$ fulfils $P=P^\dagger=P^2$.

\begin{theorem}[Spectral decomposition]
Let $A$ be a self-adjoint bounded linear operator on $\mathcal{H}$ with eigenvalues $\{a_i\}$.  
Then $\mathcal{H}$ has an orthonormal basis $\{\ket{\phi_{i,k}}\}_{k=1,\dotsc,d_i}$ of eigenvectors of $A$ and 
\begin{align}
\nonumber A &= \sum_i a_i P_{a_i}\ ,
\end{align}
where $P_{a_i}=\sum_k \ket{\phi_{i,k}}\bra{\phi_{i,k}}$ is the projector onto the eigenspace associated with the eigenvalue $a_i$. 
\end{theorem}
Note that the eigenspaces associated with different eigenvalues of $A$ are orthogonal.

\begin{example}
An example of a Hilbert space is $\mathbb{C}^n$ with the scalar product $\langle \phi,\psi \rangle=\sum_{i=1}^n \bar{\phi}_i\psi_i$. In this case, every vector $\ket{\psi}\in \mathcal{H}$ and dual vector $\bra{\phi}\in \mathcal{H}^*$ can be expressed as 
\begin{align}
\nonumber \ket{\psi} &= \left( 
\begin{array}{@{\hspace{0mm}}c@{\hspace{0mm}}}
\psi_1\\
\vdots \\
\psi_n
\end{array}
\right) \ 
\text{with} \ 
\psi_i\in \mathbb{C} \ 
\text{and for}\ 
\ket{\phi}\in \mathcal{H}\ \bra{\phi}=\left(
\begin{array}{@{\hspace{0mm}}ccc@{\hspace{0mm}}}
\bar{\phi}_1 & \cdots & \bar{\phi}_n
\end{array}\right)\ .
\end{align}
The scalar product is 
\begin{align}
\nonumber \braket{\phi}{\psi}&= \left(
\begin{array}{@{\hspace{0mm}}ccc@{\hspace{0mm}}}
\bar{\phi}_1 & \cdots & \bar{\phi}_n
\end{array}\right)\left( 
\begin{array}{@{\hspace{0mm}}c@{\hspace{0mm}}}
\psi_1\\
\vdots \\
\psi_n
\end{array}
\right)=\sum_{i=1}^n\bar{\phi}_i\psi_i \in \mathbb{C}\ .
\end{align}
and the operator $\ket{\phi}\bra{\psi}$ is a complex $n\times n$ matrix
\begin{align}
\nonumber \ket{\psi}\bra{\phi} &=
\left( 
\begin{array}{@{\hspace{0mm}}c@{\hspace{0mm}}}
\psi_1\\
\vdots \\
\psi_n
\end{array}
\right)
\left(
\begin{array}{@{\hspace{0mm}}ccc@{\hspace{0mm}}}
\bar{\phi}_1 & \cdots & \bar{\phi}_n
\end{array}\right)
=
\left( 
\begin{array}{@{\hspace{0mm}}ccc@{\hspace{0mm}}}
\psi_1 \bar{\phi}_1 & \cdots & \psi_1 \bar{\phi}_n\\
\vdots & \ddots & \vdots \\
\psi_n \bar{\phi}_1 & \cdots & \psi_n \bar{\phi}_n
\end{array}
\right)
\in \mathbb{M}_n(\mathbb{C})\ .
\end{align}
\end{example}

\begin{example}
Another example of a Hilbert space is $\mathcal{L}^2(\mathbb{R}^3,d^3x)$, the set of complex square integrable functions over $\mathbb{R}^3$. In this case, \linebreak[4] $\psi=\psi (\overrightarrow{x})\colon \mathbb{R}^3\rightarrow \mathbb{C}$, such that 
\begin{align}
\nonumber \int_{\mathbb{R}^3} d^3x\ \left|\psi \left( \overrightarrow{x}\right) \right|^2 &< \infty\ .
\end{align}
The scalar product is given by $\langle \phi,\psi \rangle=\int_{\mathbb{R}^3} d^3x\ \bar{\phi}(\overrightarrow{x})\psi(\overrightarrow{x})$.
\end{example}

When the Hilbert space is $\mathbb{C}^n$, we will often denote the canonical basis vectors by $\ket{0},\dotsc, \ket{n-1}$. We will call systems with $n=2$ a~\emph{qubit} and denote their basis states as
\begin{align}
\nonumber \ket{0} =
\left(
\begin{array}{@{\hspace{0mm}}c@{\hspace{0mm}}}
1\\
0
\end{array}
\right)
\qquad &
 \ket{1}=
\left(
\begin{array}{@{\hspace{0mm}}c@{\hspace{0mm}}}
0\\
1
\end{array}
\right)\ .
\end{align}

\subsubsection{Composite systems}

When describing a system that consists of two subsystems, one being described by a Hilbert space $\mathcal{H}_1$ and the other by $\mathcal{H}_2$, then the pure state of the total system can be described by the Hilbert space that is the \emph{tensor product} of the two subspaces, i.e., $\mathcal{H}=\mathcal{H}_1\otimes \mathcal{H}_2$. 

The tensor product is defined such that for $\ket{\psi}\in \mathcal{H}_1$ and $\ket{\phi}\in \mathcal{H}_2$, it associates a vector $\ket{\psi}\otimes \ket{\phi} \in \mathcal{H}$ with the property that
\begin{align}
\nonumber c\cdot \left(\ket{\psi}\otimes \ket{\phi} \right) 
&=(c\cdot \ket{\psi})\otimes \ket{\phi}  
=\ket{\psi}\otimes (c\cdot  \ket{\phi}) 
\\
\nonumber 
(\ket{\psi_1}+\ket{\psi_2})\otimes \ket{\phi}
&=
\ket{\psi_1}\otimes \ket{\phi}+ \ket{\psi_2}\otimes \ket{\phi}
\\
\nonumber 
\ket{\psi}\otimes (\ket{\phi_1}+\ket{\phi_2})
&=
\ket{\psi}\otimes \ket{\phi_1}+\ket{\psi}\otimes\ket{\phi_2}
\end{align}
for all $c\in \mathbb{C}$, $\ket{\psi_1},\ket{\psi_2}\in \mathcal{H}_1$ and
 $\ket{\phi_1},\ket{\phi_2}\in \mathcal{H}_2$. 

We will sometimes drop the tensor product in the notation and  write
\[
\nonumber \ket{0}\otimes \ket{1}= \ket{0}\ket{1}=\ket{01}\ .
\]

The tensor product of linear operators $A$ acting on $\mathcal{H}_1$ and $B$ acting on $\mathcal{H}_2$ can be defined via the relation
\begin{align}
\nonumber \left(A \otimes B \right)[\ket{\psi}\otimes \ket{\phi}] &= \left( A\ket{\psi} \right)\otimes \left(B \ket{\phi} \right)\ .
\end{align}

Note that when $\{\ket{\psi_i}\}_i$ is an orthonormal basis of $\mathcal{H}_1$ and 
$\{\ket{\phi_j}\}_j$ is an orthonormal basis of $\mathcal{H}_2$, then 
$\{\ket{\psi_i}  \otimes \ket{\phi_j} \}_{i,j}$ is an orthonormal basis of $\mathcal{H}=\mathcal{H}_1 \otimes \mathcal{H}_2$.

Not all states in the tensor product Hilbert space can be expressed as the tensor product of a state in each of the two subsystems. 

\begin{definition}
Let $\ket{\psi}\in\mathcal{H}=\mathcal{H}_1\otimes \mathcal{H}_2$ be a pure state. Then, if $\ket{\psi}$ cannot be expressed as the tensor product of a state $\ket{\psi_1}\in \mathcal{H}_1$ and $\ket{\psi_2}\in \mathcal{H}_2$, i.e., 
\begin{align}
\nonumber \ket{\psi} &\neq \ket{\psi_1}\otimes \ket{\psi_2}\ ,
\end{align}
the state $\ket{\psi}$ is called \emph{entangled}. 
\end{definition}

\begin{example}
Examples of entangled states of two qubits are the \emph{Bell states}. The state $\ket{\Psi^-}$ is also called the \emph{singlet}. 
\begin{align}
\nonumber \ket{\Psi^-}&= \frac{1}{\sqrt{2}}\left(\ket{01}-\ket{10} \right)\\
\nonumber \ket{\Psi^+}&= \frac{1}{\sqrt{2}}\left(\ket{01}+\ket{10} \right)\\
\nonumber \ket{\Phi^-}&= \frac{1}{\sqrt{2}}\left(\ket{00}-\ket{11} \right)\\
\nonumber \ket{\Phi^+}&= \frac{1}{\sqrt{2}}\left(\ket{00}+\ket{11} \right)\ .
\end{align}
\end{example}

\subsubsection{Density operators and generalized measurements}

A useful way to represent quantum states is using \emph{density operators}, i.e., operators on the Hilbert space. This representation has the advantage that the situation where a certain pure state $\ket{\psi_i}$ occurs with probability $p_i$ can be modelled easily. 

\begin{definition}
A \emph{density operator} is a Hermitian positive operator $\rho$ with trace $1$, i.e.,
\begin{align}
\nonumber \rho &= \rho^\dagger\ ,\\
\nonumber \rho &\succeq  0\ , \\
\nonumber \tr(\rho) &= 1\ .
\end{align}
The expression $\rho \succeq  0$ means that the eigenvalues of $\rho$ are non-negative. 
\end{definition}

The density matrix $\rho$ associated with a (normalized) pure state $\ket{\psi}$ is $\ket{\psi}\bra{\psi}$.

If a measurement $A$ is performed on a state characterized by a density operator $\rho$, then 
\begin{align}
\nonumber \langle A \rangle &= \tr(A\rho)\ .
\end{align}
The probability to obtain outcome $a_i$ is 
\begin{align}
\nonumber \Pr_{\rho}[a_i] &= \tr(P_{a_i}\rho)\ ,
\end{align}
where $P_{a_i}$ is the projector onto the eigenspace associated with eigenvalue $a_i$.

We observe that for the density matrix $\rho=\ket{\psi}\bra{\psi}$ associated with the pure state $\ket{\psi}$, we obtain
\begin{align}
\nonumber \langle A \rangle &= \tr \left( A \ket{\psi}\bra{\psi}\right) = \brakket{\psi}{A}{\psi}\ \ \ \text{and}\\
\nonumber \Pr_{\rho}[a_i] &= \tr \left( P_{a_i}\ket{\psi}\bra{\psi} \right)
=\brakket{\psi}{P_{a_i}}{\psi}\ ,
\end{align}
as expected. 

Furthermore, when the system is in state $\ket{\psi_i}$ with probability $p_i$ (this is called a \emph{mixed state}), we associate the density matrix 
\begin{align}
\nonumber \rho &= \sum_i p_i \ket{\psi_i}\bra{\psi_i} 
\end{align}
with this system. Because of the linearity of the trace, $A$ and $P_{a_i}$, we obtain in this case 
\begin{align}
\nonumber \langle A \rangle &= \tr \Bigl( A \sum_i p_i \ket{\psi_i}\bra{\psi_i}\Bigr) =\sum_i p_i \brakket{\psi_i}{A}{\psi_i}\ \ \ \text{and}\\
\nonumber \Pr_{\rho} [a_i] &= \tr \Bigl( P_{a_i} \sum_i p_i \ket{\psi_i}\bra{\psi_i}\Bigr)
=\sum_i p_i \brakket{\psi_i}{P_{a_i}}{\psi_i}\ .
\end{align}
Note that the same density matrix $\rho$ can be associated with different probabilistic mixtures of pure states. 
A density matrix $\rho$ corresponds to a \emph{pure} state exactly if $\rho^2=\rho$. 

A state represented by a density matrix is called \emph{entangled} if it cannot be expressed as the convex combination of the tensor product of two density matrices, i.e., 
\begin{align}
\nonumber \rho &\neq  \sum_i p_i \rho_{1,i}\otimes \rho_{2,i}\ .
\end{align} 

For a density matrix $\rho$ on $\mathcal{H}=\mathcal{H}_1\otimes \mathcal{H}_2$, we can obtain the density matrix describing only the first part of the system by the \emph{partial trace} over the second system, i.e., 
\begin{align}
\nonumber \rho_1 &=\tr_2 \rho = \sum_{i}\left(\mathds{1}_1\otimes \bra{\phi_i}\right) {\rho}\left(\mathds{1}_1\otimes \ket{\phi_i}\right)\ ,
\end{align}
where $\{\ket{\phi_i}\}_i$ is a basis of $\mathcal{H}_2$.

In a similar way as density matrices can be seen as a generalization of the notion of a 
state, it is also possible to generalize the notion of a measurement. 
\begin{definition}
A \emph{POVM} (Positive Operator-Valued Measure) is a set of positive Hermitian operators $\{E_i\}_i$ such that $\sum_i E_i=\mathds{1}$. 
\end{definition}

Note, however, that any \emph{density matrix} can be seen as a \emph{pure state} on a larger system and any \emph{POVM} can be seen as applying a unitary transformation to the system and an ancilla  (additional system) followed by a \emph{projective measurement} (described by the eigenvalues and the projectors onto the eigenspaces of a self-adjoint operator $A$). 

In fact, a density operator on $\mathcal{H}_1$ defined by $\rho_1=\sum_i p_i \ket{\psi_i}\bra{\psi_i}$ can be expressed as the pure state
\begin{align}
\nonumber \ket{\psi^{\prime}} &=\sum_i \sqrt{p_i}\ket{\psi_i}_1\ket{\psi_i}_2
\end{align}
in a Hilbert space $\mathcal{H}=\mathcal{H}_1\otimes \mathcal{H}_2$ (where the dimension of $\mathcal{H}_2$ must be at least the dimension of $\mathcal{H}_1$).

A \emph{POVM element} $E_i$ can be expressed as $E_i=M^\dagger_i M_i$ because it is Hermitian and positive semi-definite. Define an operator $U$ by
\begin{align}
\nonumber U  [\ket{\psi}\ket{0}] &:=\sum_i M_i \ket{\psi}\ket{i}\ .
\end{align}
$U$ is unitary because of the completeness relation $\sum_i E_i =\mathds{1}$. If  the projective measurement defined by $P_i=\mathds{1}\otimes \ket{i}\bra{i}$ is applied to $U [\ket{\psi}\ket{0}]$, this corresponds exactly to applying the POVM $\{E_i\}_i$ to $\ket{\psi}$. The POVM is, therefore, equivalent to applying the above unitary transformation and then performing a projective measurement. 

This argument implies that we will always be able to restrict our analysis to 
\emph{pure states} and \emph{projective measurements} (in a potentially larger space).

\subsubsection{Classical random variables as quantum states}

A discrete random variable $X$ with probability distribution $P_X$ can be represented by the density matrix
\begin{align}
\nonumber &\sum_x P_X(x) \ket{x}\bra{x}\ ,
\end{align}
where $\{\ket{x}\}_x$ is an orthonormal basis of a Hilbert space $\mathcal{H}_X$. Measuring the state in this basis gives the measurement result $x$ with probability $P_X(x)$.  

Similarly, the case where a quantum system is described by a different state depending on the value of a random variable $X$ can also be represented by a quantum state. More precisely, by a state $\rho_{XA}$, which is called 
\emph{classical on $X$}. 
\begin{definition}
A state $\rho_{XA}$ such that
\begin{align}
\nonumber \rho_{XA} &= \sum_x P_X(x) \ket{x}\bra{x}\otimes \rho_A^x\ ,
\end{align}
where $\{\ket{x}\}_x$ is an orthonormal basis of a Hilbert space $\mathcal{H}_X$ and $\rho_A^x$ is a density matrix on $\mathcal{H}_A$ is called \emph{classical on $X$}. 
\end{definition}

\subsubsection{Min-Entropy}

In classical information theory, tasks such as data compression or randomness extraction can be characterized by the \emph{entropy} of a distribution. These entropies can also be defined for quantum states.  We will, in particular, use the notion of the \emph{min-entropy} of a system $A$ conditioned on a system $B$. For the definition of other entropies of quantum states we refer to~\cite{rennerphd}. 
\begin{definition}
The \emph{min-entropy} of $A$ given $B$ of a density matrix $\rho_{AB}$ on $\mathcal{H}_A
 \otimes \mathcal{H}_B$ is
\begin{align}
\nonumber \mathrm{H}_{\mathrm{min}}(A|B)_{\rho_{AB}}&=\max_{\sigma_B} \sup  \{ \lambda\in \mathbb{R} |2^{-\lambda}\mathds{1}_A\otimes \sigma_B \succeq \rho_{AB} \}\ ,
\end{align}
where the maximization is over all density matrices $\sigma_B$ on $\mathcal{H}_B$.
\end{definition}

In~\cite{krs}, it is shown that when the system $A$ is classical, then the min-entropy of $A$ given $B$ is just the maximal probability that someone holding the system $B$ can correctly guess the value of $A$.
\begin{theorem}[K\"onig, Renner, Schaffner~\cite{krs}]\label{th:krs}
Let $\rho_{AB}$ be classical on $\mathcal{H}_A$. Then 
\begin{align}
\nonumber \mathrm{H}_{\mathrm{min}}(A|B)_{\rho_{AB}} &= -\log_2 P_{\mathrm{guess}}(A|B)_{\rho_{AB}}\ ,
\end{align}
where $P_{\mathrm{guess}}(A|B)_{\rho_{AB}}$ is the maximal probability of decoding $A$ from $B$ with a POVM $\{E_B^a\}_a$ on $\mathcal{H}_B$, i.e., 
\begin{align}
\nonumber P_{\mathrm{guess}}(A|B)_{\rho_{AB}} &:= \max_{\{E_B^a\}_a}\sum_a P_A(a) \tr (E_B^a\rho_B^a)\ .
\end{align}
\end{theorem}

\section{Systems from Different Resources}\label{sec:systems}

Consider a bipartite system taking an input and giving an output on each side. This system is characterized by the conditional probability distribution $P_{XY|UV}$ of the outputs given a certain input pair. 
Which systems $P_{XY|UV}$ can be realized depends on the \emph{resources} that can be used to realize it. 

We can view this situation as a game, where two parties --- let us call them Alice and Bob --- are allowed to agree on a strategy, but are then put into separate rooms. Later, they are asked questions by a referee --- Alice is asked question $u$ of some set $\mathcal{U}$, but does not know Bob's question $v\in \mathcal{V}$ and Bob gets question $v$, but does not know $u$. Their goal is to give answers $x\in \mathcal{X}$ (for Alice) and $y\in \mathcal{Y}$ (for Bob) according to the distribution $P_{XY|UV}$ using the resource at their disposition. 

One possible such resource is, of course, communication. If Alice and Bob are allowed to communicate $u$ and $v$ to each other and then decide on their answers together, it should be possible for them to realize any conditional probability distribution $P_{XY|UV}$. (If Alice and Bob are also able to make coin tosses locally.) 

In the following, we will characterize $n$-party systems that can be implemented using different resources. The resources we consider are, however, such that they do \emph{not} allow for communication.   
An $n$-party system is denoted by $P_{\bof{X}|\bof{U}}$, where $\bof{X}$ is a vector of $n$ random variables $\bof{X}=X_1\dotso X_n$. In the case of two parties, we will sometimes write $P_{XY|UV}$. Sometimes, we will also consider the case when two parties, Alice and Bob, share a $(2n)$-party system and will denote this system by $P_{\bof{XY}|\bof{UV}}$ in order to make clear which random variable is associated with which party. For a $(2n+1)$-party system, associated with Alice, Bob and Eve, we will use the notation $P_{\bof{X}\bof{Y}Z|\bof{U}\bof{V}W}$.

\subsection{Local systems}\label{subsec:localsystem}

The first resource we consider is \emph{shared randomness}. More precisely, we assume that Alice and Bob are allowed to discuss a \emph{strategy} and make an arbitrary number of \emph{coin tosses}. But after they are separated, they are only allowed to base their answers on the question they have obtained, and the value of the shared randomness. 
The strategies can be considered \emph{de\-ter\-min\-istic}, i.e. given a certain value of the shared randomness $r$ the strategy of Alice tells her exactly which answer $x$ to give as function of the question $u$ and the same for Bob. Indeed, any strategy of Alice which chooses an $x$ probabilistically as function of $u$ and $r$ can be expressed as a deterministic strategy by incorporating Alice's local randomness into the shared randomness $R$. The distributions $P_{XY|UV}$ that can be generated this way by Alice and Bob are called \emph{local}. Formally, we define the following.
\begin{definition}
An $n$-party system $P_{\bof{X}|\bof{U}}$ is called \emph{local deterministic} if 
\begin{align}
\nonumber P_{\bof{X}|\bof{U}}(\bof{x},\bof{u}) &= \prod_i \delta_{x_i,f^i(u_i)}\ ,
\end{align}
where $\delta$ denotes the Kronecker delta, i.e., the function $\delta_{a,b}=1$ if $a=b$ and~$0$ otherwise, and where $f^i\colon \mathcal{U}_i\rightarrow \mathcal{X}_i$ is a function associating with each $u_i$ an $x_i$. 
\end{definition}

Local systems are all the ones which can be expressed as convex combinations of local deterministic systems.

\begin{definition}
An $n$-party  system $P_{\bof{X}|\bof{U}}$ is called \emph{local} if 
\begin{align}
\nonumber P_{\bof{X}|\bof{U}}(\bof{x},\bof{u}) &= \sum_r P_R(r)\cdot 
\prod_i \delta_{x_i,f_r^i(u_i)}
\end{align}
with $\sum_r P_R(r)=1$. A distribution which is not local is called \emph{non-local}.
\end{definition}

The space of local probability distributions is a convex polytope and its vertices are the local deterministic distributions. A convex polytope can be described either in terms of its vertices or, alternatively, as an intersection of a  finite number of halfspaces (see, e.g.,~\cite{bv}). In the context of local probability distributions, these halfspaces correspond to so-called \emph{Bell inequalities}~\cite{bellinequality}. Informally speaking, a Bell inequality is an upper bound on a linear combination of the probabilities $P_{\bof{X}|\bof{U}}(\bof{x},\bof{u})$ that must hold for all local distributions $P_{\bof{X}|\bof{U}}$. 
\begin{definition}\label{def:bellineq} 
 A \emph{Bell inequality} is an inequality of the form 
\begin{align}
\nonumber  \sum_{\bof{x},\bof{u}}q({\bof{x},\bof{u}})
P_{\bof{X}|\bof{U}}(\bof{x},\bof{u}) 
 &\leq c
\end{align}
 that must hold for any local distribution $P_{\bof{X}|\bof{U}}$, and
where $q\colon \bof{\mathcal{X}}\times \bof{\mathcal{U}}\rightarrow \mathbb{R}$ is a function associating with each value of $\bof{x}$ and $\bof{u}$ a real number, and $c$ is a real number. 
\end{definition}

If a distribution $P_{\bof{X}|\bof{U}}$ violates a Bell inequality, this proves that it is non-local. The reverse argument is also possible: Any non-local distribution lies outside the local polytope and, therefore, must violate \emph{some} Bell inequality. 

The best-known example of a Bell inequality is the one given by 
Clauser, Horne, Shimony, and Holt~\cite{CHSH}, also called \emph{CHSH inequality}. 
This inequality is the only one relevant for bipartite systems with binary inputs and outputs, in the sense that 
any non-local system of this type must violate it, 
possibly using a relabelling of the inputs and outputs. 

\begin{example}[CHSH inequality~\cite{CHSH}] \label{ex:chshineq}
For any local system $P_{XY|UV}$ with $\mathcal{X}=\mathcal{Y}=\mathcal{U}=\mathcal{V}=\{0,1\}$ it holds that\footnote{Originally~\cite{CHSH}, the Bell inequality was stated in terms of systems giving outputs in $\{-1,1\}$, in which case the inequality reads
\begin{align}
\nonumber \langle X_0Y_0 \rangle+ \langle X_0Y_1 \rangle + \langle X_1Y_0 \rangle  -\langle X_1Y_1 \rangle & \leq 2\ ,
\end{align}
where $X_0$ stands for the random variable $X$ given input $u=0$, and $\langle X_0Y_0 \rangle$ denotes the expectation value of the random variable $X_0Y_0$. 
}
\begin{align}
\nonumber \frac{1}{4} \sum_{(x,y,u,v):x\oplus y =u\cdot v} P_{XY|UV}(x,y,u,v) &\leq \frac{3}{4}\ .
\end{align}

\begin{figure}[h!]
\centering
\psset{unit=0.525cm}
\pspicture*[](-2.5,-1)(9,10)
\pspolygon[linewidth=0pt,fillstyle=hlines,hatchcolor=lightgray](0,4.5)(2,4.5)(2,6)(0,6)
\pspolygon[linewidth=0pt,fillstyle=hlines,hatchcolor=lightgray](2,3)(4,3)(4,4.5)(2,4.5)
\rput[c]{0}(0,-3){
\pspolygon[linewidth=0pt,fillstyle=hlines,hatchcolor=lightgray](0,4.5)(2,4.5)(2,6)(0,6)
\pspolygon[linewidth=0pt,fillstyle=hlines,hatchcolor=lightgray](2,3)(4,3)(4,4.5)(2,4.5)
}
\rput[c]{0}(4,0){
\pspolygon[linewidth=0pt,fillstyle=hlines,hatchcolor=lightgray](0,4.5)(2,4.5)(2,6)(0,6)
\pspolygon[linewidth=0pt,fillstyle=hlines,hatchcolor=lightgray](2,3)(4,3)(4,4.5)(2,4.5)
}
\pspolygon[linewidth=0pt,fillstyle=hlines,hatchcolor=lightgray](6,1.5)(8,1.5)(8,3)(6,3)
\pspolygon[linewidth=0pt,fillstyle=hlines,hatchcolor=lightgray](4,0)(6,0)(6,1.5)(4,1.5)
\psline[linewidth=2pt,linecolor=gray]{->}(-2.5,5.25)(-1,5.25)
\psline[linewidth=0.5pt,linecolor=gray]{-}(-1,5.25)(8,5.25)
\psline[linewidth=2pt,linecolor=gray]{->}(-2.5,2.25)(-1,2.25)
\psline[linewidth=0.5pt,linecolor=gray]{-}(-1,2.25)(8,2.25)
\psline[linewidth=2pt,linecolor=gray]{->}(1,8.5)(1,7.25)
\psline[linewidth=0.5pt,linecolor=gray]{-}(1,7.25)(1,0)
\psline[linewidth=2pt,linecolor=gray]{->}(7,8.5)(7,7.25)
\psline[linewidth=0.5pt,linecolor=gray]{-}(7,7.25)(7,0)
\psline[linewidth=0.5pt]{-}(0,6)(-1,7)
\rput[c]{0}(-0.25,6.75){\scriptsize{$X$}}
\rput[c]{0}(-0.75,6.25){\scriptsize{$Y$}}
\rput[c]{0}(-0.5,7.5){\large{$U$}}
\rput[c]{0}(-1.5,6.5){\large{$V$}}
\rput[c]{0}(2,7.5){\Large{$0$}}
\rput[c]{0}(6,7.5){\Large{$1$}}
\rput[c]{0}(1,6.5){\Large{$0$}}
\rput[c]{0}(3,6.5){\Large{$1$}}
\rput[c]{0}(5,6.5){\Large{$0$}}
\rput[c]{0}(7,6.5){\Large{$1$}}
\rput[c]{0}(-1.5,4.5){\Large{$0$}}
\rput[c]{0}(-1.5,1.5){\Large{$1$}}
\rput[c]{0}(-0.5,5.25){\Large{$0$}}
\rput[c]{0}(-0.5,3.75){\Large{$1$}}
\rput[c]{0}(-0.5,2.25){\Large{$0$}}
\rput[c]{0}(-0.5,0.75){\Large{$1$}}
\psline[linewidth=2pt]{-}(-1,0)(8,0)
\psline[linewidth=2pt]{-}(-1,6)(8,6)
\psline[linewidth=2pt]{-}(-1,3)(8,3)
\psline[linewidth=1pt]{-}(0,1.5)(8,1.5)
\psline[linewidth=1pt]{-}(0,4.5)(8,4.5)
\psline[linewidth=2pt]{-}(0,0)(0,7)
\psline[linewidth=2pt]{-}(8,0)(8,7)
\psline[linewidth=2pt]{-}(4,0)(4,7)
\psline[linewidth=1pt]{-}(2,0)(2,6)
\psline[linewidth=1pt]{-}(6,0)(6,6)
\rput[c]{0}(1,5.25){\Large{$1$}}
\rput[c]{0}(3,3.75){\Large{$0$}}
\rput[c]{0}(5,5.25){\Large{$0$}}
\rput[c]{0}(7,3.75){\Large{$0$}}
\rput[c]{0}(1,2.25){\Large{$1$}}
\rput[c]{0}(3,0.75){\Large{$0$}}
\rput[c]{0}(5,0.75){\Large{$0$}}
\rput[c]{0}(7,2.25){\Large{$1$}}
\rput[c]{0}(3,5.25){\Large{$0$}}
\rput[c]{0}(1,3.75){\Large{$0$}}
\rput[c]{0}(7,5.25){\Large{$1$}}
\rput[c]{0}(5,3.75){\Large{$0$}}
\rput[c]{0}(3,2.25){\Large{$0$}}
\rput[c]{0}(1,0.75){\Large{$0$}}
\rput[c]{0}(5,2.25){\Large{$0$}}
\rput[c]{0}(7,0.75){\Large{$0$}}
\endpspicture
\caption{A local deterministic system. In this notation, a local deterministic system corresponds to the selection of a line (column) for each input, as indicated by the arrows. The CHSH inequality (Example~\ref{ex:chshineq}) corresponds to the condition that the sum of the entries in the hatched cells is at most~$3$. This system, therefore, reaches the maximal possible value for a local system. }
\end{figure}

\end{example}
For a specific system $P_{XY|UV}$ (not necessarily local), we will sometimes call the value of the expression on the left-hand side in the above inequality the \emph{Bell value} (or \emph{CHSH value}) of this system. 
A generalization of the CHSH inequality to systems with more inputs has been given by Braunstein and Caves~\cite{braunsteincaves2}. 
\begin{example}[Braunstein-Caves inequality~\cite{braunsteincaves2}] \label{ex:braunsteincaves}
For any local system $P_{XY|UV}$ with $\mathcal{X}=\mathcal{Y}=\{0,1\}$ and $\mathcal{U}=\mathcal{V}=\{1,\dotsc, N\}$ it holds that\footnote{The Braunstein-Caves inequality was also originally given in terms of correlations of systems with outputs in $\{-1,1\}$. 
}
\begin{multline}
\nonumber
\frac{1}{2N} \cdot \left( 
 \sum_{u=1}^{N} \sum_{v=u}^{u+1} \sum_{(x,y):x=y} P_{XY|UV}(x,y,u,v)
 \right.
 \\
+
\left.
\sum_{(x,y):x\neq y} P_{XY|UV}(x,y,N,1)
\right)
\leq 1- \frac{1}{2N}\ .
\end{multline}
\end{example}

\subsection{Quantum systems}\label{subsec:quantum}

Consider the setup where Alice and Bob are allowed to discuss a strategy and use shared randomness (as above), but in addition they are allowed to share a --- possibly entangled --- quantum state. Alice and Bob can now base their answers on the shared randomness, but also on the measurement outcomes they obtain from measuring the quantum state. Which measurement they perform can, of course, depend on the shared randomness and on the question they have obtained. 

Interestingly, the set of probability distributions which can be obtained this way is strictly larger than the local set described above. I.e., these distributions can be non-local. This is what is meant by the expression `quantum mechanics is non-local'.\footnote{Quantum physics is sometimes said to be a local theory, meaning that it is not possible to act on a system that is in a distant location. We will call this property \emph{non-signalling} (see Section~\ref{subsec:nssystem}).}

\begin{definition}\label{def:qbehavior}
 An $n$-party system $P_{\bof{X}|\bof{U}}$, where $\bof{X}=X_1\dotso X_n$, is called \emph{quantum} if there exists a pure state $\ket{\psi}\in \mathcal{H}=\bigotimes_i \mathcal{H}_i$ and a set of measurement operators $\{E_{u_i}^{x_i}\}$ on $\mathcal{H}_i$ such that
\begin{align}
\nonumber  P_{\bof{X}|\bof{U}}(\bof{x},\bof{u}) &= \bigbrakket{\psi}{\bigotimes_i E_{u_i}^{x_i}}{\psi}\ .
\end{align}
The measurement operators are
\begin{enumerate}
 \item Hermitian, i.e., ${E_{u_i}^{x_i}}^{\dagger}=E_{u_i}^{x_i}$ 
 for all $x_i,u_i$, 
 \item orthogonal projectors, i.e., $E_{u_i}^{x_i} E_{u_i}^{x^{\prime}_i}=E_{u_i}^{x_i}\delta_{x_ix^{\prime}_i}$, 
 \item and sum up to the identity, i.e., $\sum_{x_i}{E_{u_i}^{x_i}}=\mathds{1}_{\mathcal{H}_i}$ 
  for all $u_i$.
\end{enumerate}
\end{definition}

As we have seen in the previous section, it is not a restriction to assume the quantum state to be pure and the measurements to be projections, since any quantum state can be represented as a pure state in a larger Hilbert space and measurements as projective measurements by introducing an ancilla (see Section~\ref{sec:quantum}). 

In finite dimensions, the requirement that the measurements act only on one part of a larger tensor-product Hilbert space is equivalent to the requirement that all operators associated with different parties commute. See, e.g.,~\cite{dltw08} or~\cite{wehnerphd} for an explicit proof. 

\begin{theorem}\label{th:commutetensor}
Let $\mathcal{H}$ be a finite dimensional Hilbert space and  $\{E_{u_i}^{x_i}\}$ be a set of Hermitian orthogonal projectors with $\sum_{u_i} E_{u_i}^{x_i}=\mathds{1}$. 
Assume further that $[E_{u_i}^{x_i},E_{u_j}^{x_j}]=0$ for all $x_i,x_j,u_i,u_j$ where $i\neq j$. 
Then there exists a unitary isomorphism between $\mathcal{H}$ and $ \mathcal{H}^{\prime}=\bigotimes_i \mathcal{H}^{\prime}_i$ such that in $\mathcal{H}^{\prime}$, $E_{u_i}^{x_i}$ are of the form $\tilde{E}_{u_i}^{x_i}\otimes \mathds{1}$ and where $\tilde{E}_{u_i}^{x_i}$ acts on $\mathcal{H}^{\prime}_i$ only. 
\end{theorem}

For any $(n+1)$-party quantum system, the marginal and conditional systems are also quantum systems. This follows from the postulates of quantum physics, but we give a direct 
proof in terms of systems below.  
\begin{lemma}\label{lemma:qmarginalconditional}
Let $P_{\bof{X}Z|\bof{U}W}$ be an $(n+1)$-party quantum system. Then the marginal system $P_{\bof{X}|\bof{U}}(\bof{x},\bof{u}):= \sum_z P_{\bof{X}Z|\bof{U},W}(\bof{x},z,\bof{u},w)$ and the conditional system $P_{\bof{X}|\bof{U},W=w,Z=z}(\bof{x},\bof{u}):=P_{\bof{X}Z|\bof{U}W}(\bof{x},z,\bof{u},w)/P_{Z|W=w}(z)$ are $n$-party quantum systems. 
\end{lemma}
\begin{proof}
For the marginal system, take the same state $\ket{\psi}$ and the measurement operators $\{E_{u_i}^{x_i}\}$ for all $i<n$. The measurement operator associated with the $n$\textsuperscript{th} party are $\{E_{u_n}^{x_n}\otimes \mathds{1}_{\mathcal{H}_{n+1}}\}$.  They fulfil the requirements because they are part of the requirements of the operators of the $(n+1)$-party quantum system and remain valid when tensored with the identity. It then holds that 
\begin{align}
\nonumber P_{\bof{X}|\bof{U}}(\bof{x},\bof{u})&= \bigbrakket{\psi}{\Bigl(\bigotimes_i E_{u_i}^{x_i}\Bigr)\otimes \mathds{1}_{\mathcal{H}_{n+1}}}{\psi}\\
\nonumber 
&= \bigbrakket{\psi}{\Bigl(\bigotimes_i E_{u_i}^{x_i}\Bigr)\otimes \Bigl(\sum_{z}{E_{w}^{z}}\Bigr)}{\psi}\\
\nonumber &=\sum_{z} \bigbrakket{\psi}{\Bigl(\bigotimes_i E_{u_i}^{x_i}\Bigr) \otimes {E_{w}^{z}}}{\psi}
\\
\nonumber 
&= \sum_{z} P_{\bof{X}Z|\bof{U}W}(\bof{x},z,\bof{u},w) \ .
\end{align}
For the conditional system take the state 
\begin{align}
\nonumber
&\frac{1}{\sqrt{\brakket{\psi}{\mathds{1}_{\mathcal{H}_{1\dotso n}} \otimes E_w^z}{\psi}}} \mathds{1}_{\mathcal{H}_{1\dotso n}} \otimes E_w^z \ket{\psi}\ ,
\end{align}
 where $\mathds{1}_{\mathcal{H}_{1\dotso n}}=\bigotimes_{i=1}^{n}\mathds{1}_{\mathcal{H}_i}$
 and the measurement operators $\{E_{u_i}^{x_i}\}$.  
\begin{align}
\nonumber &P_{\bof{X}|\bof{U},W=w,Z=z}(\bof{x},\bof{u})
\\
\nonumber &= \bigerbrakket{ \psi }{ \frac{\mathds{1}_{\mathcal{H}_{1\dotso n}} \otimes {E_w^z}^\dagger }{\sqrt{\brakket{\psi}{\mathds{1}_{\mathcal{H}_{1\dotso n}} \otimes E_w^z}{\psi}}}\bigotimes_i E_{u_i}^{x_i}\frac{\mathds{1}_{\mathcal{H}_{1\dotso n}}\otimes E_w^z}{\sqrt{\brakket{\psi}{\mathds{1}_{\mathcal{H}_{1\dotso n}} \otimes E_w^z}{\psi}}} }{\psi}\\
\nonumber &= 
\frac{1}{{\brakket{\psi}{\mathds{1}_{\mathcal{H}_{1\dotso n}} \otimes E_w^z}{\psi}}}
\bigbrakket{\psi}{\Bigl(\bigotimes_i E_{u_i}^{x_i}\Bigr) \otimes E_w^z}{\psi}
\\
\nonumber &= \frac{1}{P_{Z|W=w}(z)}P_{\bof{X}Z|\bof{U}W}(\bof{x},z,\bof{u},w) \ .\qedhere
\end{align}
\end{proof}

The set of quantum systems is convex, but it is not a polytope, i.e., the set of its 
extremal points is infinite. 
\begin{example}\label{ex:qsystem}
The system in Figure~\ref{fig:qsystem} is a quantum system. It can be obtained by measuring the state $\ket{\psi^-}=(\ket{10}-\ket{01})/\sqrt{2}$ using the operators $E_u^x=\ket{\Psi_u^x}\bra{\Psi_u^x}$ and $E_v^y=\ket{\Phi_v^y}\bra{\Phi_v^y}$ as given below. 
\begin{align}
&
\begin{array}{@{\hspace{2mm}}r@{\hspace{2mm}}c@{\hspace{2mm}}l@{\hspace{6mm}}r@{\hspace{2mm}}c@{\hspace{2mm}}l@{\hspace{2mm}}}
\nonumber \ket{\Psi_0^0}&=& \frac{1}{\sqrt{2}}(\ket{0}+\ket{1})
&
\ket{\Psi_0^1} &=& \frac{1}{\sqrt{2}}(\ket{0}-\ket{1}) \vspace{1mm} \\ 
\nonumber \ket{\Psi_1^0}&=& \ket{0}
& \ket{\Psi_1^1}&=&\ket{1}\vspace{1mm} \\ 
\nonumber \ket{\Phi_0^0} &=& \frac{\sqrt{2-\sqrt{2}}}{2}\ket{0}-\frac{\sqrt{2+\sqrt{2}}}{2}\ket{1}
&
\ket{\Phi_0^1}&= &
\frac{\sqrt{2+\sqrt{2}}}{2}\ket{0}+\frac{\sqrt{2-\sqrt{2}}}{2}\ket{1}
\vspace{1mm} \\ 
\nonumber \ket{\Phi_1^0} &=& 
\frac{\sqrt{2+\sqrt{2}}}{2}\ket{0}-\frac{\sqrt{2-\sqrt{2}}}{2}\ket{1}
&
\nonumber \ket{\Phi_1^1} &= &
\frac{\sqrt{2-\sqrt{2}}}{2}\ket{0}+\frac{\sqrt{2+\sqrt{2}}}{2}\ket{1}\ .
\end{array}
\end{align}
\end{example}

\begin{figure}[h]
\centering
\psset{unit=0.525cm}
\pspicture*[](-2,-0.1)(8.5,8)
\psline[linewidth=0.5pt]{-}(0,6)(-1,7)
\rput[c]{0}(-0.25,6.75){\scriptsize{$X$}}
\rput[c]{0}(-0.75,6.25){\scriptsize{$Y$}}
\rput[c]{0}(-0.5,7.5){\large{$U$}}
\rput[c]{0}(-1.5,6.5){\large{$V$}}
\rput[c]{0}(2,7.5){\Large{$0$}}
\rput[c]{0}(6,7.5){\Large{$1$}}
\rput[c]{0}(1,6.5){\Large{$0$}}
\rput[c]{0}(3,6.5){\Large{$1$}}
\rput[c]{0}(5,6.5){\Large{$0$}}
\rput[c]{0}(7,6.5){\Large{$1$}}
\rput[c]{0}(-1.5,4.5){\Large{$0$}}
\rput[c]{0}(-1.5,1.5){\Large{$1$}}
\rput[c]{0}(-0.5,5.25){\Large{$0$}}
\rput[c]{0}(-0.5,3.75){\Large{$1$}}
\rput[c]{0}(-0.5,2.25){\Large{$0$}}
\rput[c]{0}(-0.5,0.75){\Large{$1$}}
\psline[linewidth=2pt]{-}(-1,0)(8,0)
\psline[linewidth=2pt]{-}(-1,6)(8,6)
\psline[linewidth=2pt]{-}(-1,3)(8,3)
\psline[linewidth=1pt]{-}(0,1.5)(8,1.5)
\psline[linewidth=1pt]{-}(0,4.5)(8,4.5)
\psline[linewidth=2pt]{-}(0,0)(0,7)
\psline[linewidth=2pt]{-}(8,0)(8,7)
\psline[linewidth=2pt]{-}(4,0)(4,7)
\psline[linewidth=1pt]{-}(2,0)(2,6)
\psline[linewidth=1pt]{-}(6,0)(6,6)
\rput[c]{0}(1,5.25){\Large{$\frac{2+\sqrt{2}}{8}$}}
\rput[c]{0}(3,3.75){\Large{$\frac{2+\sqrt{2}}{8}$}}
\rput[c]{0}(5,5.25){\Large{$\frac{2+\sqrt{2}}{8}$}}
\rput[c]{0}(7,3.75){\Large{$\frac{2+\sqrt{2}}{8}$}}
\rput[c]{0}(1,2.25){\Large{$\frac{2+\sqrt{2}}{8}$}}
\rput[c]{0}(3,0.75){\Large{$\frac{2+\sqrt{2}}{8}$}}
\rput[c]{0}(5,0.75){\Large{$\frac{2+\sqrt{2}}{8}$}}
\rput[c]{0}(7,2.25){\Large{$\frac{2+\sqrt{2}}{8}$}}
\rput[c]{0}(3,5.25){\Large{$\frac{2-\sqrt{2}}{8}$}}
\rput[c]{0}(1,3.75){\Large{$\frac{2-\sqrt{2}}{8}$}}
\rput[c]{0}(7,5.25){\Large{$\frac{2-\sqrt{2}}{8}$}}
\rput[c]{0}(5,3.75){\Large{$\frac{2-\sqrt{2}}{8}$}}
\rput[c]{0}(3,2.25){\Large{$\frac{2-\sqrt{2}}{8}$}}
\rput[c]{0}(1,0.75){\Large{$\frac{2-\sqrt{2}}{8}$}}
\rput[c]{0}(5,2.25){\Large{$\frac{2-\sqrt{2}}{8}$}}
\rput[c]{0}(7,0.75){\Large{$\frac{2-\sqrt{2}}{8}$}}
\endpspicture
\caption{\label{fig:qsystem}A quantum system.}
\end{figure}

The system in Figure~\ref{fig:qsystem} fulfils 
\begin{align}
\nonumber \frac{1}{4} \sum_{(x,y,u,v):x\oplus y =u\cdot v} P_{XY|UV}(x,y,u,v) &= \frac{2+\sqrt{2}}{4}\approx 0.85\ ,
\end{align}
i.e., it violates the Bell inequality of Example~\ref{ex:chshineq} in Section~\ref{subsec:localsystem}. Although quantum systems do not need to respect Bell inequalities, there exist limitations on the violations which can be reached by quantum systems. These limitations are called \emph{Tsirelson bounds}, after Tsirelson, who showed, 
in particular, that the above quantum system reaches indeed the maximal possible CHSH value~\cite{tsirelson}.

\subsection{Non-signalling systems}\label{subsec:nssystem}

The set of systems that can be obtained by measuring a quantum state is strictly larger than the local set, but these correlations still do not imply communication. The behaviour on her side does not give Alice any information about the question Bob has obtained. This property is called \emph{non-signalling}. We can consider the systems which can be 
obtained when Alice and Bob are allowed to share as resource an abstract device taking inputs and giving outputs on each side, under the sole condition that this device cannot be used for signalling.  This set of \emph{non-signalling systems} contains the set of quantum systems as a proper subset. 

\begin{definition}\label{def:nssystem}
An $n$-party system $P_{\bof{X}|\bof{U}}$ is called \emph{non-signalling} if for any set $I\subseteq \{1,\dotsc,n\}$, 
\begin{align}
 \nonumber \sum_{x_i:i\in I} P_{\bof{X}|\bof{U}}(\bof{x},\bof{u}_I,\bof{u}_{\bar{I}})
&= \sum_{x_i:i\in I} P_{\bof{X}|\bof{U}}(\bof{x},\bof{u}^{\prime}_I,\bof{u}_{\bar{I}})
\end{align}
holds for all $\bof{x}_{\bar{I}}$, $\bof{u}_I$, $\bof{u}^{\prime}_I$, $\bof{u}_{\bar{I}}$, and where 
$\bof{u}_I$ stands for the variables with indices in the set $I$,  $\bof{u}_I=\{u_i|i\in I\}$, and $\bof{u}_{\bar{I}}$ for the variables with indices in the 
complementary set, 
i.e.,  $\bof{u}_{\bar{I}}=\{u_i|i\notin I\}$. 
\end{definition}
This definition implies that, for any partition of the interfaces of the system, from the interaction with one set of the interfaces no information can be inferred about the inputs that were given to the remaining set of interfaces. This condition 
is actually equivalent to requiring that the behaviour of all but one interfaces gives no information about the input that was given to this one interface. 

\begin{lemma}\label{lemma:nscondsimplified}
An $n$-party system $P_{\bof{X}|\bof{U}}$ is non-signalling if and only if for all  $i\in \{1,\dotsc,n\}$, 
\begin{align}
\nonumber \sum
_{x_i} P_{\bof{X}|\bof{U}}(\bof{x},u_i,\bof{u}_{\bar{i}}) &=
 \sum
_{x_i} P_{\bof{X}|\bof{U}}(\bof{x},u^{\prime}_i,\bof{u}_{\bar{i}})\ ,
\end{align}
where $\bof{u}_{\bar{i}}$ stands for $u_1\dotso u_{i-1}u_{i+1}\dotso u_n$.
\end{lemma}
\begin{proof}
The condition is necessary, because it is simply the non-signalling condition for the set $I=\{i\}$. To see that it is sufficient, note that for any set~$I$
\begin{align}
 \nonumber \sum_{x_i:i\in I} P_{\bof{X}|\bof{U}}(\bof{x},\bof{u})&=
 \sum_{x_i:i\in I} P_{\bof{X}|\bof{U}}(\bof{x},\bof{u}_{I},\bof{u}_{\bar{I}})\\
  \nonumber &= \sum_{x_i:i\in I\backslash \{j\}}
  \sum_{x_j} P_{\bof{X}|\bof{U}}(\bof{x},\bof{u}_{ I\backslash \{j\}},u_j)\\
  \nonumber 
  &=
 \sum_{x_i:i\in I\backslash \{j\}}
  \sum_{x_j} P_{\bof{X}|\bof{U}}(\bof{x},\bof{u}_{ I\backslash \{j\}},u^{\prime}_j) \\
        \nonumber &= 
   \sum_{x_i:i\in I\backslash \{j^{\prime}\}}
  \sum_{x_{j^{\prime}}} P_{\bof{X}|\bof{U}}(\bof{x},\bof{u}_{ I\backslash \{j,j^{\prime}\}},u^{\prime}_{j},u_{j^{\prime}})  \\
        \nonumber &= 
   \sum_{x_i:i\in I\backslash \{j^{\prime}\}}
  \sum_{x_{j^{\prime}}} P_{\bof{X}|\bof{U}}(\bof{x},\bof{u}_{ I\backslash \{j,j^{\prime}\}},\bof{u}^{\prime}_{\{j,j^{\prime}\}})  \\
  \nonumber &= \cdots \\
 &= 
   \nonumber \sum_{x_i:i\in I} P_{\bof{X}|\bof{U}}(\bof{x},\bof{u}^{\prime}_{I},\bof{u}_{\bar{I}})\ .\qedhere
 \end{align}
\end{proof}

Since the set of non-signalling systems can be described by linear constraints on the probabilities describing the system, it is often easier 
to deal with the strictly larger set of non-signalling systems than with the set of quantum systems. 
The set of non-signalling systems, in fact, forms again a convex polytope. 

\begin{example}[The PR~box~\cite{pr}]\label{ex:prbox}
The system in Figure~\ref{fig:prbox} is a non-sig\-nal\-ling system. It is called a \emph{PR~box} after Popescu and Rohrlich~\cite{pr}. 
\end{example}

\begin{figure}[h]
\centering
\psset{unit=0.525cm}
\pspicture*[](-2,-1)(8.5,10)
\pspolygon[linewidth=0pt,fillstyle=vlines,hatchcolor=lightgray](0,4.5)(4,4.5)(4,6)(0,6)
\pspolygon[linewidth=0pt,fillstyle=hlines,hatchcolor=lightgray](4,4.5)(8,4.5)(8,6)(4,6)
\psline[linewidth=0.5pt]{-}(0,6)(-1,7)
\rput[c]{0}(-0.25,6.75){\scriptsize{$X$}}
\rput[c]{0}(-0.75,6.25){\scriptsize{$Y$}}
\rput[c]{0}(-0.5,7.5){\large{$U$}}
\rput[c]{0}(-1.5,6.5){\large{$V$}}
\rput[c]{0}(2,7.5){\Large{$0$}}
\rput[c]{0}(6,7.5){\Large{$1$}}
\rput[c]{0}(1,6.5){\Large{$0$}}
\rput[c]{0}(3,6.5){\Large{$1$}}
\rput[c]{0}(5,6.5){\Large{$0$}}
\rput[c]{0}(7,6.5){\Large{$1$}}
\rput[c]{0}(-1.5,4.5){\Large{$0$}}
\rput[c]{0}(-1.5,1.5){\Large{$1$}}
\rput[c]{0}(-0.5,5.25){\Large{$0$}}
\rput[c]{0}(-0.5,3.75){\Large{$1$}}
\rput[c]{0}(-0.5,2.25){\Large{$0$}}
\rput[c]{0}(-0.5,0.75){\Large{$1$}}
\psline[linewidth=2pt]{-}(-1,0)(8,0)
\psline[linewidth=2pt]{-}(-1,6)(8,6)
\psline[linewidth=2pt]{-}(-1,3)(8,3)
\psline[linewidth=1pt]{-}(0,1.5)(8,1.5)
\psline[linewidth=1pt]{-}(0,4.5)(8,4.5)
\psline[linewidth=2pt]{-}(0,0)(0,7)
\psline[linewidth=2pt]{-}(8,0)(8,7)
\psline[linewidth=2pt]{-}(4,0)(4,7)
\psline[linewidth=1pt]{-}(2,0)(2,6)
\psline[linewidth=1pt]{-}(6,0)(6,6)
\rput[c]{0}(1,5.25){\Large{$\frac{1}{2}$}}
\rput[c]{0}(3,3.75){\Large{$\frac{1}{2}$}}
\rput[c]{0}(5,5.25){\Large{$\frac{1}{2}$}}
\rput[c]{0}(7,3.75){\Large{$\frac{1}{2}$}}
\rput[c]{0}(1,2.25){\Large{$\frac{1}{2}$}}
\rput[c]{0}(3,0.75){\Large{$\frac{1}{2}$}}
\rput[c]{0}(5,0.75){\Large{$\frac{1}{2}$}}
\rput[c]{0}(7,2.25){\Large{$\frac{1}{2}$}}
\rput[c]{0}(3,5.25){\Large{$0$}}
\rput[c]{0}(1,3.75){\Large{$0$}}
\rput[c]{0}(7,5.25){\Large{$0$}}
\rput[c]{0}(5,3.75){\Large{$0$}}
\rput[c]{0}(3,2.25){\Large{$0$}}
\rput[c]{0}(1,0.75){\Large{$0$}}
\rput[c]{0}(5,2.25){\Large{$0$}}
\rput[c]{0}(7,0.75){\Large{$0$}}
\endpspicture
\caption{\label{fig:prbox}The PR~box. The non-signalling condition corresponds to the requirement that the two hatched areas contain the same probability (and similar for other outputs). }
\end{figure}

The PR~box (Figure~\ref{fig:prbox}) reaches 
\begin{align}
\nonumber \frac{1}{4}\sum_{(x,y,u,v):x\oplus y =u\cdot v} P_{XY|UV}(x,y,u,v) &= 1\ ,
\end{align}
i.e., it not only violates the Bell inequality of Example~\ref{ex:chshineq} in Section~\ref{subsec:localsystem}, it also reaches the maximum of this expression.

For an $(n+1)$-party non-signalling system $P_{\bof{X}Z|\bof{U}W}$, the \emph{marginal} and \emph{conditional} systems are well-defined and, again, $n$-party non-signalling systems.  
\begin{lemma}
Let $P_{\bof{X}Z|\bof{U}W}$ be an $(n+1)$-party non-signalling system. Then the marginal system \begin{align}
\nonumber P_{\bof{X}|\bof{U}}(\bof{x},\bof{u}) &:= \sum_z P_{\bof{X}Z|\bof{U},W}(\bof{x},z,\bof{u},w)
\end{align}
 and the conditional system 
\begin{align}
\nonumber 
 P_{\bof{X}|\bof{U},W=w,Z=z}(\bof{x},\bof{u}) &:=\frac{1}{P_{Z|W=w}(z)}P_{\bof{X}Z|\bof{U},W}(\bof{x},z,\bof{u},w)
\end{align} 
  are $n$-party non-signalling systems. 
\end{lemma}
\begin{proof}
Let us first see that the conditional systems are non-signalling. By Lemma~\ref{lemma:nscondsimplified}, it holds that for each $i$, 
\begin{align}
\nonumber \sum
_{x_i} P_{\bof{X}Z|\bof{U}W}(\bof{x},z,u_i,\bof{u}_{\bar{i}},w) &=
 \sum
_{x_i} P_{\bof{X}|\bof{U}}(\bof{x},z,u^{\prime}_i,\bof{u}_{\bar{i}},w)\ .
\end{align}
Dividing both sides by ${{P_{Z|W=w}(z)}}$ implies that the conditional system $P_{\bof{X}|\bof{U},W=w,Z=z}(\bof{x},\bof{u})$ is non-signalling. \\
The marginal system is non-signalling because it is a linear combination of conditional systems and because the non-signalling condition is linear. 
\end{proof}
This property justifies dropping the input of the other parts of the system in the notation when considering the marginal system associated with a non-signalling system.

\chapter{Security Against Non-Signalling Adversaries}\label{ch:nsadversaries}

\section{Introduction}

\emph{Non-signalling cryptography} (sometimes also called \emph{relativistic cryp\-tog\-ra\-phy}), as introduced by Kent, bases its security on the impossibility of signalling between space-like separated events, as predicted by relativity theory. In secure multi-party computation, the property guaranteeing security is that any choice made during the protocol must be independent from any event occurring in a space-like separated location. In this way, realizing a secure coin toss by two mistrustful parties is straight-forward~\cite{kentcoin}: Both parties choose a value and send them to each other simultaneously. The outcome of the coin toss is the XOR of the two values. Both players only accept if they receive the message from the other player such that it must have been sent from the location of the other player before the reception of their own message. Since each player must have chosen its value independently of the other player's, they cannot bias the outcome of the coin toss. Based on the same principle, protocols for bit commitment can also be defined~\cite{kentbc1,kentbc2,colbeckphd}. 

In~\cite{bhk}, Barrett, Hardy, and Kent proposed a protocol for secure key agreement based on the non-signalling principle (see Section~\ref{sec:approaches}). The case of key agreement works slightly differently from the above description, because there are two players which cooperate and trust each other (as opposed to the case of multi-party computation, where the players cooperate but do not trust each other). On the other hand, the eavesdropper cannot be forced to interact with the legitimate parties. The non-signalling condition then enters the argument via the requirement that Alice and Bob must not be able to signal to each other by interacting with their quantum systems even \emph{given the eavesdropper's measurement outcome}. The secrecy of the key bit is based on the fact that there exist non-local correlations which imply that the outcomes must be completely independent of any information the eavesdropper can possibly hold. These correlations can be realized by measuring an entangled quantum state and additionally have the property that Alice's and Bob's outcomes are perfectly correlated. 
These properties are exactly what is necessary for a secure shared bit. 

An advantage of non-signalling key agreement is that its security proof is based on  observed correlations. It is independent from the question how these correlations were realized, such as the physical particles used to distribute them, the dimension of the Hilbert space or the exact working of the measurement device. These protocols are, therefore, naturally \emph{device-independent}. Of course, allowing an adversary to do anything compatible with the non-signalling principle might be more than what a quantum adversary can do. However, Barrett, Hardy, and Kent's protocol implies that security is possible in principle even against such powerful adversaries.

The protocol of Barrett, Hardy, and Kent (see Figure~\ref{fig:bhk}, p.~\pageref{fig:bhk}) is secure against the most general type of attacks --- in the context of quantum key distribution these are called \emph{coherent attacks}. The adversary can directly attack  the key, independently of whether the physical realization of the protocol was made using several systems. Unfortunately, the security of the resulting key bit is only proportional to the number of systems and measurement bases used. Furthermore, the correlations need to be perfect for Alice and Bob not to abort, i.e., no noise can be tolerated. These properties imply that the protocol has zero key rate. 

When restricting the type of attacks an adversary can make, these problems can be overcome. In fact, there exist (noisy) non-local correlations with a finite number of inputs that imply \emph{partial} secrecy against a non-signalling eavesdropper, i.e., the outcome can be biased but not perfectly known. When Eve has to try to guess each bit of the raw key independently and individually, i.e., she is restricted to individual attacks, this implies that Alice and Bob can extract a secure key by applying \emph{information reconciliation} and \emph{privacy amplification}~\cite{AcinGisinMasanes,AcinMassarPironio,SGBMPA}. This works in the same way as against a purely classical adversary. 
However, generally we would not like to make such a restriction and it is unclear whether these schemes remain secure. In fact, consider privacy amplification: Alice and Bob apply a public hash function to their raw key. An adversary able to do arbitrary attacks can now \emph{directly} attack the final key, without having to learn anything about the raw key. Indeed, in Chapter~\ref{ch:impossibiltiy} we show that, unless Alice and Bob apply further countermeasures, the final key is only roughly as secure as the individual bits against a non-signalling adversary able to do collective attacks. 

In this chapter, we study privacy amplification of non-signalling secrecy under the following such countermeasure: We require the non-signalling condition not only to  hold between Alice, Bob, and Eve, but also between each of the subsystems.

\subparagraph*{Chapter outline}
We first characterize the exact possibilities that a non-sig\-nal\-ling adversary has to attack a system (Section~\ref{sec:model}) and give the description of the setup we consider (Section~\ref{subsec:nssetup}). We show how non-local systems can imply partial secrecy against non-signalling adversaries in Section~\ref{subsec:prboxdistance}, and give a general way to calculate the secrecy of a bit using a linear program (Section~\ref{subsec:bitlp}). In Section~\ref{subsec:severalns}, we consider the case of several systems and express the non-signalling condition for several systems in terms of the non-signalling conditions for the subsystems. This insight leads directly to an XOR-Lemma for non-signalling secrecy, i.e., the XOR can be used as a fixed privacy-amplification function, see Section~\ref{subsec:nsxor}. In Section~\ref{sec:nskeyagreement}, we construct a general key-agreement scheme from several partially secure non-signalling systems, and give a specific protocol in Section~\ref{sec:protocol}.

\subparagraph*{Related work}
The idea of basing secrecy on the non-signalling principle was introduced by Barrett, Hardy, and Kent~\cite{bhk}. Key agreement against non-signalling adversaries when allowing restricted (individual) attacks was shown in~\cite{AcinGisinMasanes,AcinMassarPironio,SGBMPA}. In~\cite{lluis}, Masanes showed that privacy amplification against non-signalling adversaries works using a fixed function if an additional non-sig\-nal\-ling condition holds between the subsystems. The proof is specific for the case of the CHSH inequality or its generalization, the Braunstein-Caves inequality (see Section~\ref{subsec:localsystem}), and is non-constructive, i.e., no explicit function for privacy amplification is given. Recently, Masanes showed that, in the above case, choosing the privacy amplification from a two-universal set is sufficient~\cite{masanesv4}.

\subparagraph*{Contributions}
The main technical contributions of this chapter are \linebreak[4] Lem\-ma~\ref{lemma:product_form}, relating the non-signalling condition of several systems to the ones for each subsystem and the XOR-Lemma for non-signalling secrecy (Theorem~\ref{th:nsxor}). 
Some results of this chapter have previously been published in~\cite{eurocrypt}.

\section{Modelling Non-Signalling Adversaries}\label{sec:model}

\begin{figure}[h]
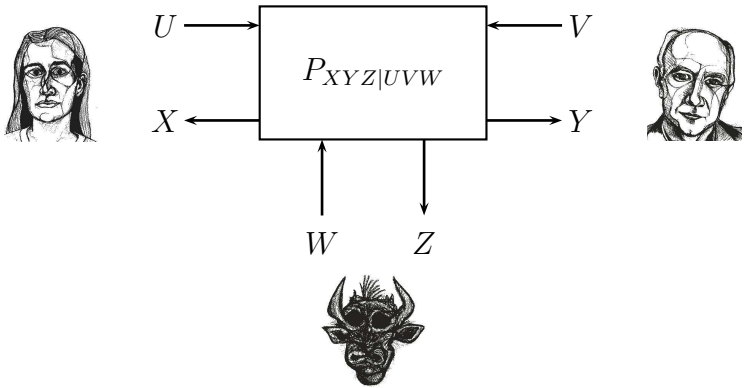

\centering
\pspicture*[](-5.4,-3.3)(5.4,1.95)
\pspolygon[](-1.5,0)(1.5,0)(1.5,1.75)(-1.5,1.75)
\rput[b]{0}(0,0.65){\large{$P_{XYZ|UVW}$}}
\psline[linewidth=1pt]{->}(-2.5,1.5)(-1.5,1.5)
\psline[linewidth=1pt]{<-}(-2.5,0.25)(-1.5,0.25)
\psline[linewidth=1pt]{->}(2.5,1.5)(1.5,1.5)
\psline[linewidth=1pt]{<-}(2.5,0.25)(1.5,0.25)
\rput[b]{0}(-2.75,1.35){\large{$U$}}
\rput[b]{0}(2.75,1.35){\large{$V$}}
\rput[b]{0}(-2.75,0.1){\large{$X$}}
\rput[b]{0}(2.75,0.1){\large{$Y$}}
\psline[linewidth=1pt]{->}(-0.675,-1)(-0.675,0)
\psline[linewidth=1pt]{<-}(0.675,-1)(0.675,0)
\rput[b]{0}(-0.6755,-1.5){\large{$W$}}
\rput[b]{0}(0.675,-1.5){\large{$Z$}}
\rput[c]{0}(-4.25,0.75){\includegraphics[width=1.25cm]{alice_small_dither.eps}}
\rput[c]{0}(4.25,0.75){\includegraphics[width=1.25cm]{bob_small_dither.eps}}
\rput[c]{0}(0,-2.5){\includegraphics[width=1.4cm]{eve_small_dither.eps}}
\endpspicture
\caption{The tripartite scenario including the eavesdropper.}
\label{tripartite-situation}
\end{figure}

In non-signalling key distribution, the measurements of Alice and Bob on some kind of physical system are abstractly modelled as a probability distribution $P_{XY|UV}$. 
This distribution must be non-signalling. 
A \emph{non-signalling adversary} is  an additional interface to the system shared by Alice and Bob, such that the resulting tripartite system $P_{XYZ|UVW}$ is still non-signalling between all parties. Of course, there is no need to limit the honest parties to two, there could be arbitrarily many: Alice, Bob, Charlie, etc. In particular, the case when Alice and Bob share $n$ different subsystems can be seen as the case of $2n$ parties (plus the eavesdropper).  The fact that we model the eavesdropper as a single interface even if the honest parties share several subsystems reflects the eavesdroppers ability to attack \emph{all} systems jointly. 

In fact, the \emph{only} restriction we will make on the ways the adversary can interact with the system is that the system between the honest parties and the adversary is non-signalling.
\begin{condition}\label{condition:ns}
The system $P_{\bof{XY}Z|\bof{UV}W}$ must be a $(2n+1)$-party non-sig\-nal\-ling system. 
\end{condition}
The non-signalling condition is motivated by quantum mechanics where measurements on different parts of an entangled quantum state cannot be used for message transmission. It, therefore, follows from the assumption usually made in quantum key distribution that, once the physical system is distributed, it can be modelled as an entangled quantum state and each party can only act (perform a measurement) on their part of the Hilbert space.  
However, Condition~\ref{condition:ns} is really equivalent to the condition that the honest parties have secure laboratories, in the sense  that no (unauthorized) information must leak to any other party --- in particular, no information is leaked via the physical system. It is clear that no cryptography is possible if this condition does not hold, for example, if Alice's laboratory contains a transmitter sending the key (or even the secret!) to the eavesdropper (see also Section~\ref{sec:whydi}).  
Note that the non-signalling condition between the honest parties and their subsystems can be guaranteed by either building several laboratories \emph{within} the laboratories or by measuring the physical systems in a space-like separated way\footnote{In special relativity, \emph{space-like separated} means 
that the coordinates of the events fulfil $c^2 \Delta t^2-|\Delta \overrightarrow{x}|^2<0$, where $c$ is the speed of light, and implies that there exists a reference frame according to which the two events occur simultaneously.}, in which case information transmission between them is ruled out by relativity theory.

\begin{figure}[h]
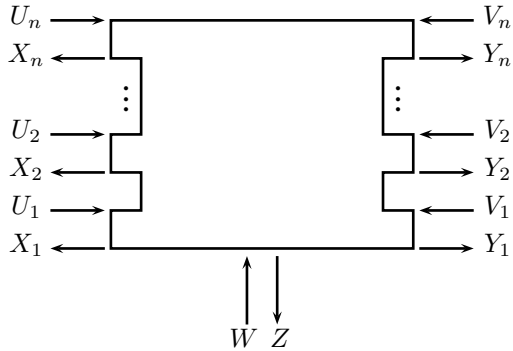

\centering
\pspicture*[](-4.5,0.2)(4.5,5)
\psset{xunit=0.8cm}
\rput[b]{0}(-2.25,3.3){\textbf{$\vdots$}}
\rput[b]{0}(2.25,3.3){\textbf{$\vdots$}}
\psline[linewidth=1pt]{->}(-3.5,4.5)(-2.6,4.5)
\rput[b]{0}(-3.9,4.4){$U_n$}
\psline[linewidth=1pt]{<-}(-3.5,4)(-2.6,4)
\rput[b]{0}(-3.9,3.9){$X_n$}
\psline[linewidth=1pt]{->}(3.5,4.5)(2.6,4.5)
\rput[b]{0}(3.9,4.4){$V_n$}
\psline[linewidth=1pt]{<-}(3.5,4)(2.6,4)
\rput[b]{0}(3.9,3.9){$Y_n$}
\psline[linewidth=1pt]{->}(-3.5,3)(-2.6,3)
\rput[b]{0}(-3.9,2.9){$U_{2}$}
\psline[linewidth=1pt]{<-}(-3.5,2.5)(-2.6,2.5)
\rput[b]{0}(-3.9,2.4){$X_{2}$}
\psline[linewidth=1pt]{->}(3.5,3)(2.6,3)
\rput[b]{0}(3.9,2.9){$V_{2}$}
\psline[linewidth=1pt]{<-}(3.5,2.5)(2.6,2.5)
\rput[b]{0}(3.9,2.4){$Y_{2}$}
\psline[linewidth=1pt]{->}(-3.5,2)(-2.6,2)
\rput[b]{0}(-3.9,1.9){$U_1$}
\psline[linewidth=1pt]{<-}(-3.5,1.5)(-2.6,1.5)
\rput[b]{0}(-3.9,1.4){$X_1$}
\psline[linewidth=1pt]{->}(3.5,2)(2.6,2)
\rput[b]{0}(3.9,1.9){$V_1$}
\psline[linewidth=1pt]{<-}(3.5,1.5)(2.6,1.5)
\rput[b]{0}(3.9,1.4){$Y_1$}
\rput[c](0,0.5){
\psline[linewidth=1pt]{->}(-0.25,0)(-0.25,0.9)
\rput[b]{0}(-0.3,-0.3){$W$}
\psline[linewidth=1pt]{<-}(0.25,0)(0.25,0.9)
\rput[b]{0}(0.3,-0.3){$Z$}
}
\psline[linewidth=1pt](-2,4.5)(-2.5,4.5)(-2.5,4)(-2,4)(-2,3)(-2.5,3)(-2.5,2.5)(-2,2.5)(-2,2)(-2.5,2)(-2.5,1.5)(-2,1.5)(2.5,1.5)(2.5,2)(2,2)(2,2.5)(2.5,2.5)(2.5,3)(2,3)(2,4)(2.5,4)(2.5,4.5)(2,4.5)(-2,4.5)
\endpspicture
\caption{\label{figure:evesposs} Alice and Bob share $n$  
systems. Eve can attack all of them at once.}
\end{figure}

\subsection{Possible attacks}

In order to define the exact possibilities a non-signalling adversary has to attack a system, we define a \emph{non-signalling partition} as a convex decomposition of the non-signalling system $P_{\bof{X}|\bof{U}}$ (see Figure~\ref{fig:boxpartition2}). 

\begin{figure}[h]
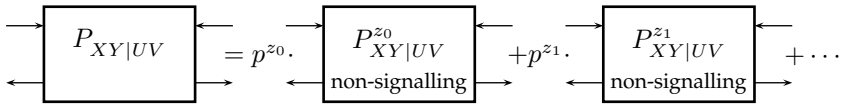

\centering
\pspicture*[](-1.1,-0.1)(10.2,1.6)
\psset{unit=1cm}
\rput[c]{0}(-0.5,0){
\pspolygon[linewidth=1pt](0,0)(0,1.25)(2,1.25)(2,0)
\rput[c]{0}(1,0.75){$P^{\phantom{z}}_{XY|UV}$}
\psline[linewidth=0.5pt]{->}(-0.5,1)(0,1)
\psline[linewidth=0.5pt]{<-}(-0.5,0.25)(0,0.25)
\psline[linewidth=0.5pt]{<-}(2,1)(2.5,1)
\psline[linewidth=0.5pt]{->}(2,0.25)(2.5,0.25)
}
\rput[c]{0}(3.2,0){
\rput[c]{0}(-0.85,0.6){$=p^{z_0}\cdot $}
\pspolygon[linewidth=1pt](0,0)(0,1.25)(2,1.25)(2,0)
\rput[c]{0}(1,0.75){$P^{z_0}_{XY|UV}$}
\rput[c]{0}(1,0.25){\footnotesize{non-signalling}}
\psline[linewidth=0.5pt]{->}(-0.5,1)(0,1)
\psline[linewidth=0.5pt]{<-}(-0.5,0.25)(0,0.25)
\psline[linewidth=0.5pt]{<-}(2,1)(2.5,1)
\psline[linewidth=0.5pt]{->}(2,0.25)(2.5,0.25)
}
\rput[c]{0}(6.9,0){
\rput[c]{0}(-0.85,0.6){$+p^{z_1}\cdot $}
\pspolygon[linewidth=1pt](0,0)(0,1.25)(2,1.25)(2,0)
\rput[c]{0}(1,0.75){$P^{z_1}_{XY|UV}$}
\rput[c]{0}(1,0.25){\footnotesize{non-signalling}}
\psline[linewidth=0.5pt]{->}(-0.5,1)(0,1)
\psline[linewidth=0.5pt]{<-}(-0.5,0.25)(0,0.25)
\psline[linewidth=0.5pt]{<-}(2,1)(2.5,1)
\psline[linewidth=0.5pt]{->}(2,0.25)(2.5,0.25)
}
\rput[c]{0}(9.7,0.6){$+\dotsb $}
\endpspicture
\caption{\label{fig:boxpartition2}By Lemmas~\ref{lemma:nsbox} and~\ref{lemma:boxns}, an attack of the eavesdropper corresponds to a choice of convex decomposition. Her outcome is an element in the convex decomposition. }
\end{figure}

\begin{definition}
A \emph{non-signalling partition} of a given $n$-party non-signalling system $P_{\bof{X}|\bof{U}}$ is a family of pairs $\{(p^{z_w}$,$P^{z_w}_{\bof{X}|\bof{U}})\}_{z_w}$, where $p^{z_w}$ is a weight and $P^{z_w}_{\bof{X}|\bof{U}}$ is an $n$-party non-signalling system, such that
\begin{align}
\label{eq:boxpart} P_{\bof{X}|\bof{U}}&=\sum_{z_w} p^{z_w}\cdot P^{z_w}_{\bof{X}|\bof{U}} \ .
\end{align}
\end{definition}
The non-signalling partition defines exactly the possible extensions of a given $n$-party non-signalling system to an $(n+1)$-party non-signalling system and, therefore, the possibilities a non-signalling adversary has to attack the system $P_{\bof{X}|\bof{U}}$. This is stated in Lemmas~\ref{lemma:nsbox} and~\ref{lemma:boxns}. 
\begin{lemma}\label{lemma:nsbox}
For any given $(n+1)$-party non-signalling system, $P_{\bof{X}Z|\bof{U}W}$, any input $w$ induces a non-signalling partition of the $n$-party non-signalling system $P_{\bof{X}|\bof{U}}$, parametrized by $z$, with $p^{z_w}:=P_{Z|W=w}(z)$ and $P^{z_w}_{\bof{X}|\bof{U}}:=P_{\bof{X}|\bof{U},Z=z,W=w}$.
\end{lemma}
\begin{proof}
Since $P_{\bof{X}Z|\bof{U}W}$ is an $(n+1)$-party non-signalling system, the marginal system $P_{\bof{X}|\bof{U}}$ and the conditional systems $P_{\bof{X}|\bof{U},Z=z,W=w}$ are $n$-party non-signalling systems. 
 For a given $W=w$, $P_{Z|W=w}$ is a probability distribution and, therefore, $p^{z_w}:=P_{Z|W=w}(z)$ is a weight. 
Equation (\ref{eq:boxpart}) holds by the definition of the marginal system. 
\end{proof} 
\begin{lemma}\label{lemma:boxns}
Given an $n$-party non-signalling system $P_{\bof{X}|\bof{U}}$, let $\mathcal{W}$ be a set of non-signalling partitions, 
$w=\{(p^{z_w},P^{z_w}_{\bof{X}Z|\bof{U}})\}_{z_w}$. 
Then the $(n+1)$-party system where the input of the last party is $w\in\mathcal{W}$, defined by 
\begin{align}
\nonumber P_{\bof{X}Z|\bof{U},W}(\bof{x},z,\bof{u},w)&:=p^{z_w}\cdot P^{z_w}_{\bof{X}|\bof{U}}(\bof{x},\bof{u})\ ,
\end{align} 
is non-signalling and has marginal system $P_{\bof{X}|\bof{U}}$.
\end{lemma}
\begin{proof}
To see that it has the correct marginal system, note that for any $w$, $\sum_{z_w} p^{z_w}\cdot P^{z_w}_{\bof{X}|\bof{U}}=P_{\bof{X}|\bof{U}}$ by (\ref{eq:boxpart}). To see that it is non-signalling, 
consider Lemma~\ref{lemma:nscondsimplified}, p.~\pageref{lemma:nscondsimplified}. We have
\begin{align}
\nonumber \sum
_{x_i} P_{\bof{X}Z|\bof{U}W}(\bof{x},z,u_i,\bof{u}_{\bar{i}},w)&=
 \sum
_{x_i} P_{\bof{X}Z|\bof{U}W}(\bof{x},z,u^{\prime}_i,\bof{u}_{\bar{i}},w)
\end{align}
because the conditional systems $P^{z_w}_{\bof{X}|\bof{U}}(\bof{x},\bof{u})$ are $n$-party non-signalling. Additionally, 
\begin{align}
\nonumber \sum
_{z} P_{\bof{X}Z|\bof{U}W}(\bof{x},z,\bof{u},w)&=
 \sum
_{z} P_{\bof{X}Z|\bof{U}W}(\bof{x},z,\bof{u},w^{\prime})\ ,
\end{align}
holds by (\ref{eq:boxpart}).
\end{proof}

\subsection{Security of our key-agreement protocol}\label{subsec:nssetup}

The setup we consider (see Figure~\ref{fig:our_system_ns}) is the one where 
Alice and Bob share a public authenticated channel plus some kind of physical system, modelled as a non-signalling system. They can interact with the physical system (i.e., give inputs and obtain outputs). Using the public authenticated channel, they can then apply a protocol to their inputs and outputs in order to obtain a shared secret key. 

Eve can wire-tap the public channel, choose an input on her part of the system and obtain an output. 
The following lemma states that it is no advantage for Eve to make several non-signalling partitions (measurements) instead of a single one, as the same information can be obtained by making a refined non-signalling partition of the initial system. Without loss of generality, we can, therefore, assume that Eve gives a single input to the system
 at the end (after all communication between Alice and Bob is finished).
\begin{lemma}
Let $w$ be a non-signalling partition of a non-signalling system $P_{\bof{X}|\bof{U}}$, with elements  $\{(p^{z_w},P^{z_w}_{\bof{X}|\bof{U}} )\}_{z_w}$, 
  and let $w^{\prime}_z$ be a set of non-signalling partitions of the non-signalling systems $P^{z_w}_{\bof{X}|\bof{U}}$, 
with elements \linebreak[4] $\{( p^{z^{\prime}_{w^{\prime}_z}},P^{z_w,z^{\prime}_{w^{\prime}_z}}_{\bof{X}|\bof{U}} )\}_{z^{\prime}_{w^{\prime}_z}}$. 
  Then there exists a non-signalling partition of $P_{\bof{X}|\bof{U}}$ with elements  
$\{(p^{z_w} p^{z^{\prime}_{w^{\prime}_z}},P^{z_w,z^{\prime}_{w^{\prime}_z}}_{\bof{X}|\bof{U}} )\}_{z_w,z^{\prime}_{w^{\prime}_z}}$. 
\end{lemma} 
 \begin{proof}
Since $p^{z_w}$ and $p^{z^{\prime}_{w^{\prime}_z}}$ are weights, their product is also a weight. The distributions $P^{z_w,z^{\prime}_{w^{\prime}_z}}_{\bof{X}|\bof{U}}$ are $n$-party non-signalling systems because they are elements of the non-signalling partition $w^{\prime}_z$. Finally, 
 \begin{align}
\nonumber \sum_{z_w,z^{\prime}_{w^{\prime}_z}}p^{z_w} p^{z^{\prime}_{w^{\prime}_z}}\cdot  P^{z_w,z^{\prime}_{w^{\prime}_z}}_{\bof{X}|\bof{U}} &= 
\sum_{z_w}p^{z_w}\cdot\Biggl( \sum_{z^{\prime}_{w^{\prime}_z}}  p^{z^{\prime}_{w^{\prime}_z}} P^{z_w,z^{\prime}_{w^{\prime}_z}}_{\bof{X}|\bof{U}}\Biggr)\\
\nonumber &= 
\sum_{z_w}p^{z_w}\cdot P^{z_w}_{\bof{X}|\bof{U}} = P_{\bof{X}|\bof{U}}\ ,
 \end{align}
 where we have first used that $w^{\prime}_z$ is a non-signalling partition of $P^{z_w}_{\bof{X}|\bof{U}}$ and then that $w$ is a non-signalling partition of $P_{\bof{X}|\bof{U}}$. 
 \end{proof}

In our \emph{real} scenario (see Figure~\ref{fig:our_system_ns}), Alice, therefore, uses the inputs and outputs $U$ and $X$ of the system and the information $Q$ exchanged over the public authenticated channel to create a string $S_A$. Bob uses $V$ and $Y$ and the information $Q$ to create $S_B$. Eve obtains all the information $Q$ exchanged over the public authenticated channel, 
can 
then choose the input to her system $W$ (which can depend on 
$Q$) and finally obtains the outcome $Z$ of the system. 

We define security by comparing this \emph{real} scenario to an \emph{ideal} scenario which is secure by definition (see Section~\ref{sec:randomsystems}). 
In the \emph{ideal} scenario, Alice and Bob output the same uniformly distributed string, and the system Eve interacts with is completely uncorrelated with it.  Our goal is to bound the distinguishing advantage between the real and ideal system.

\begin{figure}[htp!]
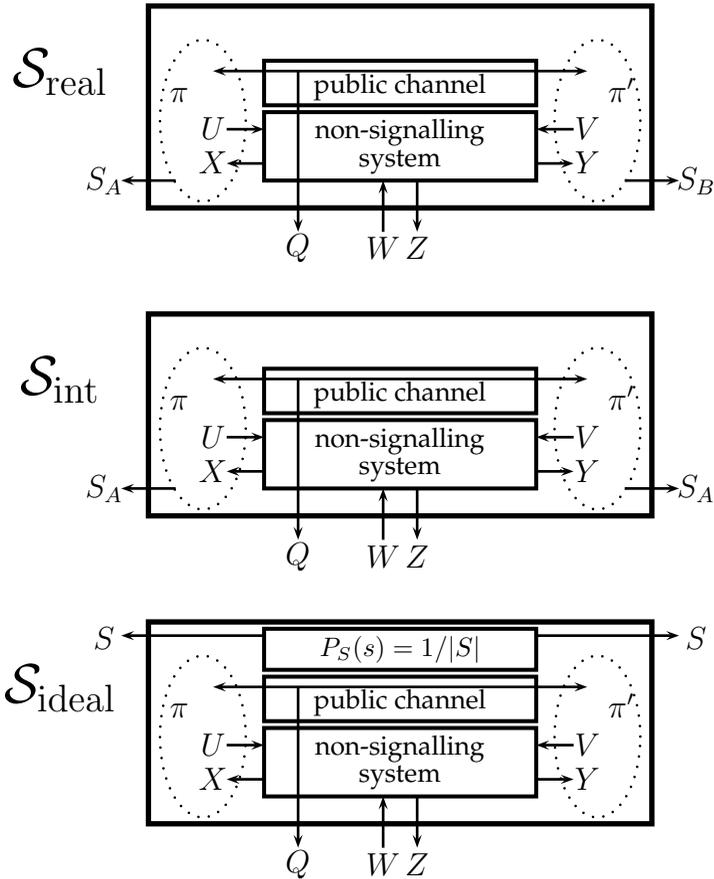

\centering
\pspicture*[](-6,-9)(5,3.5)
\psset{unit=0.9cm}
\rput[c]{0}(0,0){
\rput[b]{0}(-5,2.5){\huge{$\mathcal{S}_{\mathrm{real}}$}}
\psline[linewidth=1pt]{<->}(-2.75,2.85)(2.75,2.85)
\rput[b]{0}(0,2.425){public channel}
\pspolygon[linewidth=1.5pt](-2,3)(2,3)(2,2.35)(-2,2.35)
\pspolygon[linewidth=1.5pt](-2,1.25)(2,1.25)(2,2.25)(-2,2.25)
\psline[linewidth=1pt]{<-}(-2.55,1.5)(-2,1.5)
\rput[b]{0}(-2.75,1.35){\large{$X$}}
\psline[linewidth=1pt]{->}(-2.55,2)(-2,2)
\rput[b]{0}(-2.75,1.85){\large{$U$}}
\psline[linewidth=1pt]{<-}(2.55,1.5)(2,1.5)
\rput[b]{0}(2.75,1.35){\large{$Y$}}
\psline[linewidth=1pt]{->}(2.55,2)(2,2)
\rput[b]{0}(2.75,1.85){\large{$V$}}
\psline[linewidth=1pt]{->}(-0.25,0.5)(-0.25,1.25)
\rput[b]{0}(-0.25,0.1){\large{$W$}}
\psline[linewidth=1pt]{<-}(0.25,0.5)(0.25,1.25)
\rput[b]{0}(0.25,0.1){\large{$Z$}}
\rput[b]{0}(0,1.75){non-signalling}
\rput[b]{0}(0,1.35){system}
\psline[linewidth=1pt]{<-}(-1.5,0.5)(-1.5,2.85)
\rput[b]{0}(-1.5,0.05){\large{$Q$}}
\psellipse[linewidth=1pt,linestyle=dotted](-2.9,2.125)(0.65,1.2)
\psellipse[linewidth=1pt,linestyle=dotted](2.9,2.125)(0.65,1.2)
\pspolygon[linewidth=2pt](-3.7,0.85)(3.7,0.85)(3.7,3.8)(-3.7,3.8) 
\rput[b]{0}(-3.25,2.4){\large{$\pi$}}
\rput[b]{0}(3.25,2.4){\large{$\pi^{\prime}$}}
\psline[linewidth=1pt]{<-}(-4.1,1.25)(-3.3,1.25)
\rput[b]{0}(-4.35,1.05){\large{$S_A$}}
\psline[linewidth=1pt]{<-}(4.1,1.25)(3.3,1.25)
\rput[b]{0}(4.35,1.05){\large{$S_B$}}
}
\rput[c]{0}(0,-4.5){
\rput[b]{0}(-5,2.5){\huge{$\mathcal{S}_{\mathrm{int}}$}}
\psline[linewidth=1pt]{<->}(-2.75,2.85)(2.75,2.85)
\rput[b]{0}(0,2.425){public channel}
\pspolygon[linewidth=1.5pt](-2,3)(2,3)(2,2.35)(-2,2.35)
\pspolygon[linewidth=1.5pt](-2,1.25)(2,1.25)(2,2.25)(-2,2.25)
\psline[linewidth=1pt]{<-}(-2.55,1.5)(-2,1.5)
\rput[b]{0}(-2.75,1.35){\large{$X$}}
\psline[linewidth=1pt]{->}(-2.55,2)(-2,2)
\rput[b]{0}(-2.75,1.85){\large{$U$}}
\psline[linewidth=1pt]{<-}(2.55,1.5)(2,1.5)
\rput[b]{0}(2.75,1.35){\large{$Y$}}
\psline[linewidth=1pt]{->}(2.55,2)(2,2)
\rput[b]{0}(2.75,1.85){\large{$V$}}
\psline[linewidth=1pt]{->}(-0.25,0.5)(-0.25,1.25)
\rput[b]{0}(-0.25,0.1){\large{$W$}}
\psline[linewidth=1pt]{<-}(0.25,0.5)(0.25,1.25)
\rput[b]{0}(0.25,0.1){\large{$Z$}}
\rput[b]{0}(0,1.75){non-signalling}
\rput[b]{0}(0,1.35){system}
\psline[linewidth=1pt]{<-}(-1.5,0.5)(-1.5,2.85)
\rput[b]{0}(-1.5,0.05){\large{$Q$}}
\psellipse[linewidth=1pt,linestyle=dotted](-2.9,2.125)(0.65,1.2)
\psellipse[linewidth=1pt,linestyle=dotted](2.9,2.125)(0.65,1.2)
\pspolygon[linewidth=2pt](-3.7,0.85)(3.7,0.85)(3.7,3.8)(-3.7,3.8) 
\rput[b]{0}(-3.25,2.4){\large{$\pi$}}
\rput[b]{0}(3.25,2.4){\large{$\pi^{\prime}$}}
\psline[linewidth=1pt]{<-}(-4.1,1.25)(-3.3,1.25)
\rput[b]{0}(-4.35,1.05){\large{$S_A$}}
\psline[linewidth=1pt]{<-}(4.1,1.25)(3.3,1.25)
\rput[b]{0}(4.35,1.05){\large{$S_A$}}
}
\rput[c]{0}(0,-9){
\rput[b]{0}(-5,2.5){\huge{$\mathcal{S}_{\mathrm{ideal}}$}}
\psline[linewidth=1pt]{<->}(-2.75,2.85)(2.75,2.85)
\rput[b]{0}(0,2.425){public channel}
\pspolygon[linewidth=1.5pt](-2,3)(2,3)(2,2.35)(-2,2.35)
\pspolygon[linewidth=1.5pt](-2,1.25)(2,1.25)(2,2.25)(-2,2.25)
\psline[linewidth=1pt]{<-}(-2.55,1.5)(-2,1.5)
\rput[b]{0}(-2.75,1.35){\large{$X$}}
\psline[linewidth=1pt]{->}(-2.55,2)(-2,2)
\rput[b]{0}(-2.75,1.85){\large{$U$}}
\psline[linewidth=1pt]{<-}(2.55,1.5)(2,1.5)
\rput[b]{0}(2.75,1.35){\large{$Y$}}
\psline[linewidth=1pt]{->}(2.55,2)(2,2)
\rput[b]{0}(2.75,1.85){\large{$V$}}
\psline[linewidth=1pt]{->}(-0.25,0.5)(-0.25,1.25)
\rput[b]{0}(-0.25,0.1){\large{$W$}}
\psline[linewidth=1pt]{<-}(0.25,0.5)(0.25,1.25)
\rput[b]{0}(0.25,0.1){\large{$Z$}}
\rput[b]{0}(0,1.75){non-signalling}
\rput[b]{0}(0,1.35){system}
\psline[linewidth=1pt]{<-}(-1.5,0.5)(-1.5,2.85)
\rput[b]{0}(-1.5,0.05){\large{$Q$}}
\psellipse[linewidth=1pt,linestyle=dotted](-2.9,2.125)(0.65,1.2)
\psellipse[linewidth=1pt,linestyle=dotted](2.9,2.125)(0.65,1.2)
\pspolygon[linewidth=2pt](-3.7,0.85)(3.7,0.85)(3.7,3.8)(-3.7,3.8) 
\rput[b]{0}(-3.25,2.4){\large{$\pi$}}
\rput[b]{0}(3.25,2.4){\large{$\pi^{\prime}$}}
\pspolygon[linewidth=1.5pt](-2,3.1)(2,3.1)(2,3.7)(-2,3.7)
\rput[b]{0}(0,3.2){$P_S(s)=1/|S|$}
\psline[linewidth=1pt]{<-}(-4.1,3.6)(-2,3.6)
\rput[b]{0}(-4.35,3.4){\large{$S$}}
\psline[linewidth=1pt]{<-}(4.1,3.6)(2,3.6)
\rput[b]{0}(4.35,3.4){\large{$S$}}
}
\endpspicture
\caption{\label{fig:our_system_ns} Our \emph{real} system (top). Alice and Bob share a public authenticated channel and a non-signalling system. When they apply a protocol $(\pi,\pi^{\prime})$ to obtain a key, all this can together be modelled as a system. In our \emph{ideal} system (bottom),  the system outputs a uniform random string $S$ to both Alice and Bob. We also use an \emph{intermediate} system (middle) in our calculations, which outputs $S_A$ to both Alice and Bob. }
\end{figure}

In order to bound the distance between the real and ideal system, we introduce an intermediate system (see Figure~\ref{fig:our_system_ns}). Using the triangle inequality (Lemma~\ref{lemma:distance}, p.~\pageref{lemma:distance}) we can bound the distance between the real and ideal system by the sum of the distance between real or ideal system and the intermediate system. Note that the distance between the real and intermediate system is the parameter characterizing the \emph{correctness} of the protocol, whereas the distance between the intermediate and the ideal system characterizes the \emph{secrecy} (see Section~\ref{subsec:securitykey}). 

In order to estimate the secrecy of the protocol, we introduce the distance from uniform of the key string $S_A$ from the eavesdropper's point of view. We will in the following 
call it \emph{the distance from uniform of $S_A$ given $Z(W_{\mathrm{n-s}})$ and $Q$}, where we write $Z(W_{\mathrm{n-s}})$ because the eavesdropper 
can choose the input adaptively and the choice of input changes the output distribution. 
\begin{definition}\label{def:dist-sa-from_uniform}
Consider a system $\mathcal{S}_{\mathrm{real}}$ as depicted in Figure~\ref{fig:our_system_ns}. 
The \emph{distance from uniform of $S_A$ given $Z(W_{\mathrm{n-s}})$ and $Q$} is 
\begin{multline}
\nonumber
 d(S_A|Z(W_{\mathrm{n-s}}),Q) \\
= \frac{1}{2}
\sum_{s_A,q}  \max_{w:{\mathrm{n-s}}} \sum_{z} P_{Z,Q|W=w}(z,q)
 \cdot |P_{S_A|Z=z,Q=q,W=w}(s_A)-P_U(s_A)|\ ,
\end{multline}
where $P_U:=1/|\mathcal{S}_A|$ and the maximization is over all non-signalling systems $P_{XYZ|UVW}$. 
\end{definition}
 It will be useful to 
define the distance from uniform of a string $S$ given a \emph{specific} adversarial strategy $w$. 
To denote this difference, we will denote the strategy by a lower case letter.
\begin{definition}\label{def:distancesmallw}
The \emph{distance from uniform of $S$ given $Z(w)$ and $Q$} is 
\begin{align}
\nonumber 
 d(S|Z(w),Q)
&=
\frac{1}{2}
\sum_{s,q}\sum_{z} P_{Z,Q|W=w}(z,q)\cdot \left|P_{S|Z=z,Q=q,W=w}(s)-\frac{1}{|\mathcal{S}_A|}\right|\ .
\end{align}
\end{definition}

The following corollary is a direct consequence\footnote{Note that, because the system considered is non-signalling, we can 
think of a box giving outputs indexed by $w$, $Z_w$, of which one is selected instead of a system taking
input $W$.} of the definitions of the systems in Figure~\ref{fig:our_system_ns} and the distinguishing advantage.
\begin{corollary}
\label{corr:dist_advantage}
Assume a key $S_A$ generated by a system as given in Figure~\ref{fig:our_system_ns}. Then
\begin{align}
\nonumber 
\delta(\mathcal{S}_{\mathrm{int}},\mathcal{S}_{\mathrm{ideal}}) &= d(S_A|Z(W_{\mathrm{n-s}}),Q)\ .
\end{align}
\end{corollary}

The distance from the intermediate system to the real system is exactly the probability 
that the real system outputs different values on the two sides. This is again a direct consequence of the definitions. 
\begin{corollary}
Assume a key $S_A$ generated by the intermediate system $\mathcal{S}_{\mathrm{int}}$ depicted in  Figure~\ref{fig:our_system_ns}. 
Then
\begin{align}
\nonumber 
\delta(\mathcal{S}_{\mathrm{real}},\mathcal{S}_{\mathrm{int}}) &=
\sum_{s_A\neq s_B} P_{S_AS_B}(s_A,s_B)\ .
\end{align}
\end{corollary}

By the triangle inequality for the distinguishing advantage of systems (Lemma~\ref{lemma:distance}, p.~\pageref{lemma:distance}), we obtain the following statement.
\begin{lemma}
\begin{align}
\nonumber 
\delta(\mathcal{S}_{\mathrm{real}},\mathcal{S}_{\mathrm{ideal}}) &\leq 
\delta(\mathcal{S}_{\mathrm{real}},\mathcal{S}_{\mathrm{int}})
+
\delta(\mathcal{S}_{\mathrm{int}},\mathcal{S}_{\mathrm{ideal}})\ .
\end{align}
\end{lemma}
In order to prove security, we will, therefore, have to show that this quantity is small, more precisely, we will show that $\delta(\mathcal{S}_{\mathrm{real}},\mathcal{S}_{\mathrm{ideal}})\leq \epsilon$, which implies 
that the key-distribution scheme is $\epsilon$-secure.

\section{Security of a Single System}

\subsection{A bipartite system with binary inputs and outputs}\label{subsec:prboxdistance}

Let us consider the case where Alice and Bob share a non-signalling system which takes one bit input and gives one bit output on each side. Alice and Bob choose a random input and obtain the output. Then, they exchange their inputs over the public authenticated channel, i.e., $Q=(U=u,V=v)$,\footnote{In a certain abuse of notation, we will allow  $Q$ to consist of both random variables and 
events that a random variable takes a given value. 
In case of such events, $U=u$,  this means that the 
distance from uniform will hold \emph{given this specific value 
$u$},  whereas taking the expectation  over $Q$ will correspond to taking the expectation  over all the `free' random variables 
contained in $Q$.}
 and take directly the output bit as secret key, i.e., $S_A=X$. 

Assume that the system fulfils 
\begin{align}
\nonumber \frac{1}{4}\sum_{(x,y,u,v):x\oplus y =u\cdot v} P_{XY|UV}(x,y,u,v) &= 1-\ep\ ,
\end{align}
i.e., for $\ep<1/4$, the system is non-local (see Definition~\ref{def:bellineq}, p.~\pageref{def:bellineq} and Example~\ref{ex:chshineq}, p.~\pageref{ex:chshineq}). Our goal is to show, that the bit $X$ is partially secret. In fact, its secrecy is proportional to the parameter $\ep$. We do not consider the correctness for the moment. 

\begin{figure}[h]
\centering
\psset{unit=0.525cm}
\pspicture*[](-2,-1)(8.5,10)
\psline[linewidth=0.5pt]{-}(0,6)(-1,7)
\rput[c]{0}(-0.25,6.75){\scriptsize{$X$}}
\rput[c]{0}(-0.75,6.25){\scriptsize{$Y$}}
\rput[c]{0}(-0.5,7.5){\large{$U$}}
\rput[c]{0}(-1.5,6.5){\large{$V$}}
\rput[c]{0}(2,7.5){\Large{$0$}}
\rput[c]{0}(6,7.5){\Large{$1$}}
\rput[c]{0}(1,6.5){\Large{$0$}}
\rput[c]{0}(3,6.5){\Large{$1$}}
\rput[c]{0}(5,6.5){\Large{$0$}}
\rput[c]{0}(7,6.5){\Large{$1$}}
\rput[c]{0}(-1.5,4.5){\Large{$0$}}
\rput[c]{0}(-1.5,1.5){\Large{$1$}}
\rput[c]{0}(-0.5,5.25){\Large{$0$}}
\rput[c]{0}(-0.5,3.75){\Large{$1$}}
\rput[c]{0}(-0.5,2.25){\Large{$0$}}
\rput[c]{0}(-0.5,0.75){\Large{$1$}}
\psline[linewidth=2pt]{-}(-1,0)(8,0)
\psline[linewidth=2pt]{-}(-1,6)(8,6)
\psline[linewidth=2pt]{-}(-1,3)(8,3)
\psline[linewidth=1pt]{-}(0,1.5)(8,1.5)
\psline[linewidth=1pt]{-}(0,4.5)(8,4.5)
\psline[linewidth=2pt]{-}(0,0)(0,7)
\psline[linewidth=2pt]{-}(8,0)(8,7)
\psline[linewidth=2pt]{-}(4,0)(4,7)
\psline[linewidth=1pt]{-}(2,0)(2,6)
\psline[linewidth=1pt]{-}(6,0)(6,6)
\rput[c]{0}(1,5.25){{$\frac{1}{2}+\ep$}}
\rput[c]{0}(2.95,3.75){{$\frac{1}{2}-2\ep$}}
\rput[c]{0}(5,5.25){{$\frac{1}{2}$}}
\rput[c]{0}(7,3.75){{$\frac{1}{2}-\ep$}}
\rput[c]{0}(1,2.25){{$\frac{1}{2}+\ep$}}
\rput[c]{0}(2.95,0.75){{$\frac{1}{2}-2\ep$}}
\rput[c]{0}(5,0.75){{$\frac{1}{2}-\ep$}}
\rput[c]{0}(7,2.25){{$\frac{1}{2}$}}
\rput[c]{0}(3,5.25){{$0$}}
\rput[c]{0}(1,3.75){{${\ep}$}}
\rput[c]{0}(7,5.25){{${\ep}$}}
\rput[c]{0}(5,3.75){{$0$}}
\rput[c]{0}(3,2.25){{$0$}}
\rput[c]{0}(1,0.75){{${\ep}$}}
\rput[c]{0}(5,2.25){{${\ep}$}}
\rput[c]{0}(7,0.75){{$0$}}
\endpspicture
\caption{A system with $\Pr[X\oplus Y=U\cdot V]=1-\ep$.}
\end{figure}

More precisely, we show the following statement. 
\begin{lemma}\label{lemma:guessing_probability_single_box}
Let $P_{XYZ|UVW}$ be a non-signalling system with $\mathcal{X}=\mathcal{Y}=\mathcal{U}=\mathcal{V}=\{0,1\}$ such that the marginal $P_{XY|UV}$ fulfils
\begin{align}
\nonumber
\frac{1}{4} \sum_{(x,y,u,v):x\oplus y =u\cdot v}P_{XY|UV}(x,y,u,v) &=1-\ep
\end{align}
and let $Q:=(U=u,V=v)$. Then 
\begin{align}
 \nonumber 
d(X|Z(W_{\mathrm{n-s}}),Q) &\leq 2\ep\ .
\end{align}
\end{lemma}
\begin{proof}
Consider w.l.o.g.\ the case $X=0$. 
We call $\ep_i$ the probability that ${X\oplus Y\neq U\cdot V}$ for the inputs $\{(0,0),(0,1),(1,0),(1,1)\}$, respectively. \linebreak[4] Suppose w.l.o.g.\ that the input was $(0,0)$, so $X$ should be maximally biased for this input. 
Since it holds that $\Pr[X\oplus Y\neq U\cdot V|U,V=0,0]=\ep_1$, the bias of $Y$, given $U=V=0$, must be at least $p-\ep_1$ (see Figure~\ref{fig:bias}). Because 
of non-signalling, $X$'s bias must be $p$ as well when $V=1$, and so on. Finally, 
$\Pr[X\oplus Y\neq U\cdot V|(U,V)=(1,1)]=\ep_4$ implies
$p-\ep_2-(1-(p-\ep_1-\ep_3))\leq \ep_4$,
hence, 
$
p\leq (1+\sum_i\ep_i)/2=1/2+2\ep
$.
Now consider a non-signalling partition of $P_{XY|UV}$ parametrized by $z$. 
Let $\ep_z$ denote the error of the system 
given $Z=z$, i.e., $\ep_z=(\sum_i \ep_{i,z})/4$. Since this system must still be non-signalling, 
the bias of $X$ given $Z=z$, $U=u$ and $V=v$ is at most $2\ep_z$ by the above argument. 
However, $P_{XY|UV}=\sum_z p^z\cdot P^z_{XY|UV}$, implies $\ep=\sum_z p^z\cdot \ep_z$ and this
holds for all values of $X$, therefore, $
d(X|Z(W_{\mathrm{n-s}}),Q)\leq 
 \sum_z p^z\cdot 2\ep_z
=2\ep$. 
\end{proof}

\begin{figure}[h]
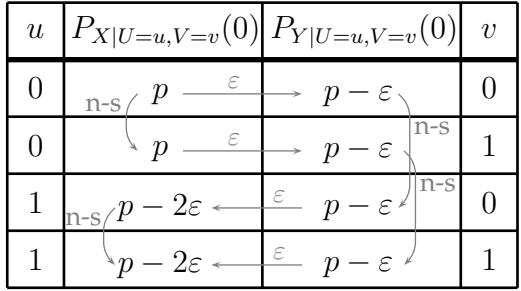

\centering
\psset{unit=0.75cm}
\pspicture*[](-2.5,-0.5)(7.5,5.5)
\psline[linewidth=1pt]{-}(-2,5)(7,5)
\psline[linewidth=2pt]{-}(-2,4)(7,4)
\psline[linewidth=1pt]{-}(-2,3)(7,3)
\psline[linewidth=1pt]{-}(-2,2)(7,2)
\psline[linewidth=1pt]{-}(-2,1)(7,1)
\psline[linewidth=1pt]{-}(-2,0)(7,0)
\psline[linewidth=1pt]{-}(-2,0)(-2,5)
\psline[linewidth=1pt]{-}(7,0)(7,5)
\psline[linewidth=1pt]{-}(2.5,0)(2.5,5)
\psline[linewidth=1pt]{-}(-1,0)(-1,5)
\psline[linewidth=1pt]{-}(6,0)(6,5)
\rput[c]{0}(0.8,4.5){\large{$P_{X|U=u,V=v}(0) $}}
\rput[c]{0}(4.3,4.5){\large{$P_{Y|U=u,V=v}(0) $}}
\rput[c]{0}(-1.5,4.5){\large{$u$}}
\rput[c]{0}(6.5,4.5){\large{$v$}}
\rput[c]{0}(-1.5,3.5){\large{$0$}}
\rput[c]{0}(-1.5,2.5){\large{$0$}}
\rput[c]{0}(-1.5,1.5){\large{$1$}}
\rput[c]{0}(-1.5,0.5){\large{$1$}}
\rput[c]{0}(6.5,3.5){\large{$0$}}
\rput[c]{0}(6.5,2.5){\large{$1$}}
\rput[c]{0}(6.5,1.5){\large{$0$}}
\rput[c]{0}(6.5,0.5){\large{$1$}}
\rput[c]{0}(0.7,3.4){\large{$p$}}
\rput[c]{0}(0.7,2.4){\large{$p$}}
\rput[c]{0}(0.7,1.4){\large{$p-2\ep$}}
\rput[c]{0}(0.7,0.4){\large{$p-2\ep$}}
\rput[c]{0}(4.2,3.4){\large{$p-\ep$}}
\rput[c]{0}(4.2,2.4){\large{$p-\ep$}}
\rput[c]{0}(4.2,1.4){\large{$p-\ep$}}
\rput[c]{0}(4.2,0.4){\large{$p-\ep$}}
\psline[linewidth=0.5pt,linearc=0.5, linecolor=gray]{->}(0.3,3.4)(0.1,3.2)(0.1,2.6)(0.3,2.4)
\rput[c]{0}(-0.3,3.2){\color{gray}{{n-s}}}
\psline[linewidth=0.5pt,linearc=0.5, linecolor=gray]{->}(-0.1,1.4)(-0.3,1.2)(-0.3,0.6)(-0.1,0.4)
\rput[c]{0}(-0.65,1.2){\color{gray}{{n-s}}}
\psline[linewidth=0.5pt,linearc=0.5, linecolor=gray]{->}(4.9,3.4)(5.1,3.2)(5.1,1.6)(4.9,1.4)
\rput[c]{0}(5.5,2.8){\color{gray}{{n-s}}}
\psline[linewidth=0.5pt,linearc=0.5, linecolor=gray]{->}(5,2.4)(5.2,2.2)(5.2,0.6)(5,0.4)
\rput[c]{0}(5.6,1.8){\color{gray}{{n-s}}}
\psline[linewidth=0.5pt,linecolor=gray]{->}(1.1,3.4)(3.2,3.4)
\rput[c]{0}(2,3.6){\color{gray}{{$\varepsilon$}}}
\psline[linewidth=0.5pt,linecolor=gray]{->}(1.1,2.4)(3.2,2.4)
\rput[c]{0}(2,2.6){\color{gray}{{$\varepsilon$}}}
\psline[linewidth=0.5pt,linecolor=gray]{<-}(1.6,1.4)(3.2,1.4)
\rput[c]{0}(2.8,1.6){\color{gray}{{$\varepsilon$}}}
\psline[linewidth=0.5pt,linecolor=gray]{<-}(1.6,0.4)(3.2,0.4)
\rput[c]{0}(2.8,0.6){\color{gray}{{$\varepsilon$}}}
\endpspicture
\caption{\label{fig:bias}The maximal bias of the output of a system with $\Pr[X\oplus Y=U\cdot V]=1-\ep$.}
\end{figure}

Note that there is a non-signalling partition, given in Section~\ref{sec:single_box},  reaching this bound.

Systems $P_{XY|UV}$ with $\ep\in [0,0.25)$ are  \emph{non-local}, i.e., they violate a Bell inequality, more precisely the CHSH inequality given in Example~\ref{ex:chshineq}, p.~\pageref{ex:chshineq}. 
For any of these systems, Eve cannot obtain perfect knowledge about Alice's output bit, and it, therefore, contains some secrecy.

\subsection{The general optimal attack on a bit}\label{subsec:bitlp}

Now consider the case when a bit $B=f(\bof{X})$ is obtained from the outputs of an $n$-party non-signalling system with arbitrary input and output alphabet. This includes, in particular, the case where Alice and Bob share a bipartite non-signalling system and the bit is a function only of Alice's outputs, i.e., the situation we will consider for key agreement. 
The inputs as well as the function $f$ are communicated over the public authenticated channel, i.e., $Q=(\bof{U}=\bof{u},F=f)$. 
What is the maximal distance from uniform given an adversary's output variable $Z$ this bit can have? 

Finding the maximal distance from uniform corresponds to finding the `best' non-signalling partition, from the adversary's point of view. 
We first show that 
it is enough to consider non-signalling partitions with two elements. 

\begin{figure}[h]
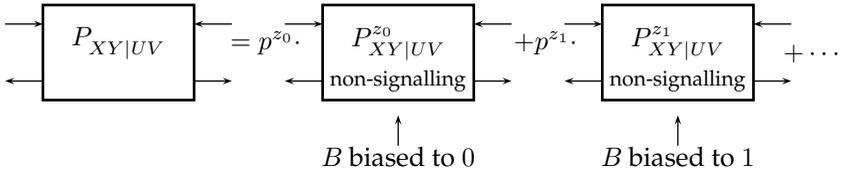

\centering
\pspicture*[](-1.1,-1)(10.2,1.6)
\psset{unit=1cm}
\rput[c]{0}(-0.5,0){
\pspolygon[linewidth=1pt](0,0)(0,1.25)(2,1.25)(2,0)
\rput[c]{0}(1,0.75){$P^{\phantom{z}}_{XY|UV}$}
\psline[linewidth=0.5pt]{->}(-0.5,1)(0,1)
\psline[linewidth=0.5pt]{<-}(-0.5,0.25)(0,0.25)
\psline[linewidth=0.5pt]{<-}(2,1)(2.5,1)
\psline[linewidth=0.5pt]{->}(2,0.25)(2.5,0.25)
}
\rput[c]{0}(3.2,0){
\psline[linewidth=0.5pt]{->}(1,-0.6)(1,-0.2)
\rput[c]{0}(1,-0.75){$B$ biased to $0$}
\rput[c]{0}(-0.75,0.75){$=p^{z_0}\cdot $}
\pspolygon[linewidth=1pt](0,0)(0,1.25)(2,1.25)(2,0)
\rput[c]{0}(1,0.75){$P^{z_0}_{XY|UV}$}
\rput[c]{0}(1,0.25){\footnotesize{non-signalling}}
\psline[linewidth=0.5pt]{->}(-0.5,1)(0,1)
\psline[linewidth=0.5pt]{<-}(-0.5,0.25)(0,0.25)
\psline[linewidth=0.5pt]{<-}(2,1)(2.5,1)
\psline[linewidth=0.5pt]{->}(2,0.25)(2.5,0.25)
}
\rput[c]{0}(6.9,0){
\psline[linewidth=0.5pt]{->}(1,-0.6)(1,-0.2)
\rput[c]{0}(1,-0.75){$B$ biased to $1$}
\rput[c]{0}(-0.75,0.75){$+p^{z_1}\cdot $}
\pspolygon[linewidth=1pt](0,0)(0,1.25)(2,1.25)(2,0)
\rput[c]{0}(1,0.75){$P^{z_1}_{XY|UV}$}
\rput[c]{0}(1,0.25){\footnotesize{non-signalling}}
\psline[linewidth=0.5pt]{->}(-0.5,1)(0,1)
\psline[linewidth=0.5pt]{<-}(-0.5,0.25)(0,0.25)
\psline[linewidth=0.5pt]{<-}(2,1)(2.5,1)
\psline[linewidth=0.5pt]{->}(2,0.25)(2.5,0.25)
}
\rput[c]{0}(9.7,0.6){$+\dotsb $}
\endpspicture
\caption{In order to find the distance from uniform of a bit, it is enough to consider non-signalling partitions with two elements (Lemma~\ref{lemma:p_half}). }
\end{figure}

\begin{lemma}\label{lemma:p_half}
Assume there exists a non-signalling partition $w^{\prime}$ with \linebreak[4] $d(f(\bof{X})|Z^{\prime}(w^{\prime}),Q)$, where $Q=(\bof{U}=\bof{u},F=f)$. Then there exists a non-signalling partition $w$ with the same distance from uniform with $Z\in \{z_0,z_1\}$ and such that $P(f(\bof{X})=0|Q,Z=z_0)>1/2$ and $P(f(\bof{X})=0|Q,Z=z_1)\leq 1/2$. 
\end{lemma}
\begin{proof}
Assume that the non-signalling partition has more than two elements. Define a new element 
 $(p^{z_0},P_{\bof{X}|\bof{U}}^{Z_0})$ by
\begin{align}
\nonumber p^{z_0} &{:=} p^{z^{\prime}_1}+\dotsb+ p^{z^{\prime}_{m}}\\
\nonumber P_{\bof{X}|\bof{U}}^{z_0} &{:=} \frac{1}{p^{z_0}}\sum\limits_{i=1}^m p^{z^{\prime}_i}P_{\bof{X}|\bof{U}}^{z^{\prime}_i},
\end{align}
where the set $z^{\prime}_1,\dotsc,z^{\prime}_m$ is defined to consist of the conditional systems $P_{\bof{X}|\bof{U}}^{z^{\prime}_i}$ such that $P(f(\bof{X})=0|\bof{U}=\bof{u},Z=z^{\prime}_i)> 1/2$ ($p^{z_0}$ can be $0$). Similarly define 
$(p^{z_1},P_{\bof{X}|\bof{U}}^{z_1})$ as the convex combination of the remaining elements of the non-signalling partition. 
Since the space of non-signalling systems is convex, this forms again a non-signalling partition, and it \linebreak[4] reaches the same distance.
\end{proof}

We can simplify the problem even further, such that we only need to consider a single element of the non-signalling partition. The reason is that given one element of a non-signalling partition with two elements, the other one is uniquely determined by the fact that the sum of the two is the marginal system, i.e.,
\begin{align}
\nonumber P_{\bof{X}|\bof{U}} &=p\cdot P^{z_0}_{\bof{X}|\bof{U}}+(1-p)\cdot P^{z_1}_{\bof{X}|\bof{U}}\ .
\end{align}

\begin{lemma}\label{lemma:termsz0}
Consider a non-signalling partition ${w}$ with element $(p,P_{\bof{X}|\bof{U}}^{z_0})$ such that $P(B=0|Q,Z=z_0)>1/2$ with $B=f(X)$ and $Q=(\bof{U}=\bof{u},F=f)$. 
Then the distance from uniform of $B$ given the non-signalling partition ${w}$ and $Q=(\bof{U}=\bof{u},F=f)$ is 
\begin{multline}
\nonumber d(B|Z({w}),Q) =
p\cdot \left(P(B=0|Q,Z=z_0)-P(B=1|Q,Z=z_0)\right)
\\ 
-\frac{1}{2}\cdot \left(P(B=0|Q)-P(B=1|Q)\right)\ ,
\end{multline}
where $P(B=0|Q)$ stands for $\sum_{\bof{x}:f(\bof{x})=0}P_{\bof{X}|\bof{U}}(\bof{x},\bof{u})$. 
\end{lemma}
\begin{proof}
W.l.o.g.\ assume that $P(B=0|Q,Z=z_1)\leq 1/2$. 
By Definition~\ref{def:distancefromuniform}, p.~\pageref{def:distancefromuniform}, the distance from uniform of $B$ given the non-signalling partition ${w}$ and $Q=(\bof{U}=\bof{u},F=f)$ is 
 \begin{align}
\nonumber  &d(B|Z(\bar{w}),Q)
\\
\nonumber &=
p\cdot 
\left(P(B=0|Q,Z=z_0)-\frac{1}{2}\right)
+(1-p)\cdot 
\left(\frac{1}{2}-P(B=0|Q,Z=z_1)\right)\\
\nonumber &= 
\frac{1}{2}\cdot p\cdot 
(P(B=0|Q,Z=z_0)-P(B=1|Q,Z=z_0))
\\ \nonumber &\quad
+\frac{1}{2}\cdot (1-p)\cdot 
(P(B=1|Q,Z=z_1)-P(B=0|Q,Z=z_1))
\\ \nonumber 
&= 
p\cdot (P(B=0|Q,Z=z_0)-P(B=1|Q,Z=z_0))
\\ \nonumber &\quad
-\frac{1}{2}\cdot (P(B=0|Q)-P(B=1|Q))\ ,
\end{align}
where we have used that $(1-p) P^{z_1}_{\bof{X}|\bof{U}}= P_{\bof{X}|\bof{U}}-p P^{z_0}_{\bof{X}|\bof{U}}$. 
\end{proof}

We have reduced the question of the maximal distance from uniform given a non-signalling partition to the problem of finding the `best' element $(p,P^{z_0}_{\bof{X}|\bof{U}})$ of a non-signalling partition. The question remains to be answered, when $(p,P^{z_0}_{\bof{X}|\bof{U}})$ is element of a non-signalling partition. The criterion is given in Lemma~\ref{lemma:zweihi}. 
\begin{lemma}\label{lemma:zweihi}
Given a non-signalling system $P_{\bof{X}|\bof{U}}$, there exists a non-sig\-nal\-ling partition with element 
$(p,P_{\bof{X}|\bof{U}}^{z_0})$ if and only if for all inputs and outputs $\bof{x},\bof{u}$,
\begin{align}
\label{eq:zcanoccurwithp}  p\cdot P_{\bof{X}|\bof{U}}^{z_0}(\bof{x},\bof{u})  &\leq P_{\bof{X}|\bof{U}}(\bof{x},\bof{u})\ .
\end{align}
\end{lemma}
\begin{proof}
The non-signalling condition is linear and the space of conditional probability distributions is convex, therefore a 
convex combination of  non-signalling systems $P_{\bof{X}|\bof{U}}^{z}$ is a non-signalling system. In order to prove that the outcome $z_0$ can occur with probability $p$ it is, therefore, sufficient to show that there exists another  outcome $z_1$ which can occur with $1-p$, and that the weighted sum of the two is $P_{\bof{X}|\bof{U}}$. If $P_{\bof{X}|\bof{U}}^{z_0}$ is a normalized and non-signalling probability distribution, then so is $P_{\bof{X}|\bof{U}}^{z_1}$, because the convex combination of the two, $P_{\bof{X}|\bof{U}}$, is also non-signalling and normalized. Therefore, we only need to verify that all entries of the complementary system $P_{\bof{X}|\bof{U}}^{z_1}$ are between $0$ and~$1$.  However, this  system is the difference
\begin{align}
\nonumber P_{\bof{X}|\bof{U}}^{z_1} &=\frac{1}{1-p}(P_{\bof{X}|\bof{U}}-p\cdot P_{\bof{X}|\bof{U}}^{z_0})\ .
\end{align}
Requesting this to be greater or equal to $0$ is equivalent to (\ref{eq:zcanoccurwithp}). We observe that all entries of $P_{\bof{X}|\bof{U}}^{z_1}$ are 
smaller or equal to  $1$ because of the normalization: If the sum of positive terms is~$1$, each of them can be at most~$1$.
\end{proof}

The above argument implies in fact, that the maximal distance from uniform can be calculated by the following optimization problem --- a linear program. 

\begin{align}
\nonumber \max :&\quad \sum_{\bof{x}:B=0}p\cdot P_{\bof{X}|\bof{U}}^{z_0}(\bof{x},\bof{u})
-\sum_{\bof{x}:B=1}p\cdot P_{\bof{X}|\bof{U}}^{z_0}(\bof{x},\bof{u})\\
\nonumber &\quad 
-\frac{1}{2}\sum_{\bof{x}:B=0}p\cdot P_{\bof{X}|\bof{U}}(\bof{x},\bof{u}) 
+\frac{1}{2} \sum_{\bof{x}:B=1}p\cdot P_{\bof{X}|\bof{U}}(\bof{x},\bof{u}) 
\\
\nonumber \st &\quad 
p\cdot P_{\bof{X}|\bof{U}}^{z_0}\ \  \text{non-signalling}\\
\nonumber &\quad p\cdot P_{\bof{X}|\bof{U}}^{z_0}(\bof{x},\bof{u}) \geq 0 \ \ \text{for all}\  \bof{x},\bof{u}\\
\nonumber &\quad p\cdot P_{\bof{X}|\bof{U}}^{z_0}(\bof{x},\bof{u}) \leq P_{\bof{X}|\bof{U}}(\bof{x},\bof{u}) \ \ \text{for all}\  \bof{x},\bof{u}\ .
\end{align}

We give a slightly different form of this optimization problem, where instead of the variable $p P_{\bof{X}|\bof{U}}^{z_0}$, we optimize over a variable 
$\Delta=2 p  P_{\bof{X}|\bof{U}}^{z_0}-P_{\bof{X}|\bof{U}}$. $\Delta$ can be seen as a non-signalling system which does not need to be normalized nor positive. Why we use this form will become clear in Section~\ref{sec:severalsystems}. 

\begin{lemma}\label{lemma:distanceislp}
The distance from uniform of $B=f(\bof{X})$ given $Z(W_{\mathrm{n-s}})$ and $Q:=(\bof{U}=\bof{u},F=f)$ is 
\begin{align}
\nonumber d(B|Z(W_{\mathrm{n-s}}),Q) &= \frac{1}{2}\cdot b^T\cdot \Delta^*\ ,
\end{align}
where $b^T \Delta^*$ is the optimal value of the linear program
\begin{align}
\label{eq:primal2} \max : &\quad \sum_{\bof{x}:B=0} \Delta(\bof{x},\bof{u})-\sum_{\bof{x}:B=1} \Delta(\bof{x},\bof{u})\\
\nonumber \st &\quad \sum_{{x_i}}\Delta(\bof{x}, u_i,\bof{u}_{\bar{i}})-\sum_{{x_i}} \Delta(\bof{x},u^{\prime}_i,\bof{u}_{\bar{i}})=0\ \ \text{for all}\ \bof{x},{u_i},{u^{\prime}_i},\bof{u}_{\bar{i}}
\\
\nonumber &\quad \Delta(\bof{x},\bof{u}) \leq P_{\bof{X}|\bof{U}}(\bof{x},\bof{u})\ \ \text{for all}\ \bof{x},\bof{u}\\
\nonumber &\quad \Delta(\bof{x},\bof{u}) \geq -P_{\bof{X}|\bof{U}}(\bof{x},\bof{u})\ \ \text{for all}\ \bof{x},\bof{u}\ .
\end{align}
\end{lemma}
\begin{proof}
We show that every element  $(p,P_{\bof{X}|\bof{U}}^{z_0})$ of a non-signalling partition corresponds to a feasible $\Delta$, and {vice versa}.\\
Assume an element of a non-signalling partition, $(p,P_{\bof{X}|\bof{U}}^{z_0})$, and define
\begin{align}
\nonumber \Delta(\bof{x},\bof{u}) &=2 p\cdot  P_{\bof{X}|\bof{U}}^{z_0}(\bof{x},\bof{u})-P_{\bof{X}|\bof{U}}(\bof{x},\bof{u})\ .
\end{align}
$\Delta$ fulfils the non-signalling conditions by linearity. The positivity of $p$ and $P_{\bof{X}|\bof{U}}^{z_0}(\bof{x},\bof{u})\geq 0$ imply $\Delta(\bof{x},\bof{u})\geq -P_{\bof{X}|\bof{U}}(\bof{x},\bof{u})$ and $p P_{\bof{X}|\bof{U}}^{z_0}(\bof{x},\bof{u}) \leq P_{\bof{X}|\bof{U}}(\bof{x},\bof{u})$ (Lemma~\ref{lemma:zweihi}) implies $\Delta(\bof{x},\bof{u})\leq P_{\bof{X}|\bof{U}}(\bof{x},\bof{u})$. $\Delta$ is, therefore, feasible. \\
To see the reverse direction, assume a feasible $\Delta$. Define
\begin{align}
\nonumber p&= \frac{1}{2}\cdot \Bigl(1+\sum_{\bof{x}}\Delta(\bof{x},0\dotso 0)\Bigr)\\
\nonumber P_{\bof{X}|\bof{U}}^{z_0}(\bof{x},\bof{u})&= \frac{P_{\bof{X}|\bof{U}}(\bof{x},\bof{u})+\Delta(\bof{x},\bof{u})}{2p}\ .
\end{align}
(For completeness, define $P_{\bof{X}|\bof{U}}^{z_0}(\bof{x},\bof{u})=P_{\bof{X}|\bof{U}}(\bof{x},\bof{u})$ in case $p=0$.)
To see that $(p,P_{\bof{X}|\bof{U}}^{z_0})$ is element of a non-signalling partition note that, because of the non-signalling constraints,  
$\sum_{\bof{x}}\Delta(\bof{x},0\dotso 0)=\sum_{\bof{x}}\Delta(\bof{x},\bof{u}^{\prime})$ for all $\bof{u}^{\prime}$. I.e., $p$ is independent of the chosen input and the above transformation is, therefore, linear. This implies that $P_{\bof{X}|\bof{U}}^{z_0}$ is non-signalling. Since 
\begin{align}
\nonumber \sum_{\bof{x}}P_{\bof{X}|\bof{U}}^{z_0}(\bof{x},\bof{u}) &=\sum_{\bof{x}}\frac{P_{\bof{X}|\bof{U}}(\bof{x},\bof{u})+\Delta(\bof{x},\bof{u})}{2p}=\frac{1+(2p-1)}{2p}
=1\ ,
\end{align}
it is normalized. Since $-P_{\bof{X}|\bof{U}}(\bof{x},\bof{u})\leq \Delta(\bof{x},\bof{u})\leq P_{\bof{X}|\bof{U}}(\bof{x},\bof{u})$ and \linebreak[4] $\sum_{\bof{x}}P_{\bof{X}|\bof{U}}(\bof{x},\bof{u})=1$, it holds that $-1\leq \sum_{\bof{x}}\Delta(\bof{x},0\dotso 0)\leq 1$ and this implies $P_{\bof{X}|\bof{U}}^{z_0}(\bof{x},\bof{u})\geq 0$
i.e., $P^{z_0}_{\bof{X}|\bof{U}}$ is a non-signalling system. By Lemma~\ref{lemma:zweihi}, $(p,P_{\bof{X}|\bof{U}}^{z_0})$ is element of a non-signalling partition because 
\begin{align}
\nonumber p\cdot P_{\bof{X}|\bof{U}}^{z_0}(\bof{x},\bof{u})&= \frac{1}{2}\cdot \Bigl(1+\sum_{\bof{x}}\Delta(\bof{x},0\dotso 0)\Bigr)
\cdot \frac{P_{\bof{X}|\bof{U}}(\bof{x},\bof{u})+\Delta(\bof{x},\bof{u})}{1+\sum_{\bof{x}}\Delta(\bof{x},0\dotso 0)}\\
\nonumber &= \frac{1}{2}\cdot (P_{\bof{X}|\bof{U}}(\bof{x},\bof{u})+\Delta(\bof{x},\bof{u}))
\\ \nonumber 
&\leq P_{\bof{X}|\bof{U}}(\bof{x},\bof{u})\ .
\end{align}
The value of the objective function for any $\Delta$ is exactly twice the distance from uniform reached by the non-signalling partition with element $(p,P^{z_0}_{\bof{X}|\bof{U}})$.
\begin{align}
\nonumber
\sum_{\bof{x}:B=0} \Delta(\bof{x},\bof{u})-\sum_{\bof{x}:B=1} \Delta(\bof{x},\bof{u})
&= 
\sum_{\bof{x}:B=0} \left(2p\cdot P_{\bof{X}|\bof{U}}^{z_0}(\bof{x},\bof{u})-P_{\bof{X}|\bof{U}}(\bof{x},\bof{u})\right)\ ,
\end{align}
which is exactly twice the distance from uniform by Lemma~\ref{lemma:termsz0}.
\end{proof}

Note that the linear program of Lemma~\ref{lemma:distanceislp} can be expressed either in its \emph{primal} or \emph{dual} form (see Section~\ref{subsec:lin_prog}). 

\begin{align}
\nonumber \mathrm{PRIMAL}\\
\label{eq:primal_delta}  \max :&\quad b^T\cdot \Delta\\
\nonumber \st &\quad 
\underbrace{
\left(\begin{array}{c}
\phantom{-}A_{\mathrm{n-s}}\\
-A_{\mathrm{n-s}}\\
\phantom{-}\mathds{1}\\
-\mathds{1}
\end{array}\right)
}_{A}
\cdot \Delta \leq
\underbrace{
\left(\begin{array}{c}
0\\
0\\
P_{\bof{X}|\bof{U}}\\
P_{\bof{X}|\bof{U}}
\end{array}\right)
}_{c}
\end{align}
The \emph{dual} of the above linear program has the form 
\begin{align}
\nonumber \mathrm{DUAL}\\
 \label{eq:dualsingle} \min : &\quad 
\overbrace{
\left(\begin{array}{c}
0\\
0\\
P_{\bof{X}|\bof{U}}\\
P_{\bof{X}|\bof{U}}
\end{array}\right)^T
}^{c^T} 
\cdot \lambda\\
\nonumber  \st &\quad 
\underbrace{
\left(\begin{array}{cccc}
A_{\mathrm{n-s}} & -A_{\mathrm{n-s}} & \mathds{1} & -\mathds{1}
\end{array}\right)
}_{A^T}
\cdot \lambda = b\\
\nonumber  &\quad \lambda \geq 0
\end{align}

As an example, consider again a system with binary inputs and outputs, 
i.e. the case we have already studied in Section~\ref{subsec:prboxdistance}. 
We give the explicit forms of $A$, $b$, and $c$ below. 
\begin{example}\label{ex:abc}
 For a bipartite system taking one bit input and giving one bit output on each side, $A$, $b$, and $c$ have the form
 
\begin{align}
\nonumber &
\begin{array}{rcl}
 A&=&\left(
\begin{array}{@{\hspace{0mm}}c@{\hspace{0mm}}}
\phantom{-} A_{\mathrm{n-s}}\\
 -A_{\mathrm{n-s}}\\
\phantom{-}\mathds{1}_{16}\\
-\mathds{1}_{16}\\
\end{array}
\right)
\\ 
\\
\\
c&=&
\left(
\begin{array}{@{\hspace{0mm}}c@{\hspace{0mm}}}
0_{16}\\
0_{16}\\
P_{XY|UV}\\
P_{XY|UV}
\end{array}
\right) 
\end{array}
\, \ \
b=
\left(
\begin{array}{@{\hspace{0mm}}r@{\hspace{0mm}}}
 1\\ 1\\ 0\\ 0\\ -1 \\-1 \\0 \\0 \\0 \\0 \\0 \\0 \\0 \\0\\ 0\\ 0\\
\end{array}
 \right)
\ \text{ with }\ 
P_{XY|UV}=
\left(
\begin{array}{@{\hspace{0mm}}r@{\hspace{0mm}}}
P(0,0,0,0)\\ P(0,1,0,0)\\ P(0,0,0,1)\\ P(0,1,0,1)\\ 
P(1,0,0,0)\\ P(1,1,0,0)\\ P(1,0,0,1)\\ P(1,1,0,1)\\ 
P(0,0,1,0)\\ P(0,1,1,0)\\ P(0,0,1,1)\\ P(0,1,1,1)\\ 
P(1,0,1,0)\\ P(1,1,1,0)\\ P(1,0,1,1)\\ P(1,1,1,1)
\end{array}
 \right)
\end{align}
and 
\begin{align}
\nonumber 
A_{\mathrm{n-s}} &=
\left(
\begin{array}{@{\hspace{0mm}}r@{\hspace{1.2mm}}r@{\hspace{1.2mm}}r@{\hspace{1.2mm}}r@{\hspace{1.2mm}}r@{\hspace{1.2mm}}r@{\hspace{1.2mm}}r@{\hspace{1.2mm}}r@{\hspace{1.2mm}}r@{\hspace{1.2mm}}r@{\hspace{1.2mm}}r@{\hspace{1.2mm}}r@{\hspace{1.2mm}}r@{\hspace{1.2mm}}r@{\hspace{1.2mm}}r@{\hspace{1.2mm}}r@{\hspace{1.2mm}}r@{\hspace{0mm}}}
 \phantom{-}1& \phantom{-}1& -1 &-1 &0 &0 &0 &0 &0 &0 &0 &0 &0 &0& 0& 0\\
 0& 0& 0& 0 & \phantom{-}1& \phantom{-}1& -1& -1& 0 &0& 0& 0& 0& 0& 0& 0\\
 0& 0& 0& 0& 0& 0& 0& 0& \phantom{-}1& \phantom{-}1& -1& -1& 0& 0& 0 &0\\
 0& 0& 0& 0& 0& 0& 0& 0& 0& 0& 0& 0& \phantom{-}1& \phantom{-}1& -1& -1\\
 \phantom{-}1& 0& 0& 0& \phantom{-}1& 0& 0& 0& -1& 0 &0& 0& -1& 0& 0& 0\\
 0& \phantom{-}1& 0& 0& 0& \phantom{-}1& 0& 0& 0& -1& 0& 0& 0& -1& 0& 0\\
 0& 0& \phantom{-}1& 0& 0& 0& \phantom{-}1& 0& 0& 0& -1& 0& 0& 0& -1& 0\\
 0& 0& 0& \phantom{-}1& 0& 0& 0& \phantom{-}1& 0& 0& 0 &-1& 0& 0& 0& -1 
\end{array}
\right)\ .
\end{align}
\end{example}

Since a linear program can be solved either in its primal or its dual form, we could as well have solved the dual problem (\ref{eq:dualsingle}) in order to obtain the distance from uniform of the bit $B$, instead of the above linear program. The dual is a minimization problem, and therefore, any feasible solution of the dual program is an upper bound on the distance from uniform of the bit $B$. We further observe that the dual feasible solutions are independent of the marginal probability distribution as seen by the honest parties, and that the value reached by the dual feasible solution can be expressed in terms of the marginal probability distribution. 
\begin{lemma}\label{lemma:dualevent}
For any dual feasible solution of the linear program (\ref{eq:primal2}) (see (\ref{eq:dualsingle})), 
there exists an event $\mathcal{E}$ defined by the inputs and outputs of the system $P_{\bof{X}|\bof{U}}$ and (independent) randomness such that the value of (\ref{eq:dualsingle}) is proportional to the probability of this event, and, therefore,  $d(f(\bof{X})|Z(W_{\mathrm{n-s}}),Q)\leq d/2 =c^T\lambda/2 \propto  P(\mathcal{E})$. 
\end{lemma}
Note that Lemma~\ref{lemma:dualevent} holds, in particular, for the optimal dual solution, which implies that the distance from uniform is proportional to some event defined by the random variables, i.e., the secrecy of the bit can be inferred from the behaviour of the marginal system. 
\begin{proof}
The value of $d(f(\bof{X})|Z(W_{\mathrm{n-s}}),Q)$ is bounded by the value of any dual feasible solution, i.e., it is of the form $c^T\lambda/2$, where $c$ contains the probabilities $P_{\bof{X}|\bof{U}}(\bof{x},\bof{u})$ (all other entries are $0$) and $\lambda\geq 0$. Therefore, it can be expressed as a weighted sum of the probabilities $P_{\bof{X}|\bof{U}}(\bof{x},\bof{u})$. If all weights have the same value, this implies that the optimal value is proportional to an event $\mathcal{E}$ defined by $\bof{X}$ and $\bof{U}$. If not all weights have the same value, define for each $\bof{x}$ and $\bof{u}$ an additional random coin which takes value $1$ with probability $\lambda_i/\max_i(\lambda_i)$. The optimal value is then proportional to an event $\mathcal{E}$ defined by $\bof{X}$ and $\bof{U}$ and the additional random coin taking value~$1$.
\end{proof}

\begin{example}\label{ex:dual}
 Let us come back to the above example of a bipartite system with binary inputs and outputs. It can easily be verified that the following is a dual feasible solution of the linear program:
\begin{align}
\nonumber  \lambda_1^{*T} &=
\begin{array}{@{\hspace{0mm}}r@{\hspace{1.2mm}}r@{\hspace{1.2mm}}r@{\hspace{1.2mm}}r@{\hspace{1.2mm}}r@{\hspace{1.2mm}}r@{\hspace{1.2mm}}r@{\hspace{1.2mm}}r@{\hspace{1.2mm}}r@{\hspace{1.2mm}}r@{\hspace{1.2mm}}r@{\hspace{1.2mm}}r@{\hspace{1.2mm}}r@{\hspace{1.2mm}}r@{\hspace{1.2mm}}r@{\hspace{1.2mm}}r@{\hspace{1.2mm}}r@{\hspace{1.2mm}}r@{\hspace{1.2mm}}r@{\hspace{1.2mm}}r@{\hspace{1.2mm}}r@{\hspace{1.2mm}}r@{\hspace{1.2mm}}r@{\hspace{1.2mm}}r@{\hspace{1.2mm}}r@{\hspace{1.2mm}}r@{\hspace{1.2mm}}r@{\hspace{1.2mm}}r@{\hspace{1.2mm}}r@{\hspace{1.2mm}}r@{\hspace{1.2mm}}r@{\hspace{1.2mm}}r@{\hspace{1.2mm}}r@{\hspace{1.2mm}}r@{\hspace{1.2mm}}r@{\hspace{1.2mm}}r@{\hspace{1.2mm}}r@{\hspace{1.2mm}}r@{\hspace{1.2mm}}r@{\hspace{1.2mm}}r@{\hspace{1.2mm}}r@{\hspace{1.2mm}}r@{\hspace{1.2mm}}r@{\hspace{1.2mm}}r@{\hspace{1.2mm}}r@{\hspace{1.2mm}}r@{\hspace{1.2mm}}r@{\hspace{1.2mm}}r@{\hspace{1.2mm}}r@{\hspace{1.2mm}}r@{\hspace{1.2mm}}r@{\hspace{1.2mm}}r@{\hspace{1.2mm}}r@{\hspace{1.2mm}}r@{\hspace{1.2mm}}r@{\hspace{1.2mm}}r@{\hspace{1.2mm}}r@{\hspace{1.2mm}}r@{\hspace{1.2mm}}r@{\hspace{1.2mm}}r@{\hspace{1.2mm}}r@{\hspace{1.2mm}}r@{\hspace{1.2mm}}r@{\hspace{1.2mm}}r@{\hspace{1.2mm}}}
(\frac{1}{2} & 0 &  \frac{1}{2} & 0 &  \frac{1}{2}& 0 & \frac{1}{2} & 0 & 0 & \frac{1}{2} & 0 & \frac{1}{2} & 0 & \frac{1}{2} & 0 & \frac{1}{2} & \multicolumn{3}{c}{\cdots} \\ 
    0 & 1 & 0 & 1 & 0 & 0 & 0 & 0 & 0 & 0 & 1& 0 & 1 & 0 & 0 & 0 & 0 & 0 & 0 & 0 & 1 & 0 & 1 & 0 & 0 & 1 & 0 & 0 & 0 & 0 & 0 & 1 )
\end{array}
\end{align}
(it is also optimal for systems with $\varepsilon \leq 0.25$). 
To obtain the value of the objective function ($c^T \lambda_1^*$), the first part of $\lambda_1^*$ will be multiplied by $0$, i.e., does not contribute to the value. The second part is multiplied by $P_{XY|UV}$. We can easily see by comparison that for every $x,y,u,v$ such that $x\oplus y\neq u\cdot v$, there is exactly one `$1$' in the second part of $\lambda_1^*$ and everywhere else $\lambda_1^*$ is $0$, i.e., 
\begin{align}
 \nonumber c^T\cdot \lambda_1^* &=\sum_{x,y,u,v:x\oplus y\neq u\cdot v}P_{XY|UV}(x,y,u,v)\ .
\end{align}
\end{example}
This confirms the results of Section~\ref{subsec:prboxdistance}.

\section{Several Systems}\label{sec:severalsystems}

\subsection{The non-signalling condition for several systems}\label{subsec:severalns}

We have already seen that the distance from uniform of a bit obtained from any (not necessarily bipartite) non-signalling system can be \linebreak[4] obtained by a linear program (\ref{eq:primal2}). In this section, we study the structure of the space describing $n$-party non-signalling systems, and show that the non-signalling condition for $n$ parties can be expressed as function of the non-signalling condition of the different parts it consists of. More precisely, we show that the non-signalling condition of an $(n+m)$-party non-signalling system is just the tensor product of the non-signalling condition for an $n$- and an $m$-party non-signalling system.

Note that the probabilities describing an $(n+m)$-party non-signalling system can be seen as living in the tensor product space of the vector of probabilities describing each subsystem. 
\begin{lemma}\label{lemma:nsproduct} 
Let $P_{\bof{X}_1|\bof{U}_1}$ be an $n$-party non-signalling system, and write $A_{\mathrm{n-s}}$ for the matrix describing the non-signalling conditions this system fulfils, i.e., $A_{\mathrm{n-s},1} P_{\bof{X}_1|\bof{U}_1} =0$. Similarly, let $P_{\bof{X}_2|\bof{U}_2}$ be an $m$-party non-signaling system fulfilling $A_{\mathrm{n-s},2} P_{\bof{X}_2|\bof{U}_2} =0$. Then the $(n+m)$-party system $P_{\bof{X}_1\bof{X}_2|\bof{U}_1\bof{U}_2}$ is non-signalling exactly if 
\begin{align}
\nonumber (A_{\mathrm{n-s},1}\otimes \mathds{1}_{\mathrm{n-s},2})\cdot P_{\bof{X}_1 \bof{X}_2|\bof{U}_1\bof{U}_2} &=0\ \ \ \text{and}\\
\nonumber (\mathds{1}_{\mathrm{n-s},1} \otimes A_{\mathrm{n-s},2} )\cdot P_{\bof{X}_1 \bof{X}_2|\bof{U}_1\bof{U}_2}&=0\ .
\end{align}
\end{lemma}
\begin{proof}
The non-signalling conditions for the $n$-party non-signalling system are of the form 
\begin{align}
\nonumber  \sum_{x_{1i}} P_{\bof{X}_1|\bof{U}_1}(\bof{x}_1,{u_1}_i,\bof{u}_{1\bar{i}})-
\sum_{x_{1i}} P_{\bof{X}_1|\bof{U}_1}(\bof{x}_1,{u_1}^{\prime}_i,\bof{u}_{1 \bar{i}})
&=0\ .
\end{align}
The conditions $(A_{\mathrm{n-s},1}\otimes \mathds{1}_{\mathrm{n-s},2})\cdot P_{\bof{X}_1\bof{X}_2|\bof{U}_1\bof{U}_2}=0$, therefore, correspond to conditions of the form
\begin{multline}
\nonumber
\sum_{x_{1i}} P_{\bof{X}_1\bof{X}_2|\bof{U}_1\bof{U}_2}
(\bof{x}_1,\bof{x}_2,{u_1}_i,\bof{u}_{1\bar{i}},\bof{u}_{2})\\
-
\sum_{x_{1i}} P_{\bof{X}_1\bof{X}_2|\bof{U}_1\bof{U}_2}
(\bof{x}_1,\bof{x}_2,{u^{\prime}_1}_i,\bof{u}_{1\bar{i}},\bof{u}_{2})
=0\ ,
\end{multline}
(and similarly for the second system) 
which must hold for any $(n+m)$-party non-signalling system by Definition~\ref{def:nssystem}, p.~\pageref{def:nssystem}. By Lemma~\ref{lemma:nscondsimplified}, p.~\pageref{lemma:nscondsimplified} these conditions are also sufficient. 
\end{proof}
The above argument implies that in the linear program (\ref{eq:primal2}), the non-signalling condition can be replaced by this `tensor product' expression instead of directly requiring the system to be $n$-party non-signalling.

\subsection{An XOR-Lemma for non-signalling secrecy}\label{subsec:nsxor}

We can now show that non-signalling secrecy can be amplified by a deterministic privacy-amplification function, namely the XOR. 
Assume that Alice and Bob 
share a system giving rise to a (non-local) probability distribution $P_{XY|UV}$. Assume further that from this distribution a bit, $f(X)$, can be extracted and that this bit is partially secret by the non-signalling condition. Then the bit obtained from $n$ copies of the distribution \linebreak[4] $P_{XY|UV}^{\otimes n}$ and by XORing the $n$ partially secret bits together is insecure only if all the $n$ copies are insecure.

The key observation in order to show that the XOR of several partially non-signalling secure bits is highly secure, is that the linear program describing the distance from uniform of this bit is the tensor product of the `individual' linear programs in the sense that its constraint matrix $A_n$ is $A^{\otimes n}$ and the objective function $b_n=b^{\otimes n}$. The vector $c_n$ does not need to be of product form, because a $(2n)$-party non-signalling system does not necessarily need to consist of $n$ independent bipartite non-signalling distributions. The linear program can be taken to be of the following form: 
\begin{align}
\label{eq:product_lp}
\max : &\quad (b^{\otimes n})^T\cdot \Delta\\
\nonumber \st &\quad 
A^{\otimes n}
\cdot \Delta \leq
c_n\ .
\end{align}

\begin{lemma}\label{lemma:product_form} 
Let $A_1$, $b_1$, and $c_1$ be the vectors and matrices associated with the linear program (\ref{eq:primal_delta}) calculating the maximal distance from uniform of a bit $f(\bof{X}_1)$  of an $n$-party non-signalling system $P_{\bof{X}_1|\bof{U}_1}$ and similarly call $A_2$, $b_2$, and $c_2$ the vectors and matrices associated with the distance from uniform of a bit $g(\bof{X}_2)$ of an $m$-party non-signalling system $P_{\bof{X}_2|\bof{U}_2}$. Then the distance from uniform of the bit $B=f(\bof{X}_1)\oplus g(\bof{X}_2)$ is bounded by the linear program $A$, $b$, and $c$, where $A=A_1\otimes A_2$ and $b=b_1\otimes b_2$.  
\end{lemma}
\begin{proof}
Let us first verify that the constraints need to hold. Lemma~\ref{lemma:nsproduct}  implies that for any $\Delta$ associated with an  
$(n+m)$-party non-signalling system, 
$(A_1\otimes \mathds{1}) \Delta \leq 0$ must hold, and similarly with the sign on the left-hand side reversed and for the non-signalling condition of the $m$-party system. 
$(A_1\otimes A_2) \Delta \leq 0$ holds because it is a linear combination of the conditions of the form $(A_1\otimes \mathds{1}) \Delta \leq 0$. The condition of the form $(\mathds{1} \otimes \mathds{1}) \Delta \leq P_{\bof{X}_1\bof{X}_2|\bof{U}_1\bof{U}_2}$ must hold by Lemma~\ref{lemma:zweihi}. \\
It remains to see that $b$ can be taken of this form. $b$ is equal to $0$ where either $b_1$ or $b_2$ is equal to $0$, equal to $1$ exactly where both $b_1$ and $b_2$ are equal to $1$ or both are equal to $-1$. It is equal to $-1$ where $b_1$ and $b_2$ are equal to $1$, $-1$ or vice versa. This models exactly the vector $b$ associated with the bit $B=f(\bof{X}_1)\oplus g(\bof{X}_2)$. 
\end{proof}

Now switch to the dual form of this linear program. 
\begin{align}
\nonumber  \min : &\quad  c_n^T\cdot \lambda_n\\
\nonumber  \st &\quad (A^{\otimes n})^T\cdot \lambda_n = b^{\otimes n}\\*
\nonumber &\quad \lambda_n \geq 0
\end{align}
It is now straight-forward to see that if $\lambda$ was a feasible solution for a single copy of the system, then $\lambda_n=\lambda^{\otimes n}$ is a feasible solution for the dual of the $n$ copy version and, therefore, an upper bound on the distance from uniform of the bit $B=\bigoplus_i B_i$. 

\begin{lemma}\label{lemma:dual_product}
For any $\lambda_1$ which is dual feasible for the linear program $A_1$, $b_1$ associated with the non-signalling system $P_{\bof{X}_1|\bof{U}_1}$ and 
$\lambda_2$ which is dual feasible for the linear program $A_2$, $b_2$ associated with the non-signalling system $P_{\bof{X}_2|\bof{U}_2}$, $\lambda=\lambda_1\otimes \lambda_2$ is dual feasible for the linear program $A$, $b$ associated with $P_{\bof{X}_1\bof{X}_2|\bof{U}_1\bof{U}_2}$. 
\end{lemma}
\begin{proof}
By Lemma~\ref{lemma:product_form}, $A=A_1\otimes A_2$ and $b=b_1\otimes b_2$. Therefore, 
\begin{align}
\nonumber 
A\cdot \lambda &=(A_1\otimes A_2)\cdot (\lambda_1\otimes \lambda_2)=(A_1\cdot \lambda_1)\otimes (A_2\cdot \lambda_2)=b_1\otimes b_2\ .
\end{align}
Furthermore, $\lambda_1,\lambda_2\geq 0$ implies $\lambda_1\otimes \lambda_2\geq 0$. Therefore, $\lambda$ is dual feasible. 
\end{proof}

If the marginal $c_n$ of $n$ systems has product form, the value of this dual feasible solution --- and, therefore, an upper bound on the distance from uniform of the key bit --- is  $c^T_n\lambda_n=\left( \bigotimes_ic^T_i \right)  \left( \bigotimes_i\lambda_i\right) =\bigotimes_i(c^T_i\lambda_i)=\prod_i(c^T_i\lambda_i)$, i.e.,  the same value as if each of the $n$ systems was attacked individually. If $c_n$ does not have product form, then the value is still bounded by the probability that the event (defined by the input/output configurations such that the bit is insecure (Lemma~\ref{lemma:dualevent})) occurs for \emph{all} the $n$ copies of the system as stated in the following lemma.  

\begin{theorem}[XOR-Lemma for non-signalling secrecy]\label{th:nsxor}
Let $P_{\bof{X}_1|\bof{U}_1}$ be an $n$-party non-signalling system and $f(\bof{X}_1)$ a bit such that \linebreak[4] $d(f(\bof{X}_1)|Z(W_{\mathrm{n-s}}),Q)\leq k_1  P(\mathcal{E}_1)/2$, where $\mathcal{E}_1$ is an event defined by $\bof{X}_1$ and $\bof{U}_1$ (and maybe independent randomness). Similarly, let 
$P_{\bof{X}_2|\bof{U}_2}$ be an $m$-party non-signalling system with associated bit  $g(\bof{X}_2)$ and $d(g(\bof{X}_2)|Z(W_{\mathrm{n-s}}),Q)\leq k_2  P(\mathcal{E}_2)/2$. Let $Q=(\bof{U}=\bof{u},F=f,G=g)$. 
Then 
\begin{align}
\nonumber d(f(\bof{X}_1)\oplus g(\bof{X}_2)|Z(W_{\mathrm{n-s}}),Q) &\leq \frac{1}{2}\cdot k_1\cdot k_2\cdot  P(\mathcal{E}_1 \land \mathcal{E}_2)\ .
\end{align}
\end{theorem}
\begin{proof}
This follows directly from Lemma~\ref{lemma:dual_product}. 
\end{proof}

\begin{example}
Let us come back to the example of Section~\ref{subsec:prboxdistance} (see also Example~\ref{ex:abc}) where $P_{\bof{XY}|\bof{UV}}$ is a $(2n)$-party non-signalling system and each random variable is a bit. We have seen, in Example~\ref{ex:dual}, that the distance from uniform of each bit is upper-bounded by 
\begin{align}
\nonumber d(X_i|Z(W_{\mathrm{n-s}}),Q) &\leq \frac{1}{2} \sum_{x_i,y_i,u_i,v_i:x_i\oplus y_i\neq u_i\cdot v_i}P_{X_iY_i|U_iV_i}(x_i,y_i,u_i,v_i)\ .
\end{align}
 Therefore, 
\begin{align}
\nonumber d\Bigl(\bigoplus_i X_i \Bigm| Z(W_{\mathrm{n-s}}),Q\Bigr) &\leq \frac{1}{2} \sum_{\bof{x},\bof{y},\bof{u},\bof{v}:x_i\oplus y_i\neq u_i\cdot v_i\ \forall i}P_{\bof{XY}|\bof{UV}}(\bof{x},\bof{y},\bof{u},\bof{v})\ . 
\end{align}
\end{example}

\section{Key Distribution from Non-Signalling Systems}\label{sec:nskeyagreement}

In this section we show how we can use the XOR-Lemma established in the previous section to obtain a
 device-independent quantum key-a\-gree\-ment protocol. An explicit example of such a protocol can be found in Section~\ref{sec:protocol}. 

A (quantum) key-distribution protocol usually proceeds in several steps. First, Alice and Bob use the quantum channel. They distribute entangled quantum states and measure them in order to obtain (classical) input and output values. Then they sacrifice some of their systems (data) to check whether an eavesdropper was present and whether their data is good enough to establish a key. This step is called \emph{parameter estimation}. Then, they do classical post-processing  to transform their weakly correlated data into bit strings which are almost certainly equal, i.e., they do \emph{information reconciliation}. Finally, they do \emph{privacy amplification}, i.e., they apply a function to their partially secure bit strings in order to obtain a shorter, but highly secure key. 

We have to show two things about this key (see Section~\ref{subsec:securitykey}): The probability that Alice's and Bob's key are not equal is small (correctness) and, the adversary knows almost nothing about this key (secrecy). Together, (Lemma~\ref{lemma:distance}, p.~\pageref{lemma:distance}) these two properties imply that the key is close to a perfect key. Note that the key can be of zero length (i.e., Alice and Bob abort the protocol), in which case correctness and secrecy both trivially hold. This situation occurs if the parameter estimation step indicates that the systems are not good enough for key agreement. If the adversary has full control over the systems which are distributed (the channel), it is not possible to require that a key is always generated, because the adversary could just interrupt the communication line. Of course, we would like a key to be generated  if the adversary is passive. This property of a key-distribution scheme is called, \emph{robustness}. Robustness characterizes the probability that the protocol aborts even though no adversary is present.

\subsection{Parameter estimation}\label{subsec:parameter_estimation}

The goal of parameter estimation is for Alice and Bob to test whether the systems they have received are good enough to do key agreement. They execute a protocol where they interact with their systems and then output either `accept' or `reject'. If the systems have the necessary properties for key agreement, they should output `accept', while if they have not, they should output `reject'. 

\begin{definition}
A parameter estimation protocol is said to \emph{$\epsilon$-securely filter} systems $P_{\bof{XY}|\bof{UV}}$ of a set $\mathcal{P}$ (or string pairs $(\bof{x},\bof{y})$ of a set $\mathcal{B}$) if on input  $P_{\bof{XY}|\bof{UV}}\in \mathcal{P}$ 
(or $(\bof{x},\bof{y})\in \mathcal{B}$) the protocol outputs `abort' with probability at least $1-\epsilon$. 
\end{definition}
\begin{definition}
A parameter estimation protocol is said to be 
\emph{$\epsilon^{\prime}$-robust} on systems $P_{\bof{XY}|\bof{UV}}$ of a set $\mathcal{P}$ if on input  $P_{\bof{XY}|\bof{UV}}\in \mathcal{P}$ the protocol outputs `abort' with probability at most $\epsilon^{\prime}$.
\end{definition}

Before starting the protocol, Alice and Bob fix its parameters, more precisely, the probabilities $k$ and $p$ and values $\varepsilon$ and $\delta$. 
\begin{protocol}[Parameter estimation]\label{prot:pe}\ 
\begin{enumerate}
\item Alice and Bob receive $P_{\bof{XY}|\bof{UV}}$. 
\item Alice chooses $\bof{U}$ such that for each $i$ with probability $1-k$, it holds that $U_i=u_k$, where 
$u_k$ is the input from which a raw key bit can be generated, and with probability $k$ she chooses one of the  $|\mathcal{U}|$ inputs uniformly at random.  
\item Bob chooses $\bof{V}$ such that $V_i=v_k$ with probability $1-k$ and 
with probability $k$, $V_i$ is chosen uniformly at random. 
\item They input $\bof{u}$ and $\bof{v}$ into the system and obtain outputs $\bof{x}$ and $\bof{y}$. 
\item They send the inputs over the public authenticated channel. 
\item If less than $(1-k)^2 p  n$ of the inputs were $({U_i},{V_i})=({u}_k,{v}_k)$ they abort. 
\item If any combination of the possible values of $({u}_{\bar{k}},{v}_{\bar{k}})$ (where $\bar{k}$ denotes the inputs which were chosen uniformly at random) occurred less than  $k^2  p  n/|\mathcal{U}||\mathcal{V}|$ times, they abort.
\item From the inputs $({u}_{\bar{k}},{v}_{\bar{k}})$, they estimate $P_{XYUV}(x,y,u,v)$, i.e., they calculate the fraction of times they obtained a certain combination $x,y,u,v$. Call this distribution $P_{XYUV}^{\mathrm{est}}$. 
If ${|\mathcal{U}||\mathcal{V}| P_{XYUV}^{\mathrm{est}}}^T \lambda \geq \ep$, where  $\lambda$ is a  dual feasible solution of (\ref{eq:dualsingle}), or if $P^{\mathrm{est}}(X\neq Y|U=u_k,V=v_k)\geq \delta$, they abort. 
Else they accept. 
\end{enumerate}
\end{protocol}
We define the set $\mathcal{P}$ as the set where we would expect an adversary not to know a lot about the output of the system. 
\begin{definition}\label{def:psets}
The set  $\mathcal{P}$ are all distributions $P_{\bof{XY}|\bof{UV}}$ such that
\begin{align}
\nonumber P^T_{\bof{XY}|\bof{UV}}\cdot \lambda^{\otimes n} &\leq  \ep^n 
\end{align}
for some dual feasible $\lambda$. 
Furthermore, the set $\mathcal{P}^{\eta}$ are all distributions $P_{\bof{XY}|\bof{UV}}$ such that $P^T_{\bof{XY}|\bof{UV}} \lambda^{\otimes n}\geq (\ep+\eta)^n$
\end{definition}

The quantity relevant for our security parameter is $P_{\bof{X}_k\bof{Y}_k|\bof{U}_k\bof{V}_k}^T {\lambda^{\otimes k}}$, \linebreak[4] where $P_{\bof{X}_k\bof{Y}_k|\bof{U}_k\bof{V}_k}$ is the marginal distribution of the systems which will be used to create the key. This quantity is directly proportional to the frequency of a certain event defined by $\bof{X},\bof{Y},\bof{U},\bof{V}$ of the system \linebreak[4] $P_{\bof{XYUV}}=P_{\bof{XY}|\bof{UV}}/ |\mathcal{U}||\mathcal{V}|$ and we will be able to apply classical \linebreak[4] sampling. 
\begin{lemma}\label{lemma:firstsample}
Protocol~\ref{prot:pe} $\epsilon$-securely filters $\mathcal{P^{\eta}}$ with
\begin{align}
\nonumber \epsilon &=2 e^{-\frac{t}{16}\left(\frac{\eta}{|\mathcal{U}||\mathcal{V}| \lambda_{\mathrm{max}}}\right)^2}\ ,
\end{align}
where $t=k^2  p n$ and $\lambda_{\mathrm{max}}=\max_i \lambda_i$.  
\end{lemma}
\begin{proof}
We want to bound $P^T_{\bof{XY}|\bof{UV}}  \lambda^{\otimes n} $. Note that 
\begin{align}
\nonumber P^T_{{XY}|{UV}}\cdot \lambda &= |\mathcal{U}||\mathcal{V}| \cdot P^T_{{XY}{UV}}\cdot \lambda =
|\mathcal{U}||\mathcal{V}|\cdot  \lambda_{\mathrm{max}} \cdot P^T_{{XY}{UV}}\cdot \lambda/\lambda_{\mathrm{max}}
\end{align}
 if the inputs are chosen uniformly. Since $\lambda/\lambda_{\mathrm{max}}\leq 1$, the last part is directly the probability of an event described by $x,y,u,v$ (and maybe independent randomness, see Lemma~\ref{lemma:dualevent}). Estimating $P^T_{\bof{XY}|\bof{UV}} \lambda^{\otimes n} $ within an error $\eta$ corresponds to estimating the probability of this event within $\eta/|\mathcal{U}||\mathcal{V}|\lambda_{\mathrm{max}}$. The claim now follows directly by applying Lemma~\ref{lemma:sampling}, p.~\pageref{lemma:sampling}. 
\end{proof}

\begin{definition}
The sets \emph{$\mathcal{P}_k^{\eta}$} is defined as in Definition~\ref{def:psets}, but where $P_{\bof{XY}|\bof{UV}}=P_{\bof{X}_k\bof{Y}_k|\bof{U}_k\bof{V}_k}$ is the $(2k^{\prime})$-party marginal of a $(2n)$-party non-signalling system.  
\end{definition}

\begin{lemma}\label{lemma:secondsample}
Let $P_{\bof{XY}|\bof{UV}}$ be a $(2n)$-party non-signalling system not in $\mathcal{P}^{\eta}$. And let  $P_{\bof{X}_k\bof{Y}_k|\bof{U}_k\bof{V}_k}$ be the $(2k^{\prime})$-party marginal non-signalling  system for some randomly chosen set of size $k^{\prime}$. Then $P_{\bof{X}_k\bof{Y}_k|\bof{U}_k\bof{V}_k}\notin \mathcal{P}_k^{\eta+\bar{\eta}}$, except with probability
\begin{align}
\nonumber \epsilon &=2 e^{-\frac{k^{\prime}}{16}\left(\frac{\bar{\eta}}{|\mathcal{U}||\mathcal{V}| \lambda_{\mathrm{max}}}\right)^2}\ .
\end{align}
\end{lemma}
\begin{proof}
This is again a direct application of the Sampling Lemma (Lem\-ma~\ref{lemma:sampling}, p.~\pageref{lemma:sampling}), the same way as in the proof of Lemma~\ref{lemma:firstsample}.
\end{proof}

Lemmas~\ref{lemma:firstsample} and~\ref{lemma:secondsample} imply that, if the parameter estimation protocol does not abort, then, almost certainly, the systems which will be used for key generation are such that a secure key can be generated. 
\begin{lemma}\label{lemma:filter}
Protocol~\ref{prot:pe} $\epsilon_1$-securely filters $\mathcal{P}_k^{\eta+\bar{\eta}}$ with 
\begin{align}
\nonumber \epsilon_1 &=2 e^{-\frac{t}{16}\left(\frac{\eta}{|\mathcal{U}||\mathcal{V}| \lambda_{\mathrm{max}}}\right)^2}
+2 e^{-\frac{k^{\prime}}{16}\left(\frac{\bar{\eta}}{|\mathcal{U}||\mathcal{V}| \lambda_{\mathrm{max}}}\right)^2}\ ,
\end{align}
where $t=k^2 p n$ and $\lambda_{\mathrm{max}}=\max_i \lambda_i$. 
\end{lemma}
\begin{proof}
This is a direct consequence of Lemmas~\ref{lemma:firstsample} and~\ref{lemma:secondsample}.
\end{proof}

Now let us also see that the parameter estimation protocol will abort on inputs for which the information 
reconciliation might not work and where Alice and Bob might, therefore, obtain different keys. 
\begin{definition}
The set  \emph{$\mathcal{B}$} are all pairs of $n$-bit strings $({\bof{x},\bof{y}})$ such that $d_{\mathrm{H}}({\bof{x},\bof{y}})\leq \delta n$. The set  \emph{$\mathcal{B}^{\eta}$} are all pairs of $n$-bit strings $({\bof{x},\bof{y}})$ such that $d_{\mathrm{H}}({\bof{x},\bof{y}})\geq (\delta+\eta) n$. By \emph{$\mathcal{B}^{\eta}_k$} we denote are all pairs of $k^{\prime}$-bit strings $({\bof{x}_k,\bof{y}_k})$ such that $d_{\mathrm{H}}({\bof{x}_k,\bof{y}_k})\geq (\delta+\eta) k^{\prime}$. 
\end{definition}

\begin{lemma}\label{lemma:irsample}
Let $({\bof{x}_k,\bof{y}_k})$ be the outputs on input $({U_i},{V_i})=({u}_k,{v}_k)$. Then protocol~\ref{prot:pe} 
$\epsilon_2$-securely filters $({\bof{x}_k,\bof{y}_k})\in \mathcal{B}^{\eta+\bar{\eta}}_k$, for any $\eta,\bar{\eta}>0$,  with
\begin{align}
\epsilon_2 &= 2 e^{-\frac{t^{\prime}}{16}\eta^2}+2 e^{-\frac{k^{\prime}}{16}\bar{\eta}^2}
\end{align}
with $t^{\prime}=k^2 p n/|\mathcal{U}||\mathcal{V}|$ and $k^{\prime}=(1-k)^2 p n$. 
\end{lemma}
\begin{proof}
This follows from applying Lemma~\ref{lemma:sampling}, p.~\pageref{lemma:sampling} twice. 
\end{proof}

Lemmas~\ref{lemma:filter} and~\ref{lemma:irsample} imply that Protocol~\ref{prot:pe} either aborts, or 
the key created will be both secret and correct. The probability that the parameter-estimation protocol 
lets a `bad' system pass is at most $\epsilon=\epsilon_1+\epsilon_2$, i.e., it is $\epsilon$-secure for 
some $\epsilon \in O(2^{-n})$.

Let us also verify, that there exist input systems on which the parameter-estimation protocol does not 
abort, i.e., it is \emph{robust}. 
\begin{definition}
The set \emph{$\mathcal{P}^{-\eta}$} are all distributions 
$P_{\bof{XY}|\bof{UV}}$ such that \linebreak[4]
$P^T_{\bof{XY}|\bof{UV}}\lambda^{\otimes n}\leq (\ep-\eta)^n$ and 
$\sum_{(\bof{x},\bof{y}):d_{\mathrm{H}}(\bof{x},\bof{y})\geq m}P_{\bof{XY}|\bof{U}=\bof{u}_k,\bof{V}=\bof{v}_k}
(\bof{x},\bof{y}) \leq  (\delta-\eta)^m$ for all $m$.
\end{definition}
Note that, for example, the distribution describing $n$ independent systems of which the individual systems are `good enough' is in this set. 

On an input in $\mathcal{P^{-\eta}}$, the probability that the parameter-estimation protocol aborts is $O(2^{-n})$, i.e., the protocol is \emph{robust}. 
\begin{lemma}\label{lemma:perobust}
Protocol~\ref{prot:pe} is $\epsilon^{\prime}$-robust on $\mathcal{P^{-\eta}}$ 
with 
\begin{align}
\nonumber \epsilon^{\prime} &=2 e^{-\frac{t}{16}\left(\frac{\eta}{|\mathcal{U}||\mathcal{V}| \lambda_{\mathrm{max}}}\right)^2}
+ 2 e^{-\frac{t^{\prime}}{16}
\eta^2}\\
\nonumber 
&\quad
+e^{-2n \left((1-p)(1-k)^2 \right)^2}
+|\mathcal{U}||\mathcal{V}|\cdot e^{-2n \left(\frac{(1-p)k^2}{ |\mathcal{U}||\mathcal{V}|} \right)^2} \ ,
\end{align}
where $t=k^2 p n$ and $t^{\prime}=k^2 p n/|\mathcal{U}||\mathcal{V}|$
\end{lemma}
\begin{proof}
The probability to wrongly estimate the frequency is given by Lem\-mas~\ref{lemma:firstsample} and~\ref{lemma:irsample}. The last two terms are the probability that any input combination does not occur often enough and follow directly from a Chernoff bound (see Lemma~\ref{lemma:chernoff}, p.~\pageref{lemma:chernoff}). 
\end{proof}

\subsection{Information reconciliation}

\emph{Information reconciliation}~\cite{brassardsalvail} is the process responsible to make Alice's and Bob's data highly correlated, i.e., if we consider Bob's string as an erroneous version of Alice's, then information reconciliation corresponds to error correction. The idea is that Alice applies a function to her data and sends the function value to Bob. Bob searches the value `closest' to his data that maps to this function value and should, almost certainly, be able to recover Alice's data.

\begin{figure}[h]
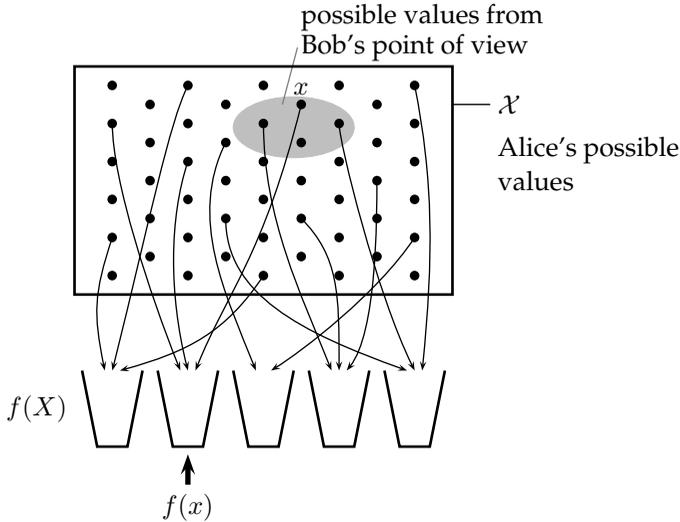

\centering
\pspicture*[](-2,-3.2)(8,4)
\psellipse[linewidth=0pt,linecolor=lightgray,fillstyle=solid,fillcolor=lightgray](2.9,2.2)(0.8,0.4)
\psline[linewidth=0.5pt,linecolor=gray]{-}(2.75,2.5)(2.95,3.25)
\rput[c]{0}(-0.5,-1.5){$f(X)$}
\pspolygon[linewidth=1pt](0,0)(5,0)(5,3)(0,3)
\rput[c]{0}(0.0,-0.25){
\psdots(0.5,0.5)(0.5,1)(0.5,1.5)(0.5,2)(0.5,2.5)(0.5,3)
}
\rput[c]{0}(0.5,-0){
\psdots(0.5,0.5)(0.5,1)(0.5,1.5)(0.5,2)(0.5,2.5)
}
\rput[c]{0}(1,-0.25){
\psdots(0.5,0.5)(0.5,1)(0.5,1.5)(0.5,2)(0.5,2.5)(0.5,3)
}
\rput[c]{0}(1.5,0){
\psdots(0.5,0.5)(0.5,1)(0.5,1.5)(0.5,2)(0.5,2.5)
}
\rput[c]{0}(2,-0.25){
\psdots(0.5,0.5)(0.5,1)(0.5,1.5)(0.5,2)(0.5,2.5)(0.5,3)
}
\rput[c]{0}(2.5,0){
\psdots(0.5,0.5)(0.5,1)(0.5,1.5)(0.5,2)(0.5,2.5)
}
\rput[c]{0}(3,-0.25){
\psdots(0.5,0.5)(0.5,1)(0.5,1.5)(0.5,2)(0.5,2.5)(0.5,3)
}
\rput[c]{0}(3.5,0){
\psdots(0.5,0.5)(0.5,1)(0.5,1.5)(0.5,2)(0.5,2.5)
}
\rput[c]{0}(4,-0.25){
\psdots(0.5,0.5)(0.5,1)(0.5,1.5)(0.5,2)(0.5,2.5)(0.5,3)
}
\psline[linewidth=1pt]{-}(0.1,-1)(0.3,-2)(0.7,-2)(0.9,-1)
\rput[c]{0}(1,0){
\psline[linewidth=1pt]{-}(0.1,-1)(0.3,-2)(0.7,-2)(0.9,-1)
}
\rput[c]{0}(2,0){
\psline[linewidth=1pt]{-}(0.1,-1)(0.3,-2)(0.7,-2)(0.9,-1)
}
\rput[c]{0}(3,0){
\psline[linewidth=1pt]{-}(0.1,-1)(0.3,-2)(0.7,-2)(0.9,-1)
}
\rput[c]{0}(4,0){
\psline[linewidth=1pt]{-}(0.1,-1)(0.3,-2)(0.7,-2)(0.9,-1)
}
\psbezier[linewidth=0.5pt]{->}(0.5,0.75)(0.25,0)(0.25,-0.25)(0.4,-1)
\psbezier[linewidth=0.5pt]{->}(1.5,2.75)(1.25,2.25)(1,1)(0.5,-1)
\psbezier[linewidth=0.5pt]{->}(2.5,0.25)(2,-0.5)(1.5,-0.75)(0.6,-1)
\psbezier[linewidth=0.5pt]{->}(0.5,2.25)(0.5,1.5)(1,0)(1.4,-1)
\psbezier[linewidth=0.5pt]{->}(1.5,1.75)(1.25,1.0)(1.25,0)(1.5,-1)
\psbezier[linewidth=0.5pt]{->}(3,2.5)(2.5,0.5)(2,0)(1.6,-1)
\rput[c]{0}(3,2.7){$x$}
\psbezier[linewidth=0.5pt]{->}(2,2)(1.5,1)(2,0)(2.4,-1)
\psbezier[linewidth=0.5pt]{->}(4.5,0.75)(4,0)(3,-0.75)(2.6,-1)
\psbezier[linewidth=0.5pt]{->}(2.5,2.25)(2.5,1)(3,0)(3.4,-1)
\psbezier[linewidth=0.5pt]{->}(3,1)(3.5,0.5)(3.5,0)(3.5,-1)
\psbezier[linewidth=0.5pt]{->}(4,1.5)(4,1)(4,-0.5)(3.6,-1)
\psbezier[linewidth=0.5pt]{->}(2,1)(2,0)(3,-0.5)(4.4,-1)
\psbezier[linewidth=0.5pt]{->}(3.5,2.25)(3.5,2)(4,0)(4.5,-1)
\psbezier[linewidth=0.5pt]{->}(4.5,2.75)(4.75,1.5)(4.75,0.5)(4.6,-1)
\psline[linewidth=2pt]{<-}(1.5,-2.1)(1.5,-2.5)
\rput[c]{0}(1.5,-2.8){$f(x)$}
\psline[linewidth=0.5pt]{-}(5,2.5)(5.5,2.5)
\rput[l]{0}(5.6,2.5){$\mathcal{X}$}
\rput[l]{0}(5.6,1.9){Alice's possible} 
\rput[l]{0}(5.6,1.5){values}
\rput[l]{0}(3,3.65){possible values from}
\rput[l]{0}(3,3.25){Bob's point of view}
\endpspicture
\caption{The principle of information reconciliation. Alice sends to Bob the function $f$ and the value of the function applied to $x$,  $f(x)$. Bob can then recover the value of $x$.}
\end{figure}

\begin{definition}
Let $\mathcal{P}$ be a set of distributions $P_{\bof{XY}}$ (or $\mathcal{B}$ a set of bit-string pairs $(\bof{x},\bof{y})$). We say that an information reconciliation protocol is \emph{$\epsilon$-correct} on $\mathcal{P}$ (or $\mathcal{B}$), if on input $P_{\bof{XY}}\in \mathcal{P}$ ($(\bof{x},\bof{y})\in \mathcal{B}$) it outputs $\bof{x}^{\prime}$, $\bof{y}^{\prime}$ such that $\bof{x}^{\prime}\neq \bof{y}^{\prime}$ with probability at most $\epsilon$. 
\end{definition}
We will only consider \emph{one-way protocols}, where Alice sends information about her string to Bob, but Bob does not send anything. In that case, $\bof{x}=\bof{x}^{\prime}$, and only Bob changes his string. 
\begin{definition}
Let $\mathcal{P}$ be a set of distributions $P_{\bof{XY}}$. We say that an information reconciliation protocol is \emph{$\epsilon$-robust} on $\mathcal{P}$ if on input $P_{\bof{XY}}\in \mathcal{P}$ it aborts with probability at most $\epsilon$. 
\end{definition}
We will actually consider protocols where Alice and Bob never abort. But it is possible to introduce different protocols where Bob has a small chance to abort, for example, if he cannot find a suitable $\bof{y}^{\prime}$ or if he finds more than one suitable $\bof{y}^{\prime}$.

\begin{protocol}[Information reconciliation]\label{prot:ir}\ 
\begin{enumerate}
\item Alice obtains $\bof{x}$ and Bob $\bof{y}$ (distributed according to $P_{\bof{XY}}$) with $\mathcal{\bof{X}}=\mathcal{\bof{Y}}=\{0,1\}^n$.  
Alice outputs $\bof{x}^{\prime}=\bof{x}$. 
\item Alice chooses a matrix $A\in M_{m\times n}(GF(2))$ at random and calculates $r=A\odot \bof{x}$ (where `$\odot$' denotes the multiplication over $GF(2)$). 
\item She sends the matrix $A$ and $r$ to Bob.
\item Bob chooses $\bof{y}^{\prime}$ such that $d_{\mathrm{H}}(\bof{y},\bof{y}^{\prime})$ is minimal among all strings $\bof{z}$ with $f(\bof{z})=A\odot \bof{x}$ (if there are two possibilities, he chooses one at random) and outputs $\bof{y}^{\prime}$. 
\end{enumerate}
\end{protocol}

To see that this protocol works, we need a result from~\cite{carterwegman} about two-universal sets of hash functions and from~\cite{brassardsalvail} about information reconciliation. 
\begin{definition}\label{def:twouniv}
A set of functions $\mathcal{F}$ such that $f\colon \mathcal{X}\rightarrow \mathcal{Z}$ is called \emph{two-universal} if $\Pr_f[f(\bof{x})=f(\bof{x}^{\prime})]\leq {1}/{|\mathcal{Z}|}$ for 
any $\bof{x},\bof{x}^{\prime}\in \mathcal{X}$, and where the function $f$ is chosen uniformly at random from $\mathcal{F}$.
\end{definition}
\begin{theorem}[Carter, Wegman~\cite{carterwegman}]
The set of functions $f_A(\bof{x}):=A\odot \bof{x}$, where $A$ is an $n\times m$-matrix over $GF(2)$, is two-universal.
\end{theorem}
Brassard and Salvail~\cite{brassardsalvail} (see Theorem~\ref{th:ir}, p.~\pageref{th:ir}) showed that information reconciliation can be achieved by a two-universal function. We give a slightly modified 
version of their result in Lemma~\ref{lemma:ir}.  
\begin{lemma}\label{lemma:ir}
Let $\bof{x}$ be an $n$-bit string and $\bof{y}$ another $n$-bit string such that 
$d_{\mathrm{H}}(\bof{x},\bof{y})\leq \delta^{\prime} n$. 
 Assume the 
 function $f\colon \{0,1\}^n\rightarrow \{0,1\}^m$ is chosen at random amongst a two-universal  set of functions.  
 Choose $\bof{y}^{\prime}$ such that $d_{\mathrm{H}}(\bof{y},\bof{y}^{\prime})$ is minimal among all strings $\bof{r}$ with $f(\bof{r})=f(\bof{x})$. 
 Then 
\begin{align} 
\nonumber \Pr[{\bof{x}\neq \bof{y}^{\prime}}] &\leq 
2^{n\cdot h(\delta^{\prime}
)-m} \ ,
\end{align}
 where $h(p)=-p\cdot \log_2 p - (1-p)\log_2 (1-p)$ is the binary entropy function.
\end{lemma}
\begin{proof}
The probability that a $\bof{y}^{\prime}\neq \bof{x}$ with 
$d_{\mathrm{H}}(\bof{x},\bof{y}^{\prime})\leq \delta^{\prime} n$ are mapped to the same value by $f$, when $f\in \mathcal{F}$ is chosen at random, is
\begin{align}
\nonumber \Pr[f(\bof{x})=f(\bof{y}^{\prime})|d_{\mathrm{H}}(\bof{x},\bof{y}^{\prime})\leq \delta^{\prime}\cdot n
 ] &\leq  2^{-m}\cdot \sum_{i=0}^{\delta^{\prime}\cdot n 
}\binom{n}{i} \\
\nonumber &\leq  2^{-m}2^{n\cdot h(\delta^{\prime}
)}\ . \qedhere
\end{align} 
\end{proof}

\begin{lemma}\label{lemma:ircorrect}
Protocol~\ref{prot:ir} is $\epsilon$-correct on input $(\bof{x},\bof{y})$
such that 
$d_{\mathrm{H}}(\bof{x},\bof{y}^{\prime})\leq \delta^{\prime} n$, with
\begin{align}
\nonumber \epsilon &= 
2^{n\cdot h(\delta^{\prime}
)-m}\ .
\end{align}
\end{lemma}
\begin{proof}
This follows directly from Lemma~\ref{lemma:ir}. 
\end{proof}

\begin{lemma}
Protocol~\ref{prot:ir} is $0$-robust on all inputs. 
\end{lemma}
\begin{proof}
There is always a $\bof{y}^{\prime}$ such that $A\odot \bof{y}^{\prime}=r$, because $A\odot \bof{x}=r$. Therefore, the protocol never aborts. 
\end{proof}

The above lemmas show that in the limit of large $n$, $m=\lceil n\cdot h(\delta^{\prime})\rceil$ (where $\delta^{\prime}$ is the fraction of Bob's 
bits which are different from Alice's and $h$ the binary entropy function),  is both necessary and sufficient for Bob to correct the 
errors in his raw key, i.e., 
the protocol is $\epsilon$-correct for some $\epsilon\in O(2^{-n})$.

\subsection{Privacy amplification}

After Alice and Bob have done information reconciliation, they hold (almost certainly) the same strings. Eve might have some information about this string. \emph{Privacy amplification}~\cite{bbr,ILL} is the process making from this string a highly secure key. The idea of privacy amplification is very similar to the one of information reconciliation: Alice and Bob apply a (public) function to their data. As long as Eve does not know the initial data perfectly, she will know almost nothing about the function value. 

\begin{figure}[h]
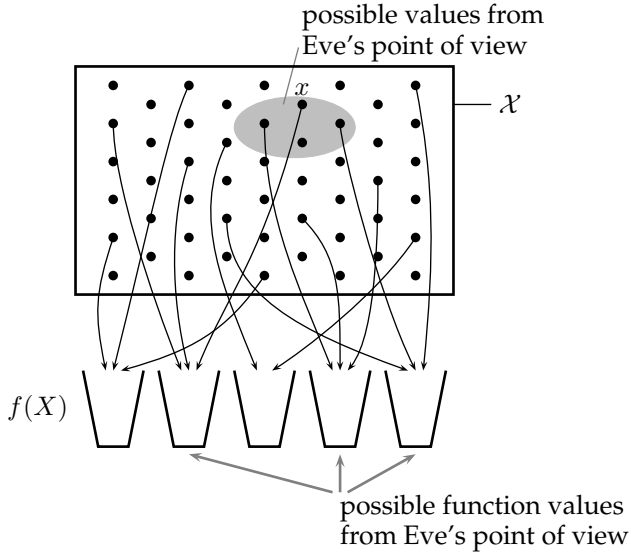

\centering
\pspicture*[](-2,-3.5)(8,4)
\psellipse[linewidth=0pt,linecolor=lightgray,fillstyle=solid,fillcolor=lightgray](2.9,2.2)(0.8,0.4)
\psline[linewidth=0.5pt,linecolor=gray]{-}(2.75,2.5)(2.95,3.25)
\rput[c]{0}(-0.5,-1.5){$f(X)$}
\pspolygon[linewidth=1pt](0,0)(5,0)(5,3)(0,3)
\rput[c]{0}(0.0,-0.25){
\psdots(0.5,0.5)(0.5,1)(0.5,1.5)(0.5,2)(0.5,2.5)(0.5,3)
}
\rput[c]{0}(0.5,-0){
\psdots(0.5,0.5)(0.5,1)(0.5,1.5)(0.5,2)(0.5,2.5)
}
\rput[c]{0}(1,-0.25){
\psdots(0.5,0.5)(0.5,1)(0.5,1.5)(0.5,2)(0.5,2.5)(0.5,3)
}
\rput[c]{0}(1.5,0){
\psdots(0.5,0.5)(0.5,1)(0.5,1.5)(0.5,2)(0.5,2.5)
}
\rput[c]{0}(2,-0.25){
\psdots(0.5,0.5)(0.5,1)(0.5,1.5)(0.5,2)(0.5,2.5)(0.5,3)
}
\rput[c]{0}(2.5,0){
\psdots(0.5,0.5)(0.5,1)(0.5,1.5)(0.5,2)(0.5,2.5)
}
\rput[c]{0}(3,-0.25){
\psdots(0.5,0.5)(0.5,1)(0.5,1.5)(0.5,2)(0.5,2.5)(0.5,3)
}
\rput[c]{0}(3.5,0){
\psdots(0.5,0.5)(0.5,1)(0.5,1.5)(0.5,2)(0.5,2.5)
}
\rput[c]{0}(4,-0.25){
\psdots(0.5,0.5)(0.5,1)(0.5,1.5)(0.5,2)(0.5,2.5)(0.5,3)
}
\psline[linewidth=1pt]{-}(0.1,-1)(0.3,-2)(0.7,-2)(0.9,-1)
\rput[c]{0}(1,0){
\psline[linewidth=1pt]{-}(0.1,-1)(0.3,-2)(0.7,-2)(0.9,-1)
}
\rput[c]{0}(2,0){
\psline[linewidth=1pt]{-}(0.1,-1)(0.3,-2)(0.7,-2)(0.9,-1)
}
\rput[c]{0}(3,0){
\psline[linewidth=1pt]{-}(0.1,-1)(0.3,-2)(0.7,-2)(0.9,-1)
}
\rput[c]{0}(4,0){
\psline[linewidth=1pt]{-}(0.1,-1)(0.3,-2)(0.7,-2)(0.9,-1)
}
\psbezier[linewidth=0.5pt]{->}(0.5,0.75)(0.25,0)(0.25,-0.25)(0.4,-1)
\psbezier[linewidth=0.5pt]{->}(1.5,2.75)(1.25,2.25)(1,1)(0.5,-1)
\psbezier[linewidth=0.5pt]{->}(2.5,0.25)(2,-0.5)(1.5,-0.75)(0.6,-1)
\psbezier[linewidth=0.5pt]{->}(0.5,2.25)(0.5,1.5)(1,0)(1.4,-1)
\psbezier[linewidth=0.5pt]{->}(1.5,1.75)(1.25,1.0)(1.25,0)(1.5,-1)
\psbezier[linewidth=0.5pt]{->}(3,2.5)(2.5,0.5)(2,0)(1.6,-1)
\rput[c]{0}(3,2.7){$x$}
\psbezier[linewidth=0.5pt]{->}(2,2)(1.5,1)(2,0)(2.4,-1)
\psbezier[linewidth=0.5pt]{->}(4.5,0.75)(4,0)(3,-0.75)(2.6,-1)
\psbezier[linewidth=0.5pt]{->}(2.5,2.25)(2.5,1)(3,0)(3.4,-1)
\psbezier[linewidth=0.5pt]{->}(3,1)(3.5,0.5)(3.5,0)(3.5,-1)
\psbezier[linewidth=0.5pt]{->}(4,1.5)(4,1)(4,-0.5)(3.6,-1)
\psbezier[linewidth=0.5pt]{->}(2,1)(2,0)(3,-0.5)(4.4,-1)
\psbezier[linewidth=0.5pt]{->}(3.5,2.25)(3.5,2)(4,0)(4.5,-1)
\psbezier[linewidth=0.5pt]{->}(4.5,2.75)(4.75,1.5)(4.75,0.5)(4.6,-1)
\psline[linewidth=1pt,linecolor=gray]{->}(3.4,-2.6)(1.5,-2.1)
\psline[linewidth=1pt,linecolor=gray]{->}(3.5,-2.6)(3.5,-2.1)
\psline[linewidth=1pt,linecolor=gray]{->}(3.6,-2.6)(4.5,-2.1)
\rput[l]{0}(3.5,-2.8){possible function values}
\rput[l]{0}(3.5,-3.2){from Eve's point of view}
\psline[linewidth=0.5pt]{-}(5,2.5)(5.5,2.5)
\rput[l]{0}(5.6,2.5){$\mathcal{X}$}
\rput[l]{0}(3,3.65){possible values from}
\rput[l]{0}(3,3.25){Eve's point of view}
\endpspicture
\caption{The principle of privacy amplification. Alice and Bob apply a public function to $x$ to obtain $f(x)$. Eve, who does not know $x$ exactly, knows almost nothing about $f(x)$.}
\end{figure}

We now want to show that privacy amplification against non-signalling adversaries is possible using a random linear function, i.e., by applying the XOR to randomly chosen subsets of the bits. In Section~\ref{subsec:nsxor}, we have seen that a secure bit can be created using the XOR. Let us first estimate what the security of the XOR of a random subset of the outputs of a system $\in \mathcal{P}$ can be.

\begin{lemma}\label{lemma:distancrandombit}
Let $c$ be a random vector of length $n$ over $GF(2)$, and $P_{\bof{XY}|\bof{UV}}\in \mathcal{P}$ an $(2n)$-party non-signalling system. Call $S_c=c\odot \bof{X}$.  Then 
\begin{align}
\nonumber d(S_c|Z(W_{\mathrm{n-s}}),Q) &\leq 
 \frac{1}{2}\left(\frac{1+\ep+\tilde{\eta}}{2}\right)^n
+ e^{-\frac{n}{8}} +  
e^{-\frac{n}{64}\left(\frac{\tilde{\eta}}{|\mathcal{U}||\mathcal{V}| \lambda_{\mathrm{max}}}\right)^2}
\end{align}
where $Q=(\bof{U}=\bof{u},\bof{V}=\bof{v},C)$.
\end{lemma}
\begin{proof}
We need to estimate $P^T_{\bof{X}_s\bof{Y}_s|\bof{U}_s\bof{V}_s} \lambda^{\otimes s}$ for some randomly chosen set $S$. 
We distinguish two cases depending on the size $s$ of the set $S$. By the Chernoff bound (see Lemma~\ref{lemma:chernoff}, p.~\pageref{lemma:chernoff}), $s\leq n/4$ happens with probability at most $e^{-\frac{n}{8}}$. For $s>n/4$, by Lemma~\ref{lemma:firstsample}, $P^T_{\bof{X}_s\bof{Y}_s|\bof{U}_s\bof{V}_s} \lambda^{\otimes s}$ is at most $(\ep+\tilde{\eta})^s$, except with probability $2e^{-\frac{s}{16}\left(\frac{\tilde{\eta}}{|\mathcal{U}||\mathcal{V}| \lambda_{\mathrm{max}}}\right)^2}\leq 2 e^{-\frac{n}{64}\left(\frac{\tilde{\eta}}{|\mathcal{U}||\mathcal{V}| \lambda_{\mathrm{max}}}\right)^2}$.  The distance is bounded by half the sum of the two terms. 
We obtain the statement by taking the average over all possible choices of sets $S$, using the binomial formula, i.e., $\sum_i \binom{n}{i}x^i=(1+x)^n$, and the union bound. 
\end{proof}

Let us now calculate the security of a key $S$, where each key bit is the XOR of a random subset of the raw key. 
We first reduce the security of the key $S$ to the question of the security of every single bit.
\begin{lemma}\label{lemma:distanceseveralbits}
 Assume $S:=[S_1,\dotsc, S_s]$, where the $S_i$ are bits. Then
\begin{align}
\nonumber d(S|Z(W_{\mathrm{n-s}}),Q) &\leq 
\sum_i d(S_i|Z(W_{\mathrm{n-s}}),Q,S_{1},\dotsc,S_{i-1})\ .
\end{align}
\end{lemma}
\begin{proof}
\begin{align}
\nonumber  d(S|&Z(W_{\mathrm{n-s}}),Q)
\\
\nonumber  &= \sum_{s,q}\max_{w:\mathrm{n-s}} \sum_z \left|P_{S,Z,Q|W=w}(s,z,q)-\frac{1}{2^s}\cdot P_{Z,Q|W=w}(z,q)\right|\\
 \nonumber &\leq  \sum_{s,q}\max_w \sum_z \Biggl[ |P_{S,Z,Q|W=w}(s,z,q)
\\
\nonumber
&\quad 
 -\frac{1}{2}\cdot P_{S_1\dotso S_{s-1},Z,Q|W=w}(s_1,\dotsc ,s_{s_1},z,q)| \\
\nonumber &\quad 
 +\dotsb +\frac{1}{2^{s-1}}\left|P_{S_1,Z,Q|W=w}(s_1,z,q)-\frac{1}{2}\cdot P_{Z,Q|W=w}(z,q)]\right|
 \Biggr] \\
\nonumber 
&\leq \sum_i d(S_i|Z(W_{\mathrm{n-s}}),Q,S_{1},\dotsc ,S_{i-1})\ ,
\end{align}
where the first equation is by the definition of the distance from uniform and the second inequality holds by the triangle inequality. 
\end{proof}

We, therefore, need to bound the distance from uniform of the $i$\textsuperscript{th} key bit given all previous bits. 

For this, we need to show a few lemmas. The first one states that the linear combination of two random bit vectors (modulo $2$) is again a random vector. The second one implies that in order to bound the distance from uniform of the $i$\textsuperscript{th} bit given all previous bits, it is enough to bound the distance from uniform given all linear combinations of these bits. 
\begin{lemma}\label{lemma:lin_com_of_random_vect}
Assume $\bof{u}$ and $\bof{v}$ are $n$-bit vectors and $P_U$ is the uniform distribution over all these vectors. 
Define the vector $\bof{w}=\bof{u}\oplus \bof{v}$. Then $\bof{w}$ is again distributed according to the uniform distribution, i.e., 
\begin{align}
\nonumber P_{\bof{u}\leftarrow P_U}P_{\bof{v}\leftarrow P_U}(\bof{u}\oplus \bof{v}) &=P_{\bof{w}\leftarrow P_U}(\bof{w})\ .
\end{align}
\end{lemma}
\begin{proof}
The uniform distribution over all $n$-bit vectors can be obtained by drawing each of the $n$-bits at random, 
i.e., $P(0)=P(1)=1/2$. The XOR of two random bits is again a random bit, i.e., $P(0)=P(1)=1/2$ and therefore, 
$\bof{w}$ is also a vector drawn according to the uniform distribution over all $n$-bit vectors. 
\end{proof}

\begin{lemma}\label{lemma:distance_set_given_other_sets}
Let $S_1,\dotsc,S_k$ be random bits. If $S_k$ is uniform given all linear combinations over $GF(2)$ of $S_1,\dotsc,S_{k-1}$, i.e., it holds that  
$P_{S_k|\bigoplus_{i\in I}S_i}(0)=P_{S_k|\bigoplus_{i\in I}S_i}(1)$ for all $I\subseteq \{1,\dotsc,k-1\}$, then $S_k$ is uniform given 
$S_1,\dotsc, \linebreak[4] S_{k-1}$, i.e., $P_{S_k|S_1\dotso S_{k-1}}(0)=P_{S_k|S_1 \dotso S_{k-1}}(1)$.
\end{lemma}
\begin{proof}
We prove the case $k=3$, the general case follows by induction. We have to show that if $P_{S_3|S_1}$, $P_{S_3|S_2}$ and 
$P_{S_3|S_1\oplus S_2}$ are uniform, then $P_{S_3|S_1 S_2}$ is uniform. Consider the probabilities $P_{S_1 S_2 S_3}$. Since  
$P_{S_3|S_1}$ is uniform, we obtain the constraints on $P_{S_1 S_2 S_3}$ (we drop the index) 
\begin{align}
\label{eq:1} P
(0,0,0)+P
(0,1,0)&= P
(0,0,1)+P
(0,1,1)\\
\nonumber P
(1,0,0)+P
(1,1,0)&= P
(1,0,1)+P
(1,1,1)\ .
\end{align}
Since $P_{S_3|S_2}$ is uniform, 
\begin{align}
\nonumber P
(0,0,0)+P
(1,0,0)&= P
(0,0,1)+P
(1,0,1)\\
\label{eq:2} P
(0,1,0)+P
(1,1,0)&= P
(0,1,1)+P
(1,1,1)\ .
\end{align}
And from the fact that  $P_{S_3|S_1\oplus S_2}$ is uniform, we obtain
\begin{align}
\label{eq:3} P
(0,0,0)+P
(1,1,0)&= P
(0,0,1)+P
(1,1,1)\\
\nonumber P
(0,1,0)+P
(1,0,0)&= P
(0,1,1)+P
(1,0,1)\ .
\end{align}
Subtract (\ref{eq:2}) from (\ref{eq:1}) and add (\ref{eq:3}) to obtain
\begin{align}
\nonumber 2\cdot P
(0,0,0)&= 2\cdot P
(0,0,1)
\end{align}
which implies 
\begin{align}
\nonumber 
P_{S_3|S_1=0,S_2=0}(0)&= \frac{P_{S_1 S_2 S_3}(0,0,0)}{P_{S_1 S_2 S_3}(0,0,0)+P_{S_1 S_2 S_3}(0,0,1)}\\
\nonumber &= P_{S_3|S_1=0,S_2=0}(1)\ .
\end{align}
Uniformity of $S_1$ and $S_2$ follows in an analogous way. 
\end{proof}

Now we can calculate the distance from uniform of the $i$\textsuperscript{th} bit given the bits $1$ to $i-1$ by the union bound. 
\begin{lemma}\label{lemma:distance_k_th_bit}
Let $P_{\bof{XY}|\bof{UV}}\in \mathcal{P}$ and $S:=A\odot \bof{X}$, where $A$ is a $i\times n$-matrix over $GF(2)$ and be $P_A$ the uniform distribution 
over all these matrices. $Q:=(\bof{U}=\bof{u},\bof{V}=\bof{v},A)$. Then
\begin{multline}
\nonumber d(S_i|Z(W_{\mathrm{n-s}}),Q,S_1,\dotsc,S_{i-1})\\
\leq 
 \frac{2^{i-1}}{2}\left(\frac{1+\ep+\tilde{\eta}}{2}\right)^n
+2^{i-1} e^{-\frac{n}{8}} +  
2^{i-1} e^{-\frac{n}{64} \left(\frac{\tilde{\eta}}{|\mathcal{U}||\mathcal{V}| \lambda_{\mathrm{max}}}\right)^2}\ .
\end{multline}
\end{lemma}
\begin{proof}
By Lemma~\ref{lemma:distance_set_given_other_sets}, bounding the distance from uniform of $S_i$ given $S_1,\dotsc,S_{i-1}$ corresponds to bounding the distance from uniform of 
$S_i$ given all linear combinations over $GF(2)$ of $S_1,\dotsc,S_{i-1}$. 
For each linear combination $\bigoplus_{j\in I}S_j$ define the random bit $S_c=c\odot \bof{X}$, where 
$c=\bigoplus_{j\in I}a_j\oplus a_i$ and $a_j$ denotes the $j$\textsuperscript{th} line of the matrix $A$. Note that $S_c$ is a random linear function 
over $\bof{X}$. If $S_c$ is uniform and independent of 
$S_1,\dotsc,S_{i-1}$, then $S_i$ is uniform given this specific linear combination. However, the distance from uniform and independent of 
$S_c$ is given by Lemma~\ref{lemma:distancrandombit}. 
By the union bound over all 
$2^{i-1}$ 
possible linear combinations of 
$S_1,\dotsc,S_{i-1}$, we obtain the probability that $S_i$ is uniform given $S_1,\dotsc,S_{i-1}$, i.e.,  
\begin{align}
\nonumber d(S_i|Z(W_{\mathrm{n-s}}),Q,S_1,\dotsc ,S_{i-1})
&\leq  2^{i-1}\cdot d(S_c|Z(W_{\mathrm{n-s}}),Q)\ . \qedhere
\end{align}
\end{proof}

Now we can bound the distance from uniform of a key $S:=S_1\dotso S_s$ by Lemma~\ref{lemma:distanceseveralbits} 
and~\ref{lemma:distance_k_th_bit}. 
\begin{lemma}\label{lemma:distance_key_string}
Assume $S:=A\odot \bof{X}$, where $A$ is a $s\times n$-matrix over $GF(2)$ and be 
$P_A$ the uniform distribution over all these matrices. $Q:=(\bof{U}=\bof{u},A)$. Then
\begin{align}
\nonumber d(S|Z(W_{\mathrm{n-s}}),Q)  
& \leq 
 \frac{2^{s}}{2}\left(\frac{1+\ep+\tilde{\eta}}{2}\right)^n
+2^{s} e^{-\frac{n}{8}} +  
2^{s} e^{-\frac{n}{64} \left(\frac{\tilde{\eta}}{|\mathcal{U}||\mathcal{V}| \lambda_{\mathrm{max}}}\right)^2}\ .
\end{align}
\end{lemma}
\begin{proof}
This follows from Lemmas~\ref{lemma:distanceseveralbits} and~\ref{lemma:distance_k_th_bit}, 
when using the expression for geometric series, i.e., 
\begin{align}
\nonumber \sum_{i=1}^s 2^{i-1} &=\frac{2^s-1}{2-1}\leq 2^s\ . \qedhere
\end{align}
\end{proof}

This expression is in $O(2^{-n})$ whenever $s=q\cdot n$ for some (constant) $q$ with $2^{q-1}(1+\ep)<1$.

\subsection{Key distribution}\label{subsec:keydistr}

Now, we can put everything together in order to create a key-agreement scheme using the steps above.

\begin{definition}\label{def:protocolsecure}
A key-distribution protocol is said to be \emph{$\epsilon$-secret} against non-signalling adversaries if, on all inputs, $d(S_A|Z(W_{\mathrm{n-s}}),Q)\leq \epsilon$. It is said to be \emph{$\epsilon^{\prime}$-correct}  if, on all inputs, $\Pr[S_A\neq S_B]\leq \epsilon^{\prime}$, and it is said to be \emph{$\epsilon^{\prime\prime}$-secure} if it is both secret and correct, i.e., $\delta(\mathcal{S}_{\mathrm{real}},\mathcal{S}_{\mathrm{ideal}})<\epsilon^{\prime\prime}$.  
\end{definition}
\begin{lemma}
A key-distribution protocol which is $\epsilon$-secret and $\epsilon^{\prime}$-correct, is $(\epsilon+\epsilon^{\prime})$-secure. 
\end{lemma}
\begin{proof}
This follows directly from the triangle inequality (Lemma~\ref{lemma:distance}, \linebreak[4] p.~\pageref{lemma:distance}). 
\end{proof}

\begin{protocol}[Key distribution secure against non-signalling adversaries]\label{prot:keyagreement}\ 
\begin{enumerate}
\item Alice and Bob obtain a system $P_{\bof{XY}|\bof{UV}}$
\item They do parameter estimation using Protocol~\ref{prot:pe}.
\item Information reconciliation and privacy amplification: Alice chooses a matrix $A\in M_{(s+r)\times n}$ and calculates $[S_A,R]=A\odot \bof{x}$. 
\item Alice sends the matrix $A$ and $R$ to Bob and outputs $S_A$. 
\item Bob calculates $\bof{y}^{\prime}$ with minimal $d_{\mathrm{H}}(\bof{y},\bof{y}^{\prime})$ such that $R=A_r\odot \bof{y}^{\prime}$ and outputs $S_B=A_s\odot \bof{y}^{\prime}$.  
\end{enumerate}
\end{protocol}

\begin{theorem}\label{th:nskeyworks}
Protocol~\ref{prot:keyagreement} is $\epsilon$-correct, $\epsilon^{\prime}$-secret with $\epsilon,\epsilon^{\prime}\in O(2^{-n})$ for $s=q\cdot n$ and $r>n\cdot h(\delta)$ and where $q$ is such that
$2^{q-r/n-1}(1+\ep)<1$. Additionally, Protocol~\ref{prot:keyagreement} is $\epsilon^{\prime\prime}$-robust on  $\mathcal{P^{-\eta}}$ with $\epsilon^{\prime\prime}\in O(2^{-n})$. 
\end{theorem}
\begin{proof}
This follows directly from Lemmas~\ref{lemma:filter},~\ref{lemma:ircorrect} and~\ref{lemma:distance_key_string}. Note that in order to do information reconciliation, a key of length $s+r$ has to be created. Robustness follows from Lemma~\ref{lemma:perobust}.
\end{proof}

The secret key rate is the length of the key a secure protocol can output, divided by the number of systems used, in the asymptotic limit of a large number of systems. 
\begin{lemma}\label{lemma:keyrate}
Protocol~\ref{prot:keyagreement} reaches a key rate $q$ of
\begin{align}
\nonumber q &=1- h(\delta)-\log_2 ({1+\ep})\ .
\end{align}
\end{lemma}

\section{The Protocol}\label{sec:protocol}

In this section, we analyse a protocol with an implementation similar to the one given in~\cite{ekert}. We compute its key rate in the presence of a non-signalling adversary. The protocol can be \emph{implemented} using quantum mechanics, the security relies, however, only on the non-signalling condition. A slightly different protocol reaching a positive key rate in the quantum regime is given in~\cite{eurocrypt}.

\begin{figure}[h!]
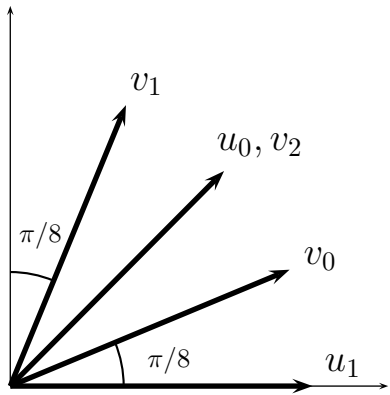

\centering
\pspicture*[](-0.5,-0.5)(5.5,5.5)
\psarc(0,0){1.5}{67.5}{90}
\rput[B]{0}(0.4,1.9){\normalsize{$\pi/8$}}
\psarc(0,0){1.5}{0}{22.5}
\rput[B]{0}(2.1,0.25){\normalsize{$\pi/8$}}
\rput[br]{90}(0,0){\psline[linewidth=0.5pt]{->}(0,0)(5,0)}
\rput[br]{0}(0,0){\psline[linewidth=0.5pt]{->}(0,0)(5,0)}
\rput[B]{0}(1.8,3.9){\Large{$v_1$}}
\rput[br]{67.5}(0,0){\psline[linewidth=2pt]{->}(0,0)(4,0)}
\rput[br]{22.5}(0,0){\psline[linewidth=2pt]{->}(0,0)(4,0)}
\rput[B]{0}(4.1,1.6){\Large{$v_0$}}
\rput[br]{45}(0,0){\psline[linewidth=2pt]{->}(0,0)(4,0)}
\rput[B]{0}(3.3,3.1){\Large{$u_0,v_2$}}
\rput[br]{0}(0,0){\psline[linewidth=2pt]{->}(0,0)(4,0)}
\rput[B]{0}(4.4,0.2){\Large{$u_1$}}
\endpspicture
\caption{Alice's and Bob's measurement bases in terms of polarization.}
\label{fig:basenekert}
\end{figure}

\begin{protocol}\label{prot}\ 
\begin{enumerate}
 \item Alice creates $n$ singlet states $\ket{\Psi^-}=(\ket{01}-\ket{10})/\sqrt{2}$, and sends one qubit of every state to Bob.
 \item Alice and Bob randomly measure the $i$\textsuperscript{th} system in either the basis $u_0$ or $u_1$ (for Alice) or $v_0$, $v_1$ or $v_2$ (Bob); the five bases are shown
 in Figure~\ref{fig:basenekert}. Bob inverts his measurement result. 
They make sure that no signal can travel between the subsystems. 
 \item The measurement results from the cases where both measured $u_0=v_2$ form the raw key.
\item For the remaining measurements, they announce the results over the public authenticated channel and estimate the parameters $\ep$ and $\delta$ (see Section~\ref{subsec:parameter_estimation}). 
If the parameters are such that key agreement is possible, they continue; otherwise they abort.
\item Information reconciliation and privacy amplification: Alice \linebreak[4] ran\-dom\-ly chooses an $(m+s)\times n$-matrix $A$ such that $p(0)=p(1)=1/2$ for all entries and $m:=\lceil n\cdot h(\delta)\rceil$. She calculates
 $A\odot \bof{x}$ (where $\bof{x}$ is Alice's raw key) and sends the first $m$ bits to Bob over the public authenticated 
 channel. The remaining bits form the key. 
\end{enumerate}
\end{protocol}

Assume that Alice and Bob execute the above protocol using a noisy quantum channel. More precisely, their final state is 
a mixture of a singlet with weight $1-\rho$ and a fully mixed state with weight $\rho$. The key rate as function of the parameter $\rho$ is given in Figure~\ref{fig:nskeyrate}. 
\begin{figure}[ht!]
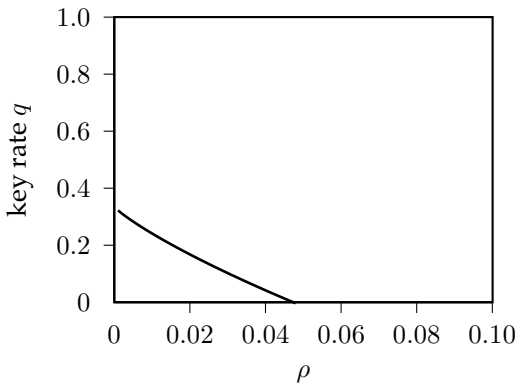

\centering
 \pspicture[](-2,-1)(7,4)
  \psset{xunit=50cm,yunit=3.75cm}
  \rput[c](0.05,-0.25){{$\rho$}}
  \rput[c]{90}(-0.025,0.5){key rate {$q$}}
     \psaxes[Dx=0.02,Dy=0.2,  showorigin=true,tickstyle=bottom,axesstyle=frame](0,0)(0.1001,1.0001)
\psplot[linewidth=1pt, linestyle=solid]{0.001}{0.048}{
1 x sub 1 x sub exp x x exp mul 3 2 sqrt sub 2 sqrt x mul add div ln 2 ln div  1 add
}   
\endpspicture
\caption{\label{fig:nskeyrate} The key rate of Protocol~\ref{prot} secure against a non-signalling adversary in terms of the channel noise.}
\end{figure}

\section{Concluding Remarks}

We have shown that privacy amplification of non-signalling secrecy is possible, if a non-signalling condition holds between all subsystems. It follows from the results in Chapter~\ref{ch:impossibiltiy} that \emph{some} kind of additional requirement is, in general, necessary. The question remains open whether it could be partially relaxed, for example such that signalling is only allowed in one direction (as it would be the case when the systems are measured one after the other). 

Another challenge is to find different non-local correlations, inequivalent to the CHSH inequality or Braunstein-Caves inequality, which imply partial secrecy in this setup and can be used as building block for a key-distribution scheme.

\chapter{Device-Independent Security Against Quantum Adversaries}\label{ch:quantum}
\markboth{Security Against Quantum Adversaries}{}

\section{Introduction}

The key-distribution scheme studied in Chapter~\ref{ch:nsadversaries} is secure against all non-signalling adversaries. Since it is not possible to signal by measuring different parts of an entangled quantum state, this holds, in particular, for an adversary limited by quantum physics. However, a non-signalling adversary is, in general, much stronger than a quantum adversary. For example, she can even have significant knowledge about a system that violates the CHSH inequality (Section~\ref{subsec:localsystem}) by its maximum quantum value. As discussed in Section~\ref{sec:approaches}, a quantum system reaching this value must be (equivalent to) a singlet state. A quantum adversary could, therefore, not have \emph{any} knowledge about the measurement outcome. For a key-agreement scheme, this means that tolerating a non-signalling adversary leads to an unnecessarily low key rate, or even the impossibility to agree on a key in a range allowed in the presence of quantum adversaries. 

In this chapter, we consider key agreement secure against quantum adversaries. It is already known that classical post-processing, in particular privacy amplification~\cite{rennerkoenig}, works even if the adversary holds quantum information. The problem is to estimate the \emph{entropy}, i.e., the uncertainty, the adversary has about the raw key.

\subparagraph*{Chapter outline}
We study the possible attacks by a quantum adversary and explain our setup in Section~\ref{sec:qmodel}. We then study the security of a single quantum system and show how the probability that an eavesdropper can guess the measurement result (this quantity is equivalent to the min-entropy) can be expressed as the solution of a semi-definite program. We first give a version which depends on the exact state and measurements of the honest parties (Section~\ref{subsec:qguess}) and then modify it to a device-independent version in Section~\ref{subsec:observableprob}. We also give a slightly different form which can be used to calculate the security of a bit (Sections~\ref{subsec:qbit} and~\ref{subsec:qbitobservable}). We then turn to several systems and show how the conditions they need to fulfil can be expressed in terms of the conditions of the individual systems (Section~\ref{subsec:qseveral}) if measurements on different subsystems commute. This leads directly to a product theorem for the guessing probability (Section~\ref{subsec:productguess}) (i.e., additivity of the min-entropy) and an XOR-Lemma for partially secure bits against quantum adversaries (Section~\ref{subsec:qxor}). This insight can be used to construct a key-distribution scheme. We first assume that the honest parties' systems behave independently (Section~\ref{sec:qkd}) and then remove this requirement in Section~\ref{sec:notindependent}. Finally, we give an explicit protocol in Section~\ref{sec:qprotocol}.

\subparagraph*{Related work}
The question of device-independent quantum key distribution has been raised, and security in a noiseless scenario been shown by Mayers and Yao in~\cite{my}. In~\cite{mmmo}, this result has been extended to allow for noise. In~\cite{abgs}, a protocol secure against collective attacks has been given. Under a plausible, but unproven conjecture, it 
remains secure against coherent attacks if the devices are memoryless~\cite{mckaguephd}. All these results use the fact that for binary outcomes, the effective dimension of the Hilbert space can be reduced. 

The question of security against quantum adversaries is related to the question which correlations can be obtained from measurements on a quantum system~\cite{tsirelson,wehnertsi,masanestsi}. In fact, our approach bases on such a criterion given in~\cite{npa07,dltw08,npa}.

\subparagraph*{Contributions}
The main technical contribution of this chapter is Lem\-ma~\ref{lemma:qproductconditions}, which shows that the conditions several quantum systems must fulfil can be expressed in terms of the conditions on the individual subsystems. The resulting product lemma for the guessing probability  of a~quantum adversary is Theorem~\ref{th:guessprod}, and the XOR-Lemma for quantum secrecy is given in Theorem~\ref{th:xorquant}.

\section{Modelling Quantum Adversaries}\label{sec:qmodel}

\subsection{Possible attacks}

Consider the scenario where Alice, Bob, and Eve share a tripartite quantum state. They can each measure their part of the system and obtain a~measurement outcome. We can, of course, also consider the state Alice and Bob share after Eve's part has been traced out, and this is also a quantum state. In accordance with the non-signalling principle, the marginal state Alice and Bob share is independent of what Eve does with her part of the state (in particular, from her measurement). We can even consider the state Alice and Bob share conditioned an a certain measurement outcome of Eve and this is, of course, still a quantum state. Finally, in case Alice and Bob share \emph{several} systems (living in a tensor product Hilbert space and such that measurements are preformed on the individual subspaces), then even conditioned on the measurement outcomes of one system, the remaining systems are still quantum systems. 

We will consider the case where Alice and Bob share $n$ bipartite quantum systems and ask the question whether they can agree on a secret key unknown to Eve by interacting with them. 
We make the following requirement. 
\begin{condition}
The system $P_{\bof{XY}Z|\bof{UV}W}$ must be a $(2n+1)$-party quantum system. 
\end{condition}

In quantum cryptography, when Alice and Bob share a certain quantum state described by a density operator $\rho_{AB}$, it is usually assumed that Eve controls the whole environment, i.e., the total quantum state between Alice, Bob, and Eve is pure. Any measurement on the purifying system corresponds to a partition of the form $\rho_{AB}=\sum_z p^z \rho_{AB}^z$, where $\rho_{AB}^z$ is the state conditioned on the measurement outcome $z$.
Considering the resulting \emph{systems}, each of these $\rho_{AB}^z$ gives rise to a quantum system when measured, i.e., any measurement Eve does on her part of the quantum state induces a `convex decomposition' of the quantum system Alice and Bob share into several quantum systems. This limits the possibilities an eavesdropper has to attack the systems. 
\begin{lemma}\label{lemma:qpartition}
Let $P_{\bof{X}Z|\bof{U}W}$ be an $(n+1)$-party quantum system. Then any input~$W$ induces a family of pairs $\{(p^z$,$P^z_{\bof{X}|\bof{U}})\}_z$, where $p^z$ is a weight and $P^z_{\bof{X}|\bof{U}}$ is an $n$-party quantum system, such that
\begin{align}
\label{eq:qpartition} P_{\bof{X}|\bof{U}}&=\sum_z p^z\cdot P^z_{\bof{X}|\bof{U}}\ .
\end{align}
\end{lemma}
\begin{proof}
For any $(n+1)$-party quantum system $P_{\bof{X}Z|\bof{U}W}$, the marginal and conditional systems are $n$-party quantum systems (see Lemma~\ref{lemma:qmarginalconditional}, p.~\pageref{lemma:qmarginalconditional}). Equation (\ref{eq:qpartition}) holds by the definition of the marginal system. 
\end{proof}

\subsection{Security definition}

The system we consider (see Figure~\ref{fig:our_system_quantum}) is the one where 
Alice and Bob share a public authenticated channel plus a quantum state (modelled abstractly as a device taking inputs and giving outputs). 
Alice and Bob apply a protocol $(\pi,\pi^{\prime})$ to the inputs and 
outputs of their systems in order to obtain a key. 
Eve can wire-tap the public channel and choose a measurement on her part of the quantum state. 
It is no advantage for Eve to make several measurements instead of a single one, as the same information can be obtained by making a refined measurement on the initial state. Without loss of generality, we can, therefore, assume that Eve makes a single measurement 
 at the end (after all communication between Alice and Bob is finished).
In our scenario, Eve, therefore, obtains all the 
communication exchanged over the public channel $Q$, can 
then choose a measurement $W$ (which can depend on 
$Q$) and finally obtains an outcome~$Z$.

\begin{figure}[hp]
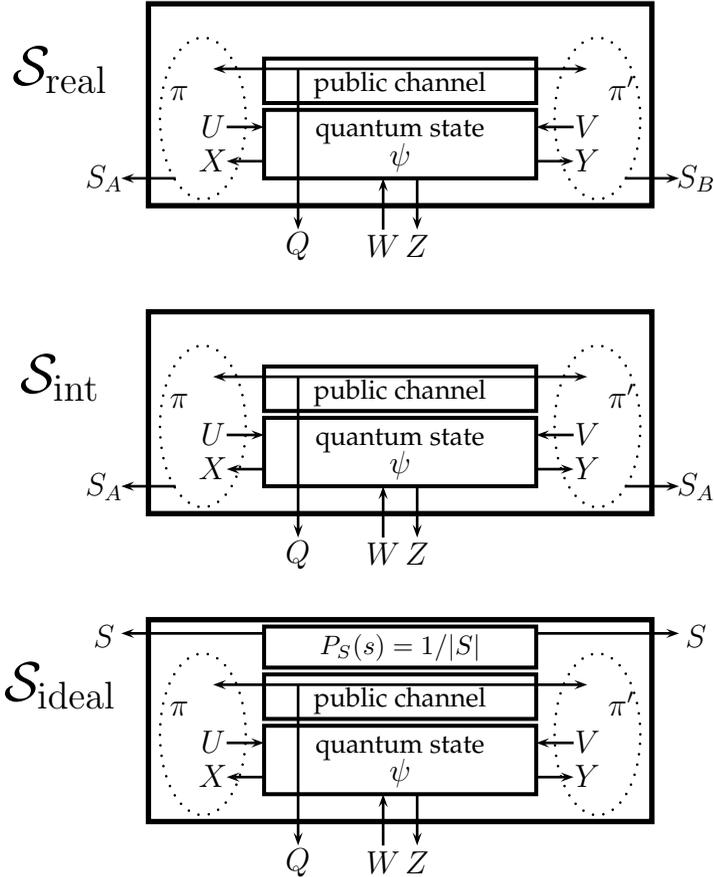

\centering
\pspicture*[](-6,-9)(5,3.5)
\psset{unit=0.9cm}
\rput[c]{0}(0,0){
\rput[b]{0}(-5,2.5){\huge{$\mathcal{S}_{\mathrm{real}}$}}
\psline[linewidth=1pt]{<->}(-2.75,2.85)(2.75,2.85)
\rput[b]{0}(0,2.425){public channel}
\pspolygon[linewidth=1.5pt](-2,3)(2,3)(2,2.35)(-2,2.35)
\pspolygon[linewidth=1.5pt](-2,1.25)(2,1.25)(2,2.25)(-2,2.25)
\psline[linewidth=1pt]{<-}(-2.55,1.5)(-2,1.5)
\rput[b]{0}(-2.75,1.35){\large{$X$}}
\psline[linewidth=1pt]{->}(-2.55,2)(-2,2)
\rput[b]{0}(-2.75,1.85){\large{$U$}}
\psline[linewidth=1pt]{<-}(2.55,1.5)(2,1.5)
\rput[b]{0}(2.75,1.35){\large{$Y$}}
\psline[linewidth=1pt]{->}(2.55,2)(2,2)
\rput[b]{0}(2.75,1.85){\large{$V$}}
\psline[linewidth=1pt]{->}(-0.25,0.5)(-0.25,1.25)
\rput[b]{0}(-0.25,0.1){\large{$W$}}
\psline[linewidth=1pt]{<-}(0.25,0.5)(0.25,1.25)
\rput[b]{0}(0.25,0.1){\large{$Z$}}
\rput[b]{0}(0,1.75){quantum state}
\rput[b]{0}(0,1.35){\large{$\psi$}}
\psline[linewidth=1pt]{<-}(-1.5,0.5)(-1.5,2.85)
\rput[b]{0}(-1.5,0.05){\large{$Q$}}
\psellipse[linewidth=1pt,linestyle=dotted](-2.9,2.125)(0.65,1.2)
\psellipse[linewidth=1pt,linestyle=dotted](2.9,2.125)(0.65,1.2)
\pspolygon[linewidth=2pt](-3.7,0.85)(3.7,0.85)(3.7,3.8)(-3.7,3.8) 
\rput[b]{0}(-3.25,2.4){\large{$\pi$}}
\rput[b]{0}(3.25,2.4){\large{$\pi^{\prime}$}}
\psline[linewidth=1pt]{<-}(-4.1,1.25)(-3.3,1.25)
\rput[b]{0}(-4.35,1.05){\large{$S_A$}}
\psline[linewidth=1pt]{<-}(4.1,1.25)(3.3,1.25)
\rput[b]{0}(4.35,1.05){\large{$S_B$}}
}
\rput[c]{0}(0,-4.5){
\rput[b]{0}(-5,2.5){\huge{$\mathcal{S}_{\mathrm{int}}$}}
\psline[linewidth=1pt]{<->}(-2.75,2.85)(2.75,2.85)
\rput[b]{0}(0,2.425){public channel}
\pspolygon[linewidth=1.5pt](-2,3)(2,3)(2,2.35)(-2,2.35)
\pspolygon[linewidth=1.5pt](-2,1.25)(2,1.25)(2,2.25)(-2,2.25)
\psline[linewidth=1pt]{<-}(-2.55,1.5)(-2,1.5)
\rput[b]{0}(-2.75,1.35){\large{$X$}}
\psline[linewidth=1pt]{->}(-2.55,2)(-2,2)
\rput[b]{0}(-2.75,1.85){\large{$U$}}
\psline[linewidth=1pt]{<-}(2.55,1.5)(2,1.5)
\rput[b]{0}(2.75,1.35){\large{$Y$}}
\psline[linewidth=1pt]{->}(2.55,2)(2,2)
\rput[b]{0}(2.75,1.85){\large{$V$}}
\psline[linewidth=1pt]{->}(-0.25,0.5)(-0.25,1.25)
\rput[b]{0}(-0.25,0.1){\large{$W$}}
\psline[linewidth=1pt]{<-}(0.25,0.5)(0.25,1.25)
\rput[b]{0}(0.25,0.1){\large{$Z$}}
\rput[b]{0}(0,1.75){quantum state}
\rput[b]{0}(0,1.35){\large{$\psi$}}
\psline[linewidth=1pt]{<-}(-1.5,0.5)(-1.5,2.85)
\rput[b]{0}(-1.5,0.05){\large{$Q$}}
\psellipse[linewidth=1pt,linestyle=dotted](-2.9,2.125)(0.65,1.2)
\psellipse[linewidth=1pt,linestyle=dotted](2.9,2.125)(0.65,1.2)
\pspolygon[linewidth=2pt](-3.7,0.85)(3.7,0.85)(3.7,3.8)(-3.7,3.8) 
\rput[b]{0}(-3.25,2.4){\large{$\pi$}}
\rput[b]{0}(3.25,2.4){\large{$\pi^{\prime}$}}
\psline[linewidth=1pt]{<-}(-4.1,1.25)(-3.3,1.25)
\rput[b]{0}(-4.35,1.05){\large{$S_A$}}
\psline[linewidth=1pt]{<-}(4.1,1.25)(3.3,1.25)
\rput[b]{0}(4.35,1.05){\large{$S_A$}}
}
\rput[c]{0}(0,-9){
\rput[b]{0}(-5,2.5){\huge{$\mathcal{S}_{\mathrm{ideal}}$}}
\psline[linewidth=1pt]{<->}(-2.75,2.85)(2.75,2.85)
\rput[b]{0}(0,2.425){public channel}
\pspolygon[linewidth=1.5pt](-2,3)(2,3)(2,2.35)(-2,2.35)
\pspolygon[linewidth=1.5pt](-2,1.25)(2,1.25)(2,2.25)(-2,2.25)
\psline[linewidth=1pt]{<-}(-2.55,1.5)(-2,1.5)
\rput[b]{0}(-2.75,1.35){\large{$X$}}
\psline[linewidth=1pt]{->}(-2.55,2)(-2,2)
\rput[b]{0}(-2.75,1.85){\large{$U$}}
\psline[linewidth=1pt]{<-}(2.55,1.5)(2,1.5)
\rput[b]{0}(2.75,1.35){\large{$Y$}}
\psline[linewidth=1pt]{->}(2.55,2)(2,2)
\rput[b]{0}(2.75,1.85){\large{$V$}}
\psline[linewidth=1pt]{->}(-0.25,0.5)(-0.25,1.25)
\rput[b]{0}(-0.25,0.1){\large{$W$}}
\psline[linewidth=1pt]{<-}(0.25,0.5)(0.25,1.25)
\rput[b]{0}(0.25,0.1){\large{$Z$}}
\rput[b]{0}(0,1.75){quantum state}
\rput[b]{0}(0,1.35){\large{$\psi$}}
\psline[linewidth=1pt]{<-}(-1.5,0.5)(-1.5,2.85)
\rput[b]{0}(-1.5,0.05){\large{$Q$}}
\psellipse[linewidth=1pt,linestyle=dotted](-2.9,2.125)(0.65,1.2)
\psellipse[linewidth=1pt,linestyle=dotted](2.9,2.125)(0.65,1.2)
\pspolygon[linewidth=2pt](-3.7,0.85)(3.7,0.85)(3.7,3.8)(-3.7,3.8) 
\rput[b]{0}(-3.25,2.4){\large{$\pi$}}
\rput[b]{0}(3.25,2.4){\large{$\pi^{\prime}$}}
\pspolygon[linewidth=1.5pt](-2,3.1)(2,3.1)(2,3.7)(-2,3.7)
\rput[b]{0}(0,3.2){$P_S(s)=1/|S|$}
\psline[linewidth=1pt]{<-}(-4.1,3.6)(-2,3.6)
\rput[b]{0}(-4.35,3.4){\large{$S$}}
\psline[linewidth=1pt]{<-}(4.1,3.6)(2,3.6)
\rput[b]{0}(4.35,3.4){\large{$S$}}
}
\endpspicture
\caption{\label{fig:our_system_quantum} Our \emph{real} system (top). Alice and Bob share a public authenticated channel and a quantum state. 
In our \emph{ideal} system (bottom), instead of outputting the key generated by the protocol $(\pi,\pi^{\prime})$, the system outputs a uniform random string $S$ to both Alice and Bob. We also use an \emph{intermediate} system (middle) in our calculations. }
\end{figure}

To show security, we need to bound the distance of this \emph{real} system from an \emph{ideal} system (see Section~\ref{subsec:securitykey}), where Alice and Bob both obtain the same random string uncorrelated with anything else. 
In order to bound the distance between our real system and the ideal system, we introduce an intermediate system $\mathcal{S}_{\mathrm{int}}$, which is equal to our real system, but which outputs $S_A$ on both sides (i.e., $S_B$ is replaced by $S_A$). 

We introduce the distance from uniform of the key from the eavesdropper's point of view.
\begin{definition}
Consider a system $\mathcal{S}_{\mathrm{real}}$ as depicted in Figure~\ref{fig:our_system_quantum}. 
The \emph{distance from uniform of $S_A$ given $Z(W_{\mathrm{q}})$ and $Q$} is 
\begin{multline}
\nonumber
 d(S_A|Z(W_{\mathrm{q}}),Q)=
\frac{1}{2}
\sum_{s_A,q}  \max_{w:{\mathrm{quantum}}} \sum_{z} P_{Z,Q|W=w}(z,q)
\\  
\cdot \left|P_{S_A|Z=z,Q=q,W=w}(s_A)-P_U(s_A)\right|
\ ,
\end{multline}
where the maximization is over all quantum systems $P_{XYZ|UVW}$. 
\end{definition}
The following statement is a direct consequence of the definitions of the systems in Figure~\ref{fig:our_system_quantum} and the distinguishing advantage.
\begin{corollary}
\label{corr:q_dist_advantage} Consider the intermediate system $\mathcal{S}_{\mathrm{int}}$ and the ideal system as depicted in Figure~\ref{fig:our_system_quantum}.  
Then
\begin{align}
\nonumber 
\delta(\mathcal{S}_{\mathrm{int}},\mathcal{S}_{\mathrm{ideal}})&=  d(S_A|Z(W_{\mathrm{q}}),Q)\ .
\end{align}
\end{corollary}
This quantity will be the one that is relevant for the \emph{secrecy} of the protocol. 

Furthermore, the \emph{correctness} of the protocol, i.e., the probability that \linebreak[4] Alice's and Bob's key are equal, is determined by the distinguishing advantage from the intermediate system to the real system, more precisely, the probability that the real system outputs different values on the two sides. This is again a direct consequence of the definitions. 
\begin{corollary}
Consider the intermediate system $\mathcal{S}_{\mathrm{int}}$ and the real system  $\mathcal{S}_{\mathrm{real}}$ as defined above. 
Then
\begin{align}
\nonumber 
\delta(\mathcal{S}_{\mathrm{real}},\mathcal{S}_{\mathrm{int}}) &=
\sum_{s_A\neq s_B} P_{S_AS_B}(s_A,s_B)\ .
\end{align}
\end{corollary}

Finally, by the triangle inequality for the distinguishing advantage of systems (see Lemma~\ref{lemma:distance}, p.~\pageref{lemma:distance}), we obtain the following statement relating the security of our protocol to the secrecy and correctness.
\begin{lemma}
\begin{align}
\nonumber 
\delta(\mathcal{S}_{\mathrm{real}},\mathcal{S}_{\mathrm{ideal}}) &\leq 
\delta(\mathcal{S}_{\mathrm{real}},\mathcal{S}_{\mathrm{int}})
+
\delta(\mathcal{S}_{\mathrm{int}},\mathcal{S}_{\mathrm{ideal}})\ .
\end{align}
\end{lemma}
Since a system with $\delta(\mathcal{S}_{\mathrm{real}},\mathcal{S}_{\mathrm{ideal}})\leq \epsilon$ is $\epsilon$-secure, we will be interested in bounding this quantity.

\section{Security of a Single System}\label{sec:qsingle}

\subsection{A bound on the guessing probability}\label{subsec:qguess}

It will be our goal to show the security of a key-distribution protocol of the form as given in Figure~\ref{fig:our_system_quantum}. The crucial part hereby is to bound the \emph{min-entropy} an adversary 
has about the (raw) key. However, the min-entropy is equivalent to the probability that an eavesdropper interacting with her part of the quantum state can correctly guess the value of Alice's raw key $\bof{X}$ (see 
Theorem~\ref{th:krs}, p.~\pageref{th:krs}). Once this probability is bounded, a secure key can be obtained using standard techniques, such as information reconciliation and privacy amplification, which are already known to work in the quantum case~\cite{rennerkoenig}, \cite{rennerphd}.  

We will, in the following, study the scenario where Eve can choose an input $W$,  depending on some additional information $Q$, and then obtains an output $Z$ (depending on $W$). She then has to try to guess a value $f(\bof{X})$ of range $\mathcal{F}$. In the context of key distribution,  $f$ will be the identity function on the outputs on Alice's side.  

\begin{definition}
Consider a system $\mathcal{S}_{\mathrm{real}}$ as depicted in Figure~\ref{fig:our_system_quantum}. 
The \emph{guessing probability of $f(\bof{X})$ given $Z(W_{\mathrm{q}})$ and $Q$} is 
\begin{multline}
\nonumber P_{\mathrm{guess}}(f(\bof{X})|Z(W_{\mathrm{q}}),Q)= \sum_q \max_{w:{\mathrm{quantum}}} \sum_z P_{ZQ|W=w}(z,q)
\\ 
\cdot
 \max_{f(x)} P_{f(X)|Z=z,Q=q,W=w}(f(x))\ ,
\end{multline}
where the maximization is over all quantum systems $P_{XYZ|UVW}$. 
The \emph{min-entropy of $f(\bof{X})$ given $Z(W_{\mathrm{q}})$ and $Q$} is 
\begin{align}
\nonumber \mathrm{H}_{\mathrm{min}}(f(\bof{X})|Z(W_{\mathrm{q}}),Q) &= -\log_2 P_{\mathrm{guess}}(f(\bof{X})|Z(W_{\mathrm{q}}),Q)\ .
\end{align}
\end{definition}
Theorem~\ref{th:krs}, p.~\pageref{th:krs} justifies this definition of the min-entropy.

Lemma~\ref{lemma:qpartition} gives a bound on the probability that a quantum adversary can guess Alice's outcome by the following maximization problem. (We assume that the inputs $\bof{u}$ are public, i.e., $Q=(\bof{U}=\bof{u},F=f)$). 
\begin{lemma}\label{lemma:qattackopt}
The value of $P_{\mathrm{guess}}(f(\bof{X})|Z(W_{\mathrm{q}}),Q)$, where $P_{\bof{X}Z|\bof{U}W}$ is an $(n+1)$-party quantum system and $Q=(\bof{U}=\bof{u})$, is upper-bounded by the optimal value of the following optimization problem

\begin{align}
\nonumber \max :&\quad \sum_{z=1}^{|\mathcal{F}|} p^z \sum_{\bof{x}:f(\bof{x})=z} P^z_{\bof{X}|\bof{U}}(\bof{x},\bof{u})\\
\nonumber \st &\quad P_{\bof{X}|\bof{U}}= \sum_{z=1}^{|\mathcal{F}|}  p^z\cdot P^z_{\bof{X}|\bof{U}} \\
\nonumber &\quad P^z_{\bof{X}|\bof{U}}\ n\text{-party quantum system, for all }z\ .
\end{align}
\end{lemma}
\begin{proof}
The first condition follows by the definition of the marginal system and the second by the fact that for any $(n+1)$-party quantum system the conditional systems are $n$-party quantum systems (see Lemma~\ref{lemma:qmarginalconditional}, p.~\pageref{lemma:qmarginalconditional}). The objective function is the definition of the guessing probability. 
It is sufficient to consider the case $|\mathcal{Z}|=|\mathcal{F}|$ because any system where $Z$ has larger range can be made into a system reaching the same guessing probability by combining the system where the same value $f(\bof{X})$ has maximal probability. By the convexity of quantum systems, this remains a quantum system. 
\end{proof}

In~\cite{npa07}, a criterion in terms of a semi-definite program is given, which any quantum system must fulfil. The idea is that if a system is quantum, then it is possible to associate a matrix $\Gamma$ with it which needs to be positive semi-definite. 
\begin{definition}
 A \emph{sequence of length $k$} of a set of operators $\{E_{u_i}^{x_i}:x_i\in \mathcal{X}_i, u_i\in \mathcal{U}_i,i\in 1,\dotsc,n\}$ is a product of $k$ operators of this set.
  The sequence of length $0$ is defined as the identity operator. 
\end{definition} 
\begin{definition}
The matrix \emph{$\Gamma$} is defined as 
\begin{align}
\nonumber \Gamma_{ij}&:= \brakket{\Psi}{O_i^\dagger O_j}{\Psi}\ ,
\end{align}
where $O_i=E_{u_m}^{x_m} E_{u^{\prime}_n}^{x^{\prime}_n}\dotsm $ is a sequence of the measurement operators  $\{E_{u_i}^{x_i}\}$. The matrix \emph{$\Gamma^k$} is defined in the same way as $\Gamma$, but restricting the operators to sequences of length at most $k$. 
\end{definition}
In the above notation we consider the measurement operators as operators on the whole Hilbert space $\mathcal{H}$. These operators must fulfil the conditions of 
Definition~\ref{def:qbehavior}, p.~\pageref{def:qbehavior}, (i.e., they must be Hermitian orthogonal projectors and sum up to the identity for each input). If we additionally require them to commute, this is equivalent to a tensor-product structure by Theorem~\ref{th:commutetensor}, p.~\pageref{th:commutetensor}, if we consider only finite dimensional Hilbert spaces. Note that the requirements the measurement operators fulfil translate to requirements on the entries of the matrix $\Gamma$. For example, certain entries must be equal to others or the sum of some must be equal to the sum of others. 

In order to decide whether a certain system is quantum, we can ask the question whether such a matrix $\Gamma$ exists; because if it is, it must be possible to associate a matrix with it, which is consistent with the probabilities describing the system and fulfil the above requirements. The problem of finding a consistent matrix $\Gamma$ is a semi-definite programming problem. 
\begin{theorem}[Navascu\'es, Pironio, Ac\'in~\cite{npa07}]
For every quantum system $P_{\bof{X}|\bof{U}}$ there exists a symmetric matrix $\Gamma^k$ with $\Gamma^k_{ij}=\smallbrakket{\Psi}{O_i^\dagger O_j}{\Psi}$ and where $O_i=E_{u_m}^{x_m} E_{u^{\prime}_n}^{x^{\prime}_n}\dotsm $ is a sequence of length $k$ of the operators $\{E_{u_i}^{x_i}\}$. 
Furthermore, 
\begin{align}
\nonumber A_{\mathrm{qb}}\cdot \Gamma^k&= 0\ ,\ \text{and}\\
\nonumber \Gamma^k &\succeq  0\ ,
\end{align}
where $A_{\mathrm{qb}}$ corresponds to the conditions 
\begin{itemize}
 \item orthogonal projectors: $\smallbrakket{\Psi}{O E_{u_i}^{x_i} E_{u_i}^{x^{\prime}_i}O^{\prime}}{\Psi}-\smallbrakket{\Psi}{OE_{u_i}^{x_i}\delta_{x_ix^{\prime}_i} O^{\prime}}{\Psi}=0$~,  
 \item completeness: $\sum_{x_i}\smallbrakket{\Psi}{OE_{u_i}^{x_i} O^{\prime}}{\Psi}-\smallbrakket{\Psi}{O O^{\prime}}{\Psi}=0$ for all ${u}_i$~,
 \item commutativity: $\smallbrakket{\Psi}{O {E_{u_i}^{x_i}} {E_{u_j}^{x_j}} O^{\prime}}{\Psi}=\smallbrakket{\Psi}{O{E_{u_j}^{x_j}} {E_{u_i}^{x_i}}O^{\prime}}{\Psi}$ for $i\neq j$~,
\end{itemize}
where  $O$ and $O^{\prime}$ stand for arbitrary sequences from the set $\{E_{u_i}^{x_i}\}$.
\end{theorem}
\begin{proof}
Orthogonality, completeness, and Hermiticity follow directly from Definition~\ref{def:qbehavior}, p.~\pageref{def:qbehavior}. Let us see that the matrix is positive semi-definite. 
 For all $v\in \mathbb{C}^m$
\begin{align}
 \nonumber v^T \Gamma^k v &=\sum_{ij}v_i^T \Gamma^k_{ij} v_i=\sum_{ij}v_i^* \brakket{\Psi}{O_i^\dagger O_j}{\Psi} v_j =\brakket{\Psi}{V^\dagger V}{\Psi}\geq 0
\end{align}
where $V:=\sum_iv_iO_i$. Finally, the matrix can be taken to be real, because for any complex $\Gamma^k$, the matrix $(\Gamma^k+{\Gamma^k}^*)/2$ is real 
and fulfils the conditions. 
\end{proof}
We do not require this matrix to be normalized. Note that the matrix $\Gamma^k$ contains, in particular, the (potentially not normalized) probabilities $P_{\bof{X}|\bof{U}}(\bof{x},\bof{u})$ associated with an $n$-party quantum system, for $n\leq 2k$. 
\begin{definition}[Navascu\'es, Pironio, Ac\'in~\cite{npa07}]
Let $P_{\bof{X}|\bof{U}}$ be an $n$-party system. If there exists a positive semi-definite matrix $\Gamma^k$ such that $A_{\mathrm{qb}}\Gamma^k=0$ and with the entries of $\Gamma^k_{ij}=P_{\bof{X}|\bof{U}}(\bof{x},\bof{u})$ where $\Gamma^k_{ij}$ is the entry associated with $\brakket{\Psi}{\prod_i E_{x_i}^{u_i}}{\Psi} $, then this 
$\Gamma^k$ is called \emph{quantum certificate of order $k$} associated with the system $P_{\bof{X}|\bof{U}}$. 
\end{definition}

In~\cite{dltw08,npa}, it is shown that if a certificate of order $k$ 
can be associated with a certain system $P_{\bof{X}|\bof{U}}$ for all $k\rightarrow \infty$, then this system is indeed quantum. More precisely, it corresponds to a quantum system where operators associated with different parties commute, but do not necessarily have a tensor product structure. For any finite dimensional system, this is, of course, equivalent, as we have seen in Theorem~\ref{th:commutetensor}, p.~\pageref{th:commutetensor}. 

\begin{figure}[hp]
\centering
\pspicture*[](-0.9,3)(13,16)
\psset{unit=0.84cm}
\psline[linewidth=1pt]{-}(0,16)(12,16)
\psline[linewidth=1pt]{-}(0,4)(0,16)
\psline[linewidth=1pt]{-}(0,4)(12,4)
\psline[linewidth=1pt]{-}(12,4)(12,16)
\psline[linewidth=0.5pt,linecolor=gray]{-}(0,15)(12,15)
\psline[linewidth=0.5pt,linecolor=gray]{-}(0,14)(12,14)
\psline[linewidth=0.5pt,linecolor=gray]{-}(0,13)(12,13)
\psline[linewidth=0.5pt,linecolor=gray]{-}(0,12)(12,12)
\psline[linewidth=0.5pt,linecolor=gray]{-}(0,11)(12,11)
\psline[linewidth=0.5pt,linecolor=gray]{-}(0,10)(12,10)
\psline[linewidth=0.5pt,linecolor=gray]{-}(0,9)(12,9)
\psline[linewidth=0.5pt,linecolor=gray]{-}(0,8)(12,8)
\psline[linewidth=0.5pt,linecolor=gray]{-}(0,7)(12,7)
\psline[linewidth=0.5pt,linecolor=gray]{-}(0,6)(12,6)
\psline[linewidth=0.5pt,linecolor=gray]{-}(0,5)(12,5)
\psline[linewidth=0.5pt,linecolor=gray]{-}(0,4)(12,4)
\psline[linewidth=0.5pt,linecolor=gray]{-}(0,3)(12,3)
\psline[linewidth=0.5pt,linecolor=gray]{-}(0,2)(12,2)
\psline[linewidth=0.5pt,linecolor=gray]{-}(0,1)(12,1)
\psline[linewidth=0.5pt,linecolor=gray]{-}(1,4)(1,16)
\psline[linewidth=0.5pt,linecolor=gray]{-}(2,4)(2,16)
\psline[linewidth=0.5pt,linecolor=gray]{-}(3,4)(3,16)
\psline[linewidth=0.5pt,linecolor=gray]{-}(4,4)(4,16)
\psline[linewidth=0.5pt,linecolor=gray]{-}(5,4)(5,16)
\psline[linewidth=0.5pt,linecolor=gray]{-}(6,4)(6,16)
\psline[linewidth=0.5pt,linecolor=gray]{-}(7,4)(7,16)
\psline[linewidth=0.5pt,linecolor=gray]{-}(8,4)(8,16)
\psline[linewidth=0.5pt,linecolor=gray]{-}(9,4)(9,16)
\psline[linewidth=0.5pt,linecolor=gray]{-}(10,4)(10,16)
\psline[linewidth=0.5pt,linecolor=gray]{-}(11,4)(11,16)
\psline[linewidth=0.5pt,linecolor=gray]{-}(12,4)(12,16)
\rput[c]{0}(0.5,16.5){\footnotesize{$\mathds{1}$}}
\rput[c]{0}(1.5,16.5){\footnotesize{$E_u^{x\vphantom{^{\prime}}}$}}
\rput[c]{0}(2.5,16.5){\footnotesize{$E_u^{x^{\prime}}$}}
\rput[c]{0}(3.5,16.5){\footnotesize{$\cdots$}}
\rput[c]{0}(4.5,16.5){\footnotesize{$E_{u^{\prime}}^{x\vphantom{^{\prime}}}$}}
\rput[c]{0}(5.5,16.5){\footnotesize{$\cdots$}}
\rput[c]{0}(6.5,16.5){\footnotesize{$F_v^{y\vphantom{^{\prime}}}$}}
\rput[c]{0}(7.5,16.5){\footnotesize{$F_v^{y^{\prime}}$}}
\rput[c]{0}(8.5,16.5){\footnotesize{$\cdots$}}
\rput[c]{0}(9.5,16.5){\footnotesize{$E_u^{x\vphantom{^{\prime}}}E_u^{x\vphantom{^{\prime}}}$}}
\rput[c]{0}(10.5,16.5){\footnotesize{$E_u^{x\vphantom{^{\prime}}}E_u^{x^{\prime}}$}}
\rput[c]{0}(11.5,16.5){\footnotesize{$\cdots$}}
\rput[c]{0}(-0.5,15.5){\footnotesize{$\mathds{1}$}}
\rput[c]{0}(-0.5,14.5){\footnotesize{$E_u^{x\vphantom{^{\prime}}}$}}
\rput[c]{0}(-0.5,13.5){\footnotesize{$E_u^{x^{\prime}}$}}
\rput[c]{0}(-0.5,12.5){\footnotesize{$\vdots$}}
\rput[c]{0}(-0.5,11.5){\footnotesize{$E_{u^{\prime}}^{x\vphantom{^{\prime}}}$}}
\rput[c]{0}(-0.5,10.5){\footnotesize{$\vdots$}}
\rput[c]{0}(-0.5,9.5){\footnotesize{$F_v^{y\vphantom{^{\prime}}}$}}
\rput[c]{0}(-0.5,8.5){\footnotesize{$F_v^{y^{\prime}}$}}
\rput[c]{0}(-0.5,7.5){\footnotesize{$\vdots$}}
\rput[c]{0}(-0.5,6.5){\footnotesize{$E_u^{x\vphantom{^{\prime}}}E_u^{x\vphantom{^{\prime}}}$}}
\rput[c]{0}(-0.5,5.5){\footnotesize{$E_u^xE_u^{x^{\prime}}$}}
\rput[c]{0}(-0.5,4.5){\footnotesize{$\vdots$}}
\pspolygon[linewidth=0.5pt,fillstyle=vlines,hatchcolor=lightgray](0.1,14.1)(0.1,14.9)(0.9,14.9)(0.9,14.1)
\pspolygon[linewidth=0.5pt,fillstyle=vlines,hatchcolor=lightgray](1.1,14.1)(1.1,14.9)(3.9,14.9)(3.9,14.1)
\pspolygon[linewidth=0.5pt,fillstyle=vlines,hatchcolor=lightgray](6.1,14.1)(6.1,14.9)(8.9,14.9)(8.9,14.1)
\pspolygon[linewidth=0.5pt,fillstyle=vlines,hatchcolor=lightgray](0.1,6.1)(0.1,6.9)(0.9,6.9)(0.9,6.1)
\rput[c]{0}(1.5,13.5){\footnotesize{$0$}}
\rput[c]{0}(1.5,12.5){\footnotesize{$0$}}
\rput[c]{0}(2.5,12.5){\footnotesize{$0$}}
\rput[c]{0}(2.5,14.5){\footnotesize{$0$}}
\rput[c]{0}(3.5,14.5){\footnotesize{$0$}}
\rput[c]{0}(3.5,13.5){\footnotesize{$0$}}
\rput[c]{0}(0.5,5.5){\footnotesize{$0$}}
\rput[c]{0}(6.5,8.5){\footnotesize{$0$}}
\rput[c]{0}(7.5,9.5){\footnotesize{$0$}}
\pspolygon[linewidth=0.5pt,fillstyle=hlines,hatchcolor=lightgray](0.1,8.1)(0.1,8.9)(0.9,8.9)(0.9,8.1)
\pspolygon[linewidth=0.5pt,fillstyle=hlines,hatchcolor=lightgray](7.1,8.1)(7.1,8.9)(7.9,8.9)(7.9,8.1)
\pspolygon[linewidth=0.5pt,fillstyle=crosshatch,hatchcolor=lightgray](0.1,9.1)(0.1,9.9)(0.9,9.9)(0.9,9.1)
\pspolygon[linewidth=0.5pt,fillstyle=crosshatch,hatchcolor=lightgray](6.1,9.1)(6.1,9.9)(6.9,9.9)(6.9,9.1)
\pspolygon[linewidth=0.5pt,fillstyle=crosshatch,hatchcolor=lightgray](9.1,9.1)(9.1,9.9)(11.9,9.9)(11.9,9.1)
\endpspicture
\caption{\label{fig:sdpcriterion} 
The matrix corresponding to the second order criteria of~\cite{npa07} associated with a bipartite system. We denote the operators associated with the first party by $E$ and with the second by $F$. If the system is quantum, the entry of the row associated with operator $A$ and column associated with operator $B$ corresponds to 
$\brakketsmall{\Psi}{A^\dagger B}{\Psi}$, 
and the resulting matrix is positive semi-definite. The constraints are such that certain entries of the matrix are $0$, or that the sum of certain entries are equal to the sum of other entries (for example, entries in areas hatched the same way are equal).
}
\end{figure}

The above criterion allows to replace the condition that $P_{\bof{XY}|\bof{UV}}^z$ is a quantum system by the condition that a certain matrix is positive semi-definite and allows us to bound Eve's guessing probability by a semi-definite program. 
\begin{lemma}\label{lemma:guessissdp}
The maximum guessing probability of $f(\bof{X})$, given $Z(W_{\mathrm{q}})$ and $Q:=(\bof{U}=\bof{u},F=f)$, is upper-bounded by\footnote{In the following, we sometimes write matrices as vectors by writing the columns `on top of each other'. When we write that a vector needs to be positive semi-definite, we mean that the matrix obtained by the inverse of this transformation must be positive semi-definite.}
\begin{align}
\nonumber P_{\mathrm{guess}}(f(\bof{X})|Z(W_{\mathrm{q}}),Q) &\leq   \sum_{z=1}^{|\mathcal{F}|} b_z^T\cdot \Gamma^{z}\ ,
\end{align}
where  $\sum_{z=1}^{|\mathcal{F}|} b_z^T\cdot \Gamma^{z}$ is the optimal value of the optimization problem 
\begin{align}
\max : &\quad \sum_{z=1}^{|\mathcal{F}|} \sum_{\bof{x}:f(\bof{x})=z} \Gamma^{z}(\bof{x},\bof{u})\\
\nonumber \st &\quad A_{\mathrm{qb}}\cdot \Gamma^z=0\ \text{ for all } z 
\\
\nonumber &\quad \Gamma^z\succeq 0\\
\nonumber &\quad \sum_z \Gamma^z = \Gamma^k_{\mathrm{marg}}
\end{align}
where $\Gamma^{z}(\bof{x},\bof{u})$ denotes the entry of the matrix $\Gamma^z$ corresponding to \linebreak[4]  $\smallbrakket{\Psi}{\prod_i E_{u_i}^{x_i} E_w^z}{\Psi}$, i.e., it 
contains in particular the probabilities $P^z_{\bof{X}|\bof{U}}(\bof{x},\bof{u})$; $b_z$ is a matrix of the same size as $\Gamma^z$ and it has a $1$ at the positions where $\Gamma$ has the entry $\smallbrakket{\psi}{O_i^\dagger O_i}{\psi}$, where $O_i=\prod_m E_{u_m}^{x_m}$ 
 such that $f(\bof{x})=z$.  The matrix $\Gamma^k_{\mathrm{marg}}$ denotes the certificate of order $k$ associated with the marginal system $P_{\bof{X}|\bof{U}}$. 
\end{lemma}
\begin{proof}
This follows from Lemma~\ref{lemma:qattackopt}, the fact that any quantum system $P_{\bof{X}|\bof{U}}^z$ has a quantum certificate of order $k$ and $\sum_z E_w^z=\mathds{1}$. 
\end{proof}

The primal and dual program can be expressed as:
\begin{align}
\nonumber \mathrm{PRIMAL}\\
\label{eq:primalguess}\max :&\quad \sum_{z=1}^{|\mathcal{F}|}  b_z^T\cdot \Gamma_z\\
\nonumber\\
\nonumber\st &\quad 
\underbrace{
\left(
\begin{array}{ccc}
A_{\mathrm{qb}} & \cdots  & 0\\
& \ddots & \\
0 & \cdots & A_{\mathrm{qb}}\\
\mathds{1} & \cdots &\mathds{1}
\end{array}
\right)
}_{A}
\cdot 
\left( 
\begin{array}{c}
 \Gamma_1\\
\vdots \\
\Gamma_{|\mathcal{F}|}
\end{array}
\right)
=
\underbrace{
\left(
\begin{array}{c}
0\\
\vdots \\
0 \\
\Gamma_{\mathrm{marg}}^k
\end{array}
\right)
}_{c} \\
\nonumber \\
\nonumber &\quad \Gamma_i \succeq 0\ \text{ for all } i
\end{align}
\begin{align}
\nonumber \mathrm{DUAL}\\
\label{eq:dualguess} \min : &\quad {\Gamma^k_{\mathrm{marg}}}^T\cdot  \lambda_{{|\mathcal{F}|}+1}\\
\nonumber \\
 \st &\quad 
 \underbrace{
\left(
\begin{array}{cccc}
A_{\mathrm{qb}}^T & \cdots  & 0 & \mathds{1} \\
& \ddots &&\\
0 & \cdots & A_{\mathrm{qb}}^T & \mathds{1}\\
\end{array}
\right)
}_{A^T}
\cdot 
\left(
\begin{array}{c}
\lambda_1 \\
\vdots \\
\lambda_{|\mathcal{F}|}\\
\lambda_{{|\mathcal{F}|}+1}
\end{array}
\right)
\succeq 
\underbrace{
\left(
\begin{array}{c}
b_1 \\
\vdots \\
b_{|\mathcal{F}|}
\end{array}
\right)
}_{b}
\nonumber \\
\nonumber \\
\nonumber &\quad \lambda_i\ \text{unrestricted}
\end{align}
Note that any dual feasible solution gives an upper bound on the guessing probability (linear) in terms of the matrix associated with the marginal system of Alice and Bob, $\Gamma^k_{\mathrm{marg}}$. Furthermore, the dual feasible region is  \emph{independent} of Alice's and Bob's marginal system, it only depends on the number of inputs and outputs and the step in the semi-definite hierarchy considered. 

However, the matrix $\Gamma^k_{\mathrm{marg}}$ contains entries which do not correspond to observable probabilities and are only known if the state and measurement operators are known (i.e., in a not device-independent scenario). It will be the goal of the next section to express the guessing probability in terms of observable quantities.

\subsection[\ldots\ \  in terms of observable probabilities]{Guessing probability in terms of observable probabilities}\label{subsec:observableprob}

Certain entries of the matrix $\Gamma^k_{\mathrm{marg}}$ do not correspond to observable probabilities, and it is, therefore, impossible to know their value by testing the system. In this section, we will modify the above optimization problem in such a way as to get a solution only in terms of observable probabilities. More precisely, we will modify the optimization problem to take the `worst' possible quantum certificate  consistent with observed probabilities. This leads to the following, modified, semi-definite program. The matrix $A_{IJ}$ is defined such that, multiplied with a quantum certificate, the observable probabilities are obtained, i.e., $A_{IJ} \Gamma^k=P_{\bof{X}|\bof{U}}$ (where $P_{\bof{X}|\bof{U}}$ denotes the vector containing the values $P_{\bof{X}|\bof{U}}(\bof{x},\bof{u})$ for all $\bof{x},\bof{u}$).

\begin{align}
\nonumber \mathrm{PRIMAL}\\
\label{eq:primalguessprob}
\max : &\quad \sum_{z=1}^{|\mathcal{F}|} b_z^T\cdot \Gamma_z\\
\nonumber \\
\nonumber \st &\quad \left(
\begin{array}{cccc}
A_{\mathrm{qb}} & \cdots  & 0 & \phantom{-}0\\
 & \ddots & & \phantom{-}0\\
0 & \cdots & A_{\mathrm{qb}} & \phantom{-}0\\
\mathds{1} & \cdots & \mathds{1} & -\mathds{1}\\
0 & \cdots &0 & A_{IJ}\\
\end{array}
\right)
\cdot 
\left( 
\begin{array}{c}
 \Gamma_1\\
\vdots \\
\Gamma_{|\mathcal{F}|}\\
\Gamma^k_{\mathrm{marg}}
\end{array}
\right)
=
\left(
\begin{array}{c}
0\\
\vdots \\
0\\
0 \\
P_{\bof{X}|\bof{U}}
\end{array}
\right)\\
\nonumber \\
\nonumber &\quad \Gamma_i \succeq 0,\ \Gamma^k_{\mathrm{marg}}\ \text{unrestricted}
\end{align} 
\begin{align}
\nonumber \mathrm{DUAL}\\
\label{eq:dualguessprob}
\min : &\quad P_{\bof{X}|\bof{U}}^T \cdot \lambda_{|\mathcal{F}|+2}\\
\nonumber \\
 \st &\quad 
\left(
\begin{array}{ccccc}
A_{\mathrm{qb}}^T & \cdots  & 0 & \phantom{-}\mathds{1} & 0 \\
& \ddots & 0 &\\
0 & \cdots & A_{\mathrm{qb}}^T & \phantom{-}\mathds{1} & 0\\
0 & \cdots & 0 & -\mathds{1} & A_{IJ}^T\\
\end{array}
\right)
\cdot 
\left(
\begin{array}{c}
\lambda_1 \\
\vdots \\
\lambda_{|\mathcal{F}|}\\
\lambda_{|\mathcal{F}|+1}\\
\lambda_{|\mathcal{F}|+2}
\end{array}
\right)
\begin{array}{c}
\\
 \succeq \\
\\
=
\end{array}
\left(
\begin{array}{c}
b_1 \\
\vdots \\
b_{|\mathcal{F}|} \\
0
\end{array}
\right)
\nonumber  \\
\nonumber  &\quad \lambda_i\ \text{unrestricted}
\end{align}
Note that we have changed $\Gamma^k_{\mathrm{marg}}$ to be a variable (instead of a constant). Obviously $\Gamma^k_{\mathrm{marg}}\succeq 0$ holds because it is the sum of positive semi-definite matrices. However, it is easier when we do not make this restriction explicit in the program.

\begin{lemma}\label{lemma:qguessdualprod}
Let $\lambda_1,\dotsc,\lambda_{|\mathcal{F}|+2}$ be dual feasible for (\ref{eq:dualguessprob}). Then  
$\lambda_1,\dotsc,\lambda_{|\mathcal{F}|+1}$ are dual feasible for (\ref{eq:dualguess}) reaching the same objective value. 
\end{lemma}
\begin{proof}
We use the fact that $A_{IJ}\Gamma^k_{\mathrm{marg}}=P_{\bof{X}|\bof{U}}$. 
Since $\lambda_1,\dotsc,\lambda_{|\mathcal{F}|+2}$ are dual feasible for (\ref{eq:dualguessprob}), it holds that $A_{IJ}^T\lambda_{|\mathcal{F}|+2}=\lambda_{|\mathcal{F}|+1}$. Therefore, 
\begin{align}
\nonumber {\Gamma^k_{\mathrm{marg}}}^T\cdot \lambda_{|\mathcal{F}|+1}&= {\Gamma^k_{\mathrm{marg}}}^T\cdot A_{IJ}^T\cdot \lambda_{|\mathcal{F}|+2} = P_{\bof{X}|\bof{U}}^T\cdot \lambda_{|\mathcal{F}|+2} \ . \qedhere
\end{align}
\end{proof}
Lemma~\ref{lemma:qguessdualprod} implies that any dual feasible solution of (\ref{eq:dualguessprob}) gives an upper bound on the guessing probability linear in terms of the observable probabilities. 

Furthermore, in terms of the min-entropy, it means that Eve's min-en\-tro\-py about Alice's value $f(\bof{X})$,  is 
at least  $\mathrm{H}_{\mathrm{min}}(f(\bof{X})|Z(W_{\mathrm{q}}),Q)\geq \linebreak[4] -\log_2 P_{\bof{X}|\bof{U}}^T \lambda_{|\mathcal{F}|+2}$ for 
any dual feasible $\lambda$. 

\begin{example}\label{ex:qguesssingle}
Consider a bipartite quantum system with binary inputs and outputs given by the mixture of the system in Figure~\ref{fig:qsystem}, p.~\pageref{fig:qsystem}, with weight $1-\rho$ and a perfectly random bit with weight $\rho$ (this system can be achieved by measuring a mixture of a singlet and a fully mixed state using the measurements given in Example~\ref{ex:qsystem}, p.~\pageref{ex:qsystem}). The guessing probability of the output bit $X$ as function of the parameter $\rho$ is given in Figure~\ref{fig:pguess}.\footnote{The data plotted in Figure~\ref{fig:pguess} has been obtained by solving (\ref{eq:primalguessprob}) numerically, using the programs MATLAB\textsuperscript{\textregistered}, Yalmip and Sedumi~\cite{matlab,sedumi,yalmip}.}
 
\begin{figure}[h]
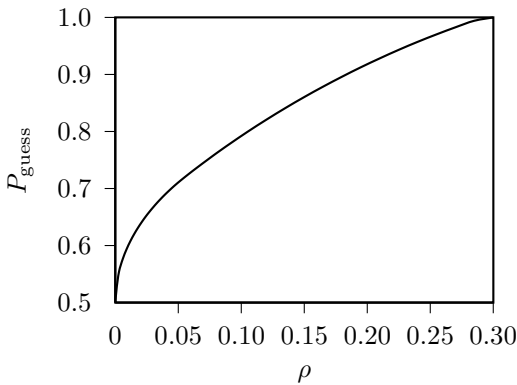

\centering
 \pspicture[](-2,2.75)(7,7.75)
 \psset{xunit=16.66666666cm,yunit=7.5cm}
 \savedata{\mydata}[
 {
{0.0000,	0.5000},
{0.0030,	0.5546},
{0.0060,	0.5769},
{0.0090,	0.5938},
{0.0120,	0.6079},
{0.0150,	0.6202},
{0.0180,	0.6312},
{0.0210,	0.6412},
{0.0240,	0.6505},
{0.0270,	0.6590},
{0.0300,	0.6670},
{0.0330,	0.6746},
{0.0360,	0.6817},
{0.0390,	0.6885},
{0.0420,	0.6949},
{0.0450,	0.7011},
{0.0480,	0.7070},
{0.0510,	0.7126},
{0.0540,	0.7180},
{0.0570,	0.7233},
{0.0600,	0.7285},
{0.0630,	0.7336},
{0.066,	0.7386},
{0.069,	0.7436},
{0.072,	0.7485},
{0.075,	0.7534},
{0.078,	0.7583},
{0.081,	0.7631},
{0.084,	0.7678},
{0.087,	0.7725},
{0.09,	0.7771},
{0.105,	0.7996},
{0.12,	0.8209},
{0.135,	0.8412},
{0.15,	0.8605},
{0.1650,	0.8788},
{0.1800,	0.8962},
{0.1950,	0.9128},
{0.2100,	0.9284},
{0.2250,	0.9433},
{0.2400,	0.9573},
{0.2550,	0.9704},
{0.2700,	0.9828},
{0.2850,	0.9943},
{0.3000,	1.0000}
 }
 ]
 \rput[c](0.15,0.375){{$\rho$}}
  \rput[c]{90}(-0.075,0.75){{$P_{\mathrm{guess}}$}}
   \psaxes[Dx=0.05,Dy=0.1, Oy=0.5,  showorigin=true,tickstyle=bottom,axesstyle=frame](0,0.5)(0.3001,1.0001)
\dataplot[plotstyle=curve,showpoints=false,dotstyle=o]{\mydata}  
 \endpspicture
 \caption{\label{fig:pguess} The bound on the guessing probability of the measurement outcomes of Example~\ref{ex:qguesssingle}.}
 \end{figure}
\end{example}

\subsection{Best attack on a bit}\label{subsec:qbit}

The above analysis can also be used to find the best attack in case the function $f$ maps $\bof{X}$ to a bit. However, in this case, we can give a slightly different form to calculate the distance from uniform of a bit. This will allow us to show an XOR-Lemma for quantum secrecy in Section~\ref{subsec:qxor}. 
\begin{lemma}\label{lemma:distanceissdp}
The distance from uniform of 
$B=f(\bof{X})\in \{0,1\}$ given $Z(W_{\mathrm{q}})$ and $Q:=(\bof{U}=\bof{u},F=f)$ is upper-bounded by
\begin{align}
\nonumber d(B|Z(W_{\mathrm{q}}),Q) &= \frac{1}{2}\cdot b^T\cdot \Gamma_\Delta^*\ ,
\end{align}
where $b^T\Gamma_\Delta^*$ is the optimal value of the optimization problem
\begin{align}
\label{eq:qprimal2} \max :&\quad \sum_{\bof{x}:B=0} \Gamma_\Delta(\bof{x},\bof{u})-\sum_{\bof{x}:B=1} \Gamma_\Delta(\bof{x},\bof{u})\\
\nonumber \st &\quad A_{\mathrm{qb}}\Gamma_\Delta =0 
\\
\nonumber &\quad \Gamma_\Delta \preceq \Gamma^k_{\mathrm{marg}} \\
\nonumber &\quad \Gamma_\Delta\succeq-\Gamma^k_{\mathrm{marg}}\ ,
\end{align}
where $\Gamma^k_{\mathrm{marg}}$ is the matrix associated with the marginal system  
$P_{\bof{X}|\bof{U}}$. 
\end{lemma}
\begin{proof}
Define 
\begin{align}
\nonumber \Gamma_\Delta &= 2 p\cdot  \Gamma^{z_0}-\Gamma_{\mathrm{marg}}\ ,
\end{align}
and note that with this definition $\Gamma^{z_0}=({\Gamma_{\mathrm{marg}}+\Gamma_\Delta})/({2p})$
and $\Gamma^{z_1}=({\Gamma_{\mathrm{marg}}-\Gamma_\Delta})/({2(1-p)})$. \\
The distance from uniform of a bit can be expressed as 
\begin{align}
\nonumber d(B|Z(W_{\mathrm{q}}),Q)&=  \frac{1}{2}\cdot \Biggl[p\cdot \Bigl(\sum_{\bof{x}:B=0}\Gamma^{z_0}(\bof{x},\bof{u})-\sum_{\bof{x}:B=1}\Gamma^{z_0}(\bof{x},\bof{u})\Bigr)
\\
\nonumber
&\quad 
+(1-p)\cdot \Bigl(\sum_{\bof{x}:B=1}\Gamma^{z_1}(\bof{x},\bof{u})-\sum_{\bof{x}:B=0}\Gamma^{z_1}(\bof{x},\bof{u})\Bigr)  
\Biggr]\\
\nonumber &= \frac{1}{2}\cdot b^T\cdot \Gamma_\Delta^*\ ,
\end{align}
Now notice that $\Gamma^{z_0}$ and $\Gamma^{z_1}$ are actually quantum certificates of order $k$ if $\Gamma_\Delta$ fulfils the above requirements. The conditions the matrix $\Gamma$ needs to fulfil are all linear and, therefore, because $\Gamma^k_{\mathrm{marg}}$ fulfils them, $\Gamma^{z_0}$ and $\Gamma^{z_1}$ fulfil them exactly if $\Gamma_\Delta$  does. 
The semi-definite constraints correspond exactly to the requirement that $\Gamma^{z_0}$ and $\Gamma^{z_0}$ are positive semi-definite, using the fact that the space of positive semi-definite matrices forms a convex cone. 
\end{proof}
The above semi-definite program can be expressed in the following form:
\begin{align}
\nonumber \mathrm{PRIMAL}\\
\label{eq:primalbit} \max : &\quad b^T\cdot \Gamma_\Delta\\
\nonumber \st &\quad
\underbrace{
 \left(
\begin{array}{c}
\phantom{-}\mathds{1} \\
-\mathds{1} \\
A_{\mathrm{qb}}
\end{array}
\right)
}_{A}
\cdot 
\Gamma_\Delta 
\begin{array}{c}
\preceq\\
\preceq\\
=
\end{array}
\underbrace{
\begin{array}{c}
\Gamma^k_{\mathrm{marg}}\\
\Gamma^k_{\mathrm{marg}}\\
0
\end{array}
}_{c}
\end{align}
\begin{align}
\nonumber  \mathrm{DUAL}\\
\label{eq:dualbit} \min :&\quad (\Gamma^k_{\mathrm{marg}})^T (\lambda_1+\lambda_2)\\
\nonumber  \st &\quad 
\underbrace{
\left(
\begin{array}{ccc}
\mathds{1} & -\mathds{1}& A_{\mathrm{qb}}^T 
\end{array}
\right)
}_{A^T}
\cdot 
\left(
\begin{array}{c}
\lambda_1 \\
\lambda_2\\
\lambda_3
\end{array}
\right)
=
b
\\
\nonumber  &\quad \lambda_1,\lambda_2\succeq 0,\ \lambda_3 \text{ unrestricted}
\end{align}

\subsection[\ldots\ \  in terms of observable probabilities]{Best attack on a bit in terms of observable probabilities}\label{subsec:qbitobservable}

Any dual solution of (\ref{eq:dualbit}) leads to a bound on the distance from uniform of the bit $B$ in terms of the matrix elements $\Gamma^k_{\mathrm{marg}}$. We will now change our primal program to one where we \emph{optimize} over all $\Gamma^k_{\mathrm{marg}}$ compatible with the observable probabilities. The dual of this program has a solution only in terms these probabilities. We then show how we can transform any dual feasible solution of this program into a dual feasible solution of the program above reaching the same value.

The new program we consider is the following:
\begin{align}
\nonumber \mathrm{PRIMAL}\\*
\label{eq:qbitprimalobs} \max :&\quad b^T\cdot \Gamma_\Delta\\*
\nonumber \st &\quad \left(
\begin{array}{cc}
\phantom{-}\mathds{1} & -\mathds{1}\\
-\mathds{1} & -\mathds{1}\\
A_{\mathrm{qb}}& 0\\
0 & A_{IJ}\\
0 & A_{\mathrm{qb}}
\end{array}
\right)
\cdot 
\left(
\begin{array}{c}
\Gamma_\Delta \\
\Gamma^k_{\mathrm{marg}}
\end{array}
\right)
\begin{array}{c}
\preceq\\
\preceq\\
=\\
=\\
=
\end{array}
\begin{array}{c}
0\\
0\\
0\\
P_{\bof{X}|\bof{U}}\\
0
\end{array}
\\*
\nonumber  &\quad \Gamma_\Delta, \Gamma^k_{\mathrm{marg}}\text{ unrestricted}
\end{align}
\begin{align}
\nonumber  \mathrm{DUAL}\\
\label{eq:newoptdual} \min :&\quad P_{\bof{X}|\bof{U}}^T \cdot \lambda_4\\
\nonumber  \st &\quad 
\left(
\begin{array}{ccccc}
\phantom{-}\mathds{1} & -\mathds{1} & A_{\mathrm{qb}}^T & 0 & 0\\
-\mathds{1} & -\mathds{1} & 0 & A_{IJ} & A_{\mathrm{qb}}
\end{array}
\right)
\cdot 
\left(
\begin{array}{c}
\lambda_1 \\
\lambda_2\\
\lambda_3\\
\lambda_4\\
\lambda_5
\end{array}
\right)
=
\left(
\begin{array}{c}
b \\
0
\end{array}
\right)
\\
\nonumber  &\quad \lambda_1,\lambda_2\succeq 0,\ \lambda_3, \lambda_4, \lambda_5 \text{ unrestricted}
\end{align}
where the matrix $A_{IJ}$ is such that $A_{IJ} \Gamma^k_{\mathrm{marg}}=P_{\bof{X}|\bof{U}}$. 
We claim that any dual feasible solution of (\ref{eq:newoptdual}) can be transformed into a dual feasible solution of (\ref{eq:dualbit}) reaching the same value. The solution of (\ref{eq:newoptdual}), therefore, gives a bound on the distance from uniform only in terms of the observable probabilities. 
\begin{lemma}
Let $ \lambda_1,\lambda_2,\lambda_3, \lambda_4, \lambda_5$ be a dual feasible solution of (\ref{eq:newoptdual}). Then $ \lambda_1,\lambda_2, \lambda_3$ is a dual feasible solutions of (\ref{eq:dualbit}) reaching the same objective value. 
\end{lemma}
\begin{proof}
The condition that  $\lambda_1,\lambda_2, \lambda_3$ is feasible for (\ref{eq:dualbit}) follows directly from the (upper row) feasibility condition of  (\ref{eq:newoptdual}). To see that it reaches the same value, we use that fact that $\Gamma^k_{\mathrm{marg}}$ is a quantum certificate, i.e., 
\begin{align}
\nonumber A_{\mathrm{qb}}\cdot \Gamma^k_{\mathrm{marg}} &=0
\end{align}
and the (lower row) condition of (\ref{eq:newoptdual}), i.e.,
\begin{align}
\nonumber -\lambda_1-\lambda_2+A_{IJ}^T\cdot\lambda_4+A_{\mathrm{qb}}^T\cdot\lambda_5 &=0\ .
\end{align}
We then obtain
\begin{align}
\nonumber {\Gamma^k}_{\mathrm{marg}}^T\cdot(\lambda_1+\lambda_2)&= {\Gamma^k}_{\mathrm{marg}}^T\cdot(\lambda_1 +\lambda_2)
\\
 \nonumber &\quad
+\Gamma_{\mathrm{marg}}^T\cdot
( -\lambda_1-\lambda_2+A_{IJ}^T\cdot\lambda_4+A_{\mathrm{qb}}^T\cdot\lambda_5)\\
\nonumber &= {\Gamma^k}_{\mathrm{marg}}^T\cdot
(A_{IJ}^T\cdot\lambda_4+A_{\mathrm{qb}}^T\cdot\lambda_5)\\
\nonumber &= (A_{IJ}\cdot {\Gamma^k}_{\mathrm{marg}})^T \cdot \lambda_4\\
\nonumber &= P_{\bof{X}|\bof{U}}^T \cdot \lambda_4\ . \qedhere
\end{align}
\end{proof}

\section{Several Systems}

\subsection{Conditions on several quantum systems}\label{subsec:qseveral}

In this section, we will show our main technical result, namely that the conditions in 
the above semi-definite program behave in a product form if the measurements on different subsystems commute. Roughly, we will show the following: Consider a system $P_{XY|UV}$ associated with a single pair of systems and the matrix $\Gamma^k$ associated with the $k$\textsuperscript{th} step of the hierarchy, 
fulfilling $A_{\mathrm{qb}} \Gamma^k=0$. 
Then, with two pairs of systems, it is possible to associate a matrix ${\Gamma^{\prime}}^k$ living in the tensor product space of two $\Gamma^k$. Furthermore, this matrix must fulfil $(\mathds{1}\otimes A_{\mathrm{qb}}) {\Gamma^{\prime}}^k=0$. 
\begin{definition}
Assume an $(n+m)$-party quantum system. The \emph{reduced quantum certificate of order $k$} is the matrix ${\Gamma^{\prime}}_{n+m}^k$, defined as 
\begin{align}
\nonumber ( {\Gamma^{\prime}}_{n+m}^k)_{ij} &= \brakket{\Psi}{O_{i_1}^\dagger O_{i_2}^\dagger O_{j_2} O_{j_1}}{\Psi}\ ,
\end{align}
where $i=l(i_1-1)+i_2$ and $j=l(j_1-1)+j_2$ and $l$ is the number of rows of a quantum certificate of order $k$ for the $n$-party quantum system. $O_{i_1}$ is the operator associated with the $i$\textsuperscript{th} row of the quantum certificate of order $k$ of the marginal $n$-party system (and similar for $O_{i_2}$ and the $m$-party system). 
\end{definition}
\begin{lemma}
The matrix $ {\Gamma^{\prime}}_{n+m}^k$ is positive semi-definite, i.e., ${\Gamma^{\prime}}_{n+m}^k\succeq 0$\ .
\end{lemma}
\begin{proof}
 This follows directly form the fact that $ {\Gamma^{\prime}}_{n+m}^k$ is a sub-matrix  
 of the $(2k)$\textsuperscript{th} order 
  quantum certificate associated with the $(n+m)$-party quantum system.  
\end{proof}

The main insight, which will lead directly to the product theorems, is the following lemma.

\begin{lemma}\label{lemma:qproductconditions}
Let $P_{\bof{X}_1|\bof{U}_1}$ be an $n$-party and 
$P_{\bof{X}_2|\bof{U}_2}$ an $m$-party quantum system. Call the associated certificates of order $k$ $\Gamma_1^k$ and $\Gamma_2^k$ and write the linear conditions they fulfil 
as $A_{\mathrm{qb},1} \Gamma_1^k=0$ and 
$A_{\mathrm{qb},2} \Gamma_2^k=0$. Then the reduced quantum certificate of order $k$ associated with the $(n+m)$-party quantum system, fulfils
\begin{align}
\nonumber (A_{\mathrm{qb},1}\otimes \mathds{1}_{\Gamma_2^k})\cdot {\Gamma^{\prime}}_{n+m}^k=0\ \ \ &\text{and}\ \ \ 
(\mathds{1}_{\Gamma_1^k} \otimes A_{\mathrm{qb},2} )\cdot {\Gamma^{\prime}}_{n+m}^k=0\ .
\end{align}
\end{lemma}
This can be interpreted in the following way: Even conditioned on any specific outcome (i.e., matrix entry) of the second system, the first system must still be a quantum system. 
\begin{proof}
The matrix $A_{\mathrm{qb},1}$ contains entries of the form 
\begin{align}
\nonumber 
\brakket{\Psi}{O_{i_1}O_{j_1}}{\Psi}-\brakket{\Psi}{O_{i^{\prime}_1}O_{j^{\prime}_1}}{\Psi} &=0
\end{align}
which all operators associated with an $n$-party quantum system must fulfil, because $O_{i_1}O_{j_1}-O_{i^{\prime}_1}O_{j^{\prime}_1}=0$. 
By the definition of ${\Gamma^{\prime}}_{n+m}^k$, the conditions 
$(A_{\mathrm{qb},1}\otimes \mathds{1}_{\Gamma_2^k}) {\Gamma^{\prime}}_{n+m}^k$ correspond to 
\begin{align}
\nonumber  \brakket{\Psi}{O_{i_1} O_{i_2} O_{j_2} O_{j_1}}{\Psi}&- 
\brakket{\Psi}{O_{i^{\prime}_1} O_{i_2} O_{j_2} O_{j^{\prime}_1}}{\Psi}
\\
\nonumber 
&=
\brakket{\Psi}{O_{i_1} O_{j_1} O_{i_2} O_{j_2}}{\Psi}- 
\brakket{\Psi}{O_{i^{\prime}_1}  O_{j^{\prime}_1}O_{i_2} O_{j_2}}{\Psi}\\
\nonumber &=
\brakket{\Psi}{(O_{i_1} O_{j_1} -O_{i^{\prime}_1}  O_{j^{\prime}_1})\, O_{i_2} O_{j_2}}{\Psi}=0\ .
\end{align}
where we have used the fact that operators associated with different par\-ties commute, linearity, and the fact that the operators associated with an $(n+m)$-party quantum system must still fulfil the conditions 
 associated with a single system (as given in 
 Definition~\ref{def:qbehavior}, p.~\pageref{def:qbehavior}). 
\end{proof}

\subsection{A product lemma for the guessing probability}\label{subsec:productguess}

Using the above property, we can show a product lemma for the guessing probability. 
\begin{lemma}\label{lemma:qproduct}
Let $A_1$, $b_1$, and $c_1$ be the parameters associated with the semi-definite program (\ref{eq:primalguess}) bounding the guessing probability of $f(\bof{X}_1)$ of an $n$-party quantum system $P_{\bof{X}_1|\bof{U}_1}$, where $Q_1=(\bof{U}_1=\bof{u}_1,F=f)$.
Similarly, associate 
 $A_2$, $b_2$, and $c_2$ with an $m$-party quantum system $P_{\bof{X}_2|\bof{U}_2}$, where $g(\bof{X}_2)$ and $Q_2=(\bof{U}_2=\bof{u}_2,G=g)$. 
Then the guessing probability of $f(\bof{X}_1)\parallel g(\bof{X}_2)$ (denoting the concatenation) of the $(n+m)$-party system $P_{\bof{X}_1\bof{X}_2|\bof{U}_1\bof{U}_2}$ where $Q=(\bof{U}=\bof{u},F=f,G=g)$ is bounded by 
the semi-definite program defined by $A$, $b$, and $c$, where 
$b=b_1\otimes b_2$, $A=A_1\otimes A_2$.  
\end{lemma}
\begin{proof}
This follows form the fact that any $(n+m)$-party quantum system must fulfil Lemma~\ref{lemma:qproductconditions} and that $b_i\otimes b_j$ has a $1$ exactly at the entry associated with $\smallbrakket{\psi}{O_1^\dagger O_2^\dagger O_2 O_1}{\psi}$, where $O_1$ is the operator associated with the probability of the outcome $\bof{x}_1$ mapped to a certain $f(\bof{x}_1)$, and similarly for $O_2$ and $g(\bof{x}_2)$. 
\end{proof}
Consider now the \emph{dual} of this `tensor product' problem. We will use a product theorem from~\cite{mittalszegedy} (see also~\cite{leemittal}) to show that for any dual feasible $\lambda$ (for a single system), $\lambda\otimes \dotsm \otimes \lambda$ is dual feasible for the dual of the tensor product problem, therefore, forming an upper bound on the guessing probability. 

\begin{theorem}[Mittal, Szegedy~\cite{mittalszegedy}]\label{th:mittalszegedy}
Consider a semi-definite program $\min :\ c_1^T \lambda_1$, $\st\ A_1^T\lambda_1-b_1\succeq 0$  and a feasible $\lambda_1$, and similarly for $A_2$, $b_2$, $c_2$, and $\lambda_2$. 
Assume $b_1\succeq 0$ and $b_2 \succeq 0$.
Then  $\lambda=\lambda_1\otimes \lambda_2$ is feasible for the semi-definite program
$\min:\ (c_1\otimes c_2)^T \lambda$, $\st\ (A_1\otimes A_2)^T\lambda-(b_1\otimes b_2)\succeq 0$ 
\end{theorem}
\begin{proof}
We use the fact that for a $\lambda$ such that $A^T\lambda-b\succeq 0$, where $b\succeq 0$, it holds that $A^T\lambda-b+2b=A^T\lambda+b\succeq 0$, because we consider a convex cone. The tensor product of two positive semi-definite matrices is positive semi-definite. We obtain
\begin{align}
\nonumber (A_1^T\lambda_1-b_1)&\otimes (A_2^T\lambda_2+b_2)
\\
\nonumber &= 
A_1^T\lambda_1\otimes A_2^T\lambda_2 - b_1\otimes A_2^T\lambda_2 
+ A_1^T\lambda_1\otimes b_2 -b_1\otimes b_2\succeq 0\\
\nonumber  (A_1^T\lambda_1+ b_1)&\otimes (A_2^T\lambda_2- b_2)
\\
\nonumber &= 
A_1^T\lambda_1\otimes A_2^T\lambda_2 + b_1\otimes A_2^T\lambda_2 
- A_1^T\lambda_1\otimes b_2 -b_1\otimes b_2\succeq 0\ .
\end{align}
Adding the two inequalities and dividing by two, implies that 
\begin{align}
\nonumber A_1^T\lambda_1\otimes A_2^T\lambda_2-b_1\otimes b_2 &= (A_1^T\otimes A_2^T)(\lambda_1\otimes \lambda_2)-b_1\otimes b_2 \succeq 0\ ,
\end{align}
which means that $\lambda_1\otimes \lambda_2$ is feasible for the product problem. 
\end{proof}

\begin{lemma}\label{lemma:qdualproduct}
Let $\lambda_1$ be a dual feasible solution of (\ref{eq:dualguess}) defined by $A_1$, $b_1$, and $c_1$ (see Lemma~\ref{lemma:qproduct}), and similarly for $\lambda_2$ and $A_2$, $b_2$, and $c_2$. 
Then $\lambda=\lambda_1\otimes \lambda_2$ is dual feasible for the program $A$, $b$, $c$ where $A=A_1\otimes A_2$ and $b=b_1\otimes b_2$. 
\end{lemma}
\begin{proof}
Note that $b_i$ is of the form 
\begin{align}
\nonumber & \left( \begin{array}{ccccc}
  0 & 0 & \cdots & 0\\
0 & 1 &  \\
\vdots & &\ddots & \vdots \\
0 & & \cdots & 0
 \end{array}
\right)\ ,
\end{align}
i.e., it has a $1$ in the place where the matrix $\Gamma$ has the entry $\smallbrakket{\Psi}{{E_u^x}^{\dagger}E_u^x}{\Psi}$ for $f(x)=i$ and $0$ everywhere else. It, therefore, only has positive entries on the diagonal and $0$ everywhere else.  
Clearly, $b_i\succeq 0$. 
The claim then follows by Theorem~\ref{th:mittalszegedy}. 
\end{proof}

We can now formulate the product lemma for the guessing probability. 
\begin{theorem}[Product lemma for the guessing probability]\label{th:guessprod}
Let $P_{\bof{X}_1|\bof{U}_1}$ be an $n$-party quantum system and $f(\bof{X}_1)$ a function $f\colon \mathcal{X}_1\rightarrow \mathcal{F}$ such that $P_{\mathrm{guess}}(f(\bof{X}_1)|Z(W_{\mathrm{q}},Q)\leq P_{\bof{X}_1|\bof{U}_1}^T  \lambda_1$, where $Q=(\bof{U}_1=\bof{u}_1,F=f)$. Similarly, associate the guessing probability $P_{\mathrm{guess}}(g(\bof{X}_2)|Z(W_{\mathrm{q}},Q)\leq P_{\bof{X}_2|\bof{U}_2}^T \lambda_2$ with an $m$-party quantum system $P_{\bof{X}_2|\bof{U}_2}$ where
$Q=(\bof{U}_2=\bof{u}_2,G=g)$. Then the guessing probability of $f(\bof{X}_1)\Vert g(\bof{X}_2)$ obtained from the $(n+m)$-party quantum system $P_{\bof{X}_1\bof{X}_2|\bof{U}_1\bof{U}_2}$ with 
$Q=(\bof{U}=\bof{u},F=f,G=g)$
 is bounded by
\begin{align}
\nonumber P_{\mathrm{guess}}(f(\bof{X}_1)\Vert g(\bof{X}_2)|Z(W_{\mathrm{q}}),Q) &\leq 
P_{\bof{X}_1\bof{X}_2|\bof{U}_1\bof{U}_2}^T \cdot (\lambda_1 \otimes \lambda_2)\ .
\end{align}
\end{theorem}
\begin{proof}
This is a direct consequence of Lemma~\ref{lemma:qdualproduct}. 
\end{proof}
When the marginal system is of the form $P_{\bof{X}_1|\bof{U}_1}\otimes P_{\bof{X}_2|\bof{U}_2}$, this implies that the guessing probability is the product of the guessing probabilities of the two subsystems. Or, in terms 
of the min-entropy, that it is additive. More precisely, the min-entropy of $n$ identical systems 
$\bigotimes_{i=1}^n P_{\bof{X}_i|\bof{U}_i}$ is $n$ times the min-entropy of the individual system.

\subsection{An XOR-Lemma for quantum secrecy}\label{subsec:qxor}

Let us also consider the case where we obtain a partially secure bit from each of the subsystems. We will show that the XOR of the two partially secure bits is highly secure.

\begin{lemma}\label{lemma:qxorlemma}
Let $A_1$, $b_1$, and $c_1$ be the parameters associated with the semi-definite program (\ref{eq:primalbit}) bounding the distance from uniform of a bit $f(\bof{X}_1)\in \{0,1\}$ obtained from an $n$-party quantum system $P_{\bof{X}_1|\bof{U}_1}$ where $Q=(\bof{U}_1=\bof{u}_1,F=f)$. Similarly, 
associate $A_2$, $b_2$, and $c_2$ with the distance from uniform of a bit $g(\bof{X}_2)\in \{0,1\}$ obtained from an $m$-party quantum system $P_{\bof{X}_2|\bof{U}_2}$. 
Then then the distance from uniform of the bit $f(\bof{X}_1)\oplus g(\bof{X}_2)$ obtained from the $(n+m)$-party system $P_{\bof{X}_1\bof{X}_2|\bof{U}_1\bof{U}_2}$, where $Q=(\bof{U}=\bof{u},F=f,G=g)$ is bounded by the semi-definite program  defined by $A$, $b$, and $c$ with $A=A_1\otimes A_2$ and $b=b_1\otimes b_2$.
\end{lemma}
\begin{proof}
This follows form the fact that any $(n+m)$-party quantum system must fulfil Lemma~\ref{lemma:qproductconditions} and $b$ describing the XOR of two bits can be described as the tensor product of the ones associated with each of the two bits. 
\end{proof}

This  implies that for any dual feasible solution, the tensor product is dual feasible for the tensor product problem.

\begin{lemma}\label{lemma:qdualxorproduct}
Let $\lambda_1$ be a dual feasible for (\ref{eq:dualbit}) with $A_1$, $b_1$, and $c_1$ associated with an $n$-party quantum system and $\lambda_2$ dual feasible for an $m$-party quantum system described by $A_2$, $b_2$, and $c_2$. Then $\lambda=\lambda_1\otimes \lambda_2$ is dual feasible for the program $A$, $b$, and $c$ where $A=A_1\otimes A_2$ and $b=b_1\otimes b_2$.
\end{lemma}
\begin{proof}
 $\lambda_1 \otimes \lambda_2$ fulfils the dual constraints because 
\begin{align}
\nonumber [ A_1\otimes A_2](\lambda_1 \otimes \lambda_2) &= b_1\otimes b_2 \ .
\end{align}
Furthermore, the tensor product of two positive semi-definite matrices is again positive semi-definite. 
\end{proof}

We can now formulate the XOR-Lemma for quantum secrecy. 
\begin{theorem}[XOR-Lemma for quantum secrecy]\label{th:xorquant}
Let $P_{\bof{X}_1|\bof{U}_1}$ be an $n$-party quantum system and $f(\bof{X}_1)$ a bit such that $d(f(\bof{X}_1)| Z(W_{\mathrm{q}}),Q)\leq \linebreak[4] P_{\bof{X}_1|\bof{U}_1}^T \lambda_1/2$ with  $Q=(\bof{U}_1=\bof{u}_1,F=f)$. 
 Similarly, associate \linebreak[4] $d(g(\bof{X}_2)| Z(W_{\mathrm{q}}),Q)\leq P_{\bof{X}_2|\bof{U}_2}^T \lambda_2/2$ with a bit from an $m$-party quantum system $P_{\bof{X}_2|\bof{U}_2}$ where
$Q=(\bof{U}_2=\bof{u}_2,G=g)$.
Then the distance from uniform of $f(\bof{X}_1)\oplus g(\bof{X}_2)$ obtained from the $(n+m)$-party quantum system $P_{\bof{X}_1\bof{X}_2|\bof{U}_1\bof{U}_2}$ with 
$Q=(\bof{U}=\bof{u},F=f,G=g)$
 is bounded by
\begin{align}
\nonumber d(f(\bof{X}_1)\oplus g(\bof{X}_2)|Z(W_{\mathrm{q}}),Q) &\leq \frac{1}{2}\cdot
P_{\bof{X}_1\bof{X}_2|\bof{U}_1\bof{U}_2}^T \cdot (\lambda_1 \otimes \lambda_2)\ .
\end{align}
\end{theorem}
\begin{proof}
This follows directly from Lemma~\ref{lemma:qdualxorproduct}. 
\end{proof}

\section{Key Distribution from Product Systems}\label{sec:qkd}

We can now relate the above technical lemmas to the security of quantum key distribution. In a first step, we will show the security of key distribution if the marginal distribution as seen by Alice and Bob is the product of several (identical) independent systems. In the next section, we will remove the condition of independence, since knowing that we are in permutation invariant scenario, we will be able to relate the security of an arbitrary distribution to the security of independent distributions.

In the quantum case, most steps on the way to a secure key are already known. The crucial step is to bound Eve's guessing probability about the raw key, which directly relates to Eve's min-entropy. Once the min-entropy is bounded, Alice and Bob can do information reconciliation and privacy amplification to obtain a secure key. 

The key-distribution protocol proceeds in three steps:
\begin{itemize}
 \item Parameter estimation: Alice and Bob obtain a distribution $P_{XY|UV}^{\otimes n}$. In order to be able to bound Eve's knowledge about the raw key, they need to estimate the probability distribution $P_{XY|UV}$ of the individual systems. 
\item Information reconciliation: Alice sends some information about her raw key to Bob, such that he can correct his errors. 
\item Privacy amplification: Alice and Bob apply a public hash function to their raw keys in order to create a highly secure key. 
\end{itemize}
For a more detailed explanation of these steps, we refer to Chapter~\ref{ch:nsadversaries}.

\subsection{Parameter estimation}\label{subsec:qpe}

A parameter estimation protocol should $\epsilon$-securely filter `bad' input systems and should be $\epsilon^{\prime}$-robust on some `good' input systems (see \linebreak[4] Sec\-tion~\ref{subsec:parameter_estimation}). 

In order to estimate the quality of their systems, Alice and Bob fix as parameters the probabilities $k$ and $p$ and values $P_{\mathrm{guess}}$ and $\delta$. 
\begin{protocol}[Parameter estimation]\label{prot:qpe}\ 
\begin{enumerate}
\item Alice and Bob receive a system $P_{\bof{XY}|\bof{UV}}=P_{XY|UV}^{\otimes n}$. 

\item Alice chooses $\bof{U}$ such that for each $i$ with probability $1-k$, it holds that $U_i=u_k$, where 
where $u_k$ is the input from which a raw key bit can be generated, and with probability $k$ she chooses one of the  $|\mathcal{U}|$ inputs uniformly at random.  
\item Bob chooses $\bof{V}$ such that $V_i=v_k$ with probability $1-k$ and 
with probability $k$, $V_i$ is chosen uniformly at random. 
\item They input $\bof{u}$ and $\bof{v}$ into the system and obtain 
the outputs $\bof{x}$ and $\bof{y}$. 
\item They exchange the inputs over the public authenticated channel. 
\item If less than $(1-k)^2 p  n$ inputs were $({U_i},{V_i})=({u}_k,{v}_k)$, they abort. 
\item Let $t$ be the number of inputs where both did not chose ${u}_k$ nor ${v}_k$. 
If any combination $(u,v)$ occurred less than $k^2  p  n/|\mathcal{U}| |\mathcal{V}|$ times
they abort. 
\item From the inputs where they both chose a uniform input they estimate the distribution by $P^{\mathrm{est}}_{XYUV}(x,y,u,v)=|\{i|(x_i,y_i,u_i,v_i)=(x,y,u,v)\}|/t$. 
Define ${\mathcal{P}}$ as the set of all $P_{XYUV}$ such that \linebreak[4] $|\mathcal{U}| |\mathcal{V}| P_{XYUV}^T  \lambda \leq P_{\mathrm{guess}}$ for some dual feasible $\lambda$ (see (\ref{eq:dualguess})) and $P(X\neq Y|U=u_k,V=v_k)\leq \delta$. 
If $d(P_{XYUV}^{\mathrm{est}},P_{XYUV}^{\mathcal{P}})>\eta$ they abort, otherwise, they accept. 
\end{enumerate}
\end{protocol}

\begin{definition}
Let $\mathcal{P}$ be a set of distributions $P_{XYUV}$. The set of systems  \emph{$\mathcal{P}^{\eta}$} are all distributions which have distance at least $\eta$ with the set $\mathcal{P}$, i.e., 
\begin{align}
\nonumber \mathcal{P}^{\eta} &= \lbrace
 P_{XYUV}| d(P_{XYUV},P_{XYUV}^{\mathcal{P}}) > \eta \ \text{for all}\ P_{XYUV}^{\mathcal{P}}\in 
 \mathcal{P}
  \rbrace
\end{align}
\end{definition}

\begin{definition}
Let $\mathcal{P}$ be a set of distributions $P_{XYUV}$. The set of systems  \emph{$\mathcal{P}^{-\eta}$} are all distributions which have distance at least $\eta$ with the complement of the set $\mathcal{P}$, i.e., 
\begin{align}
\nonumber \mathcal{P}^{-\eta} &= \lbrace
 P_{XYUV}| d(P_{XYUV},P_{XYUV}^{\bar{\mathcal{P}}}) > \eta \ \text{for all}\ P_{XYUV}^{\bar{\mathcal{P}}}\notin 
 \mathcal{P}
  \rbrace \ .
\end{align}
\end{definition}

We further define the set of \emph{conditional systems} which are $\eta$-far or $\eta$-close to a certain set by the closeness of the distributions which can be obtained from them by choosing the input distribution to be uniform. 
\begin{definition}
Let $\mathcal{P}_{\mathrm{cond}}$ be a set of systems $P_{XY|UV}^{\mathcal{P}}$. For any system $P_{XY|UV}$, consider the distribution $P_{XYUV}=P_{XY|UV} /{|\mathcal{U}| |\mathcal{V}|}$. Then a system $P_{XY|UV}$ is in \emph{$\mathcal{P}^{\eta}_{\mathrm{cond}}$} if $P_{XYUV}\in  \mathcal{P}^{\eta}$
and $P_{XY|UV}$ is in \emph{$\mathcal{P}^{-\eta}_{\mathrm{cond}}$} if $P_{XYUV}\in  \mathcal{P}^{-\eta}$.
\end{definition}

The reason to take exactly this definition of $\mathcal{P}^{\eta}_{\mathrm{cond}}$ is that  it is useful to estimate $P_{XY|UV}^T\lambda$, where $P_{XY|UV}^T$ is the vector of all probabilities in the conditional distribution and $\lambda$ is some vector. This is in fact exactly the form of the bound on the guessing probability. 
\begin{lemma} \label{lemma:etaenvir}
Let $\mathcal{P}=P_{XY|UV}$.  
For all $P^{\bar{\eta}}_{XY|UV}\notin \mathcal{P}^{\eta}_{\mathrm{cond}} $, it holds that
\begin{align}
\nonumber {P^{\bar{\eta}}_{XY|UV}}^T \cdot \lambda - P_{XY|UV}^T \cdot \lambda  &\leq  {P_{XY|UV}}^T \cdot  \lambda +|\mathcal{U}| |\mathcal{V}|\cdot  \eta \cdot \Bigl(\sum_i|\lambda_i| \Bigr). 
\end{align}
\end{lemma}
\begin{proof}
\begin{align}
\nonumber \Bigl({P^{\bar{\eta}}_{XY|UV}}^T  - P_{XY|UV}^T \Bigr)\cdot  \lambda &= |\mathcal{U}| |\mathcal{V}|\cdot  {P^{\bar{\eta}}_{XYUV}}^T\lambda-|\mathcal{U}| |\mathcal{V}|\cdot  {P_{XYUV}}^T\cdot  \lambda  \\
\nonumber 
&= |\mathcal{U}|  |\mathcal{V}|\cdot ({P^{\eta}_{XYUV}}^T-{P_{XYUV}}^T)\cdot \lambda\\
\nonumber & \leq  | \mathcal{U}|  |\mathcal{V}|\cdot  \eta \cdot \Bigl(\sum_i|\lambda_i| \Bigr)\ . \qedhere
\end{align}
\end{proof}

We will need the Sampling Lemma (Lemma~\ref{lemma:sampling}, p.~\pageref{lemma:sampling}) to show that our protocol is secure, i.e., it $\epsilon$-securely filters input states with $\tilde{P}_{\mathrm{guess}}\geq \linebreak[4] P_{\mathrm{guess}}+|\mathcal{U}| |\mathcal{V}| \eta \sum_i |\lambda_i|$ for the individual systems. 
\begin{lemma}\label{lemma:qfilters}
Protocol~\ref{prot:qpe} $\epsilon$-securely filters $\left(\mathcal{P}^{+\eta}_{\mathrm{cond}}\right)^{\otimes n}$ with 
\begin{align}
\nonumber \epsilon &=
 |\mathcal{X}| |\mathcal{Y}| |\mathcal{U}| |\mathcal{V}|\cdot e^{-\left(\frac{t^{\prime} \eta^2}{8 |\mathcal{X}|  |\mathcal{Y}|}\right)}\ , 
 \end{align} 
 where $t^{\prime}=k^2pn/|\mathcal{U}| |\mathcal{V}|$.
\end{lemma}
\begin{proof}
If for each of the conditional distributions $P_{XY|U=u,V=v}$ the estimate is within $\eta$, this also holds for the total distribution $P_{XYUV}$. By Lem\-ma~\ref{lemma:sampling}, p.~\pageref{lemma:sampling}, the probability that for any conditional distribution the estimate is $\eta$-far is at most $|\mathcal{X}| |\mathcal{Y}|e^{-t^{\prime} \eta^2/8|\mathcal{X}||\mathcal{Y}|}$, where $t^{\prime}=k^2pn/|\mathcal{U}| |\mathcal{V}|$. We obtain the statement by the union bound over all inputs. 
\end{proof}
Note that $\epsilon\in O(2^{-n})$ for any constant $0< k,p<1$ and $\eta>0$.

\begin{lemma}\label{lemma:peqrobust}
Protocol~\ref{prot:pe} is $\epsilon^{\prime}$-robust on $\left( \mathcal{P^{-\eta}}\right)^{\otimes n}$  with 
\begin{align}
\nonumber \epsilon^{\prime} &=
 |\mathcal{X}| |\mathcal{Y}||\mathcal{U}| |\mathcal{V}|\cdot e^{-\left(\frac{t^{\prime} \eta^2}{8 |\mathcal{X}|  |\mathcal{Y}|}\right)}
 \\ \nonumber &\quad
+e^{-2n \left((1-p)(1-k)^2 \right)^2}
+|\mathcal{U}||\mathcal{V}|\cdot e^{-2n \left(\frac{(1-p)k^2}{ |\mathcal{U}||\mathcal{V}|} \right)^2} \ , 
 \end{align} 
where $t^{\prime}=k^2pn/|\mathcal{U}| |\mathcal{V}|$.
\end{lemma}
\begin{proof}
This follows by the same argument as Lemma~\ref{lemma:qfilters} and a Chernoff bound (see Lemma~\ref{lemma:chernoff}, p.~\pageref{lemma:chernoff}) on the probability that the protocol aborts because any of the inputs did not occur sufficiently often. 
\end{proof}
It holds that $\epsilon^{\prime}\in O(2^{-n})$ for any constant $0< k,p<1$ and $\eta>0$.

\begin{lemma}\label{lemma:qfilters2}
The protocol $\epsilon$-securely filters systems with $\tilde{P}_{\mathrm{guess}}\geq P_{\mathrm{guess}} + \eta^{\prime}$ for the individual system, where $\eta^{\prime}=|\mathcal{U}||\mathcal{V}| \eta   \sum_i |\lambda_i|$.
\end{lemma}
\begin{proof}
This is a direct consequence of Lemma~\ref{lemma:qfilters} and Lemma~\ref{lemma:etaenvir}  and the fact that the guessing probability is given by $P_{XY|UV}^T \lambda$, see (\ref{eq:dualguess}) . 
\end{proof}
Lemma~\ref{lemma:qfilters2} also implies that the protocol filters systems with small min-entropy, i.e.,  $\tilde{\mathrm{H}}_{\mathrm{min}}(X|Z(W_{\mathrm{q}}))\leq -\log_2 \tilde{P}_{\mathrm{guess}}$.

\begin{lemma}
The protocol $\epsilon$-securely filters systems with $\tilde{\delta}\geq \delta + \eta^{\prime}$ for the individual systems, where $\eta^{\prime}=|\mathcal{U}||\mathcal{V}| \eta   \sum_i |\lambda_i|$.
\end{lemma}
\begin{proof}
This follows from the definition of $\mathcal{P}^{+\eta}_{\mathrm{cond}}$. 
\end{proof}

\subsection{Information reconciliation}\label{subsec:qir}

Having estimated the probability of error $\delta$ of their key bits in the previous section, Alice and Bob can do information reconciliation by applying a two-universal hash function\footnote{Information reconciliation using a two-universal hash function has the disadvantage that the decoding procedure (i.e., for Bob to find $y^{\prime}$) cannot be done in a computationally efficient way, in general. It is possible to use a \emph{code} for information reconciliation instead, and there exist codes which can be efficiently decoded~\cite{holensteinphd}. However, in our setup the \emph{theoretical} efficiency of the decoding procedure is actually not important, since there exist codes with very good decoding properties in \emph{practice} and Alice and Bob can test whether they have correctly decoded using a short hash value of their strings. In case decoding does not succeed, they can repeat the protocol, resulting in some loss of robustness.} with output length $m$ bits, where $m=n\cdot h(\delta)+\kappa^{\prime}$ and they can almost surely correct their errors, i.e., the keys will be equal except with exponentially small probability.

\begin{protocol}[Information reconciliation]\label{prot:qir}\ 
\begin{enumerate}
\item Alice obtains $\bof{x}$ and Bob $\bof{y}$ distributed according to $P_{XY}^{\otimes n}$ with $\mathcal{X}=\mathcal{Y}=\{0,1\}$. Alice outputs $\bof{x}^{\prime}=\bof{x}$. 
\item Alice chooses a function $f\in \mathcal{F}\colon \{0,1\}^n\rightarrow \{0,1\}^m$ at random, where $\mathcal{F}$ is a two-universal set of functions. 
\item She sends the function $f$ and the value $f(\bof{x})$ to Bob.
\item Bob chooses $\bof{y}^{\prime}$ such that $d_{\mathrm{H}}(\bof{y},\bof{y}^{\prime})$ is minimal among all strings $\bof{z}$ with $f(\bof{z})=f(\bof{x})$ (if there are two possibilities, he chooses one at random) and outputs $\bof{y}^{\prime}$. 
\end{enumerate}
\end{protocol}

The following theorem by Brassard and Salvail states that information reconciliation can be achieved this way. We state the theorem with a slightly stronger bound on the error probability than the one originally given in~\cite{brassardsalvail}. 
\begin{theorem}[Information reconciliation~\cite{brassardsalvail}]\label{th:ir}
Let $\bof{x}$ be an $n$-bit string and $\bof{y}$ another $n$-bit string  obtained by sending $\bof{x}$ over a binary symmetric 
 channel with error parameter $\delta$. Assume the 
 function $f\colon \{0,1\}^n\rightarrow \{0,1\}^m$ is chosen at random amongst a two-universal  set of functions.  
 Choose $\bof{y}^{\prime}$ such that $d_{\mathrm{H}}(\bof{y},\bof{y}^{\prime})$ is minimal among all strings $\bof{r}$ with $f(\bof{r})=f(\bof{x})$. 
 Then 
\begin{align} 
\nonumber \Pr[{\bof{x}\neq \bof{y}^{\prime}}] &\leq 
e^{-2\kappa ^2\cdot n}+
2^{n\cdot h(\delta+\kappa)-m} \ ,
\end{align}
 where $h(p)=-p\cdot \log_2 p - (1-p)\log_2 (1-p)$ is the binary entropy function.
\end{theorem}
\begin{proof}
$\bof{x}\neq \bof{y}^{\prime}$ either if $d_{\mathrm{H}}(\bof{x},\bof{y})$ is large or if $f(\bof{x})=f(\bof{y}^{\prime})$. 
The probability that the strings $\bof{x}$ and $\bof{y}$ differ at more than $n(\delta+\kappa)$ positions is bounded by 
\begin{align}
\nonumber \Pr[d_{\mathrm{H}}(\bof{x},\bof{y})]\geq n\cdot (\delta+\kappa)] &\leq e^{-2\kappa^2\cdot n}\ .
\end{align}
The probability that a $\bof{y}^{\prime}\neq \bof{x}$ with small $d_{\mathrm{H}}(\bof{x},\bof{y}^{\prime})$ is mapped to the same value by $f$ is
\begin{align}
\nonumber \Pr[f(\bof{x})=f(\bof{y}^{\prime}),d_{\mathrm{H}}(\bof{x},\bof{y}^{\prime})\leq n(\delta+\kappa) ] &\leq  2^{-m}\cdot \sum_{i=0}^{ n(\delta+\kappa) }\binom{n}{i} \\
\nonumber &\leq  2^{-m}2^{n\cdot h(\delta+\kappa)}\ .
\end{align}
The theorem follows by the union bound. 
\end{proof}

\begin{lemma}\label{lemma:qirworks}
The protocol is $\epsilon$-correct on input $P_{XY}^{\otimes n}$ such that $P(X\neq Y)\leq \delta$ where, for any $\kappa>0$, 
\begin{align}
\nonumber \epsilon &= e^{-2\kappa^2\cdot n}+
2^{n\cdot h(\delta+\kappa)-m} \ ,
\end{align}
and it is $0$-robust on all inputs. 
\end{lemma}
\begin{proof}
Correctness follows directly from Theorem~\ref{th:ir}. Robustness follows from the fact that there always exists a $\bof{y}^{\prime}$ such that $f(\bof{y}^{\prime})=f(\bof{x})$.  
\end{proof}
For any $\kappa>0$ and $m>n\cdot h(\delta+\kappa)$, this value is $\in O(2^{-n})$.

When some information about the raw key is released --- such as, for example, when Alice and Bob do information reconciliation --- the min-entropy can at most be reduced by the number of bits com\-mu\-ni\-ca\-ted, see~\cite{rennerphd}. 
\begin{theorem}[Chain rule~\cite{rennerphd}]\label{th:chainrule}
Let $\rho_{XEC}$ be classical on $C$. Then 
\begin{align}
\nonumber H_{\mathrm{min}}(X|E,C)_{\rho}\geq H_{\mathrm{min}}(X|E)_{\rho}- H_{\mathrm{max}}(C) &\geq H_{\mathrm{min}}(X|E)_{\rho}-m\ ,
\end{align}
where $m=\log_2 |C|$ is the number of bits of $C$. 
\end{theorem}

\subsection{Privacy amplification}\label{subsec:qpa}

It is possible to create a highly secure key from a partially secure string by applying a 
 two-universal hash function. 
The distance from uniform of the final key string is given by the following theorem. 
\begin{theorem}[Privacy amplification~\cite{rennerkoenig,rennerphd}]\label{th:qpa}
Let $\rho_{XE}$ be classical on $\mathcal{H}_X$ and let $\mathcal{F}$ be a family of two-universal hash functions from $|\mathcal{X}|$ to $\{0,1\}^s$. Then
\begin{align}
\nonumber d(\rho_{F(X)EF}|EF) &\leq  \sqrt{\tr\rho_{XE}}\cdot 2^{-\frac{1}{2}(H_{\mathrm{min}}(\rho_{XE}|E)}-s)
 \leq 2^{-\frac{1}{2}(H_{\mathrm{min}}(\rho_{XE}|E)-s)}\ .
\end{align}
\end{theorem}

\subsection{Key distribution on product inputs}\label{subsec:qkd}

We can now put everything together to obtain a key-distribution scheme. As discussed in Section~\ref{subsec:keydistr}, a key-distribution protocol should be \emph{secure}. This means that it should output the same key to Alice and Bob (\emph{correctness}) and that Eve should not know anything about the key (\emph{secrecy}) (the exact definitions are given in Definition~\ref{def:protocolsecure}, p.~\pageref{def:protocolsecure}). Furthermore, the protocol should output a key when the adversary is passive, i.e., it should be \emph{robust}.

\begin{protocol}[Key distribution]\label{prot:qkey}\ 
\begin{enumerate}
\item Alice and Bob receive $P_{XY|UV}^{\otimes n}$
\item They apply parameter estimation using Protocol~\ref{prot:qpe}. 
\item They do information reconciliation using Protocol~\ref{prot:qir}. 
\item Privacy amplification: Alice chooses a function $f\colon \{0,1\}^n\rightarrow \{0,1\}^s \linebreak[4]\in \mathcal{F}$ from a two-universal set and sends $f$ to Bob. Alice outputs $f(\bof{x})$ and Bob $f(\bof{y}^{\prime})$. 
\end{enumerate}
\end{protocol}

\begin{lemma}\label{lemma:qkdsecure}
The protocol is $\epsilon$-secret with $\epsilon\in O(2^{-n})$ and $\epsilon^{\prime}$-correct with $\epsilon^{\prime}\in O(2^{-n})$ for $m>n\cdot h(\delta)$ and $s=q\cdot n< \log_2P_{\mathrm{guess}}-m/n$. It is $\epsilon^{\prime\prime}$-robust on $\left( \mathcal{P}^{-\eta}\right)^{\otimes n}$ with $\epsilon^{\prime\prime}\in O(2^{-n})$. 
\end{lemma}
\begin{proof}
This is a direct consequence of the fact that each step in the protocol is secure (Lemmas~\ref{lemma:qfilters} and~\ref{lemma:qirworks}, and Theorem~\ref{th:qpa}), taking into account Theorem~\ref{th:chainrule}. Robustness follows from the robustness of the parameter-estimation protocol, Lemma~\ref{lemma:peqrobust}. 
\end{proof}

The secret key rate is the length of the key $S$ that the protocol can output and still remain secure. We obtain the following.  
\begin{lemma}
The scheme reaches a key rate $q$ of 
\begin{align}
\nonumber q  &=- \log_2 P_{\mathrm{guess}}-h(\delta)\ .
\end{align}
\end{lemma}

\begin{lemma}
The scheme reaches a positive key rate $q$ whenever 
\begin{align}
\nonumber -\log_2 P_{\mathrm{guess}}-h(\delta) &> 0\ .
\end{align}
\end{lemma}

\section{Removing the Requirement of Independence}\label{sec:notindependent}

We have seen that Alice and Bob can do key agreement (i.e., they either agree on a secret key or abort) if they share i.i.d.\ distributions. We now want to remove the requirement of independence.

\subsection{A special case: the CHSH inequality}\label{sec:qchsh}

First, we consider a special case: the one where Alice and Bob have two inputs and two outputs. In this case, Alice and Bob can apply a (classical) map to their inputs and outputs such that the distribution they share afterwards actually \emph{is} i.i.d., more precisely a convex combination of i.i.d.\ distributions. The systems obtained this way, furthermore still violate the CHSH inequality by the same amount.\footnote{A similar map also exists for the generalization of the CHSH inequality, the Braunstein-Caves inequalities. }

Assume Alice and Bob share an arbitrary distribution 
$P_{\bof{XY}|\bof{UV}}$ where $\bof{X},\bof{Y},\bof{U},\bof{V}$ is an n-bit string. 
They can perform a sequence of local operations and 
public communication in order to obtain a system which 
corresponds to the convex combination of $n$ independent 
unbiased PR~boxes with error $\ep$, i.e., systems, such that
$\Pr[X\oplus Y=u\cdot v]=1-\ep$ for all $u,v$ and where $X$ and $Y$ are random bits (see Figure~\ref{fig:eps-box}, p.~\pageref{fig:eps-box}).  

The local 
operations achieving this, are given in~\cite{mag,masanesv4}. 
We restate them here briefly: For each system $i$, Alice and Bob choose the local map independently in two steps. 
First, with probability $1/2$, they do either of the following:
\begin{enumerate}
 \item nothing
 \item both flip their outcome bits, i.e., $x_i\rightarrow x_i\oplus 1$ and $y_i\rightarrow y_i\oplus 1$\ .
\end{enumerate}
Then, with probability $1/4$ each, they do either of the following:
\begin{enumerate}
 \item nothing
 \item $x_i\rightarrow x_i\oplus u_i$ and $v_i\rightarrow v_i\oplus 1$
 \item $u_i\rightarrow u_i\oplus 1$ and  $y_i\rightarrow y_i\oplus v_i$ 
 \item $u_i\rightarrow u_i\oplus 1$, $x_i\rightarrow x_i\oplus u_i\oplus 1$, $v_i\rightarrow v_i\oplus 1$ and $y_i\rightarrow y_i\oplus v_i$\ .
\end{enumerate}
The choice of local operation needs $3$ random bits per system which have to be communicated from Alice to Bob. Since, each of these operations conserves the `probability of error' $\ep_i$, a system with the same error parameter --- but now an unbiased one with the same error for all inputs --- is obtained. When this transformation is applied to each input/output bit of a distribution $P_{\bof{XY}|\bof{UV}}$ taking $n$ bits input and giving $n$ bits output, a convex combination of products of such systems is obtained. 

When using systems based on the CHSH or Braunstein-Caves inequalities for a key-distribution scheme, we can, therefore, obtain any system as input, apply the above transformation and, hereby, enforce the situation in which we already know that the key-distrubtions scheme is secure (as seen in Section~\ref{sec:qkd}).

\subsection{The general case}

In general, we do not know of a map, such as the one given in Section~\ref{sec:qchsh} to transform arbitrary systems into product systems. Nevertheless, we will be able to relate the security of the key-distribution scheme on \emph{any} input to the security of the scheme on product inputs $P_{\bof{XY}|\bof{UV}}=P_{XY|UV}^{\otimes n}$, for which we have already seen that it is secure, in Section~\ref{sec:qkd}. 
The reason is that we know that security is `permutation invariant' under the systems because each step of the protocol --- parameter estimation, information reconciliation and privacy amplification --- is permutation invariant\footnote{Otherwise permutation-invariance can be enforced by applying a random permutation on the systems at the start of the protocol.}. The post-selection theorem allows us to relate security of permutation invariant states to the security of product states.

The post-selection theorem states that any permutation-invariant state can be obtained from the convex combination of i.i.d.\ (product) states by a measurement, and furthermore this measurement `works' sufficiently often. Therefore, if our key-distribution scheme is secure for product distributions, it is still `almost as secure' on a permutation invariant one.

Technically, the post-selection technique~\cite{postselection} gives a bound on the \emph{diamond norm} between two \emph{completely positive trace-preserving maps} (i.e., quantum channels) acting symmetrically on an $n$-party system. The diamond norm is directly related to the maximal probability of guessing whether one or the other map has been applied (on an input of choice), through the formula $p=1/2+ (1/4) \lVert \mathcal{E}-\mathcal{F} \rVert_\diamond $ (i.e., the distinguishing \emph{advantage} is then $(1/4)  \lVert \mathcal{E}-\mathcal{F} \rVert_\diamond$.) Therefore, it is especially useful in the context of cryptography, where a \emph{real} map is compared to an \emph{ideal} map --- such as one that creates a key that is secure by construction. While the diamond norm is defined as a maximization over all possible input states, the post-selection technique states that in the case of permutation invariant maps it is enough to consider them acting on a \emph{de Finetti state}, i.e., a convex combination of product states $\tau_{\mathcal{H}^n}=\int \sigma_{\mathcal{H}}^{\otimes n}\mu(\sigma_{\mathcal{H}})$, where $\mu$ is the measure induced by the Hilbert-Schmidt metric. We now restate the main result of~\cite{postselection}. 
\begin{theorem}[Post-selection~\cite{postselection}]\label{th:postselection}
Consider a linear map from \linebreak[4] $\mathrm{End}(\mathcal{H}^{\otimes n})$ to $\mathrm{End}(\mathcal{H}^{\prime})$.\footnote{Note that, in particular, $\Delta$ can be the difference between two completely positive trace-preserving maps $\mathcal{E}$ and $\mathcal{F}$.} 
If for any permutation $\pi$ there exists a completely positive trace-preserving map $\mathcal{K}_{\pi}$ such that $\Delta \circ \pi=\mathcal{K}_{\pi}\circ \Delta$, then 
\begin{align}
\nonumber \lVert \Delta \rVert_\diamond &\leq  g_{n,d} \lVert (\Delta \otimes \mathds{1}_{\mathcal{R}})\tau_{\mathcal{H}^n\mathcal{R}} \rVert_1\ ,
\end{align}
where $\mathds{1}_{\mathcal{R}}$ denotes the identity map on $\mathrm{End}(\mathcal{R})$ and the factor $g_{n,d}=\binom{n+d^2-1}{n} \linebreak[4] \leq (n+1)^{d^2-1}$, where $d=\mathrm{dim}\mathcal{H}$. 
\end{theorem}

For our purposes, this means roughly 
\begin{align}
\nonumber  \Pr[\mathcal{E}(\sigma^{\pi})=\mathrm{insecure}]&\leq (n+1)^{(d^2-1)} \int \Pr[\mathcal{E}(\sigma^{\otimes n})=\mathrm{insecure}]d\sigma\ ,
\end{align}
where $\sigma^{\pi}$ is a permutation invariant input, and $\mathcal{E}$ denotes the event that the scheme is insecure. 
The very right-hand side is what we have analysed in the previous section, and because this is exponentially small, it remains exponentially small even when multiplied by the polynomial factor in front of it.

In our case, $\sigma$ represents the system $P_{XY|UV}$. We, therefore, need to model $P_{XY|UV}$ by a quantum state (note that this is only a mathematical tool and does not have any physical meaning). More precisely, we represent the distribution $P_{XYUV}$ by $\sigma$. Since our parameter estimation protocol is such that it filters the conditional distribution independently of the input distribution (it aborts if any input does not occur often enough), this is equivalent to the conditional distribution.

\begin{lemma}
Let $P_{XYUV}$ be a probability distribution. Then there exists a density matrix $\sigma$ in a Hilbert space $\mathcal{H}$ with $\mathrm{dim}(\mathcal{H})=|\mathcal{X}| |\mathcal{Y}||\mathcal{U}| |\mathcal{V}|$ such that measuring $\sigma$ in the standard basis gives the distribution $P_{XYUV}$.
\end{lemma}
\begin{proof}
Associate with each element of the standard basis $\{\ket{i}\}_i$ an outcome $x,y,u,v$. Take $\sigma =\sum_{i=1}^{|\mathcal{X}||\mathcal{Y}| |\mathcal{U}| |\mathcal{V}|}p_i \ket{i}\bra{i}$ where the weights are $p_i=P_{XYUV}(x,y,u,v)$. 
\end{proof}
This implies, that we can use $d=|\mathcal{X}||\mathcal{Y}||\mathcal{U}||\mathcal{V}|$ in the above formula describing the security of our protocol. We can now state, that the key-distribution protocol is secure on any input (not only product). The protocol furthermore, reaches essentially the same key rate as in the product case. Robustness remains, of course, unchanged. 
\begin{theorem}\label{th:qsecure}
Protocol~\ref{prot:qkey} is $\epsilon$-secure with $\epsilon\in O(2^{-n})$ on any input for $m>n\cdot h(\delta)$ and $s=q\cdot n< \log_2P_{\mathrm{guess}}-m/n$. It is $\epsilon^{\prime\prime}$-robust on $\left( \mathcal{P}^{-\eta}\right)^{\otimes n}$ with $\epsilon^{\prime\prime}\in O(2^{-n})$. 
\end{theorem}
\begin{proof}
This follows directly from Lemma~\ref{lemma:qkdsecure}, using Theorem~\ref{th:postselection}.  
\end{proof}

\section{The Protocol}\label{sec:qprotocol}

We can apply the generic security proof of Section~\ref{sec:qkd} to a specific protocol. The implementation of this protocol is similar to~\cite{ekert}, i.e., it is an entanglement-based quantum key-distribution protocol. By the analysis given in Section~\ref{sec:qkd}, it is secure in the device-independent scenario.

\begin{figure}[h]
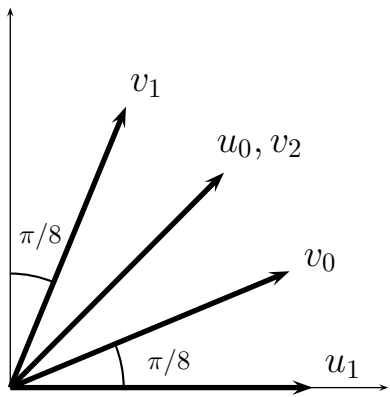

\centering
\pspicture*[](-0.5,-0.5)(5.5,5.5)
\psarc(0,0){1.5}{67.5}{90}
\rput[B]{0}(0.4,1.9){\normalsize{$\pi/8$}}
\psarc(0,0){1.5}{0}{22.5}
\rput[B]{0}(2.1,0.25){\normalsize{$\pi/8$}}
\rput[br]{90}(0,0){\psline[linewidth=0.5pt]{->}(0,0)(5,0)}
\rput[br]{0}(0,0){\psline[linewidth=0.5pt]{->}(0,0)(5,0)}
\rput[B]{0}(1.8,3.9){\Large{$v_1$}}
\rput[br]{67.5}(0,0){\psline[linewidth=2pt]{->}(0,0)(4,0)}
\rput[br]{22.5}(0,0){\psline[linewidth=2pt]{->}(0,0)(4,0)}
\rput[B]{0}(4.1,1.6){\Large{$v_0$}}
\rput[br]{45}(0,0){\psline[linewidth=2pt]{->}(0,0)(4,0)}
\rput[B]{0}(3.3,3.1){\Large{$u_0,v_2$}}
\rput[br]{0}(0,0){\psline[linewidth=2pt]{->}(0,0)(4,0)}
\rput[B]{0}(4.4,0.2){\Large{$u_1$}}
\endpspicture
\caption{\label{fig:basen} Alice's and Bob's measurement bases in terms of polarization used in Protocol~\ref{qprot}.}
\end{figure}

\begin{protocol}\label{qprot}\ 
\begin{enumerate}
 \item Alice creates $n$ maximally entangled states $\ket{\Psi^-}=(\ket{01}-\ket{10})/\sqrt{2}$, and sends one qubit of every state to Bob.
 \item Alice and Bob randomly measure the $i$\textsuperscript{th} system in either the basis $u_0$ or $u_1$ (for Alice) or $v_0$, $v_1$ or $v_2$ (Bob); the five bases are shown
 in Figure~\ref{fig:basen}. Bob flips his measurement result. 
They make sure that measurements on different subsystems commute. 
 \item The measurement results when both measured $u_0,v_2$ form the raw key.
\item For the remaining measurements they announce the results over the public authenticated channel and estimate the guessing probability $P_{\mathrm{guess}}$ and $\delta$ 
(see Section~\ref{subsec:qpe}). 
If the parameters are such that key agreement is possible, they continue; otherwise they abort.
\item They do information reconciliation and privacy amplification as \linebreak[4] given in Sections~\ref{subsec:qir} and~\ref{subsec:qpa}. 
\end{enumerate}
\end{protocol}

When Alice and Bob use a noisy quantum channel for the above protocol, they will not obtain a perfect singlet state. Let us assume that they obtain a mixture of the singlet with weight $1-\rho$ and a fully mixed state with weight $\rho$. The guessing probability for each individual system is then given in Figure~\ref{fig:pguess} and we give the key rate as function of the parameter $\rho$ in Figure~\ref{fig:qkeyrate}. 

\begin{figure}[h]
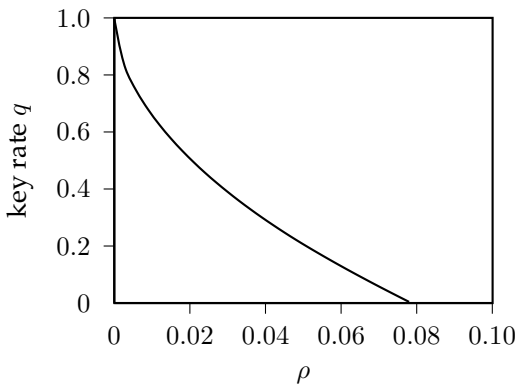

\centering
 \pspicture[](-2,-1)(7,4)
 \psset{xunit=50cm,yunit=3.75cm}
  \savedata{\mydataa}[
 {
{0.0000,1},
{0.0030,	0.821016426665049},
{0.0060,	0.740691751143548},
{0.0090,	0.677862617875497},
{0.0120,	0.624316166255036},
{0.0150,	0.57683385883745},
{0.0180,	0.533772043356894},
{0.0210,	0.494134634173186},
{0.0240,	0.457033494760889},
{0.0270,	0.422533310472383},
{0.0300,	0.389849475645926},
{0.0330,	0.35867529866217},
{0.0360,	0.329149447359709},
{0.0390,	0.300783150964238},
{0.0420,	0.273734569316761},
{0.0450,	0.247542826109787},
{0.0480,	0.222378305894379},
{0.0510,	0.198205818087642},
{0.0540,	0.174792660472844},
{0.0570,	0.151914854287686},
{0.0600,0.129554203443098},
{0.0630,0.107694014159461},
{0.066,0.0863189024482308},
{0.069,0.0652206073728114},
{0.072,0.0445824575238165},
{0.075,0.0242005174755597},
{0.078,0.00406681145525334}
} ]
 \rput[c](0.05,-0.25){{$\rho$}}
  \rput[c]{90}(-0.025,0.5){key rate {$q$}}
   \psaxes[Dx=0.02,Dy=0.2,  showorigin=true,tickstyle=bottom,axesstyle=frame](0,0)(0.1001,1.0001)
   \dataplot[plotstyle=curve,showpoints=false,dotstyle=o]{\mydataa} 
 \endpspicture
 \caption{\label{fig:qkeyrate} The key rate of Protocol~\ref{qprot} secure against quantum adversaries in the device-independent scenario as function of the channel noise.}
 \end{figure}

\section{Concluding Remarks}

In this chapter, we have shown that secure device-independent quantum  key distribution is possible even against the most general attacks of the eavesdropper, under the additional requirement that the measurements of the honest parties on different subsystems must commute. Our security analysis does not use any Hilbert space formalism, only convex optimization and works for any type of system. 

It is an open question whether the requirement of commuting measurements is necessary. When basing security only on the (weaker) non-sig\-nal\-ling condition, the analysis given in  Chapter~\ref{ch:impossibiltiy} implies that \emph{some} additional requirement is indeed needed, but this does not imply that the same is the case against a (weaker) quantum adversary. If this is possible, it would, of course, be interesting to give a security proof of device-independent quantum key distribution where both the state and measurements are completely arbitrary. 

A further open question is to find other systems which are (partially) secure against quantum adversaries. In particular, it is unknown whether there exist systems which are partially secure against quantum adversaries but completely insecure against non-signalling adversaries (in the context of key agreement, where we analyse security under public inputs). Finally, it would be interesting to see how the key rate of our key-dis\-tri\-bu\-tion scheme behaves in the non-asymptotic scenario, i.e., where only a finite number of systems are considered and a key of finite length is created.

\chapter{Necessity of the Non-Signalling Condition}\label{ch:impossibiltiy}

\section{Introduction}

Privacy amplification~\cite{bbr,ILL,koenigrenner} is the technique of applying a function to a partially secure string in order to obtain a (shorter) highly secure string. It can be used if the adversary holds classical as well as when she holds quantum information and might suggest, that the same is true against non-sig\-nal\-ling adversaries. 
In Chapter~\ref{ch:nsadversaries}, we have seen that this is indeed the case if we impose further non-sig\-nal\-ling conditions between the different subsystems. Privacy amplification is then even possible using a deterministic function. In this chapter, we will show that such an additional non-signalling condition \emph{within} Alice's and Bob's laboratories is necessary, in the sense that without it, no privacy amplification is possible. 

We will consider the case where Alice, Bob, and Eve share a system which is non-signalling between the three of them (but not between the subsystems). The system is such that it outputs a partially secure $n$-bit string to Alice and Bob. We then show that, no matter what function Alice and Bob apply to this string, they cannot obtain a highly secure bit. Put differently, Eve can attack the final key bit \emph{directly} (without trying to learn the bit string).\footnote{Maybe this is not so surprising, after all, Eve can delay her measurement and choose an attack depending on the hash function Alice and Bob have chosen. It might be more surprising that privacy amplification against quantum (or non-signalling) adversaries \emph{does} work in certain cases.} 

As an example, consider the case where Alice and Bob share $n$ systems, each taking one bit input and giving one bit output on both sides and such that the outputs are uniform and fulfil $P[X\oplus Y=U\cdot V]=1-\ep$ for each input pair (see Figure~\ref{fig:eps-box}). Note that this system can be expressed as a mixture of a system with error $\ep^{\prime}$ ($<\ep$) of weight $1-p=1-{(\ep-\ep^{\prime})}/{(1/2-2\ep^{\prime})}$ and a completely random bit with weight $p={(\ep-\ep^{\prime})}/{(1/2-2\ep^{\prime})}$. It is now easy to see that the XOR cannot be used as privacy amplification function in this case because of the following attack~\cite{personal}. 
Eve sends a system to Alice and Bob such that the first $n-1$ bits are just the outputs of $n-1$ independent systems, i.e., they have exactly error $\ep$. The last system is created by Eve as a probabilistic mixture of the systems as described above. She first tosses a coin such that `heads' has probability $p$. In case the result is `heads', Eve chooses the last bit pair such that it corresponds to a system with error $\ep^{\prime}$ and accepts to know nothing about the XOR. If the result is `tails', she tosses another coin and decides whether the outcome of the XOR should be $0$ or $1$. The system then outputs a random bit on Bob's side, and on Alice's side it outputs exactly the bit such that the XOR of all outputs corresponds to the result of the coin toss. Obviously, with probability $p$ Eve knows the XOR perfectly and this probability is independent of the number of systems $n$ Alice and Bob share. Furthermore, this attack works both in the non-signalling case (in which case $\ep^{\prime}=0$), as well as in the quantum case (where $\ep^{\prime}\approx 0.15$), and it even works when signalling is only permitted in the `forward' direction, i.e., when considering an even stronger restriction than what we will consider now.

The above attack is such that the marginal systems of Alice and Bob are \emph{exactly} as expected. Alternatively, 
Eve  could
\emph{always} send a local system such that she knows the outcome of the XOR with certainty. In that case, the probability not to get caught is the same as the probability not to get caught on a single system. In either scenario, the security only depends on the security of a single system and is independent of the total number of systems $n$. 

This already shows that the proof techniques we have used in the previous chapters do not carry over to this case.

\subparagraph*{Chapter outline}
The security definition and the eavesdroppers possibilities to attack are given in 
Section~\ref{sec:scenario}. We then give an intuition why privacy amplification does not work without an additional non-signalling assumptions 
by describing the system as probabilistic mixture of two systems, one of which is completely local. The maximum weight this local system can have is the \emph{local part} and gives a lower bound on the adversarial knowledge.  
In Section~\ref{sec:twobox}, we show that one or two systems have the same local part and that we can, therefore, not hope to create a more secret bit by applying a function to the 
outputs of two systems than when considering a single system. 

We then consider an arbitrary number of systems and give a good joint attack in Section~\ref{subsec:wbar}. In Section~\ref{subsec:linfct}, we show that 
 applying the XOR to a randomly selected subset of systems (i.e., applying a linear function) is actually counter-productive: The more systems are XORed together, the better Eve can know the outcome. 
Finally in Section~\ref{subsec:anyhash}, we show (again using the above attack) that even for arbitrary functions there exists a constant lower bound on the adversary's knowledge, and this bound is independent of the number of systems. This implies
 that there does not exist \emph{any} function that can be used to obtain an secure bit, no matter how many systems are shared by Alice and Bob. 

\subparagraph*{Related work}
The local part has been introduced in the context of quantum systems in~\cite{epr2} and further studied in~\cite{scarani} and~\cite{bgs}. 
We are not aware of any work considering directly the possibility or impossibility of privacy amplification in this setting. 
Since a higher violation of the CHSH inequality (Section~\ref{subsec:localsystem}, Example~\ref{ex:chshineq}, p.~\pageref{ex:chshineq}) corresponds to more secrecy, the question of privacy amplification is related to the question of non-locality distillation, i.e., whether several partially non-local systems can be used to obtain a more non-local one. This question has been investigated and both positive as well as negative answers have been found for special systems~\cite{fww,bs,dw,short}. 

\subparagraph*{Contributions}
The contributions of this chapter are Lemma~\ref{lemma:two_symm_boxes} about the local part of $2$ systems, the attack of a non-signalling adversary \linebreak[4] against an arbitrary number of systems (Lemma~\ref{lemma:distancewbar}) and the resulting impossibility of privacy amplification in the tripartite non-signalling case, given in Lemma~\ref{lemma:imppalin} and Theorem~\ref{th:imppa}. The results of this chapter have previously been published in~\cite{localpart} and~\cite{HRW08}.

\section{Scenario and Security Criteria}\label{sec:scenario}

We study the scenario where Alice and Bob share several approximations of PR~boxes (see Example~\ref{fig:prbox}, p.~\pageref{fig:prbox}); more precisely, $n$ independent and unbiased PR~boxes with error $\ep$, defined below.
\begin{definition}
An \emph{unbiased PR~box with error $\ep$} is a system $P_{XY|UV}$, where $\mathcal{X}=\mathcal{Y}=\mathcal{U}=\mathcal{V}=\{0,1\}$, and for every pair $(u,v)$, $X$ and $Y$ are uniform random bits, and
\begin{align}
\nonumber \Pr[X\oplus Y=u \cdot v]&=1-\ep
\end{align}
(see also Figure~\ref{fig:eps-box}).
\end{definition}

\begin{figure}[h]
\centering
\psset{unit=0.525cm}
\pspicture*[](-2,-1)(8.5,10)
\psline[linewidth=0.5pt]{-}(0,6)(-1,7)
\rput[c]{0}(-0.25,6.75){\scriptsize{$X$}}
\rput[c]{0}(-0.75,6.25){\scriptsize{$Y$}}
\rput[c]{0}(-0.5,7.5){\large{$U$}}
\rput[c]{0}(-1.5,6.5){\large{$V$}}
\rput[c]{0}(2,7.5){\Large{$0$}}
\rput[c]{0}(6,7.5){\Large{$1$}}
\rput[c]{0}(1,6.5){\Large{$0$}}
\rput[c]{0}(3,6.5){\Large{$1$}}
\rput[c]{0}(5,6.5){\Large{$0$}}
\rput[c]{0}(7,6.5){\Large{$1$}}
\rput[c]{0}(-1.5,4.5){\Large{$0$}}
\rput[c]{0}(-1.5,1.5){\Large{$1$}}
\rput[c]{0}(-0.5,5.25){\Large{$0$}}
\rput[c]{0}(-0.5,3.75){\Large{$1$}}
\rput[c]{0}(-0.5,2.25){\Large{$0$}}
\rput[c]{0}(-0.5,0.75){\Large{$1$}}
\psline[linewidth=2pt]{-}(-1,0)(8,0)
\psline[linewidth=2pt]{-}(-1,6)(8,6)
\psline[linewidth=2pt]{-}(-1,3)(8,3)
\psline[linewidth=1pt]{-}(0,1.5)(8,1.5)
\psline[linewidth=1pt]{-}(0,4.5)(8,4.5)
\psline[linewidth=2pt]{-}(0,0)(0,7)
\psline[linewidth=2pt]{-}(8,0)(8,7)
\psline[linewidth=2pt]{-}(4,0)(4,7)
\psline[linewidth=1pt]{-}(2,0)(2,6)
\psline[linewidth=1pt]{-}(6,0)(6,6)
\rput[c]{0}(1,5.25){\Large{$\frac{1-\ep}{2}$}}
\rput[c]{0}(3,3.75){\Large{$\frac{1-\ep}{2}$}}
\rput[c]{0}(5,5.25){\Large{$\frac{1-\ep}{2}$}}
\rput[c]{0}(7,3.75){\Large{$\frac{1-\ep}{2}$}}
\rput[c]{0}(1,2.25){\Large{$\frac{1-\ep}{2}$}}
\rput[c]{0}(3,0.75){\Large{$\frac{1-\ep}{2}$}}
\rput[c]{0}(5,0.75){\Large{$\frac{1-\ep}{2}$}}
\rput[c]{0}(7,2.25){\Large{$\frac{1-\ep}{2}$}}
\rput[c]{0}(3,5.25){\Large{$\frac{\ep}{2}$}}
\rput[c]{0}(1,3.75){\Large{$\frac{\ep}{2}$}}
\rput[c]{0}(7,5.25){\Large{$\frac{\ep}{2}$}}
\rput[c]{0}(5,3.75){\Large{$\frac{\ep}{2}$}}
\rput[c]{0}(3,2.25){\Large{$\frac{\ep}{2}$}}
\rput[c]{0}(1,0.75){\Large{$\frac{\ep}{2}$}}
\rput[c]{0}(5,2.25){\Large{$\frac{\ep}{2}$}}
\rput[c]{0}(7,0.75){\Large{$\frac{\ep}{2}$}}
\endpspicture
\caption{\label{fig:eps-box}An unbiased PR~box with error $\ep$.}
\end{figure}

The system we consider behaves like $n$ systems, but we only require it to be non-signalling between Alice and Bob, i.e., two sets of interfaces. This means that even though the marginal system of Alice and Bob is actually a $2n$-party non-signalling system, we only consider it as a $2$-party non-signalling system which takes an $n$-bit string as input and gives an $n$-bit string as output on each side.

We define a short notation for the bipartite non-signalling system that behaves like $n$ unbiased PR~boxes with error $\ep$.
\begin{definition}\label{def:n-eps-box}
The system $P_{{XY}|{UV}}^{n,\ep}$ is a bipartite non-signalling system with $\mathcal{X}=\mathcal{Y}=\mathcal{U}=\mathcal{V}=\{0,1\}^n$, such that
\begin{align}
\nonumber P_{{XY}|{UV}}^{n,\ep} &:= \prod_{i=1}^n P_{X_iY_i|U_iV_i}\ ,
\end{align}
and where $P_{X_iY_i|U_iV_i}$ is an unbiased PR~box with error $\ep$. 
\end{definition}

\begin{figure}[h]
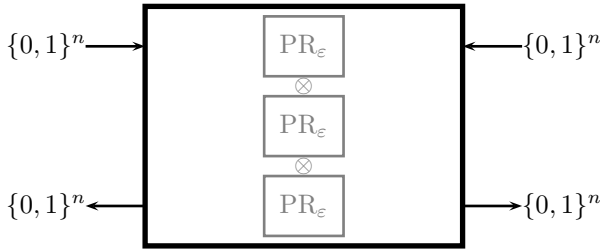

\centering
\psset{unit=0.525cm}
\pspicture*[](-5,-0.2)(11.5,6.2)
\pspolygon[linewidth=2pt](0,0)(8,0)(8,6)(0,6)
\rput[c]{0}(-2.5,5){$\{0,1\}^n$}
\rput[c]{0}(-2.5,1){$\{0,1\}^n$}
\rput[c]{0}(10.5,5){$\{0,1\}^n$}
\rput[c]{0}(10.5,1){$\{0,1\}^n$}
\psline[linewidth=1pt]{->}(-1.5,5)(0,5)
\psline[linewidth=1pt]{<-}(-1.5,1)(0,1)
\psline[linewidth=1pt]{<-}(8,5)(9.5,5)
\psline[linewidth=1pt]{->}(8,1)(9.5,1)
\rput[c]{0}(4,2){\color{gray}{\small{$\otimes$}}}
\rput[c]{0}(4,4){\color{gray}{\small{$\otimes$}}}
\rput[c]{0}(3,0.25){
\pspolygon[linewidth=1pt,linecolor=gray](0,0)(2,0)(2,1.5)(0,1.5)
\rput[c]{0}(1,0.75){\color{gray}{$\mathrm{PR}_{\varepsilon}$}}
}
\rput[c]{0}(3,2.25){
\pspolygon[linewidth=1pt,linecolor=gray](0,0)(2,0)(2,1.5)(0,1.5)
\rput[c]{0}(1,0.75){\color{gray}{$\mathrm{PR}_{\varepsilon}$}}
}
\rput[c]{0}(3,4.25){
\pspolygon[linewidth=1pt,linecolor=gray](0,0)(2,0)(2,1.5)(0,1.5)
\rput[c]{0}(1,0.75){\color{gray}{$\mathrm{PR}_{\varepsilon}$}}
}
\endpspicture
\caption{\label{fig:marg}Alice's and Bob's system looks like $n$ independent systems.}
\end{figure}

For an impossibility proof, we can make the assumption that the system behaves \emph{exactly} this way (i.e., Alice and Bob do not need to do parameter estimation) and only consider the distance from uniform of Alice's key (i.e., they do not do information reconciliation and Bob does not output anything). Since the distance from uniform of a key string is lower-bounded by the distance of each bit, 
it will be enough to consider the case when Alice's key consists of a single bit, and to give a specific (explicit) attack which reaches a high distance from uniform of this bit. 

More specifically, we will consider the case where Alice, Bob, and Eve share a tripartite non-signalling system such that the marginal of Alice and Bob corresponds to $n$ unbiased PR~boxes with error $\ep$ (see Definition~\ref{def:n-eps-box} and Figure~\ref{fig:marg}), plus a classical public authenticated channel (see Figure~\ref{fig:real-ideal-impossibility}). Alice applies a (public) function $f\colon \{0,1\}^n\rightarrow \{0,1\}$ to her $n$-bit string to obtain a single bit. Bob outputs nothing. Eve receives the information sent over the channel $Q$, where $Q=(U=u,V=v,F=f)$\footnote{As in the previous chapters, lower case letters in $Q$ mean that we consider the distance from uniform given \emph{this specific value}.} (we have included the inputs in analogy with the situation in the previous chapters, although we will see that for the specific attack we will use, it is not necessary to know the inputs). Eve can then choose an input (possibly depending on $Q$) to her interface of the non-signalling system and obtains an output. The only restriction hereby, is the following condition. 
\begin{figure}[h]
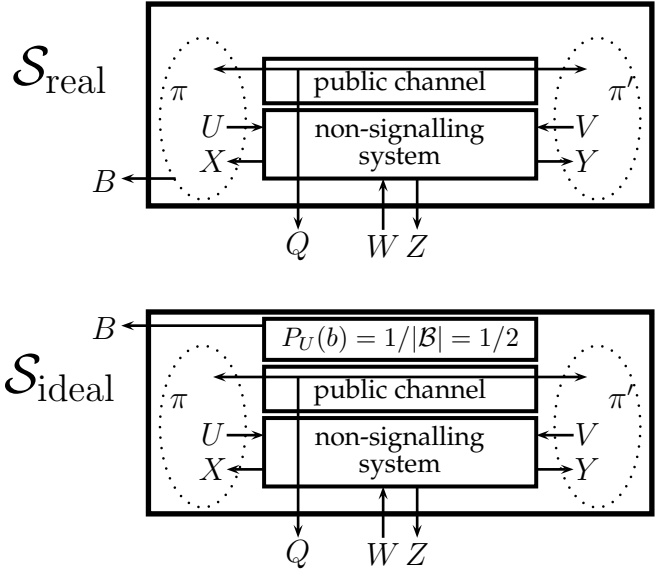

\centering
\pspicture*[](-6,-4.5)(5,3.5)
\psset{unit=0.9cm}
\rput[c]{0}(0,0){
\rput[b]{0}(-5,2.5){\huge{$\mathcal{S}_{\mathrm{real}}$}}
\psline[linewidth=1pt]{<->}(-2.75,2.85)(2.75,2.85)
\rput[b]{0}(0,2.425){public channel}
\pspolygon[linewidth=1.5pt](-2,3)(2,3)(2,2.35)(-2,2.35)
\pspolygon[linewidth=1.5pt](-2,1.25)(2,1.25)(2,2.25)(-2,2.25)
\psline[linewidth=1pt]{<-}(-2.55,1.5)(-2,1.5)
\rput[b]{0}(-2.75,1.35){\large{$X$}}
\psline[linewidth=1pt]{->}(-2.55,2)(-2,2)
\rput[b]{0}(-2.75,1.85){\large{$U$}}
\psline[linewidth=1pt]{<-}(2.55,1.5)(2,1.5)
\rput[b]{0}(2.75,1.35){\large{$Y$}}
\psline[linewidth=1pt]{->}(2.55,2)(2,2)
\rput[b]{0}(2.75,1.85){\large{$V$}}
\psline[linewidth=1pt]{->}(-0.25,0.5)(-0.25,1.25)
\rput[b]{0}(-0.25,0.1){\large{$W$}}
\psline[linewidth=1pt]{<-}(0.25,0.5)(0.25,1.25)
\rput[b]{0}(0.25,0.1){\large{$Z$}}
\rput[b]{0}(0,1.75){non-signalling}
\rput[b]{0}(0,1.35){system}
\psline[linewidth=1pt]{<-}(-1.5,0.5)(-1.5,2.85)
\rput[b]{0}(-1.5,0.05){\large{$Q$}}
\psellipse[linewidth=1pt,linestyle=dotted](-2.9,2.125)(0.65,1.2)
\psellipse[linewidth=1pt,linestyle=dotted](2.9,2.125)(0.65,1.2)
\pspolygon[linewidth=2pt](-3.7,0.85)(3.7,0.85)(3.7,3.8)(-3.7,3.8) 
\rput[b]{0}(-3.25,2.4){\large{$\pi$}}
\rput[b]{0}(3.25,2.4){\large{$\pi^{\prime}$}}
\psline[linewidth=1pt]{<-}(-4.1,1.25)(-3.3,1.25)
\rput[b]{0}(-4.35,1.05){\large{$B$}}
}
\rput[c]{0}(0,-4.5){
\rput[b]{0}(-5,2.5){\huge{$\mathcal{S}_{\mathrm{ideal}}$}}
\psline[linewidth=1pt]{<->}(-2.75,2.85)(2.75,2.85)
\rput[b]{0}(0,2.425){public channel}
\pspolygon[linewidth=1.5pt](-2,3)(2,3)(2,2.35)(-2,2.35)
\pspolygon[linewidth=1.5pt](-2,1.25)(2,1.25)(2,2.25)(-2,2.25)
\psline[linewidth=1pt]{<-}(-2.55,1.5)(-2,1.5)
\rput[b]{0}(-2.75,1.35){\large{$X$}}
\psline[linewidth=1pt]{->}(-2.55,2)(-2,2)
\rput[b]{0}(-2.75,1.85){\large{$U$}}
\psline[linewidth=1pt]{<-}(2.55,1.5)(2,1.5)
\rput[b]{0}(2.75,1.35){\large{$Y$}}
\psline[linewidth=1pt]{->}(2.55,2)(2,2)
\rput[b]{0}(2.75,1.85){\large{$V$}}
\psline[linewidth=1pt]{->}(-0.25,0.5)(-0.25,1.25)
\rput[b]{0}(-0.25,0.1){\large{$W$}}
\psline[linewidth=1pt]{<-}(0.25,0.5)(0.25,1.25)
\rput[b]{0}(0.25,0.1){\large{$Z$}}
\rput[b]{0}(0,1.75){non-signalling}
\rput[b]{0}(0,1.35){system}
\psline[linewidth=1pt]{<-}(-1.5,0.5)(-1.5,2.85)
\rput[b]{0}(-1.5,0.05){\large{$Q$}}
\psellipse[linewidth=1pt,linestyle=dotted](-2.9,2.125)(0.65,1.2)
\psellipse[linewidth=1pt,linestyle=dotted](2.9,2.125)(0.65,1.2)
\pspolygon[linewidth=2pt](-3.7,0.85)(3.7,0.85)(3.7,3.8)(-3.7,3.8) 
\rput[b]{0}(-3.25,2.4){\large{$\pi$}}
\rput[b]{0}(3.25,2.4){\large{$\pi^{\prime}$}}
\pspolygon[linewidth=1.5pt](-2,3.1)(2,3.1)(2,3.7)(-2,3.7)
\rput[b]{0}(0,3.2){$P_U(b)=1/|\mathcal{B}|=1/2$}
\psline[linewidth=1pt]{<-}(-4.1,3.6)(-2,3.6)
\rput[b]{0}(-4.35,3.4){\large{$B$}}
}
\endpspicture
\caption{\label{fig:real-ideal-impossibility}
Our \emph{real} system (top). Alice and Bob share a public authenticated channel and a non-signalling system. The key bit is $B=f(X)$. 
In the \emph{ideal} system (bottom), the bit $B$ is a perfectly uniform bit unrelated to the other parts of the system. 
  }
\end{figure}
\begin{condition}
The system $P_{XYZ|UVW}$ is tripartite non-signalling with \linebreak[4] marginal $P_{XY|UV}=P_{XY|UV}^{n,\ep}$. 
\end{condition}

The distinguishing advantage between the real and ideal system is the distance from uniform of $B=f(X)$ given $Z(W_{\mathrm{n-s}})$ and $Q$, denoted by $d(B|Z(W_{\mathrm{n-s}}),Q)$. 
We recall the definition here, see Definition~\ref{def:dist-sa-from_uniform}, p.~\pageref{def:dist-sa-from_uniform} for more details.
\begin{multline}
\nonumber
 d(B|Z(W_{\mathrm{n-s}}),Q)=\\
\frac{1}{2}
\sum_{b,q}  \max_{w:{\mathrm{n-s}}} \sum_{z} P_{Z,Q|W=w}(z,q)
\cdot \left|P_{B|Z=z,Q=q,W=w}(b)-P_U(b)\right|\ .
\end{multline}
The distance from uniform is exactly the advantage Eve has when guessing the bit $B$ and it, therefore, quantifies the knowledge Eve has about the key bit. 
 Obviously, this quantity depends on the system Alice and Bob share, Eve's strategy (the non-signalling partition she uses), and the information $Q$ sent over the public channel, in particular, the hash  function $f$ that is applied to the output bits. 
 
The quantity that we are interested in --- the distance from uniform of $B$ given $Z(W_{\mathrm{n-s}})$ and $Q$ --- is defined as a maximization over all possible non-signalling strategies of Eve. We will sometimes also consider the distance from uniform given a \emph{specific}
adversarial strategy, defined as 
\begin{align}
\nonumber
 d(B|Z(w),Q)
&=\frac{1}{2}
\sum_{b,q}\sum_{z} P_{Z,Q|W=w}(z,q)\cdot
\left|P_{B|Z=z,Q=q,W=w}(b)-P_U(b)\right|\ .
\end{align}
(see Definition~\ref{def:distancesmallw}, p.~\pageref{def:distancesmallw} for details).

Since we will only consider the case when Alice tries to create a single secure bit, we can further simplify this expression. 
\begin{lemma}\label{lemma:distancesinglebit}
For the case $B=f(X)$ with $f\colon \{0,1\}^n\rightarrow \{0,1\}$ and $Q=(U=u,V=v,F=f)$
\begin{align}
\nonumber d(B|Z(w),Q) &=
\frac{1}{2} \sum\limits_{z_w} p^{z_w}\cdot \left|\sum_{(x,y)}(-1)^{f(x)} P_{{XY}|{UV}}^{z_w}(x,y,u,v)\right|\ ,
\end{align}
where $\{(p^{z_w},P_{{XY}|{UV}}^{z_w})\}_{z_w}$ are the elements of the non-signalling partition \linebreak[4] defined by $w$. 
\end{lemma}
\begin{proof}
\begin{align}
\nonumber  d(B|Z(w),Q)
\nonumber &=
\frac{1}{2}
\sum_{b,z} P_{Z,Q|W=w}(z,q)\cdot \left|P_{B|Z=z,Q=q,W=w}(b)-\frac{1}{2}\right|\\
\nonumber &=
\frac{1}{2} \sum\limits_{z_w} p^{z_w}\cdot \left|\sum_{(x,y)
} 
(-1)^{f(x)}
P_{{XY}|{UV}}^{z_w}(x,y,u,v)
\right| \qedhere
\end{align}
\end{proof}
$d(B|Z(w),Q)=0$ means that the eavesdropper has no knowledge about the bit $B$ 
while $d(B|Z(w),Q)=1/2$ corresponds to complete knowledge. 
We will now show that there exists a strategy $w$ such that this distance from uniform is high, 
independent of the number of systems and 
what function $f$ is applied to the output bits.

\section{Best Non-Signalling Partition of a Single System}
\label{sec:single_box}

In this section, we show that the bound on the distance from uniform of the outputs of a  bipartite system with binary inputs and outputs derived in Lemma~\ref{lemma:guessing_probability_single_box}, p.~\pageref{lemma:guessing_probability_single_box} is tight. 

\begin{lemma}\label{lemma:best_partition_single_box}
Let $P_{XY|UV}$ be a non-signaling system 
with  $\mathcal{X}=\mathcal{Y}=\mathcal{U}=\mathcal{V}=\{0,1\}$ and
$\sum_{x\oplus y=u\cdot v}P_{XY|UV}(x,y,u,v)/4=1-\ep$,  
where $\ep\leq 1/4$. Then for $Q=(U=u,V=v)$, there exists a non-signalling partition $w$ such that 
\begin{align}
\nonumber 
d(X|Z(w),Q) &=2\ep\ .
\end{align}
\end{lemma}
\begin{proof}
The proof is given by the non-signalling partition given in Figure~\ref{fig:partition-single-box}.
To see that Figure~\ref{fig:partition-single-box}
defines a non-signalling partition, notice that the parameters $a_2$, $a_3$, $b_2$, $b_3$, $c_2$, $c_3$, $d_1$, $d_4$ (the ones for which
$x\oplus y\neq u\cdot v$) fully characterize the non-signalling system. 
By the normalization 
($\sum_i a_i=1$; 
and similar for $b$, $c$, and $d$) 
and the non-signalling condition
($a_1+a_2=b_1+b_2$, etc.)
we can express $a_1$ as 
\begin{align}
\nonumber a_1 &=\frac{1}{2}\cdot (1-a_2-a_3+b_2-b_3-c_2+c_3+d_1-d_4)\ .
\end{align}
This shows that the right-hand side and left-hand side of the equation are indeed equal. 
Because we assumed $4\ep\leq 1$, all weights are positive. 
To see that it reaches this distance from uniform, note that 
with probability $a_2-a_3+b_2-b_3-c_2+c_3+d_1-d_4=4\ep$, $z$ is such that $P_{XY|UV}^z$ is local deterministic 
(i.e., $X$ and $Y$) are deterministic functions of $U$ and $V$), in which case knowing $U=u$ and $V=v$, $z$ gives perfect information 
about $X$. 
With probability $1-4\ep$, $z$ is such that $P_{XY|UV}^z$ is a PR~box, 
in which case even knowing $U=u$ and $V=v$, $X$ is a uniform random bit.
\end{proof}
\begin{figure}[h]
\centering
\psset{unit=0.21cm}
\pspicture*[](-4.5,-21)(48,10)
\rput[c](0,0){
\psline[linewidth=0.5pt]{-}(0,6)(-1,7)
\rput[c]{0}(-0.25,6.75){\tiny{$X$}}
\rput[c]{0}(-0.75,6.25){\tiny{$Y$}}
\rput[c]{0}(-0.5,7.5){\small{$U$}}
\rput[c]{0}(-1.5,6.5){\small{$V$}}
\rput[c]{0}(2,7.5){\small{$0$}}
\rput[c]{0}(6,7.5){\small{$1$}}
\rput[c]{0}(1,6.5){\small{$0$}}
\rput[c]{0}(3,6.5){\small{$1$}}
\rput[c]{0}(5,6.5){\small{$0$}}
\rput[c]{0}(7,6.5){\small{$1$}}
\rput[c]{0}(-1.5,4.5){\small{$0$}}
\rput[c]{0}(-1.5,1.5){\small{$1$}}
\rput[c]{0}(-0.5,5.25){\small{$0$}}
\rput[c]{0}(-0.5,3.75){\small{$1$}}
\rput[c]{0}(-0.5,2.25){\small{$0$}}
\rput[c]{0}(-0.5,0.75){\small{$1$}}
\psline[linewidth=1pt]{-}(-1,0)(8,0)
\psline[linewidth=1pt]{-}(-1,6)(8,6)
\psline[linewidth=1pt]{-}(-1,3)(8,3)
\psline[linewidth=0.5pt]{-}(0,1.5)(8,1.5)
\psline[linewidth=0.5pt]{-}(0,4.5)(8,4.5)
\psline[linewidth=1pt]{-}(0,0)(0,7)
\psline[linewidth=1pt]{-}(8,0)(8,7)
\psline[linewidth=1pt]{-}(4,0)(4,7)
\psline[linewidth=0.5pt]{-}(2,0)(2,6)
\psline[linewidth=0.5pt]{-}(6,0)(6,6)
\rput[c]{0}(1,5.25){\small{$a_1$}}
\rput[c]{0}(3,3.75){\small{$a_4$}}
\rput[c]{0}(5,5.25){\small{$b_1$}}
\rput[c]{0}(7,3.75){\small{$b_4$}}
\rput[c]{0}(1,2.25){\small{$c_1$}}
\rput[c]{0}(3,0.75){\small{$c_4$}}
\rput[c]{0}(5,0.75){\small{$d_1$}}
\rput[c]{0}(7,2.25){\small{$d_4$}}
\rput[c]{0}(3,5.25){\small{$a_2$}}
\rput[c]{0}(1,3.75){\small{$a_3$}}
\rput[c]{0}(7,5.25){\small{$b_2$}}
\rput[c]{0}(5,3.75){\small{$b_3$}}
\rput[c]{0}(3,2.25){\small{$c_2$}}
\rput[c]{0}(1,0.75){\small{$c_3$}}
\rput[c]{0}(5,2.25){\small{$d_2$}}
\rput[c]{0}(7,0.75){\small{$d_3$}}
}
\rput[c](10,3){$=a_2\cdot$}
\rput[c](13,0){
\psline[linewidth=0.5pt]{-}(0,6)(-1,7)
\rput[c]{0}(-0.25,6.75){\tiny{$X$}}
\rput[c]{0}(-0.75,6.25){\tiny{$Y$}}
\rput[c]{0}(-0.5,7.5){\small{$U$}}
\rput[c]{0}(-1.5,6.5){\small{$V$}}
\rput[c]{0}(2,7.5){\small{$0$}}
\rput[c]{0}(6,7.5){\small{$1$}}
\rput[c]{0}(1,6.5){\small{$0$}}
\rput[c]{0}(3,6.5){\small{$1$}}
\rput[c]{0}(5,6.5){\small{$0$}}
\rput[c]{0}(7,6.5){\small{$1$}}
\rput[c]{0}(-1.5,4.5){\small{$0$}}
\rput[c]{0}(-1.5,1.5){\small{$1$}}
\rput[c]{0}(-0.5,5.25){\small{$0$}}
\rput[c]{0}(-0.5,3.75){\small{$1$}}
\rput[c]{0}(-0.5,2.25){\small{$0$}}
\rput[c]{0}(-0.5,0.75){\small{$1$}}
\psline[linewidth=1pt]{-}(-1,0)(8,0)
\psline[linewidth=1pt]{-}(-1,6)(8,6)
\psline[linewidth=1pt]{-}(-1,3)(8,3)
\psline[linewidth=0.5pt]{-}(0,1.5)(8,1.5)
\psline[linewidth=0.5pt]{-}(0,4.5)(8,4.5)
\psline[linewidth=1pt]{-}(0,0)(0,7)
\psline[linewidth=1pt]{-}(8,0)(8,7)
\psline[linewidth=1pt]{-}(4,0)(4,7)
\psline[linewidth=0.5pt]{-}(2,0)(2,6)
\psline[linewidth=0.5pt]{-}(6,0)(6,6)
\rput[c]{0}(1,5.25){\small{$0$}}
\rput[c]{0}(3,3.75){\small{$0$}}
\rput[c]{0}(5,5.25){\small{$1$}}
\rput[c]{0}(7,3.75){\small{$0$}}
\rput[c]{0}(1,2.25){\small{$0$}}
\rput[c]{0}(3,0.75){\small{$1$}}
\rput[c]{0}(5,0.75){\small{$1$}}
\rput[c]{0}(7,2.25){\small{$0$}}
\rput[c]{0}(3,5.25){\small{$1$}}
\rput[c]{0}(1,3.75){\small{$0$}}
\rput[c]{0}(7,5.25){\small{$0$}}
\rput[c]{0}(5,3.75){\small{$0$}}
\rput[c]{0}(3,2.25){\small{$0$}}
\rput[c]{0}(1,0.75){\small{$0$}}
\rput[c]{0}(5,2.25){\small{$0$}}
\rput[c]{0}(7,0.75){\small{$0$}}
}
\rput[c](23,3){$+a_3\cdot$}
\rput[c](26,0){
\psline[linewidth=0.5pt]{-}(0,6)(-1,7)
\rput[c]{0}(-0.25,6.75){\tiny{$X$}}
\rput[c]{0}(-0.75,6.25){\tiny{$Y$}}
\rput[c]{0}(-0.5,7.5){\small{$U$}}
\rput[c]{0}(-1.5,6.5){\small{$V$}}
\rput[c]{0}(2,7.5){\small{$0$}}
\rput[c]{0}(6,7.5){\small{$1$}}
\rput[c]{0}(1,6.5){\small{$0$}}
\rput[c]{0}(3,6.5){\small{$1$}}
\rput[c]{0}(5,6.5){\small{$0$}}
\rput[c]{0}(7,6.5){\small{$1$}}
\rput[c]{0}(-1.5,4.5){\small{$0$}}
\rput[c]{0}(-1.5,1.5){\small{$1$}}
\rput[c]{0}(-0.5,5.25){\small{$0$}}
\rput[c]{0}(-0.5,3.75){\small{$1$}}
\rput[c]{0}(-0.5,2.25){\small{$0$}}
\rput[c]{0}(-0.5,0.75){\small{$1$}}
\psline[linewidth=1pt]{-}(-1,0)(8,0)
\psline[linewidth=1pt]{-}(-1,6)(8,6)
\psline[linewidth=1pt]{-}(-1,3)(8,3)
\psline[linewidth=0.5pt]{-}(0,1.5)(8,1.5)
\psline[linewidth=0.5pt]{-}(0,4.5)(8,4.5)
\psline[linewidth=1pt]{-}(0,0)(0,7)
\psline[linewidth=1pt]{-}(8,0)(8,7)
\psline[linewidth=1pt]{-}(4,0)(4,7)
\psline[linewidth=0.5pt]{-}(2,0)(2,6)
\psline[linewidth=0.5pt]{-}(6,0)(6,6)
\rput[c]{0}(1,5.25){\small{$0$}}
\rput[c]{0}(3,3.75){\small{$0$}}
\rput[c]{0}(5,5.25){\small{$0$}}
\rput[c]{0}(7,3.75){\small{$1$}}
\rput[c]{0}(1,2.25){\small{$1$}}
\rput[c]{0}(3,0.75){\small{$0$}}
\rput[c]{0}(5,0.75){\small{$0$}}
\rput[c]{0}(7,2.25){\small{$1$}}
\rput[c]{0}(3,5.25){\small{$0$}}
\rput[c]{0}(1,3.75){\small{$1$}}
\rput[c]{0}(7,5.25){\small{$0$}}
\rput[c]{0}(5,3.75){\small{$0$}}
\rput[c]{0}(3,2.25){\small{$0$}}
\rput[c]{0}(1,0.75){\small{$0$}}
\rput[c]{0}(5,2.25){\small{$0$}}
\rput[c]{0}(7,0.75){\small{$0$}}
}
\rput[c](36,3){$+b_2\cdot$}
\rput[c](39,0){
\psline[linewidth=0.5pt]{-}(0,6)(-1,7)
\rput[c]{0}(-0.25,6.75){\tiny{$X$}}
\rput[c]{0}(-0.75,6.25){\tiny{$Y$}}
\rput[c]{0}(-0.5,7.5){\small{$U$}}
\rput[c]{0}(-1.5,6.5){\small{$V$}}
\rput[c]{0}(2,7.5){\small{$0$}}
\rput[c]{0}(6,7.5){\small{$1$}}
\rput[c]{0}(1,6.5){\small{$0$}}
\rput[c]{0}(3,6.5){\small{$1$}}
\rput[c]{0}(5,6.5){\small{$0$}}
\rput[c]{0}(7,6.5){\small{$1$}}
\rput[c]{0}(-1.5,4.5){\small{$0$}}
\rput[c]{0}(-1.5,1.5){\small{$1$}}
\rput[c]{0}(-0.5,5.25){\small{$0$}}
\rput[c]{0}(-0.5,3.75){\small{$1$}}
\rput[c]{0}(-0.5,2.25){\small{$0$}}
\rput[c]{0}(-0.5,0.75){\small{$1$}}
\psline[linewidth=1pt]{-}(-1,0)(8,0)
\psline[linewidth=1pt]{-}(-1,6)(8,6)
\psline[linewidth=1pt]{-}(-1,3)(8,3)
\psline[linewidth=0.5pt]{-}(0,1.5)(8,1.5)
\psline[linewidth=0.5pt]{-}(0,4.5)(8,4.5)
\psline[linewidth=1pt]{-}(0,0)(0,7)
\psline[linewidth=1pt]{-}(8,0)(8,7)
\psline[linewidth=1pt]{-}(4,0)(4,7)
\psline[linewidth=0.5pt]{-}(2,0)(2,6)
\psline[linewidth=0.5pt]{-}(6,0)(6,6)
\rput[c]{0}(1,5.25){\small{$1$}}
\rput[c]{0}(3,3.75){\small{$0$}}
\rput[c]{0}(5,5.25){\small{$0$}}
\rput[c]{0}(7,3.75){\small{$0$}}
\rput[c]{0}(1,2.25){\small{$1$}}
\rput[c]{0}(3,0.75){\small{$0$}}
\rput[c]{0}(5,0.75){\small{$0$}}
\rput[c]{0}(7,2.25){\small{$1$}}
\rput[c]{0}(3,5.25){\small{$0$}}
\rput[c]{0}(1,3.75){\small{$0$}}
\rput[c]{0}(7,5.25){\small{$1$}}
\rput[c]{0}(5,3.75){\small{$0$}}
\rput[c]{0}(3,2.25){\small{$0$}}
\rput[c]{0}(1,0.75){\small{$0$}}
\rput[c]{0}(5,2.25){\small{$0$}}
\rput[c]{0}(7,0.75){\small{$0$}}
}
\rput[c](-3,-7){$+b_3\cdot$}
\rput[c](0,-10){
\psline[linewidth=0.5pt]{-}(0,6)(-1,7)
\rput[c]{0}(-0.25,6.75){\tiny{$X$}}
\rput[c]{0}(-0.75,6.25){\tiny{$Y$}}
\rput[c]{0}(-0.5,7.5){\small{$U$}}
\rput[c]{0}(-1.5,6.5){\small{$V$}}
\rput[c]{0}(2,7.5){\small{$0$}}
\rput[c]{0}(6,7.5){\small{$1$}}
\rput[c]{0}(1,6.5){\small{$0$}}
\rput[c]{0}(3,6.5){\small{$1$}}
\rput[c]{0}(5,6.5){\small{$0$}}
\rput[c]{0}(7,6.5){\small{$1$}}
\rput[c]{0}(-1.5,4.5){\small{$0$}}
\rput[c]{0}(-1.5,1.5){\small{$1$}}
\rput[c]{0}(-0.5,5.25){\small{$0$}}
\rput[c]{0}(-0.5,3.75){\small{$1$}}
\rput[c]{0}(-0.5,2.25){\small{$0$}}
\rput[c]{0}(-0.5,0.75){\small{$1$}}
\psline[linewidth=1pt]{-}(-1,0)(8,0)
\psline[linewidth=1pt]{-}(-1,6)(8,6)
\psline[linewidth=1pt]{-}(-1,3)(8,3)
\psline[linewidth=0.5pt]{-}(0,1.5)(8,1.5)
\psline[linewidth=0.5pt]{-}(0,4.5)(8,4.5)
\psline[linewidth=1pt]{-}(0,0)(0,7)
\psline[linewidth=1pt]{-}(8,0)(8,7)
\psline[linewidth=1pt]{-}(4,0)(4,7)
\psline[linewidth=0.5pt]{-}(2,0)(2,6)
\psline[linewidth=0.5pt]{-}(6,0)(6,6)
\rput[c]{0}(1,5.25){\small{$0$}}
\rput[c]{0}(3,3.75){\small{$1$}}
\rput[c]{0}(5,5.25){\small{$0$}}
\rput[c]{0}(7,3.75){\small{$0$}}
\rput[c]{0}(1,2.25){\small{$0$}}
\rput[c]{0}(3,0.75){\small{$1$}}
\rput[c]{0}(5,0.75){\small{$1$}}
\rput[c]{0}(7,2.25){\small{$0$}}
\rput[c]{0}(3,5.25){\small{$0$}}
\rput[c]{0}(1,3.75){\small{$0$}}
\rput[c]{0}(7,5.25){\small{$0$}}
\rput[c]{0}(5,3.75){\small{$1$}}
\rput[c]{0}(3,2.25){\small{$0$}}
\rput[c]{0}(1,0.75){\small{$0$}}
\rput[c]{0}(5,2.25){\small{$0$}}
\rput[c]{0}(7,0.75){\small{$0$}}
}
\rput[c](10,-7){$+c_2\cdot$}
\rput[c](13,-10){
\psline[linewidth=0.5pt]{-}(0,6)(-1,7)
\rput[c]{0}(-0.25,6.75){\tiny{$X$}}
\rput[c]{0}(-0.75,6.25){\tiny{$Y$}}
\rput[c]{0}(-0.5,7.5){\small{$U$}}
\rput[c]{0}(-1.5,6.5){\small{$V$}}
\rput[c]{0}(2,7.5){\small{$0$}}
\rput[c]{0}(6,7.5){\small{$1$}}
\rput[c]{0}(1,6.5){\small{$0$}}
\rput[c]{0}(3,6.5){\small{$1$}}
\rput[c]{0}(5,6.5){\small{$0$}}
\rput[c]{0}(7,6.5){\small{$1$}}
\rput[c]{0}(-1.5,4.5){\small{$0$}}
\rput[c]{0}(-1.5,1.5){\small{$1$}}
\rput[c]{0}(-0.5,5.25){\small{$0$}}
\rput[c]{0}(-0.5,3.75){\small{$1$}}
\rput[c]{0}(-0.5,2.25){\small{$0$}}
\rput[c]{0}(-0.5,0.75){\small{$1$}}
\psline[linewidth=1pt]{-}(-1,0)(8,0)
\psline[linewidth=1pt]{-}(-1,6)(8,6)
\psline[linewidth=1pt]{-}(-1,3)(8,3)
\psline[linewidth=0.5pt]{-}(0,1.5)(8,1.5)
\psline[linewidth=0.5pt]{-}(0,4.5)(8,4.5)
\psline[linewidth=1pt]{-}(0,0)(0,7)
\psline[linewidth=1pt]{-}(8,0)(8,7)
\psline[linewidth=1pt]{-}(4,0)(4,7)
\psline[linewidth=0.5pt]{-}(2,0)(2,6)
\psline[linewidth=0.5pt]{-}(6,0)(6,6)
\rput[c]{0}(1,5.25){\small{$0$}}
\rput[c]{0}(3,3.75){\small{$1$}}
\rput[c]{0}(5,5.25){\small{$0$}}
\rput[c]{0}(7,3.75){\small{$1$}}
\rput[c]{0}(1,2.25){\small{$0$}}
\rput[c]{0}(3,0.75){\small{$0$}}
\rput[c]{0}(5,0.75){\small{$0$}}
\rput[c]{0}(7,2.25){\small{$1$}}
\rput[c]{0}(3,5.25){\small{$0$}}
\rput[c]{0}(1,3.75){\small{$0$}}
\rput[c]{0}(7,5.25){\small{$0$}}
\rput[c]{0}(5,3.75){\small{$0$}}
\rput[c]{0}(3,2.25){\small{$1$}}
\rput[c]{0}(1,0.75){\small{$0$}}
\rput[c]{0}(5,2.25){\small{$0$}}
\rput[c]{0}(7,0.75){\small{$0$}}
}
\rput[c](23,-7){$+c_3\cdot$}
\rput[c](26,-10){
\psline[linewidth=0.5pt]{-}(0,6)(-1,7)
\rput[c]{0}(-0.25,6.75){\tiny{$X$}}
\rput[c]{0}(-0.75,6.25){\tiny{$Y$}}
\rput[c]{0}(-0.5,7.5){\small{$U$}}
\rput[c]{0}(-1.5,6.5){\small{$V$}}
\rput[c]{0}(2,7.5){\small{$0$}}
\rput[c]{0}(6,7.5){\small{$1$}}
\rput[c]{0}(1,6.5){\small{$0$}}
\rput[c]{0}(3,6.5){\small{$1$}}
\rput[c]{0}(5,6.5){\small{$0$}}
\rput[c]{0}(7,6.5){\small{$1$}}
\rput[c]{0}(-1.5,4.5){\small{$0$}}
\rput[c]{0}(-1.5,1.5){{\small$1$}}
\rput[c]{0}(-0.5,5.25){\small{$0$}}
\rput[c]{0}(-0.5,3.75){\small{$1$}}
\rput[c]{0}(-0.5,2.25){\small{$0$}}
\rput[c]{0}(-0.5,0.75){\small{$1$}}
\psline[linewidth=1pt]{-}(-1,0)(8,0)
\psline[linewidth=1pt]{-}(-1,6)(8,6)
\psline[linewidth=1pt]{-}(-1,3)(8,3)
\psline[linewidth=0.5pt]{-}(0,1.5)(8,1.5)
\psline[linewidth=0.5pt]{-}(0,4.5)(8,4.5)
\psline[linewidth=1pt]{-}(0,0)(0,7)
\psline[linewidth=1pt]{-}(8,0)(8,7)
\psline[linewidth=1pt]{-}(4,0)(4,7)
\psline[linewidth=0.5pt]{-}(2,0)(2,6)
\psline[linewidth=0.5pt]{-}(6,0)(6,6)
\rput[c]{0}(1,5.25){\small{$1$}}
\rput[c]{0}(3,3.75){\small{$0$}}
\rput[c]{0}(5,5.25){\small{$1$}}
\rput[c]{0}(7,3.75){\small{$0$}}
\rput[c]{0}(1,2.25){\small{$0$}}
\rput[c]{0}(3,0.75){\small{$0$}}
\rput[c]{0}(5,0.75){\small{$1$}}
\rput[c]{0}(7,2.25){\small{$0$}}
\rput[c]{0}(3,5.25){\small{$0$}}
\rput[c]{0}(1,3.75){\small{$0$}}
\rput[c]{0}(7,5.25){\small{$0$}}
\rput[c]{0}(5,3.75){\small{$0$}}
\rput[c]{0}(3,2.25){\small{$0$}}
\rput[c]{0}(1,0.75){\small{$1$}}
\rput[c]{0}(5,2.25){\small{$0$}}
\rput[c]{0}(7,0.75){\small{$0$}}
}
\rput[c](36,-7){$+d_1\cdot$}
\rput[c](39,-10){
\psline[linewidth=0.5pt]{-}(0,6)(-1,7)
\rput[c]{0}(-0.25,6.75){\tiny{$X$}}
\rput[c]{0}(-0.75,6.25){\tiny{$Y$}}
\rput[c]{0}(-0.5,7.5){\small{$U$}}
\rput[c]{0}(-1.5,6.5){\small{$V$}}
\rput[c]{0}(2,7.5){\small{$0$}}
\rput[c]{0}(6,7.5){\small{$1$}}
\rput[c]{0}(1,6.5){\small{$0$}}
\rput[c]{0}(3,6.5){\small{$1$}}
\rput[c]{0}(5,6.5){\small{$0$}}
\rput[c]{0}(7,6.5){\small{$1$}}
\rput[c]{0}(-1.5,4.5){\small{$0$}}
\rput[c]{0}(-1.5,1.5){\small{$1$}}
\rput[c]{0}(-0.5,5.25){\small{$0$}}
\rput[c]{0}(-0.5,3.75){\small{$1$}}
\rput[c]{0}(-0.5,2.25){\small{$0$}}
\rput[c]{0}(-0.5,0.75){\small{$1$}}
\psline[linewidth=1pt]{-}(-1,0)(8,0)
\psline[linewidth=1pt]{-}(-1,6)(8,6)
\psline[linewidth=1pt]{-}(-1,3)(8,3)
\psline[linewidth=0.5pt]{-}(0,1.5)(8,1.5)
\psline[linewidth=0.5pt]{-}(0,4.5)(8,4.5)
\psline[linewidth=1pt]{-}(0,0)(0,7)
\psline[linewidth=1pt]{-}(8,0)(8,7)
\psline[linewidth=1pt]{-}(4,0)(4,7)
\psline[linewidth=0.5pt]{-}(2,0)(2,6)
\psline[linewidth=0.5pt]{-}(6,0)(6,6)
\rput[c]{0}(1,5.25){\small{$1$}}
\rput[c]{0}(3,3.75){\small{$0$}}
\rput[c]{0}(5,5.25){\small{$1$}}
\rput[c]{0}(7,3.75){\small{$0$}}
\rput[c]{0}(1,2.25){\small{$1$}}
\rput[c]{0}(3,0.75){\small{$0$}}
\rput[c]{0}(5,0.75){\small{$0$}}
\rput[c]{0}(7,2.25){\small{$0$}}
\rput[c]{0}(3,5.25){\small{$0$}}
\rput[c]{0}(1,3.75){\small{$0$}}
\rput[c]{0}(7,5.25){\small{$0$}}
\rput[c]{0}(5,3.75){\small{$0$}}
\rput[c]{0}(3,2.25){\small{$0$}}
\rput[c]{0}(1,0.75){\small{$0$}}
\rput[c]{0}(5,2.25){\small{$1$}}
\rput[c]{0}(7,0.75){\small{$0$}}
}
\rput[c](-3,-17){$+d_4\cdot$}
\rput[c](0,-20){
\psline[linewidth=0.5pt]{-}(0,6)(-1,7)
\rput[c]{0}(-0.25,6.75){\tiny{$X$}}
\rput[c]{0}(-0.75,6.25){\tiny{$Y$}}
\rput[c]{0}(-0.5,7.5){\small{$U$}}
\rput[c]{0}(-1.5,6.5){\small{$V$}}
\rput[c]{0}(2,7.5){\small{$0$}}
\rput[c]{0}(6,7.5){\small{$1$}}
\rput[c]{0}(1,6.5){\small{$0$}}
\rput[c]{0}(3,6.5){\small{$1$}}
\rput[c]{0}(5,6.5){\small{$0$}}
\rput[c]{0}(7,6.5){\small{$1$}}
\rput[c]{0}(-1.5,4.5){\small{$0$}}
\rput[c]{0}(-1.5,1.5){\small{$1$}}
\rput[c]{0}(-0.5,5.25){\small{$0$}}
\rput[c]{0}(-0.5,3.75){\small{$1$}}
\rput[c]{0}(-0.5,2.25){\small{$0$}}
\rput[c]{0}(-0.5,0.75){\small{$1$}}
\psline[linewidth=1pt]{-}(-1,0)(8,0)
\psline[linewidth=1pt]{-}(-1,6)(8,6)
\psline[linewidth=1pt]{-}(-1,3)(8,3)
\psline[linewidth=0.5pt]{-}(0,1.5)(8,1.5)
\psline[linewidth=0.5pt]{-}(0,4.5)(8,4.5)
\psline[linewidth=1pt]{-}(0,0)(0,7)
\psline[linewidth=1pt]{-}(8,0)(8,7)
\psline[linewidth=1pt]{-}(4,0)(4,7)
\psline[linewidth=0.5pt]{-}(2,0)(2,6)
\psline[linewidth=0.5pt]{-}(6,0)(6,6)
\rput[c]{0}(1,5.25){\small{$0$}}
\rput[c]{0}(3,3.75){\small{$1$}}
\rput[c]{0}(5,5.25){\small{$0$}}
\rput[c]{0}(7,3.75){\small{$1$}}
\rput[c]{0}(1,2.25){\small{$0$}}
\rput[c]{0}(3,0.75){\small{$1$}}
\rput[c]{0}(5,0.75){\small{$0$}}
\rput[c]{0}(7,2.25){\small{$0$}}
\rput[c]{0}(3,5.25){\small{$0$}}
\rput[c]{0}(1,3.75){\small{$0$}}
\rput[c]{0}(7,5.25){\small{$0$}}
\rput[c]{0}(5,3.75){\small{$0$}}
\rput[c]{0}(3,2.25){\small{$0$}}
\rput[c]{0}(1,0.75){\small{$0$}}
\rput[c]{0}(5,2.25){\small{$0$}}
\rput[c]{0}(7,0.75){\small{$1$}}
}
\rput[c](27,-17){$+(1$-$a_2$-$a_3$-$b_2$-$b_3$-$c_2$-$c_3$-$d_1$-$d_4)\cdot$}
\rput[c](39,-20){
\psline[linewidth=0.5pt]{-}(0,6)(-1,7)
\rput[c]{0}(-0.25,6.75){\tiny{$X$}}
\rput[c]{0}(-0.75,6.25){\tiny{$Y$}}
\rput[c]{0}(-0.5,7.5){\small{$U$}}
\rput[c]{0}(-1.5,6.5){\small{$V$}}
\rput[c]{0}(2,7.5){\small{$0$}}
\rput[c]{0}(6,7.5){\small{$1$}}
\rput[c]{0}(1,6.5){\small{$0$}}
\rput[c]{0}(3,6.5){\small{$1$}}
\rput[c]{0}(5,6.5){\small{$0$}}
\rput[c]{0}(7,6.5){\small{$1$}}
\rput[c]{0}(-1.5,4.5){\small{$0$}}
\rput[c]{0}(-1.5,1.5){\small{$1$}}
\rput[c]{0}(-0.5,5.25){\small{$0$}}
\rput[c]{0}(-0.5,3.75){\small{$1$}}
\rput[c]{0}(-0.5,2.25){\small{$0$}}
\rput[c]{0}(-0.5,0.75){\small{$1$}}
\psline[linewidth=1pt]{-}(-1,0)(8,0)
\psline[linewidth=1pt]{-}(-1,6)(8,6)
\psline[linewidth=1pt]{-}(-1,3)(8,3)
\psline[linewidth=0.5pt]{-}(0,1.5)(8,1.5)
\psline[linewidth=0.5pt]{-}(0,4.5)(8,4.5)
\psline[linewidth=1pt]{-}(0,0)(0,7)
\psline[linewidth=1pt]{-}(8,0)(8,7)
\psline[linewidth=1pt]{-}(4,0)(4,7)
\psline[linewidth=0.5pt]{-}(2,0)(2,6)
\psline[linewidth=0.5pt]{-}(6,0)(6,6)
\rput[c]{0}(1,5.25){\scriptsize{$\nicefrac{1}{2}$}}
\rput[c]{0}(3,3.75){\scriptsize{$\nicefrac{1}{2}$}}
\rput[c]{0}(5,5.25){\scriptsize{$\nicefrac{1}{2}$}}
\rput[c]{0}(7,3.75){\scriptsize{$\nicefrac{1}{2}$}}
\rput[c]{0}(1,2.25){\scriptsize{$\nicefrac{1}{2}$}}
\rput[c]{0}(3,0.75){\scriptsize{$\nicefrac{1}{2}$}}
\rput[c]{0}(5,0.75){\scriptsize{$\nicefrac{1}{2}$}}
\rput[c]{0}(7,2.25){\scriptsize{$\nicefrac{1}{2}$}}
\rput[c]{0}(3,5.25){\small{$0$}}
\rput[c]{0}(1,3.75){\small{$0$}}
\rput[c]{0}(7,5.25){\small{$0$}}
\rput[c]{0}(5,3.75){\small{$0$}}
\rput[c]{0}(3,2.25){\small{$0$}}
\rput[c]{0}(1,0.75){\small{$0$}}
\rput[c]{0}(5,2.25){\small{$0$}}
\rput[c]{0}(7,0.75){\small{$0$}}
}
\endpspicture
\caption{\label{fig:partition-single-box}The optimal non-signalling partition of a bipartite system with binary inputs and outputs.}
\end{figure}

In the above non-signalling partition, with probability $4\ep$, the outcome $z$ is such that $P_{XY|UV}^z$ is local deterministic. We define the \emph{local part} as the maximum weight  a local system can have in a non-signalling partition. 
\begin{definition}
Let $P_{\bof{X}|\bof{U}}$ be an $n$-party non-signalling system. The \emph{local part} of $P_{\bof{X}|\bof{U}}$ is the maximal $p$ such that 
\begin{align}
\nonumber P_{\bof{X}|\bof{U}} &=p\cdot P^{\mathrm{local}}_{\bof{X}|\bof{U}}+(1-p)\cdot P^{\mathrm{n-s}}_{\bof{X}|\bof{U}}
\end{align}
and where $P^{\mathrm{local}}_{\bof{X}|\bof{U}}$ is an ($n$-party) local system and $P^{\mathrm{n-s}}_{\bof{X}|\bof{U}}$ is an ($n$-party) non-signalling system. 
\end{definition}

For any bit obtained from the outputs of a non-signalling system, the local part is a lower bound on the distance from uniform as seen from a non-signalling adversary. 
\begin{lemma}\label{lemma:dgeqlocal}
Let $P_{\bof{X}|\bof{U}}$ be a system with local part $p$. 
Then for any function $f\colon \mathcal{\bof{X}}\rightarrow \{0,1\}$ such that $B=f(\bof{X})$ and  $Q=(\bof{U}=\bof{u}, F=f)$, 
\begin{align}
\nonumber d(B|Z(W_{\mathrm{n-s}}),Q) &\geq \frac{1}{2}\cdot p\ .
\end{align}
\end{lemma}
\begin{proof}
Any local system can be expressed as a convex combination of local deterministic systems. 
A system $P_{\bof{X}|\bof{U}}$ with local part $p$, therefore, has a non-sig\-nal\-ling partition $w$ such that with probability 
$p$, $z$ is such that $P^z_{\bof{X}|\bof{U}}$ is local deterministic. For any local deterministic $P^z_{\bof{X}|\bof{U}}$, the output $\bof{X}$  and, therefore, also $B=f(\bof{X})$ is a deterministic function of $\bof{U}$. Therefore, 
\begin{align}
\nonumber d(B|Z(W_{\mathrm{n-s}}),Q)&\geq  d(B|Z(w),Q) 
\\
\nonumber &=
\frac{1}{2} \sum\limits_{b} \sum\limits_{z_w} p^{z_w}\cdot \biggl|\sum_{\bof{x}:f(\bof{x})=b} P_{\bof{X}|\bof{U}}^{z_w}(\bof{x},\bof{u})-\frac{1}{2} \biggr|\\
\nonumber &= 
\frac{1}{2} \sum\limits_{z_w} p^{z_w}\cdot \biggl|\sum_{\bof{x}}(-1)^{f(\bof{x})} P_{\bof{X}|\bof{U}}^{z_w}(\bof{x},\bof{u})
\biggr|
\\
\nonumber &\geq  
\frac{1}{2} \sum\limits_{z_w\ \mathrm{local}} p^{z_w}
=
\frac{1}{2} \cdot p \ . \qedhere
\end{align}
\end{proof}

\begin{lemma}
The local part of a system $P_{XY|UV}$ with $\mathcal{X}=\mathcal{Y}=\mathcal{U}=\mathcal{V}=\{0,1\}$ and 
$\sum_{x\oplus y=u\cdot v}P_{XY|UV}(x,y,u,v)/4=1-\ep$ 
where $\ep\leq 1/4$ is $4\ep$.
\end{lemma}
\begin{proof}
That this value can be reached follows from the non-signalling partition given in Figure~\ref{fig:partition-single-box}. The optimality of this value follows from Lemma~\ref{lemma:dgeqlocal} and the bound on the distance from uniform of the bit $X$ given in Lemma~\ref{lemma:guessing_probability_single_box}, p.~\pageref{lemma:guessing_probability_single_box}. 
\end{proof}

\section{A Special Case: Two Unbiased PR~Boxes with Error $\ep$}\label{sec:twobox}

In the previous section, we have studied the local part of a bipartite system with binary inputs and outputs. 
In this section, we study the special case of two unbiased PR~boxes with error $\ep$. We show that the local part remains the same as for the case of a single unbiased PR~box with error $\ep$. 
By Lemma~\ref{lemma:dgeqlocal}, this implies directly that privacy amplification of the outputs of two systems is impossible, independently of the function that is applied. The fact that the local part of several systems can be significantly higher than what would be expected when the local part of each individual system is analysed could give an intuition why privacy amplification 
of non-signalling secrecy is impossible also for an arbitrary number of systems, as we will see in Section~\ref{sec:impossibilityseveral}.

\begin{lemma}\label{lemma:two_symm_boxes}
The local part of $P^{2,\varepsilon}_{XY|UV}$ is $4\ep$. 
\end{lemma}
\begin{proof}
 The local part of two unbiased PR~boxes with error $\ep$
 cannot be larger than
 $4\varepsilon$ as this would contradict the fact that $4\varepsilon$ is the local part
 of a single PR~box with error $\ep$. 
To see that this value can be reached we provide
 an explicit non-signalling partition:  
With probability $4\ep-8\ep^2$, the system is one of the $64$ local deterministic strategies which can be obtained from 
the strategy
\begin{align}
\nonumber u_1u_2 &\rightarrow x_1x_2: &\quad 00 \mapsto 00,\ 
01 \mapsto 00,\ 
10 \mapsto 00,\ 
11 \mapsto 01\\ 
\nonumber 
v_1v_2 &\rightarrow y_1y_2: &\quad 00 \mapsto 00,\ 
01 \mapsto 00,\ 
10 \mapsto 10,\ 
11 \mapsto 00
\end{align}
by depolarization (see Section~\ref{sec:qchsh}, p.~\pageref{sec:qchsh}). With probability $8\ep^2$, it is one of the $64$ local deterministic strategies which can be obtained from 
\begin{align}
\nonumber u_1u_2 &\rightarrow x_1x_2: &\quad 00 \mapsto 00,\ 
01 \mapsto 00,\ 
10 \mapsto 00,\ 
11 \mapsto 00\\ 
\nonumber 
v_1v_2 &\rightarrow y_1y_2: &\quad 00 \mapsto 00,\ 
01 \mapsto 00,\ 
10 \mapsto 00,\ 
11 \mapsto 00
\end{align}
by depolarization.
With probability $1-4\ep$ it is two PR~boxes (see also Figure~\ref{fig:localparttwosystb}). 
\end{proof}

\begin{figure}[h]
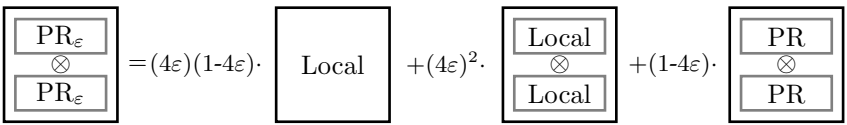

\centering
\psset{unit=0.3cm}
\pspicture*[](-1.2,-8.1)(37.5,0)
\rput[c]{0}(4.9,-4.5){\small{$=$}}
\rput[c]{0}(7.9,-4.5){\small{$(4\varepsilon)(1$-$4\varepsilon)\cdot$}}
\rput[c]{0}(18.5,-4.5){\small{$+(4\varepsilon)^2\cdot$}}
\rput[c]{0}(28.5,-4.5){\small{$+(1$-$4\varepsilon)\cdot$}}
\rput[c]{0}(-2,0){
 \pspolygon[linewidth=1pt](1,-7)(6,-7)(6,-2)(1,-2)
 \pspolygon[linewidth=1pt,linecolor=gray](1.5,-6.5)(5.5,-6.5)(5.5,-5)(1.5,-5)
 \pspolygon[linewidth=1pt,linecolor=gray](1.5,-4)(5.5,-4)(5.5,-2.5)(1.5,-2.5)
\rput[c]{0}(3.5,-4.5){$\otimes$}
\rput[c]{0}(3.5,-3.25){$\mathrm{PR}_{\varepsilon}$}
\rput[c]{0}(3.5,-5.75){$\mathrm{PR}_{\varepsilon}$}
}
\rput[c]{0}(10,0){
 \pspolygon[linewidth=1pt](1,-7)(6,-7)(6,-2)(1,-2)
\rput[c]{0}(3.5,-4.5){$\mathrm{Local}$}
}
\rput[c]{0}(20,0){
 \pspolygon[linewidth=1pt](1,-7)(6,-7)(6,-2)(1,-2)
 \pspolygon[linewidth=1pt,linecolor=gray](1.5,-6.5)(5.5,-6.5)(5.5,-5)(1.5,-5)
 \pspolygon[linewidth=1pt,linecolor=gray](1.5,-4)(5.5,-4)(5.5,-2.5)(1.5,-2.5)
\rput[c]{0}(3.5,-4.5){$\otimes$}
\rput[c]{0}(3.5,-3.25){$\mathrm{Local}$}
\rput[c]{0}(3.5,-5.75){$\mathrm{Local}$}
}
\rput[c]{0}(30,0){
 \pspolygon[linewidth=1pt](1,-7)(6,-7)(6,-2)(1,-2)
 \pspolygon[linewidth=1pt,linecolor=gray](1.5,-6.5)(5.5,-6.5)(5.5,-5)(1.5,-5)
 \pspolygon[linewidth=1pt,linecolor=gray](1.5,-4)(5.5,-4)(5.5,-2.5)(1.5,-2.5)
\rput[c]{0}(3.5,-4.5){$\otimes$}
\rput[c]{0}(3.5,-3.25){$\mathrm{PR}$}
\rput[c]{0}(3.5,-5.75){$\mathrm{PR}$}
}
\endpspicture
\caption{\label{fig:localparttwosystb} The local part of two systems is as large as the one of a single system. Some of the local deterministic strategies correspond to independent  local strategies for each of the two systems, while others are joint strategies for the two systems. }
\end{figure}

This has direct consequences for the amount of (non-signalling) secrecy which can be extracted from the
outputs of two unbiased PR~boxes with error $\ep$. In fact, it is not possible to apply a (public) function to the outputs of two systems such that the resulting bit is more secret than the output of a single system. 
\begin{lemma}
 Assume a system $P_{XY|UV}^{2,\ep}$ and $Q=(U=u,V=v,F=f)$ with $f\colon \{0,1\}^2\rightarrow \{0,1\}$. Then
\begin{align}
\nonumber d(f({X})|Z(W_{\mathrm{n-s}}),Q) &\geq  2\ep \ .
\end{align}
\end{lemma}
\begin{proof}
This follows directly from Lemmas~\ref{lemma:two_symm_boxes} and~\ref{lemma:dgeqlocal}. 
\end{proof}

The above result also implies, that by applying a function to the inputs and outputs of two unbiased PR~boxes with error $\ep$, it is not possible to create an unbiased PR~box with error $\ep^{\prime}$, where $\ep^{\prime}<\ep$. This fact was already known, even when not restricting the transformations to the application of functions~\cite{short}.

If the local part is large, we know that the distance from uniform of any bit we can extract from this system is also large. However, as it has been shown in~\cite{localpart}, the local part of $n$ unbiased PR~boxes with error $\ep$ behaves as $O(2^{\lfloor n/2 \rfloor})$. If we want to show that the distance from uniform of a bit extracted from any number of systems is always high, we, therefore, need to give a different attack than the one determined by the local part.

\section{Several Systems}\label{sec:impossibilityseveral}

\subsection{The general optimal attack on a bit}

What is the best attack a non-signalling adversary can do on a single bit which is obtained from the outcome of a non-signalling system with public inputs? According to Lemma~\ref{lemma:distanceislp}, p.~\pageref{lemma:distanceislp}, this corresponds to finding the non-signalling partition with two outputs $z_0$ and $z_1$ such that for $P^{z_0}_{XY|UV}$ the bit $B$ is maximally biased towards $0$ while for $P^{z_1}_{XY|UV}$ it is maximally biased towards $1$. This optimization can be expressed as a linear programming problem. 

\begin{lemma}\label{lemma:distanceofbitislp}
Let $P_{XYZ|UVW}$ be a tripartite non-signalling system. 
The distance from uniform of $B=f(X)\in \{0,1\}$ given $Z(W_{\mathrm{n-s}})$ and $Q:=({U}={u},{V}={v},F=f)$ is given by the 
 optimal value of the following optimization problem (we drop the index of the probability distribution in the notation). 
 
\begin{align}
\nonumber  \max : &\quad \frac{1}{2}\cdot 
 \Biggl[
p^{z_0}\cdot \Bigl(
 \sum_{({x},{y}):B=0} P^{z_0}({x,y},{u,v})
 - \sum_{({x},{y}):B=1} P^{z_0}({x,y},{u,v})
 \Bigr) 
  \\
\nonumber &\quad
 + 
 p^{z_1}\cdot \Bigl(
 \sum_{({x},{y}):B=1} P^{z_1}({x,y},{u,v})
 - \sum_{({x},{y}):B=0} P^{z_1}({x,y},{u,v})
 \Bigr)
 \Biggr]
\\
\nonumber \st &\quad
 \sum_{{x}}P^{z_0}({x,y},{u,v})-\sum_{{x}} P^{z_0}({x,y},{u^{\prime},v})=0
\\
 \nonumber &\quad \sum_{{x}}P^{z_1}({x,y},{u,v})-\sum_{{x}} P^{z_1}({x,y},{u^{\prime},v})=0
\\
\nonumber &\quad \sum_{{y}}P^{z_0}({x,y},{u,v})-\sum_{{y}} P^{z_0}({x,y},{u,v^{\prime}})=0
\\
\nonumber &\quad \sum_{{y}}P^{z_1}({x,y},{u,v})-\sum_{{y}} P^{z_1}({x,y},{u,v^{\prime}})=0
\\
\nonumber &\quad p^{z_0}\cdot P^{z_0}({x,y},{u,v}) \geq 0\\
\nonumber &\quad p^{z_1}\cdot P^{z_1}({x,y},{u,v}) \geq 0\\
\nonumber &\quad p^{z_0}\cdot P^{z_0}({x,y},{u,v}) +p^{z_1}\cdot P^{z_1}({x,y},{u,v}) =P({x,y},{u,v})\\
\nonumber &\quad \text{for all } {x},{y},{u},{u^{\prime}},{v},{v^{\prime}}\ .
\end{align}
\end{lemma}
\begin{proof}
This follows directly from Lemma~\ref{lemma:distanceislp}, p.~\pageref{lemma:distanceislp} and the definition of a tripartite non-signalling system. 
\end{proof}
Note that when expressed in terms of the variables ${P^{\prime}}^{z_0}({x,y},{u,v})= p^{z_0}\cdot P^{z_0}({x,y},{u,v})$ and 
${P^{\prime}}^{z_1}({x,y},{u,v})= p^{z_1}\cdot P^{z_1}({x,y},{u,v})$ this is a linear program.

\subsection{A concrete (good) adversarial strategy}
\label{subsec:wbar}

We now describe a special non-signalling partition $\bar{w}$ of the system \linebreak[4] $P_{XY|UV}^{n,\ep}$,  which gives a large distance from uniform of the key bit $B=f({X})$.
The non-signalling partition is of the form (see also Figure~\ref{fig:impboxpartitionb})
\begin{align}
\nonumber P_{XY|UV}^{n,\ep} &=  \frac{1}{2}\cdot P^{\bar{z}_0}_{{XY}|{UV}}+ \frac{1}{2}\cdot P^{\bar{z}_1}_{{XY}|{UV}}\ .
\end{align}
It will, therefore, be enough to give $P^{\bar{z}_0}_{{XY}|{UV}}$ and to show that \linebreak[4]
 $(1/2,P^{\bar{z}_0}_{{XY}|{UV}})$ is an element of a non-signalling partition of $P_{XY|UV}^{n,\ep}$. 
 
\begin{figure}[h]
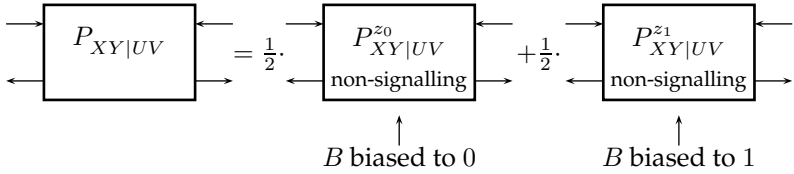

\centering
\pspicture*[](-1.1,-1)(10.2,1.6)
\psset{unit=1cm}
\rput[c]{0}(-0.5,0){
\pspolygon[linewidth=1pt](0,0)(0,1.25)(2,1.25)(2,0)
\rput[c]{0}(1,0.75){$P^{\phantom{z}}_{XY|UV}$}
\psline[linewidth=0.5pt]{->}(-0.5,1)(0,1)
\psline[linewidth=0.5pt]{<-}(-0.5,0.25)(0,0.25)
\psline[linewidth=0.5pt]{<-}(2,1)(2.5,1)
\psline[linewidth=0.5pt]{->}(2,0.25)(2.5,0.25)
}
\rput[c]{0}(3.2,0){
\psline[linewidth=0.5pt]{->}(1,-0.6)(1,-0.2)
\rput[c]{0}(1,-0.75){$B$ biased to $0$}
\rput[c]{0}(-0.85,0.6){$= \frac{1}{2}\cdot $}
\pspolygon[linewidth=1pt](0,0)(0,1.25)(2,1.25)(2,0)
\rput[c]{0}(1,0.75){$P^{z_0}_{XY|UV}$}
\rput[c]{0}(1,0.25){\footnotesize{non-signalling}}
\psline[linewidth=0.5pt]{->}(-0.5,1)(0,1)
\psline[linewidth=0.5pt]{<-}(-0.5,0.25)(0,0.25)
\psline[linewidth=0.5pt]{<-}(2,1)(2.5,1)
\psline[linewidth=0.5pt]{->}(2,0.25)(2.5,0.25)
}
\rput[c]{0}(6.9,0){
\psline[linewidth=0.5pt]{->}(1,-0.6)(1,-0.2)
\rput[c]{0}(1,-0.75){$B$ biased to $1$}
\rput[c]{0}(-0.85,0.6){$+ \frac{1}{2}\cdot $}
\pspolygon[linewidth=1pt](0,0)(0,1.25)(2,1.25)(2,0)
\rput[c]{0}(1,0.75){$P^{z_1}_{XY|UV}$}
\rput[c]{0}(1,0.25){\footnotesize{non-signalling}}
\psline[linewidth=0.5pt]{->}(-0.5,1)(0,1)
\psline[linewidth=0.5pt]{<-}(-0.5,0.25)(0,0.25)
\psline[linewidth=0.5pt]{<-}(2,1)(2.5,1)
\psline[linewidth=0.5pt]{->}(2,0.25)(2.5,0.25)
}
\endpspicture
\caption{\label{fig:impboxpartitionb}The successful attack in the tripartite non-signalling case is such that, with probability $1/2$, Eve obtains an outcome such that the bit $B$ is biased to $0$. }
\end{figure} 
 
The probabilities $P^{\bar{z}_0}({x},{y},{u},{v})$ are defined in four cases, according to the values of ${x}$ and ${y}$ and the properties of the system $P_{{XY}|{UV}}$. 
For simplicity, let us use the following notation:
\begin{align}
\nonumber
{y}_<&{:=} \biggl\{ {y} \biggm| \sum\limits_{{x}| f({x})=0} P({{x,y},{u,v}})<\sum\limits_{{x}|f({x})=1} P({{x,y},{u,v}}) \biggr\}\ ,\\
\nonumber
{y}_> &{:=} \bigg\{ {y}\biggm| \sum\limits_{{x}|f({x})=0} P({{x,y},{u,v}})>\sum\limits_{{x}|f({x})=1} P({{x,y},{u,v}}) \biggr\} \ ,\\
\nonumber
{x}_0 &{:=} \{{x}|f({x})=0\}\ ,\\
\nonumber
{x}_1 &{:=} \{{x}|f({x})=1\}\ .
\end{align}

\begin{definition}\label{def:pz0}
For a given system $P_{XY|UV}$ and function $f\colon \mathcal{X}\rightarrow \{0,1\}$, the system 
\emph{$P^{\bar{z}_0}_{XY|UV}$} is defined as (see Figure~\ref{fig:intuitionb})
\begin{align}
\nonumber P^{\bar{z}_0}_{XY|UV}({x,y},{u,v})&:= c({x,y},{u,v})\cdot P_{XY|UV}({{x,y},{u,v}})\ ,
\end{align}
where the factor $c({x,y},{u,v})$ is defined as follows. \\ 
\begin{align}
\nonumber &\text{For all }{x}\in {x}_0,{y} \in {y}_<, &&c({x,y},{u,v}):=2\ .\\
\nonumber &\text{For all }{x}\in {x}_1,{y}\in {y}_<, &&c({x,y},{u,v}):=\frac{\sum\limits_{{x}}(-1)^{(f(x)+1)} P({{x,y},{u,v}})}{\sum\limits_{{x}:f({x})=1} P({{x,y},{u,v}})}
\ .\\
\nonumber &\text{For all }{x}\in {x}_0,{y}\in {y}_>, &&c({x,y},{u,v}):=\frac{\sum\limits_{{x}} P({{x,y},{u,v}})}{\sum\limits_{{x}:f({x})=0} P({{x,y},{u,v}})}\ .\\
\nonumber &\text{For all }{x}\in {x}_1,{y}\in {y}_>, &&c({x,y},{u,v}):=0\ .
\end{align}
\end{definition}

\begin{lemma}\label{lemma:pz0nonsig} 
For $P_{XY|UV}^{n,\ep}$ and any $f\colon \mathcal{X}\rightarrow \{0,1\}$, $P^{\bar{z}_0}_{XY|UV}$ is a non-signalling system. 
\end{lemma}
\begin{proof}\emph{\phantom{A }}\\
For all ${u},{v}$ and ${y}\in {y}_<$:
\begin{align}
\label{ns-AtoB1} \sum\limits_{{x}} P^{\bar{z}_0}({x,y},{u,v}) &=\sum\limits_{{x}:f({x})=0}2\cdot P({{x,y},{u,v}})
\\
\nonumber &\quad + \sum\limits_{{x}:f({x})=1}\frac{
\sum\limits_{{x^{\prime}}}(-1)^{(f(x^{\prime})+1)} P({{x^{\prime},y},{u,v}})
}{\sum\limits_{{x^{\prime}}:f({x^{\prime}})=1} P({{x^{\prime},y},{u,v}})}
\cdot P({{x,y},{u,v}})
\\
\nonumber &= 2 \sum\limits_{{x}:f({x})=0} P({{x,y},{u,v}})+\sum\limits_{{x}}(-1)^{(f(x)+1)} P({{x,y},{u,v}})
\\
\nonumber &= \sum\limits_{{x}} P({{x,y},{u,v}})=\frac{1}{2^n}\ .
\end{align}
For all ${u},{v}$ and ${y}\in {y}_>$:
\begin{align}
\label{ns-AtoB2} \sum\limits_{{x}} P^{\bar{z}_0}({x,y},{u,v}) &= 
\sum\limits_{{x}:f({x})=1}0 
\\ \nonumber &\quad 
+
\sum\limits_{{x}:f({x})=0}\frac{\sum\limits_{{x^{\prime}}} P({{x^{\prime},y},{u,v}})
}{\sum\limits_{{x^{\prime}}:f({x^{\prime}})=0} P({{x^{\prime},y},{u,v}})}
\cdot P({{x,y},{u,v}})\\
\nonumber &= \sum\limits_{{x}:f({x})=1} P({{x,y},{u,v}})+\sum\limits_{{x}:f({x})=0} P({{x,y},{u,v}})=\frac{1}{2^n}\ .
\end{align}
\\
For the non-signalling condition in the other direction, note that
\begin{align}
\nonumber P({{x,y},{u,v^{\prime}}}) &= 
P({{x,y^{\prime}},{u,v}})\ ,
\end{align}
where the $i$\textsuperscript{th} bit of ${y^{\prime}}$ is defined as 
$y^{\prime}_i:=y_i\oplus u_i\cdot(v^{\prime}_i-v_i)$. 
Therefore, for all ${x}$, ${u}$, ${v^{\prime}}$:
\begin{align}
\nonumber \sum\limits_{{y}} P^{\bar{z}_0}({x,y},{u,v^{\prime}}) &=\sum\limits_{{y^{\prime}}} P^{\bar{z}_0}({x,y^{\prime}},{u,v})=\sum\limits_{{y}} P^{\bar{z}_0}({x,y},{u,v})\ .
\end{align}
Finally, the normalization follows directly from (\ref{ns-AtoB1}) and (\ref{ns-AtoB2}):
\begin{align}
\nonumber \sum\limits_{{x},{y}}  P^{\bar{z}_0}({x,y},{u,v}) &=\sum\limits_{{y}}\biggl( \sum\limits_{{x}} P^{\bar{z}_0}({x,y},{u,v})\biggr)=\sum\limits_{{y}}\frac{1}{2^n}=1\ . \qedhere
\end{align}
\end{proof}
\begin{figure}
\centering
\psset{unit=0.525cm}
\pspicture*[](-4,-2)(10.5,8.75)
\rput[c]{0}(6,0){
\pspolygon[linewidth=0pt,fillstyle=hlines,hatchcolor=lightgray](0,0)(2,0)(2,6)(0,6)
\psline[linewidth=1pt,linecolor=gray]{->}(1,7)(1,8)
\rput[c]{0}(1,8.5){\color{gray}{$B=0$}}
}
\pspolygon[linewidth=0pt,fillstyle=hlines,hatchcolor=lightgray](0,0)(2,0)(2,6)(0,6)
\psline[linewidth=1pt,linecolor=gray]{->}(1,7)(1,8)
\rput[c]{0}(1,8.5){\color{gray}{$B=0$}}
\psline[linewidth=0.5pt]{-}(0,6)(-1,7)
\rput[c]{0}(10,3){\Large{$\cdots$}}
\rput[c]{0}(4,-1.5){\Large{$\vdots$}}
\rput[c]{0}(-0.25,6.75){\scriptsize{$X$}}
\rput[c]{0}(-0.75,6.25){\scriptsize{$Y$}}
\rput[c]{0}(-0.5,7.5){\large{$U$}}
\rput[c]{0}(-1.5,6.5){\large{$V$}}
\rput[c]{0}(4,7.5){\Large{$00$}}
\rput[c]{0}(1,6.5){\Large{$00$}}
\rput[c]{0}(3,6.5){\Large{$01$}}
\rput[c]{0}(5,6.5){\Large{$10$}}
\rput[c]{0}(7,6.5){\Large{$11$}}
\rput[c]{0}(-1.7,3){\Large{$00$}}
\rput[c]{0}(-0.6,5.25){\Large{$00$}}
\rput[c]{0}(-0.6,3.75){\Large{$01$}}
\rput[c]{0}(-0.6,2.25){\Large{$10$}}
\rput[c]{0}(-0.6,0.75){\Large{$11$}}
\psline[linewidth=2pt]{-}(-1,0)(9,0)
\psline[linewidth=2pt]{-}(-1,6)(9,6)
\psline[linewidth=1pt]{-}(0,3)(9,3)
\psline[linewidth=1pt]{-}(0,1.5)(9,1.5)
\psline[linewidth=1pt]{-}(0,4.5)(9,4.5)
\psline[linewidth=2pt]{-}(0,-1)(0,7)
\psline[linewidth=2pt]{-}(8,-1)(8,7)
\psline[linewidth=1pt]{-}(4,-1)(4,6)
\psline[linewidth=1pt]{-}(2,-1)(2,6)
\psline[linewidth=1pt]{-}(6,-1)(6,6)
\rput[c]{0}(1,5.25){{$\frac{(1-\ep)^2}{4}$}}
\rput[c]{0}(3,3.75){{$\frac{(1-\ep)^2}{4}$}}
\rput[c]{0}(5,5.25){{$\frac{\ep-\ep^2}{4}$}}
\rput[c]{0}(7,3.75){{$\frac{\ep-\ep^2}{4}$}}
\rput[c]{0}(1,2.25){{$\frac{\ep-\ep^2}{4}$}}
\rput[c]{0}(3,0.75){{$\frac{\ep-\ep^2}{4}$}}
\rput[c]{0}(5,0.75){{$\frac{\ep-\ep^2}{4}$}}
\rput[c]{0}(7,2.25){{$\frac{\ep-\ep^2}{4}$}}
\rput[c]{0}(3,5.25){{$\frac{\ep-\ep^2}{4}$}}
\rput[c]{0}(1,3.75){{$\frac{\ep-\ep^2}{4}$}}
\rput[c]{0}(7,5.25){{$\frac{\ep^2}{4}$}}
\rput[c]{0}(5,3.75){{$\frac{\ep^2}{4}$}}
\rput[c]{0}(3,2.25){{$\frac{\ep^2}{4}$}}
\rput[c]{0}(1,0.75){{$\frac{\ep^2}{4}$}}
\rput[c]{0}(5,2.25){{$\frac{(1-\ep)^2}{4}$}}
\rput[c]{0}(7,0.75){{$\frac{(1-\ep)^2}{4}$}}
\psline[linewidth=2pt,linecolor=darkgray,arrowscale=1.5]{<-}(1.4,0.75)(2.4,0.75)
\psline[linewidth=2pt,linecolor=darkgray,arrowscale=1.5]{<-}(1.4,2.25)(2.4,2.25)
\psline[linewidth=2pt,linecolor=darkgray,arrowscale=1.5]{<-}(1.4,3.75)(2.4,3.75)
\psline[linewidth=2pt,linecolor=darkgray,arrowscale=1.5]{<-}(1.4,5.25)(2.4,5.25)
\psline[linewidth=2pt,linecolor=darkgray,arrowscale=1.5]{->}(5.4,0.75)(6.4,0.75)
\psline[linewidth=2pt,linecolor=darkgray,arrowscale=1.5]{->}(5.4,2.25)(6.4,2.25)
\psline[linewidth=2pt,linecolor=darkgray,arrowscale=1.5]{->}(5.4,3.75)(6.4,3.75)
\psline[linewidth=2pt,linecolor=darkgray,arrowscale=1.5]{->}(5.4,5.25)(6.4,5.25)
\endpspicture
\caption{\label{fig:intuitionb} The intuition for the construction of $P_{XY|UV}^{z_0}$ from $P_{XY|UV}$: For each value of $y$, move as much probability as possible from values mapped to $1$ to values mapped to $0$. 
}
\end{figure}

\begin{lemma}\label{lemma:z0boxpartition}
There exists a non-signalling partition of $P^{n,\ep}_{XY|UV}$ with an element $({1}/{2},P^{\bar{z}_0}_{XY|UV})$.
\end{lemma}
\begin{proof}
Lemma~\ref{lemma:pz0nonsig}  implies that $P^{\bar{z}_0}_{XY|UV}$ is a non-sig\-nal\-ling system. 
The criterion for an element of a non-signalling partition is given in Lem\-ma~\ref{lemma:zweihi}, p.~\pageref{lemma:zweihi}, which for the case $p={1}/{2}$ translates to the constraint $P^{\bar{z}_0}({x,y},{u,v})\leq 2 P({{x,y},{u,v}})$, and which is satisfied due to the definition of $c({x,y},{u,v})$.
\end{proof}

Defining the complementary system as $P^{\bar{z}_1}({x,y},{u,v})=2 P({{x,y},{u,v}})-P^{\bar{z}_1}({x,y},{u,v})$, we obtain a non-signalling partition of $P^{n,\ep}_{{XY}|{UV}}$, by
\begin{align}
\nonumber P^{n,\ep}_{{XY}|{UV}} &=\frac{1}{2}P^{\bar{z}_0}({x,y},{u,v}) +\frac{1}{2}P^{\bar{z}_1}({x,y},{u,v})\ .
\end{align}

\begin{definition}\label{def:wbar}
The \emph{non-signalling partition $\bar{w}$} of $P^{n,\ep}_{{XY}|{UV}}$ is 
\begin{align}
\nonumber &\left\{\biggl(\frac{1}{2},P^{\bar{z}_0}_{XY|UV}\biggr),\biggl(\frac{1}{2},2\cdot P^{n,\ep}_{{XY}|{UV}}-P^{\bar{z}_0}_{XY|UV}\biggr)\right\}_{\bar{z}}\ .
\end{align}
\end{definition}
We can now calculate the distance from uniform of the bit $B=f(X)$ that can be reached by this non-signalling partition. 
\begin{lemma}\label{lemma:distancewbar}
Consider the non-signalling system $P^{n,\ep}_{{XY}|{UV}}$. 
The distance from uniform of $B=f(X)\in\{0,1\}$ given $Z(\bar{w})$ and $Q=(U=u,V=v,F=f)$ is 
\begin{align}
\nonumber d(B|&Z(\bar{w}),Q)
\\
\nonumber &= \max \left\{\frac{1}{2}\cdot 
 \biggl|\sum\limits_{{(x,y)}:f({x})=0} P^{n,\ep}({{x,y},{u,v}})
 -
\sum\limits_{{(x,y)}:f({x})=1} P^{n,\ep}({{x,y},{u,v}}) \biggr|,
\right.
\\
\nonumber  
&\quad \left.
 \sum\limits_{{y}} \min\Biggl\{\sum\limits_{{x}:f({x})=0} P^{n,\ep}({{x,y},{u,v}}),\sum\limits_{{x}:f({x})=1} P^{n,\ep}({{x,y},{u,v}})\Biggr\} \right\}
 \ .
\end{align}
\end{lemma}
Note that the first term in the maximization corresponds to the bias of the bit $B$ and the second to the sum over all possible values of $y$, of the probability that given this specific value of $y$, $B$ is mapped to $0$ or $1$, whichever one of the two is smaller. 
\begin{proof}
By Definition~\ref{def:pz0}, 
\begin{align}
\nonumber 
 &\sum\limits_{{(x,y)}:f({x})=0}P^{\bar{z}_0}({{x,y},{u,v}})
-
\sum\limits_{{(x,y)}:f({x})=1} P^{\bar{z}_0}({{x,y},{u,v}})
\\
 \nonumber &= 
 \sum\limits_{y}\Biggl(
\sum\limits_{{x}:f({x})=0}
 P^{n,\ep}({{x,y},{u,v}})
 -\sum\limits_{{x}:f({x})=1}
 P^{n,\ep}({{x,y},{u,v}})
 \Biggr)\\
 \nonumber &\quad +2
 \sum\limits_{y}\left(\min \Biggl\{
\sum\limits_{{x}:f({x})=0}
P^{n,\ep}({{x,y},{u,v}}),
\sum\limits_{{x}:f({x})=1}
 P^{n,\ep}({{x,y},{u,v}})
 \Biggr\}
 \right)\ .
\end{align}
Assume w.l.o.g.\ that this quantity is positive, otherwise exchange the role of $\bar{z}_0$ and $\bar{z}_1$. 
We use $P^{\bar{z}_1}({x,y},{u,v})=2\cdot P^{n,\ep}({{x,y},{u,v}})-P^{\bar{z}_1}({x,y},{u,v})$ and 
Lemma~\ref{lemma:distancesinglebit} and distinguish two cases:\\
If $B$ given $\bar{z}_1$ is biased towards $1$, then
\begin{multline}
\nonumber
d(B|Z(\bar{w}),Q)\\
=\sum\limits_{{y}} \min\Biggl\{\sum\limits_{{x}:f({x})=0} P^{n,\ep}({{x,y},{u,v}}),\sum\limits_{{x}:f({x})=1} P^{n,\ep}({{x,y},{u,v}})\Biggr\}\ .
\end{multline}
If $B$ given $z_1$ is biased towards $0$, then
\begin{align}
\nonumber
d(B|Z(\bar{w}),Q)
&=\frac{1}{2}  \sum\limits_{{(x,y)}}(-1)^{f(x)} P^{n,\ep}({{x,y},{u,v}})\ .
\end{align}
Note that  $B$ given $\bar{z}_1$ is biased towards $1$ exactly if 
\begin{multline}
\nonumber
\sum\limits_{y}\left(
\sum\limits_{{x}:f({x})=1}
 P^{n,\ep}({{x,y},{u,v}})
 -\sum\limits_{{x}:f({x})=0}
 P^{n,\ep}({{x,y},{u,v}})\right. \\
+\left. 2  \min \Biggl\{
\sum\limits_{{x}:f({x})=0}
P^{n,\ep}({{x,y},{u,v}}),
\sum\limits_{{x}:f({x})=1}
 P^{n,\ep}({{x,y},{u,v}})
 \Biggr\}\right)
>0\ .
\end{multline}
This concludes the proof.
\end{proof}

\section{Impossibility of Privacy Amplification}
\label{sec:privamp}
\subsection{For linear functions}\label{subsec:linfct}

Using the non-signalling partition given in Section~\ref{subsec:wbar}, it is now \linebreak[4] straightforward to show that privacy amplification by applying a  linear function --- taking the XOR of some subset of the output bits --- is impossible. Moreover, we will show that the more bits we take the XOR of, the more Eve can know.
The non-signalling partition $\bar{w}$ is such that 
the distance from uniform of the key bit given $\bar{w}$ is always bigger than $\ep$, where $\ep$ is the error of the system. In the limit of large $n$ it is, however, even larger and Eve can almost perfectly know Alice's final bit. 

\begin{lemma}\label{lemma:imppalin}
For all linear functions $f\colon \{0,1\}^n\rightarrow \{0,1\}$, the distance from uniform of the bit $B=f({X})$ given the non-signalling partition $\bar{w}$ and $Q=(U=u,V=v,F=f)$ is larger than $\ep$, i.e.,
$d(f(X)|Z(\bar{w}),Q)
\geq \ep$.
\end{lemma}

\begin{figure}[h]
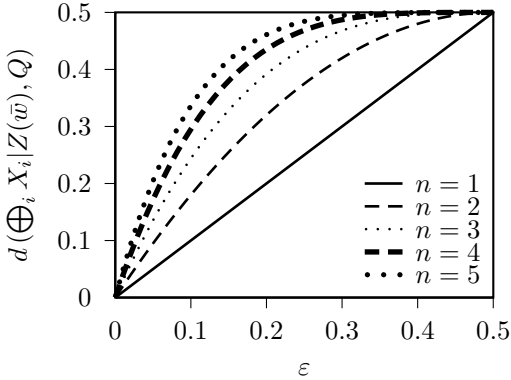

\centering
 \pspicture[](-2,-1)(7,4)
 \psset{xunit=10cm,yunit=7.5cm}
 \rput[c](0.25,-0.125){{$\ep$}}
  \rput[c]{90}(-0.125,0.25){{$d\left(\bigoplus_i X_i|Z(\bar{w}),Q\right)$}}
   \psaxes[Dx=0.1,Dy=0.1, showorigin=true,tickstyle=bottom,axesstyle=frame](0,0)(0.5001,0.5001)
\psplot[linewidth=1pt,linestyle=solid]{0}{0.5}{x} 
\psplot[linewidth=1pt,linestyle=dashed]{0}{0.5}{x 2 mul x x mul -2 mul add} 
\psplot[linewidth=1pt,linestyle=dotted]{0}{0.5}{x 3 mul x x mul -6 mul add x x mul x mul 4 mul add}      
\psplot[linewidth=2pt,linestyle=dashed]{0}{0.5}{x 4 mul x x mul -12 mul add x x mul x mul 16 mul add x x mul x x mul mul -8 mul add}  
\psplot[linewidth=2pt, linestyle=dotted]{0}{0.5}{x 5 mul x x mul -20 mul add x x mul x mul 40 mul add x x mul x x mul mul -40 mul add x x mul x x mul mul x mul 16 mul add}          
\psline[linewidth=1pt,linestyle=solid](0.33,0.2)(0.38,0.2) \uput[r](0.38,0.2){$n=1$}
\psline[linewidth=1pt,linestyle=dashed](0.33,0.16)(0.38,0.16) \uput[r](0.38,0.16){$n=2$}
\psline[linewidth=1pt,linestyle=dotted](0.33,0.12)(0.38,0.12) \uput[r](0.38,0.12){$n=3$}
\psline[linewidth=2pt,linestyle=dashed](0.33,0.08)(0.38,0.08) \uput[r](0.38,0.08){$n=4$}
\psline[linewidth=2pt, linestyle=dotted](0.33,0.04)(0.38,0.04) \uput[r](0.38,0.04){$n=5$}
 \endpspicture
 \caption{The lower bound on the distance from uniform of $\bigoplus_i X_i$ as given by (\ref{eq:biasxor}) as a function of the number of systems $n$ and the error $\ep$. Note that the non-trivial region of $\ep$ is below $1/4$.}
 \end{figure}

\begin{proof}
Any function from $n$ bits to~$1$ bit which is linear in the input bits can be expressed as 
$f(X)=\bigoplus_{i\in K}X_i$. Because all values of $X$ are output with the same probability, the probability that $B=0$ is the same as the probability that $B=1$. The first term in the maximization (Lemma~\ref{lemma:distancewbar}) is, therefore, $0$. 
To determine the distance from uniform of $B=f({X})$ given the non-signalling partition $\bar{w}$, we calculate the value of the second term. For each value of $y$ it holds that
\begin{align}
\nonumber \sum\limits_{{x}:\oplus_i x_i=0}P^{n,\ep}({{x,y},{u,v}})&=\sum_{i=0}^{\lfloor \frac{n}{2}\rfloor}
\binom{n}{n-2i}\left(\frac{1}{2}-\frac{\ep}{2}\right)^{n-2i}\left(\frac{\ep}{2}\right)^{2i}\\
\nonumber \sum\limits_{{x}:\oplus_i x_i =1}P^{n,\ep}({{x,y},{u,v}})&=
\sum_{i=0}^{\lfloor \frac{n-1}{2}\rfloor}
\binom{n}{n-2i-1}\left(\frac{1}{2}-\frac{\ep}{2}\right)^{n-2i-1}\left(\frac{\ep}{2}\right)^{2i+1}
\end{align}
 or with the value of the function flipped. The value of the second expression is always smaller than the value of the first one, because both values sum up to $1/2^n$ and the first one is larger than $(1-\ep)/2^n$, which is at least half of the sum for $\ep\leq 1/2$. Therefore,
\begin{align}
\nonumber  
d\Bigl(\bigoplus_i X_i \Bigm| Z(\bar{w}),Q\Bigr)
&= \sum_y
\sum_{i=0}^{\lfloor \frac{n-1}{2}\rfloor}
\binom{n}{n-2i-1}\left(\frac{1}{2}-\frac{\ep}{2}\right)^{n-2i-1}\left(\frac{\ep}{2}\right)^{2i+1}\\
\label{eq:biasxor}
&=\sum_{i=0}^{\lfloor \frac{n-1}{2}\rfloor}\binom{n}{n-2i-1}\left(1-\ep\right)^{n-2i-1}\ep^{2i+1}\ ,
\end{align}
which is larger than $\ep$ for all $n>1$. 
\end{proof}
This shows that there exists a constant lower bound on the knowledge Eve can  obtain about the key bit
 by using this strategy. Furthermore, in the limit of large $n$, the distance from uniform of the bit $f({X})=\bigoplus_i X_i$ tends toward $1/2$ and Eve can have almost perfect knowledge about \linebreak[4] Alice's output bit, no matter the original error of the system.

\subsection{For any  hashing}\label{subsec:anyhash}

Let us now turn to the case where $f$ can be any function and does not necessarily need to be linear. We will show that even then, privacy amplification is not possible. For the proof we will proceed in several steps: 
First, we will show that the distance from uniform of the bit $f(X)$ reached by the non-signalling partition $\bar{w}$ is independent of the input that Alice and Bob have given. This will allow us to consider the distance from uniform only for the case when the input has been the all-zero input, in which case we can express it in terms of the correlations of the output bit strings.  We will then use a result by Yang~\cite{Yang07} on (the impossibility of) non-interactive correlation distillation, limiting the correlation of bits which can be obtained from a sequence of weakly correlated bits. 

\begin{lemma}\label{lemma:barka}
The distance from uniform of the bit $f({X})$ given the non-sig\-nal\-ling partition $\bar{w}$ defined in Section~\ref{subsec:wbar} is independent of the values of ${u}$ and ${v}$, i.e., 
$d(f(X)|Z(\bar{w}),Q)=d(f(X)|Z(\bar{w}),Q^{\prime})$, where $Q=(U=u,V=v,F=f)$ and $Q^{\prime}=(F=f)$. 
\end{lemma}
\begin{proof}
The probability of the output ${x},{y}$, given input ${u},{v}$, is the same as the probability of output ${x},{y^{\prime}}$, given the all-zero input, i.e.,
\begin{align}
\nonumber P^{n,\ep}({{x,y},{u,v}})&= \left(\frac{1}{2}-\frac{\ep}{2}\right)^{\sum\limits_i 1\oplus x_i \oplus y_i\oplus u_i\cdot v_i}\cdot \left(\frac{\ep}{2}\right)^{\sum\limits_i x_i \oplus y_i\oplus u_i\cdot v_i}\\
\nonumber &= 
P^{n,\ep}({{x,y^{\prime}},0\dotso 0,0\dotso 0})\ ,
\end{align}
where we have 
defined $y^{\prime}_i=y_i\oplus u_i\cdot v_i$. 
Because the distance from uniform given the non-signalling partition $\bar{w}$ (Lemma~\ref{lemma:distancewbar}) is obtained by summing over all values of $y$ it is independent of the values $u$, $v$.
\end{proof}

Hence, we only have to find a lower bound on the distance from uniform of
$d(f(X)|Z(\bar{w}),Q^{\prime})$, where we can  assume that the input was the all-zero input. 
Note that the output probabilities given the all-zero input take a particularly simple form, more precisely, 
\begin{align}
\nonumber P^{n,\ep}_{XY|UV}({{x,y},0\dotso 0,0\dotso 0}) &= \left(\frac{1}{2}-\frac{\ep}{2}\right)^{n-d_{\mathrm{H}}({x},{y})}\cdot \left(\frac{\ep}{2}\right)^{d_{\mathrm{H}}({x},{y})}\ ,
\end{align}
where $d_{\mathrm{H}}({x},{y})$ denotes the Hamming distance between the bit strings $x$ and $y$, i.e., the number of positions where the strings differ. 

We will now show that the distance from uniform reached by the non-signalling partition $\bar{w}$ is related to the \emph{correlation} of two bits which can be obtained from the outputs. First, we need to introduce some definitions. 
\begin{definition}
The \emph{correlation $c_{XY}$} between two random bits $X$ and $Y$ is the probability for the two bits to be equal, minus the probability
for the two bits to be different, i.e., 
\begin{align}
\nonumber c_{XY}&=P(X=Y)-P(X\neq Y)\ .
\end{align}
\end{definition}
Two equal random bits  have correlation $1$ and are called \emph{completely correlated}, two random bits which are always different have correlation $-1$ and are called \emph{completely anti-correlated}.

Let us further consider the following scenario: Alice has a random  $n$-bit-string ${X}$ to which she applies a public function $f$ in order to obtain a single bit: $f\colon {X}\rightarrow \{0,1\}$. Bob has a random $n$-bit-string ${Y}$. Each bit of $Y$ is correlated with each bit of ${X}$ and Bob would like to calculate a bit $g(Y)$ that is highly correlated with $f({X})$. The best achievable correlation is $c_{g(Y)f({X})}^{\mathrm{opt}}=2\mathbb{E}_{{y}}[\max(P(f({X})=0|g(y)),P(f({X})=1|g(y)))]-1$, and it is reached by choosing $g(Y)$ to be $0$ (or $1$) if $f({X})$ is more likely to be $0$ ($1$) given the value of ${Y}$.
\begin{definition} 
Assume a random variable ${X}$, which is mapped to a bit $f({X})\in \{0,1\}$, and a random variable ${Y}$ 
with a joint distribution $P_{XY}$. 
The \emph{maximum-likelihood function $g$} of $f({X})$ given ${Y}$ is the function $g\colon {Y}\rightarrow \{0,1\}$ such that
\begin{align}
\nonumber g({y}) &=
\begin{cases}
0 & \text{if}\ \Pr[ f({X})=0|{Y}={y})\geq P(f({X})=1|{Y}={y}]\\
1 & \text{if}\ \Pr[f({X})=0|{Y}={y}) < P(f({X})=1|{Y}={y}]\ .
\end{cases}
\end{align}
\end{definition}

Using these definitions, we can show the key statement for the derivation of our result: The amount of information Eve can gain about the key bit is proportional to the error in correlation between Alice's and Bob's bits. 
\begin{lemma}\label{lemma:knowledgereduce}
The distance from uniform of $f(X)$ given the non-signalling partition $\bar{w}$ and $Q=(F=f)$ is at least $(1-c_{f({X})g({Y})})/2$, where $g$ is the maximum-likelihood function of $f({X})$ given ${Y}$, i.e.,
\begin{align}
\nonumber d(f(X)|Z(\bar{w}),Q) &\geq  \frac{1}{2}-\frac{1}{2}\cdot c_{f({X})g({Y})}\ .
\end{align}
\end{lemma}
\begin{proof}
\begin{flalign}
\nonumber 
d(f(X)|Z(\bar{w}),Q)
&\geq 
\frac{1}{2^n}\sum\limits_{{y}} \min\Biggl\{\sum\limits_{{x}:f({x})=0} (1-\ep)^{n-d_{\mathrm{H}}({x},{y})}\cdot \ep^{d_{\mathrm{H}}({x},{y})},
\\ &\quad  \nonumber
\sum\limits_{{x}:f({x})=1} (1-\ep)^{n-d_{\mathrm{H}}({x},{y})}\cdot \ep^{d_{\mathrm{H}}({x},{y})}\Biggr\}\\
\nonumber &= 1- \frac{1}{2^n}\sum\limits_{{y}} \max\Biggl\{\sum\limits_{{x}:f({x})=0} (1-\ep)^{n-d_{\mathrm{H}}({x},{y})}\cdot \ep^{d_{\mathrm{H}}({x},{y})},
 \\ &\quad \nonumber
\sum\limits_{{x}:f({x})=1} (1-\ep)^{n-d_{\mathrm{H}}({x},{y})}\cdot \ep^{d_{\mathrm{H}}({x},{y})}\Biggr\} \\
\nonumber &= 1-\expect_{{y}}[\max(P(f({X})=0|{Y}={y}),P(f({X})=1|{Y}={y}))]\ .
\end{flalign}
The last line is exactly equal to $1/2- c_{f({X})g({Y})}/2$, where $g$ is the maximum-likelihood function of $f({X})$ given ${Y}$. 
\end{proof}

This means that unless Bob is able to create a bit which is highly correlated with Alice's output bit, the adversary can always obtain significant information about the key bit. 
However, we will see now that the only way to obtain highly correlated bits is to apply a biased function. 

The following theorem, proven by Yang~\cite{Yang07}, shows the trade-off between randomness and correlation of two random bits. 
\begin{theorem}[Non-interactive correlation distillation~\cite{Yang07}] \label{th:yang}
Let $X$ and $Y$ be strings of 
$n$ uniformly random bits with correlation $1-2\ep$. Then the maximal correlation that can be reached
by locally applying 
 a function $f$ (and $g$, respectively) to 
the $n$ bits 
 is $1-2\ep(1-4\delta^2)$, where 
  $\delta:=\max(d(f({X})),d(g({Y})))$. 
\end{theorem}

Lemma~\ref{lemma:knowledgereduce} implies that if $\delta$ is small, then Eve's knowledge is high. 
We now need to see whether we can lower-bound Eve's knowledge for the case of large $\delta$. For $\delta$ to be large, either $d(f({X}))$ or $d(g({Y}))$ needs to be large. We first show that if $d(f({X}))$ is large, then so is Eve's knowledge about the bit $f({X})$.
\begin{lemma}\label{lemma:nottoounbalanced}
The distance from uniform of $f(X)$ given the non-signalling partition $\bar{w}$ and $Q=(F=f)$ is at least $d(f({X}))$, i.e.,
\begin{align}
\nonumber d(f(X)|Z(\bar{w}),Q) &\geq  d(f({X}))\ .
\end{align}
\end{lemma}
\begin{proof} 
\begin{align}
\nonumber d(f(X)|Z(\bar{w}),Q)&\geq  \frac{1}{2}\cdot \left|P({f({X})=0|{0\dotso 0}})-P({f({X})=1|{0\dotso 0}})\right|
\\
\nonumber 
&= d(f({X}))\ . \qedhere
\end{align}
\end{proof}

We have shown that Eve's knowledge about the key bit is high if either the output bits are not very correlated or one of the bits is biased. 
It remains to exclude the case that $\delta$ is large because $d(f({X}))$ is small and $d(g({Y}))$ is large. However, when the difference between these two values is large, the correlation between the two bits cannot be high. 
\begin{lemma}\label{lemma:wbarbetterthandeltadifference}
\begin{align}
\nonumber c_{f({X})g({Y})} &\geq 1-2 \left| d(g({Y}))-d(f({X}))\right| \ .
\end{align}
\end{lemma}
\begin{proof}
\begin{align}
\nonumber c_{f({X})g({Y})}&= 2\cdot \Pr \left[f({X})=g({Y})\right]-1\\
\nonumber &\leq 2\cdot \left(1-\left| d({g({Y})})-d({f({X})})\right| \right)-1\\
\nonumber &= 1-2\cdot \left| d({g({Y})})-d({f({X})})\right| \ . \qedhere
\end{align}
\end{proof}

By Lemma~\ref{lemma:knowledgereduce}, this implies directly that when the difference between the two distances from uniform is large, then the correlation is low and, therefore, the distance from uniform of the key bit is large. 
We can connect the distance from uniform 
 of the bit $f({X})$ with the value $\delta$.
\begin{lemma}\label{lemma:knowledgebiggerthanhalfdelta}
The distance from uniform of $f(X)$ given the non-signalling partition $\bar{w}$ and $Q=(F=f)$ is at least $\delta /2$, i.e.,
\begin{align}
\nonumber d(f(X)|Z(\bar{w}),Q) &\geq  \frac{1}{2}\cdot \delta\ ,
\end{align}
 where $\delta:=\max(d(f({X})),d(g({Y}))$ and $g$ is the maximum-likelihood function of $f({X})$ given ${Y}$. 
\end{lemma}
\begin{proof} 
Lemmas~\ref{lemma:knowledgereduce},~\ref{lemma:nottoounbalanced} and \ref{lemma:wbarbetterthandeltadifference} imply that 
\begin{align}
\nonumber d(f(X)|Z(\bar{w}),Q)&\geq  \max \left\{d(f({X})), |d(g({Y}))-d(f({X}))| \right\}\\
\nonumber &\geq  \frac{1}{2} \cdot \max \left\{d(f({X})),d(g({Y})) \right\}
\\ \nonumber
&\geq  \frac{1}{2} \cdot \delta \qedhere
\end{align}
\end{proof}

Now we can put Lemmas~\ref{lemma:knowledgereduce} to~\ref{lemma:knowledgebiggerthanhalfdelta} and Theorem~\ref{th:yang} together to obtain a general lower bound on the adversary's knowledge. 
\begin{theorem}\label{th:imppa}
The distance from uniform of $f(X)$ given the non-signalling partition $\bar{w}$ and $Q=(F=f)$, is at least  $({-1+\sqrt{1+64\ep^2}})/({32\ep})$, i.e.,
\begin{align}
\nonumber d(f(X)|Z(\bar{w}),Q) &\geq  \frac{-1+\sqrt{1+64\ep^2}}{32\ep}\ .
\end{align}
\end{theorem}

\begin{figure}[h]
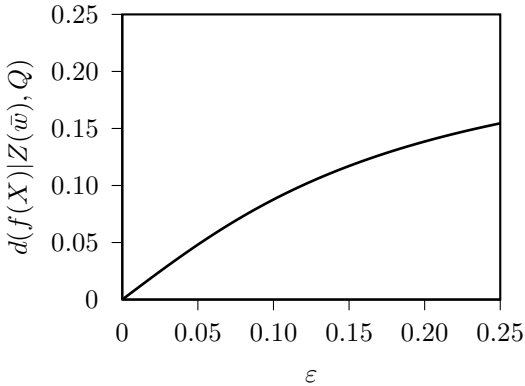

\centering
 \pspicture[](-2,-1)(7,4)
 \psset{xunit=20cm,yunit=15cm}
 \rput[c](0.125,-0.0675){{$\ep$}}
  \rput[c]{90}(-0.0675,0.125){{$d(f(X)|Z(\bar{w}),Q)$}}
   \psaxes[Dx=0.05,Dy=0.05, showorigin=true,tickstyle=bottom,axesstyle=frame](0,0)(0.25001,0.25001)  
\psplot[linewidth=1pt, linestyle=solid]{0.001}{0.25}{
x x mul 64 mul 1 add sqrt -1 add 32 x mul div
}    
\endpspicture
\caption{\label{fig:non-univ} The lower bound on the distance from uniform of the final bit as function of the error of the systems $\ep$.}
\end{figure}
\begin{proof}
By Theorem~\ref{th:yang}, it holds that $ 1/2- c_{f({X}),g({Y})}/2\geq\ep(1-4\delta^2)$. Together with 
Lemmas~\ref{lemma:knowledgereduce} and~\ref{lemma:knowledgebiggerthanhalfdelta}, this implies that 
\begin{align}
\nonumber d(f(X)|Z(\bar{w}),Q) &\geq  \max \left\{ \ep(1-4\delta^2), \frac{1}{2} \cdot \delta \right\}\geq  \frac{-1+\sqrt{1+64\ep^2}}{32\ep}\ . \qedhere
\end{align}
\end{proof}

Note that for small $\ep$, this lower bound actually takes a value close to $2\ep$; while for $\ep$ close to $1/4$, it is still larger than $\ep/2$. We obtain a constant lower bound (see Fig.~\ref{fig:non-univ}) depending only on the error $\ep$ of the individual systems and independent of the number of systems $n$. 
This implies that the distance from uniform can never become negligible in the number $n$, as it should be the case for privacy amplification. 

The above argument further implies that by applying a function to the inputs and outputs of any number of unbiased PR~boxes with error $\ep$, it is not possible to create an unbiased PR~box with error $\ep^{\prime}<\ep/4$. 

\section{Concluding Remarks}

We have shown that when a non-signalling condition holds only between Alice, Bob, and Eve, privacy amplification is, in general, not possible \linebreak[4] against non-sig\-nal\-ling adversaries. Some sort of additional non-sig\-nal\-ling condition is, therefore, necessary. 

We have also argued that the XOR is not a good privacy amplification even if non-signalling is restricted to one direction and in the quantum case. It remains an open question, whether a different function could be used in these cases. 

The question might arise, whether instead of using a fixed function $f$, it might be useful to choose a random function, i.e., a function chosen from a certain set of functions. For the impossibility result in this chapter, this would, however, not help. In fact, a non-signalling adversary always has \emph{all} the possibilities to attack a distribution of a certain marginal. In the above argument, it is, therefore, not important what set the function $f$ was chosen from, because the eavesdropper can delay the choice of her input until the function becomes public. In the quantum case, it is an open question, whether functions chosen at random from a certain set are strictly stronger than fixed functions.

\bibliographystyle{alpha}
\bibliography{main}

\end{document}